\documentclass[10pt, reqno]{amsart}
\usepackage{amsthm,amsfonts,amssymb,amsmath,euscript,comment}
\usepackage{graphics}
\usepackage{xcolor} 

\usepackage{graphicx}
\usepackage{color}
\usepackage{transparent}

\usepackage{fullpage} 

\definecolor{green}{rgb}{0,0.8,0} 

\newtheorem{theorem}{Theorem}[section]
\newtheorem{corollary}[theorem]{Corollary}
\newtheorem{lemma}[theorem]{Lemma}
\newtheorem{proposition}[theorem]{Proposition}
\theoremstyle{definition}
\newtheorem{definition}[theorem]{Definition}

\theoremstyle{remark}
\newtheorem{remark}[theorem]{Remark}

\numberwithin{equation}{section}
\numberwithin{figure}{section}

\newcommand{\abs}[1]{\vert#1\vert}

\newcommand{\set}[1]{\{#1\}}

\newcommand{\ep}{\epsilon}
\def\beaa{\begin{eqnarray*}}
\def\eeaa{\end{eqnarray*}}
\def\bea{\begin{eqnarray}}
\def\eea{\end{eqnarray}}
\def\be{\begin{equation}}
\def\ee{\end{equation}}
\def\f{\frac}

\newcommand{\rd}{\partial}
\newcommand{\frd}[1]{\frac{\partial}{\partial #1}}

\newcommand{\bb}{\Big}

\newcommand{\alp}{\alpha}
\newcommand{\bt}{\beta}
\newcommand{\gmm}{\gamma}
\newcommand{\Gmm}{\Gamma}

\newcommand{\eps}{\epsilon}

\newcommand{\tht}{\theta}

\newcommand{\om}{\omega}

\newcommand{\calB}{\mathcal B}

\newcommand{\calR}{\mathcal R}

\newcommand{\Lb}{\underline{L}}

\newcommand{\tBox}{\widetilde{\Box}}
\newcommand{\ls}{\lesssim}
\newcommand{\db}{\bar{\partial}}
\newcommand{\nab}{\nabla}

\newcommand{\de}{\delta}
\newcommand{\nabb}{{\nabla} \mkern-13mu /\,}

\begin{document}

\title[]{Global nonlinear stability of large dispersive solutions \\to the Einstein equations}
\author{Jonathan Luk}
\address{Department of Mathematics, Stanford University, Stanford, CA 94305, USA}
\email{jluk@stanford.edu}

\author{Sung-Jin Oh}%
\address{Department of Mathematics, UC Berkeley, Berkeley, CA 94720, USA and KIAS, Seoul, Korea 02455}%
\email{sjoh@math.berkeley.edu}%


\begin{abstract}
We extend the monumental result of Christodoulou--Klainerman on the global nonlinear stability of the Minkowski spacetime to the global nonlinear stability of a class of large dispersive spacetimes. More precisely, we show that any regular future causally geodesically complete, asymptotically flat solution to the Einstein--scalar field system which approaches the Minkowski spacetime sufficiently fast for large times is future globally nonlinearly stable. Combining our main theorem with results of Luk--Oh, Luk--Oh--Yang and Kilgore, we prove that a class of large data spherically symmetric dispersive solutions to the Einstein--scalar field system are globally nonlinearly stable with respect to small non-spherically symmetric perturbations. This in particular gives the first construction of an open set of large asymptotically flat initial data for which the solutions to the Einstein--scalar field system are future causally geodesically complete.
\end{abstract}
\maketitle

\tableofcontents
\section{Introduction}

The Minkowski spacetime $(\mathbb R^{3+1}, m)$ with metric
\begin{equation}\label{Minkowski.metric}
m=-dt^2+\sum_{i=1}^3 (dx^i)^2
\end{equation}
is a special solution to the Einstein equations
$$Ric_{\mu\nu}-\f 12 g_{\mu\nu}R=2 \mathbb T_{\mu\nu}$$
in vacuum, i.e.~when $\mathbb T_{\mu\nu}\equiv 0$. A monumental result in general relativity is the nonlinear stability of Minkowski spacetime, proved by Christodoulou--Klainerman in 1993:

\begin{theorem}[Christodoulou--Klainerman \cite{CK}]\label{CK.thm}
Minkowski spacetime is globally nonlinearly stable for the Einstein vacuum equations $Ric_{\mu\nu}=0$. 
\end{theorem}

In slightly more precise terms, given asymptotically flat initial data satisfying the constraint equations which are sufficiently close to the Minkowskian initial data, the maximal globally hyperbolic development to the Einstein vacuum equations is causally geodesically complete, remains close to the Minkowski spacetime and ``approaches the Minkowski spacetime at large times''.

A more restricted result, for which the initial data are posed on a hyperboloid\footnote{as oppose to an asymptotically flat Cauchy hypersurface.}, was previously proven by Friedrichs \cite{Fried}. Variations, simplifications and generalizations of the Christodoulou--Klainerman result have subsequently been achieved by various authors. We refer the readers to \cite{BZ, Graf, KN, LeMa, L2, LR1, LR2} for extensions and simplifications, and to \cite{BFJST, BZ, FJS, IP, LeMa, LT, Loizelet, Speck, Taylor, Wang.KG} for results with various matter models. See also the related works \cite{Chr.memory.effect, Chr.MG, Huneau}. We highlight in particular the work of Lindblad--Rodnianski \cite{LR1, LR2} as it is the most relevant to the present work. They gave an alternative proof of the stability of Minkowski spacetime in a wave coordinate system. Besides simplifying the original proof \cite{CK}, their argument extends to the Einstein--scalar field system for which in addition to the Lorentzian manifold $(\mathcal M,g)$, there is a real-valued scalar field $\phi:\mathcal M\to \mathbb R$, such that the following system of equations are satisfied:
\begin{equation}\label{Einstein.scalar.field}
\begin{cases}
Ric_{\mu\nu}-\f 12 g_{\mu\nu}R=2 \mathbb T_{\mu\nu},\\
\mathbb T_{\mu\nu}=\rd_\mu\phi\rd_\nu\phi-\f 12 g_{\mu\nu}(g^{-1})^{\alp\bt}\rd_\alp\phi\rd_\bt\phi,\\
\Box_g\phi:=\f{1}{\sqrt{-\det g}}\rd_\alp((g^{-1})^{\alp\bt}\sqrt{-\det g}\rd_\bt\phi)=0.
\end{cases}
\end{equation}
We summarize the Lindblad--Rodnianski theorem for the Einstein--scalar field system as follows:
\begin{theorem}[Lindblad--Rodnianski \cite{LR2}]\label{LR.thm}
Minkowski spacetime is globally nonlinearly stable for the Einstein--scalar field system \eqref{Einstein.scalar.field}.
\end{theorem}
The main mechanism for the stability of Minkowski spacetime is a combination of the dispersive nature of the equations and the special structure in the nonlinearity. In this paper, we generalize the Theorems~\ref{CK.thm} and \ref{LR.thm} to a larger class of asymptotically flat spacetime, {\bf where smallness is not required}. As is well-known, general large data solutions to the Einstein equations may be incomplete \cite{Chr.SS.BH, Chr, LOY2, Penrose}. However, we show that as long as we have a background ``dispersive'' solution which is geodesically complete and converging to the Minkowski spacetime sufficiently fast, then any sufficiently small perturbations will also be dispersed. In particular, in an appropriately chosen system of coordinates, the nonlinearity has a special structure such that any small perturbations of the initial data to the background solution lead to a geodesically complete spacetime which again converges to the Minkowski spacetime for large times. We summarize our theorem as follows:
\begin{theorem}[Main theorem, first version]\label{Main.Theorem.intro}
Any sufficiently regular, future causally geodesically complete and asymptotically flat\footnote{We in fact need only a slightly weaker assumptions on the asymptotics (toward spatial infinity) than the usual notion of strong asymptotic flatness (see \cite{CK}). We will defer this discussion to Definitions \ref{def.dispersivespt}, \ref{def.pert} and Remark \ref{SAF}.} solution to the Einstein--scalar field equations that approaches the Minkowski spacetime sufficiently fast is future globally nonlinearly stable.
\end{theorem}
We will make precise in what sense the solution is required to approach Minkowski spacetime in later sections (see in particular Section \ref{sec.assumptions}). See Theorems~\ref{main.thm.2} and \ref{main.thm} for more precise statements. Let us point out the obvious fact that Theorem \ref{Main.Theorem.intro} generalizes Theorems \ref{CK.thm} and \ref{LR.thm}.


Unlike Theorems \ref{CK.thm} and \ref{LR.thm}, our main result does not require the background solution to be small. The natural question is then whether there exist spacetimes satisfying the assumptions of Theorem~\ref{Main.Theorem.intro}. We will particularly consider examples which are \emph{spherically symmetric}. In a previous work \cite{LO1}, we showed that as long as a spherically symmetric ``asymptotically flat'' solution to the Einstein--scalar field system satisfies a weak qualitative ``BV local scattering condition'', then they in fact satisfy quantitative inverse polynomial decay estimates. In a subsequent joint work with Yang \cite{LOY}, we also demonstrated the existence of solutions which scatter locally in the BV norm with \emph{arbitrarily large} BV norms\footnote{When discussing ``large data'', we of course need to specify the topology that we consider. The spacetimes constructed in \cite{LOY} are large not only in the Lindblad--Rodnianski norm, but also with respect to the BV norm, which is scaling invariant for the Einstein--scalar field system in spherical symmetry. We mention also that these spacetimes can have arbitrarily large ADM masses. Notice however that the construction in \cite{LOY} requires the amplitudes of the scalar fields to be small (in $L^\infty$).}, so that the solution verifies the decay estimates in \cite{LO1}. Very recently, Kilgore \cite{Kilgore} proved that after constructing a suitable gauge, a subclass of the large BV solutions constructed in \cite{LOY} in fact also satisfies the estimates required in the assumptions of Theorem~\ref{Main.Theorem.intro}. We therefore obtain


\begin{corollary}[Combining Theorem~\ref{Main.Theorem.intro} with \cite{LO1, LOY, Kilgore}]\label{cor.open.set}
There exist spherically symmetric solutions to the Einstein--scalar field system with {\bf large} initial data which obey the assumptions of Theorem~\ref{Main.Theorem.intro} and are therefore future globally nonlinearly stable.

As a consequence, there exists an {\bf open set} of {\bf large} initial data for the Einstein--scalar field system such that the maximal Cauchy development is future causally geodesically complete.
\end{corollary}

While \cite{Kilgore} only verifies the assumptions of Theorem~\ref{Main.Theorem.intro} towards the future, one expects that a subclass of the spacetime in \cite{LOY} in fact obey the assumptions of Theorem~\ref{Main.Theorem.intro} both towards the future and the past. In particular, this would give an open set of solutions which are future and past causally geodesically complete. 

In addition to Corollary~\ref{cor.open.set}, we mention two other potential applications of Theorem~\ref{Main.Theorem.intro}.
\begin{itemize}
\item Klainerman--Nicolo \cite{KN} provided an alternative proof of the stability of Minkowski spacetimes for the Einstein vacuum equations restricted to the causal future of a domain near the asymptotically flat end. Their proof uses the double null foliation gauge, which allows all the estimates to be localized to the causal future of the asymptotically flat region and can therefore be applied to large initial data to show that there exists ``a non-empty piece of future null infinity''. Our main theorem could potentially be used to give a different proof of the result in \cite{KN} and also to generalize\footnote{While the Lindblad--Rodnianski theorem allows for a scalar field, it only applies when the mass is small. Moreover, since the outgoing null cones diverge logarithmically from the corresponding Minkowskian outgoing null cones, it requires extra work to localize the estimates to the exterior region. This latter problem is treated in the present paper by a different resolution of the problem of mass (see discussion in Section \ref{sec.mass}).} it to the Einstein--scalar field system
\item In \cite{Chr}, Christodoulou constructed spacetimes which are \emph{past} causally geodesically \emph{complete} while trapped surfaces form dynamically in the future so that the spacetimes are \emph{future} causally geodesically \emph{incomplete}. (See also the very recent \cite{LiMei} for a construction which moreover contains a \emph{black hole} region in the future.) One expects that after introducing an appropriate gauge, Theorem~\ref{Main.Theorem.intro} can in principle be used to show that these spacetimes are \emph{asymptotically stable} towards the \emph{past}.
\end{itemize}

Our proof is based on estimating the difference of the metric components and the scalar field with their background values in a generalized wave coordinate gauge. The estimates make use of the decay of the background solutions. As one would expect from \cite{LR1, LR2}, both the decay of the background solutions and the decay of the perturbed solutions are borderline. Thus we need to make use of the weak null condition as in \cite{LR1, LR2}. Unlike in \cite{LR1, LR2}, however, we are dealing with a large data stability problem, and we need to avoid that the energy grows as a \emph{large} power of $t$. To achieve this, we exploit the weak null condition in our energy estimates (as opposed to just in the $L^\infty$ estimates as in \cite{LR1, LR2}) and also need to localized our estimates in various regions of spacetime. In particular, all of these features require us to choose our generalized wave gauge condition to be adapted to the background solution and moreover also to take into account the perturbation of the ADM mass. In this gauge, we are able to have good control of the null hypersurfaces of the metric which is crucial for us to localize our estimates in various regions of the spacetime. This allows us to fully exploit the weak null structure in the Einstein--scalar field system for a large data stability problem. We will explain all these issues in Section~\ref{sec.outline}.

The remainder of the introduction is organized as follows. First we discuss some stability results in the spirit of Theorem~\ref{Main.Theorem.intro} for related quasilinear wave equations in \textbf{Section~\ref{sec.related}}. Afterwards, we will then give a more detailed outline of the ideas of the proof in \textbf{Section~\ref{sec.strategy}}. Finally, we end our introduction with an outline of the remainder of the paper in \textbf{Section~\ref{sec.outline}}.

%

\subsection{Global existence and stability of solutions for quasilinear wave equations}\label{sec.related}
The problem of the global nonlinear stability of the Minkowski spacetime can be viewed in the larger context of small data global existence for small data for nonlinear wave equations. It is well-known that in $3+1$ dimensions, the dispersion of the linear wave equation is barely insufficient to obtain small data global existence for wave equations with a quadratic nonlinearity. Counterexamples were first given by John \cite{J}. For a large class of quasilinear wave equations including the compressible Euler equation, it is known that arbitrarily small initial data may lead to the formation of shocks \cite{Alinhac.blowup.1, Alinhac.blowup.2, Chr.shock, CM, HKSW, Sideris, Speck.shocks}.

On the other hand, since the seminal work of Klainerman \cite{K2}, it is well-known that a small data global existence result can be obtained if the quadratic nonlinearity obeys the classical null condition. An independent alternative proof was also given by Christodoulou \cite{Chr.null}. We cannot do justice to the large literature on related problems, but will simply point out that similar results have been obtained on more general asymptotically flat manifolds \cite{WaYu}, domains in the exterior of obstacles \cite{MetSog05, MetSog07}, as well as in multiple-speed problems \cite{SiderisTu, Sogge} including that of elasticity \cite{Sideris00}. 

Motivated by the problem of small-data global existence for the Einstein equations in the wave coordinate gauge, Lindblad--Rodnianski introduced the notion of the \emph{weak null condition} \cite{LRWN}, which generalizes the classical null condition. A quasilinear system of wave equations is said to satisfy the weak null condition if the corresponding asymptotic system (see H\"ormander \cite{H1}) has global solutions. Under suitable additional structural assumptions of the equations, small data to quasilinear systems satisfying the weak null condition lead to global solutions. This is in particular the case for the Einstein vacuum equations or the Einstein--scalar field system in the wave coordinate gauge, which was crucially used in the Lindblad--Rodnianski proof of the stability of Minkowski spacetime. \cite{LR1,LR2}. In addition, small data global existence has been proven for a number of other quasilinear systems satisfying the weak null condition; see for instance \cite{DengPusateri, HY18, Keir1, Keir2,  LWeakNull92, LWeakNull}. 

In the context of nonlinear wave equations, our main theorem (Theorem~\ref{Main.Theorem.intro}) can be viewed as a stability theorem for large solutions to nonlinear wave equations with sufficiently fast decay. Such results were first obtained by Alinhac \cite{A} for quasilinear wave equations satisfying a classical null condition. This was extended in the work of Yang \cite{Y} to equations with time dependent non-decaying coefficients satisfying a classical null condition. The works \cite{A,Y} use that under the classical null condition, there is effectively additional decay for the solutions. In contrast, in our present work, only a weak null condition holds; see Section~\ref{sec.strategy}.

As mentioned earlier, our result can be specialized to study the solutions in a neighborhood of a class of spherically symmetric solutions to the Einstein--scalar field system (see Corollary~\ref{cor.open.set} above). This result has parallels with global existence theorems for nonlinear wave equations in a neighborhood of symmetric solutions. For instance, Krieger showed that two-dimensional spherically symmetric wave maps\footnote{This problem has since then been completely resolved even without the almost-spherically-symmetric assumption \cite{KrSc, StTa, Tao}.} $:\mathbb R^{2+1}\to \mathbb H^2$ are stable \cite{Kr.WP} using the precise asymptotics of the exactly spherically symmetric solutions proven by Christodoulou--Tahvildar-Zadeh \cite{ChrTaZa}. We refer the readers also to the work of Andr\'easson--Ringstr\"om \cite{AndRing} for the Einstein--Vlasov system in the cosmological setting in which the authors studied the stability of a class of global $\mathbb T^2$-symmetric spacetimes.

\subsection{Strategy of the proof}\label{sec.strategy}

Our proof, following the main strategy in \cite{LR1, LR2}, is based on controlling the metric and the scalar field (and their derivatives) in an appropriately chosen generalized wave coordinate system. In such a coordinate system, the metric and the scalar field satisfy a quasilinear system of wave equations with a weak null condition and our goal is to control the difference of the metric and the scalar field with their background values using this system of wave equations. We will call the background solution $(g_B,\phi_B)$. Define $h$ and $\beta$ respectively to be appropriate\footnote{As we will soon discuss, $h$ will not actually be the difference between the unknown metric and $g_B$, but will be defined in a way that takes into account the contribution of the mass difference.} differences of the metric and the scalar field with their background values. In contrast to the small data problem (i.e.~the stability of Minkowski spacetime in Theorems \ref{CK.thm} and \ref{LR.thm}), the equations in our problem contain linear terms in the $(h,\beta)$ with coefficients that are large and are associated to the background solution $g_B$ and $\phi_B$. In the following, we will in particular explain how these additional terms can be handled.

This section is organized as follows: We begin in \textbf{Section~\ref{sec.proof.null.condition}} recalling the decay properties of solutions to the wave equation and the (weak) null condition. In \textbf{Section~\ref{sec.proof.decay}}, we discuss the decay condition that can be expected for the background solution (which for example holds for the spherically symmetric dispersive solutions of \cite{Kilgore}). We also explain the relevance of the decay properties of the background to our problem at hand. In \textbf{Section~\ref{sec.weak.null}}, we then study a model semilinear problem, which captures some of the analytic difficulties for the Einstein--scalar field system, and explain how a large data stability problem can be treated for that model. In \textbf{Section~\ref{sec.weak.null.2}}, we then discuss the similarities of the Einstein--scalar field system and the said model problem, but we also emphasize the additional difficulties that arise in the case of the Einstein--scalar field system. In \textbf{Section~\ref{sec.mass}}, we introduce the main new ideas of the paper and discuss how by choosing an appropriate generalized wave coordinate gauge, we can on the one hand treat the difficulties associated to the difference of the ADM masses and on the other hand introduce a localization to fully exploit the weak null structure present in the Einstein--scalar field system. The remaining subsections discuss more technical aspects of the proof. Namely, in \textbf{Section~\ref{sec.localization}} we explain how to perform the localization to different regions of spacetimes; in \textbf{Section~\ref{sec.proof.quasilinear}}, we discuss the treatment of the quasilinear error terms; in \textbf{Section~\ref{sec.hierarchy}}, we conclude by discussing the higher order error terms and the hierarchy of estimates that are introduced to tackle them.

\subsubsection{The classical null condition and the weak null condition}\label{sec.proof.null.condition}
We quickly recall the classical null condition and the weak null condition for quasilinear wave equations in $(3+1)$ dimensions. The key issue is that solutions to the linear wave equation only have uniform $O(\f 1{1+t})$ decay, which is barely non-integrable, and in general leads to finite-time blowup for small data solutions. 

On the other hand, as is by now very well-understood, in addition to the uniform $O(\f 1{1+t})$ decay, we have
\begin{itemize}
\item additional decay in the variable $|t-|x||$, i.e.~the sharp $\f{1}{1+t}$ decay is only saturated when $t\sim |x|$; and 
\item derivatives which are tangential to the outgoing light cone --- which we denote by $\db$ --- decay faster.
\end{itemize} 
The classical null condition requires that in quadratic terms in the nonlinearity, at least of of the derivative is a good $\db$. Thus this gives better decay so that small data always lead to global solutions. This structure also allows one to prove a large data stability result, as long as the background large solution obeys ``wave-like'' estimates.

The Einstein equation in wave coordinates, however, violate the classical null condition. Nonetheless, as shown in the work of Lindblad--Rodnianski \cite{LR1, LR2}, they satisfy the weak null condition. The simplest model problem to capture the structure of the semilinear terms is the system
\begin{equation}\label{model.system}
\begin{cases}
\Box \phi=0,\\
\Box \psi=(\rd_t\phi)^2.
\end{cases}
\end{equation}
It is clear\footnote{In fact, for such an overly-simplified system, \underline{all} regular data lead to global solutions!} that both global existence for small data and also global stability for large data solutions hold for \eqref{model.system}. While this system does not satisfy the classical null condition, there is a reductive structure, i.e.~one can first solve the first equation in \eqref{model.system} and then solve the second equation in \eqref{model.system}. It should be noted that even in the small data case, $\psi$ would \underline{not} have the decay as for solutions to the linear wave equation, but has a logarithmic correction. 

Similar ideas of using a reductive structure allow one to consider the following less simplistic model problem:
\begin{equation}\label{model.system.2}
\begin{cases}
\Box \phi=Q(\phi,\psi)=\rd\phi\db\psi+\rd\psi\db\phi,\\
\Box \psi=(\rd_t\phi)^2.
\end{cases}
\end{equation}
where $Q$ is a classical null form such that that there is at least one good derivative. For the system \eqref{model.system.2}, small data global existence holds (and follows ). For this system of equation, a reductive structure can still be exploited to obtain global stability of large data solutions, as long as the background solution is assumed to obey ``wave-like'' estimates, i.e.~it has $O(\f{1}{1+t})$ decay with improved decay in $||x|-t|$ and for the $\db$ derivatives.

We will sketch a proof of stability for large solutions to \eqref{model.system.2} in Section~\ref{sec.weak.null}, after discussing in Section~\ref{sec.proof.decay} the decay that we can expect for the background solutions. We note already that while part of our paper is to generalize the proof for \eqref{model.system.2} to the Einstein--scalar field system, a perhaps more important part is to understand why \eqref{model.system.2} is a reasonable toy model for the problem at hand. We will postpone the latter discussion to Sections~\ref{sec.weak.null.2} and \ref{sec.mass}.

\subsubsection{Decay conditions for the background solution}\label{sec.proof.decay}

Since the main difference between our problem and the stability of Minkowski spacetime is the extra terms associated to the background $g_B$ and $\phi_B$, it is important to understand their decay properties. Indeed, if these terms decay like\footnote{Of course, we also need estimates for $g_B$, $\phi_B$ themselves and for their higher derivatives. We suppress this discussion at the moment.} \footnote{For convenience, we will also assume that $t\geq 0$ below.} $|\rd g_B|+|\rd\phi_B|\ls \f{1}{(1+t)^{1+}}$, then because they are integrable in time, they can be controlled by a Gr\"onwall type argument. The remaining (small) nonlinear terms can then be treated as in the proof of the stability of Minkowski spacetime.

However, since $g_B$ and $\phi_B$ themselves are solutions to the the Einstein--scalar field system, we can at best expect ``wave-like'' estimates. In particular, the uniform-in-time decay estimate is no better than $O(\f{1}{1+t})$. Here are the decay estimates that are reasonable for the background solution.  
\begin{itemize}
\item The derivatives of $g_B$ and $\phi_B$ obey the following uniform-in-time decay for some (small) $\gamma>0$:
$$|\rd g_B|(t,x) + |\rd \phi_B|(t,x) \ls \f{1}{(1+t)(1+|t - |x||)^{\gamma}}.$$
This captures both the uniform-in-time $O(\f 1{1+t})$ decay, and the improvement away from the light cone $\{(t,x): t = |x|\}$ typical of solutions to the wave equation.
\item Just as for the solutions to wave equation, the ``good derivatives'' $\db$ --- those that are tangential to the light cone --- of $g_B$ and $\phi_B$ decay better. For some (small) $\gamma>0$, we have
$$| \db g_B| + |\db \phi_B|\ls \f{1}{(1+t)^{1+\gamma}}$$
\item Higher order versions of the above estimates still holds after differentiating with the Minkowskian commuting vector fields\footnote{See Definition~\ref{def.Mink.vf}.}.
\item So that we can localize our solutions (see Section~\ref{sec.mass}), we need to choose a gauge for the background solution such that some of the components of $h_B:=g_B-m$ decay faster than $\f 1{1+t}$ and in fact has a decay that is integrable in time. More precisely, let $L=\rd_t+\rd_r$ and $\mathcal T$ be a spanning set of vectors tangent to the Minkowskian outgoing light cone, we choose $h_B$ such that for some (small) $\gamma>0$,
$$|h_B|_{L\mathcal T}(t,x):=\sum_{V\in \mathcal T} |L^\alp V^{\bt} (h_B)_{\alp\bt}|(t,x)\ls \f{1}{(1+t)^{1+\gamma}}.$$
However, in this gauge, general components of the background metric decay slower, namely
$$|h_B|(t,x) \ls \f{\log(2+t+|x|)}{1+t+|x|}$$
\end{itemize}

For the precise assumptions, see Definition~\ref{def.dispersivespt}. By the results of \cite{Kilgore, LO1, LOY}, all these decay estimates are indeed satisfy by the class of spherically symmetric BV-scattering solutions to the Einstein--scalar field system considered in Corollary~\ref{cor.open.set}.

\subsubsection{Toy model problem \eqref{model.system.2}}\label{sec.weak.null}

We now sketch a proof of the stability of large data dispersive solutions for the toy model \eqref{model.system.2} introduced earlier. Consider a background global solution to \eqref{model.system.2} with the property that for $|I|\leq 10$, there exists some small $\gamma>0$ so that
\begin{equation}\label{model.assumptions}
\begin{split}
|\rd\Gamma^I\phi_B|\ls \f{1}{(1+t)(1+||x|-t|)^\gamma},&\quad |\rd\Gamma^I\psi_B|\ls \f{1}{(1+t)(1+||x|-t|)^\gamma},\\
|\db\Gamma^I\phi_B|+|\db\Gamma^I\psi_B|&\ls \f{1}{(1+t)^{1+\gamma}}.
\end{split}
\end{equation}
Here, $\Gamma$ are the Minkowskian commuting vector fields, which generate symmetries of the d'Alembertian on the Minkowski spacetime (see Definition \ref{def.Mink.vf}). These assumptions are exactly consistent with those in Section~\ref{sec.proof.decay}.

To prove the stability of such solutions, one combines the following three types of estimates: the weighted energy estimates, the Klainerman--Sobolev inequality and the $L^\infty-L^\infty$ ODE estimate of \cite{L1}. First, we have the energy estimates for solutions to $\Box\xi=F$ whenever $t_2\geq t_1\geq 0$:
\begin{equation}\label{EE.intro}
\begin{split}
E(t_2;t_1;\xi):=&\sup_{t'\in[t_1,t_2]}\int_{\{t'\}\times \mathbb R^3} w(|x|-t')|\rd\xi|^2(t',x)\,dx+\int_{t_1}^{t_2} \int_{\{t'\}\times \mathbb R^3} w'(|x|-t')|\db\xi|^2(t',x)\,dx\, dt'\\
\ls &\int_{\{t_1\}\times \mathbb R^3} w(|x|-t_1)|\rd\xi|^2(t_1,x)\,dx+(\int_{t_1}^{t_2} (\int_{\{t'\}\times \mathbb R^3} w(|x|-t')|F(t',x)|^2\,dx)^{\f 12} dt')^2.
\end{split}
\end{equation}
where $w(q):=\begin{cases}
1+(1+|q|)^{1+2\gamma} &\mbox{if }q\geq 0\\
1+(1+|q|)^{-\f{\gamma}{2}} &\mbox{if }q< 0.
\end{cases}$. The weight $w$ in the energy estimates, already introduced in \cite{LR1,LR2}, serves the double purpose of giving decay in $||x|-t|$ when $|x|\geq t$ and also giving a \emph{positive} bulk term on the left hand side which gives better control of the good derivatives terms $|\db\xi|$.

The energy estimate is applied to derivatives of $\xi$ with respect to $\Gamma$, which has the property that $[\Box,\Gamma]=c\Box$ for some constant $c$. Moreover, the energy of the $\Gamma$-differentiated quantities gives rise to the decay estimates due to the following Klainerman--Sobolev inequalities (see Propositions \ref{decay.weights}, \ref{KS.ineq} and Lemma \ref{int.lemma}), which hold for all sufficiently regular functions $\xi$:
\begin{equation}\label{KS.intro}
\sup_{x}|\rd\xi(t,x)|(1+t+|x|)(1+||x|-t|)^{\frac 12} w(|x|-t)^{\frac 12}\leq C\sum_{|I|\leq 3}\|w^{\frac 12}(|\cdot|-t)\rd\Gamma^I\xi(t,\cdot) \|_{L^2(\mathbb R^3)} ,
\end{equation}
\begin{equation}\label{KS.intro.2}
\begin{split}
&\sup_{x}\left(|\db\xi(t,x)|(1+t+|x|)^2+|\xi(t,x)|(1+t+|x|)\right)(1+||x|-t|)^{-\frac 12} w(|x|-t)^{\frac 12}\\
\leq &C\sum_{|I|\leq 4}\sup_{0\leq \tau\leq t}\|w^{\frac 12}(|\cdot|-\tau)\rd\Gamma^I\xi(\tau,\cdot) \|_{L^2(\mathbb R^3)} .
\end{split}
\end{equation}

The third ingredient that we need is the following $L^\infty-L^\infty$ estimate (see Proposition \ref{decay.est}), which holds for solutions to $\Box\xi=F$:
\begin{equation}\label{Li.Li.intro}
\begin{split}
\sup_x(1+t)|\rd\xi(t,x)|&
\ls \sup_{0\leq \tau\leq t}\sum_{|I|\leq 1}\|\Gamma^I\xi(\tau,\cdot)\|_{L^\infty(\mathbb R^3)}\\
&+\int_0^t \left((1+\tau)\|F(\tau,\cdot)\|_{L^\infty(\mathbb R^3)}+\sum_{|I|\leq 2}(1+\tau)^{-1}\|\Gamma^I\xi(\tau,\cdot)\|_{L^\infty(\mathbb R^3)}\right)d\tau.
\end{split}
\end{equation}
After introducing the basic tools, let us return to the problem of stability of large solutions to \eqref{model.system.2}. Defining $\bar\phi=\phi-\phi_B$ and $\bar\psi=\psi-\psi_B$, using the bounds \eqref{model.assumptions} for the background solution and only writing a few typical terms to simplify the exposition\footnote{In particular, we drop all the terms that are lower order in the derivatives.}, we have roughly
\begin{equation}\label{model.system.2.2}
\begin{cases}
|\Box \Gamma^I\bar\phi|&\ls \f{|\rd\Gamma^I\bar\phi|+|\rd\Gamma^I\bar\psi|}{(1+t)^{1+\gamma}}+\f{|\db\Gamma^I\bar\phi|+|\db\Gamma^I\bar\psi|}{(1+t)(1+||x|-t|)^\gamma}+|\rd\Gamma^I\bar\phi||\db\bar\psi|+|\rd\bar\phi||\db\Gamma^I\bar\psi|+\cdots,\\
|\Box \bar\psi|&\ls \f{|\rd\Gamma^I\bar\phi|}{1+t}+|\rd\Gamma^I\bar\phi||\rd\bar\phi|+\cdots.
\end{cases}
\end{equation}
Assume now that the initial perturbations are small, i.e.~$\sum_{|I|\leq 10} (E(0;0;\Gamma^I\bar\psi)+E(0;0;\Gamma^I\bar\phi))\leq \ep^2$. We first note that by a standard Cauchy stability argument, for every $T>0$, (after choosing $\ep$ smaller,) there exists $C_T>0$ such that $\sum_{|I|\leq 10} (E(T;0;\Gamma^I\bar\psi)+E(T;0;\Gamma^I\bar\phi))\leq C_T\ep^2$. We then make the bootstrap assumptions $\sum_{|I|\leq 10} (E(t;0;\Gamma^I\bar\psi)+E(t;0;\Gamma^I\bar\phi))\leq \ep (1+t)^{\de}$ for $\de\ll \gamma$.  The bootstrap assumption implies some pointwise bounds using \eqref{KS.intro} and \eqref{KS.intro.2} so that we can bound the first equation in \eqref{model.system.2.2} by
\begin{equation}\label{model.eqn.phi}
|\Box \Gamma^I\bar\phi|\ls \f{|\rd\Gamma^I\bar\phi|+|\rd\Gamma^I\bar\psi|}{(1+t)^{1+\gamma}}+\f{|\db\Gamma^I\bar\phi|+|\db\Gamma^I\bar\psi|}{(1+t)^{1-\de}(1+||x|-t|)^\gamma}+\cdots.
\end{equation}
We now apply the energy estimates \eqref{EE.intro} to \eqref{model.eqn.phi} with $t\geq T>0$ and $T$ sufficiently large to be chosen.  Noticing that the term on $\{t=T\}$ can be controlled by the Cauchy stability argument described above, we then get
\begin{equation}\label{EE.intro.1}
\begin{split}
&\:\sum_{|I|\leq 10}\sup_{t'\in[T,t]} E(t';T;\Gamma^I\bar\phi)\\
\ls &\:C_T\ep^2+\sum_{|I|\leq 10}(\int_T^t \f{(\int_{\{t'\}\times \mathbb R^3} w(|x|-t')(|\rd\Gamma^I\bar\phi|^2(t',x)+|\rd\Gamma^I\bar\psi|^2(t',x)\,dx)^{\f 12}}{(1+t')^{1+\gamma}} dt')^2 \\
&+\sum_{|I|\leq 10}(\int_T^t (\int_{\{t'\}\times \mathbb R^3} \f{w(|x|-t')(|\db\Gamma^I\bar\phi|^2(t',x)+|\db\Gamma^I\bar\psi|^2(t',x))}{(1+t')^{2-2\de}(1+||x|-t'|)^{2\gamma}}\,dx)^{\f 12} dt')^2\\
\ls &\:C_T\ep^2+T^{-\gamma}\sum_{|I|\leq 10}\sup_{t'\in[T,t]} (E(t';T;\Gamma^I\bar\phi)+E(t';T;\Gamma^I\bar\psi)) \\
&+(\int_T^t \int_{\{t'\}\times \mathbb R^3} \f{w(|x|-t')(|\db\Gamma^I\bar\phi|^2(t',x)+|\db\Gamma^I\bar\psi|^2(t',x))}{(1+t')^{1-\gamma-2\de}(1+||x|-t'|)^{2\gamma}}\,dx dt')(\int_T^t \f{dt'}{(1+t')^{1+\gamma}})\\
\ls &\:C_T\ep^2+T^{-\gamma}\sum_{|I|\leq 10}\sup_{t'\in[T,t]} (E(t';T;\Gamma^I\bar\phi)+E(t';T;\Gamma^I\bar\psi))\ls C_T\ep^2+T^{-\gamma}\sum_{|I|\leq 10}\sup_{t'\in[T,t]}E(t';T;\Gamma^I\bar\psi).
\end{split}
\end{equation}
where we have used $\f{w(|x|-t')}{(1+t')^{1-\gamma-2\de}(1+||x|-t'|)^{2\gamma}}\ls w'(|x|-t')$ and the very last estimate is achieved by choosing $T$ sufficiently large and absorbing the term $T^{-\gamma}\sum_{|I|\leq 10}\sup_{t'\in[T,t]} E(t;T;\Gamma^I\bar\phi)$ to the left hand side. On the other hand, applying \eqref{Li.Li.intro} to \eqref{model.eqn.phi} (for $|I|=0$) and using the bootstrap assumption together with \eqref{KS.intro} and \eqref{KS.intro.2} to control the terms on the right hand side, we get 
$$\sup_x|\rd\bar\phi|(t,x)\ls \f{\ep}{1+t}.$$
Plugging this into the second equation in \eqref{model.system.2.2} and applying the energy estimate \eqref{EE.intro} (again using the estimates from Cauchy stability on the constant $T$-hypersurface), we get
\begin{equation}\label{EE.intro.2}
\begin{split}
\sum_{|I|\leq 10}\sup_{t'\in [T,t]} E(t';T;\Gamma^I\bar\psi)
\ls &\: C_T\ep^2+\sum_{|I|\leq 10}(\int_T^t \f{(\int_{\{t'\}\times \mathbb R^3} w(|x|-t')|\rd\Gamma^I\bar\phi|^2(t',x)\,dx)^{\f 12}}{(1+t')} dt')^2\\
\ls &\: C_T\ep^2\log^2(2+t)+T^{-\gamma}\sum_{|I|\leq 10}(\int_T^t \f{E(t';T;\Gamma^I\bar\psi)^{\f12}}{(1+t')} dt')^2,
\end{split}
\end{equation}
where in the last line we have plugged in \eqref{EE.intro.1} and the $C_T\ep^2\log^2(2+t)$ term arises from $C_T\ep^2(\int_T^t \f{dt'}{1+t'})^2$. Taking square root of \eqref{EE.intro.2}, choosing $T$ sufficiently large (depending on $\de$) and using Gr\"onwall's inequality, we thus get
$$\sum_{|I|\leq 10}\sup_{t'\in [T,t]} E(t';T;\Gamma^I\bar\psi)\ls C_T\ep^2 (1+t)^{\f{\de}{2}}.$$
Plugging this back into \eqref{EE.intro.1}, we then obtain 
$$\sum_{|I|\leq 10}\sup_{t'\in [T,t]} E(t';T;\Gamma^I\bar\phi)\ls C_T\ep^2 (1+t)^{\f{\de}{2}}.$$ 
Now fix $T>0$ so that the above argument goes through. We can then choose $\ep>0$ to be sufficiently small and improve the bootstrap assumption.

In the above argument, we see that while all the estimates are coupled and have to be treated via a bootstrap argument, one can exploit the reductive structure in the sense that by first treating the estimates for $\bar\phi$, we can obtain the crucial smallness factor $T^{-\gamma}$ (see \eqref{EE.intro.1}). Moreover, we can close the argument allowing the energy to grow with a slow rate\footnote{Indeed, it can be proven a posteriori that the solution does not have the decay as in the linear wave equation case.}.

\subsubsection{Weak null condition for the Einstein--scalar field system}\label{sec.weak.null.2}

As shown in \cite{LR1, LR2}, the Einstein--scalar field system has a weak null structure similar to that in the model problem \eqref{model.system}. This thus gives hope to generalizing the small data results of \cite{LR1, LR2} to the stability of large data solutions. However, there is an additional difficulty that the weak null structure of the semilinear terms in \cite{LR1, LR2} is not manifest in the wave equations for the components of the metric in Cartesian coordinates. To reveal the weak null structure, on the one hand one needs to use the wave coordinate condition and on the other hand one also needs to project the equation to vector fields\footnote{see Definition \ref{def.proj.vf} for definition of these vector fields.} ${\bf E^\mu}\in \{L,\Lb,E^1,E^2,E^3\}$ adapted to the Minkowskian light cone.

To explain more precisely the structure of the semilinear terms, let us first consider the setting of \cite{LR1, LR2} in which the wave coordinate condition holds. We first note that the terms in the equation for $\tBox_g g_{\mu\nu}$ take the form $(g^{-1})(g^{-1})(\rd g)(\rd g)$ or $(\rd\phi)(\rd\phi)$. The most difficult terms in \cite{LR1, LR2} are those which are quadratic in the derivatives, i.e.~either the $(\rd\phi)^2$ terms or the metric terms with $g^{-1}$ replaced by $m$. This is because the remaining terms are at least cubic and are easier to control. In our setting, since we have a large background solution, $g^{-1}-m$ is only linear in the perturbation. However, for the linear terms, we can exploit the decay of $(\rd g_B)^2$ of the background solution and these terms are also easier to treat. We will therefore restrict our attention in this subsection only to the quadratic semilinear terms in the derivatives.

For these quadratic semilinear terms in the equation for $\tBox_g g_{\mu\nu}$, it was shown in \cite{LR1, LR2} that while some of the terms obey the classical null condition, the following terms violate it:
\begin{equation}\label{bad.term.intro}
\frac 14 m^{\alp\alp'}\rd_\mu g_{\alp \alp'} m^{\bt \bt'}\rd_\nu g_{\bt \bt'}-\frac 12 m^{\alp\alp'}\rd_\mu g_{\alp \bt} m^{\bt \bt'}\rd_\nu g_{\alp' \bt'}-4\rd_\mu\phi\rd_\nu\phi.
\end{equation}
Notice that if \eqref{bad.term.intro} is contracted with ${\bf E}^\mu, {\bf E}^\nu \in \{L,\Lb,E^1,E^2,E^3\}$ with ${\bf E}^\mu\neq\Lb$ and ${\bf E}^\nu\neq\Lb$, then we have at least one good $\db$ derivative and the quadratic term behaves essentially as a term obeying the null condition. Therefore, the ``bad'' terms only appear in the equation $\Lb^\mu\Lb^\nu\tBox_g g_{\mu\nu}$. Using the properties of the vector fields $\{L,\Lb,E^1,E^2,E^3\}$ in Minkowski spacetime, we thus have
\begin{equation}\label{bad.equation}
|\Lb^\mu\Lb^\nu\tBox_g g_{\mu\nu}|\ls |\rd h|_{\mathcal T\mathcal U}|\rd h|_{\mathcal T\mathcal U}+|\rd h|_{LL}|\rd h|+|\rd\phi||\rd\phi|+\dots,
\end{equation}
where we have defined the notation for projection to ${\bf E}^\mu$ by 
$$|\rd p|^2_{\mathcal V\mathcal W}:=\sum_{U\in\mathcal U,V\in\mathcal V,W\in\mathcal W}|(\rd_\gamma p_{\alp\bt})V^{\alp}W^{\bt}U^\gamma|^2,$$ with $\mathcal T:=\{L,E^1,E^2,E^3\}$, $\mathcal U:=\{L,\Lb,E^1,E^2,E^3\}$, $L=\{L\}$ and $\mathcal V$, $\mathcal W$ can be any of these sets. To proceed, it was observed in \cite{LR1, LR2} that by using the wave coordinate condition $\tBox_g x^{\mu}=0$, one can rewrite the derivatives of some components of the metric as the good $\db$ derivative of some other components of the metric. Namely, 
\begin{equation}\label{wave.coord.intro}
|\rd h|_{L\mathcal T} \ls |\db h|+\mbox{quadratic terms}.
\end{equation}
In particular, since $L\in \mathcal T$, this gives good control of $|\rd h|_{LL}$ and using \eqref{bad.equation} together with the above observations, we have the system
\begin{equation*}
\begin{cases}
|\Lb^\mu\Lb^\nu\tBox_g g_{\mu\nu}|\ls |\rd h|_{\mathcal T\mathcal U}|\rd h|_{\mathcal T\mathcal U}+|\rd\phi||\rd\phi|+\mbox{good terms},\\
\sum_{V\in \mathcal T,\,W\in \mathcal U}|V^\mu W^\nu\tBox_g g_{\mu\nu}|=\mbox{good terms},\\
\Box_g\phi=0,
\end{cases}
\end{equation*}
which almost obeys a reductive structure analogous to \eqref{model.system.2}, except for the need to commute $\tBox_g$ and the projection to $V$ and $W$ in the second equation.

\subsubsection{Localization to the wave zone, projection to vector fields adapted to Minkowskian null cone and a generalized wave coordinate condition}\label{sec.mass}

One of the difficulties in exploiting the reductive structure for the semilinear terms is that the projection to ${\bf E}^\mu$ does not commute with\footnote{$\tBox_g$ is the reduced wave operator $\tBox_g:=(g^{-1})^{\alp\bt}\rd_{\alp\bt}^2$, which is the principal part of the equations for the metric, see Proposition \ref{Einstein.eqn.g}.} $\tBox_g$. One of the key insights in \cite{LR1, LR2} is that one can in fact prove $L^\infty$ estimates capturing this reductive structure {\bf without} commuting the projection to ${\bf E}^\mu$ with $\tBox_g$. More precisely, they adapted a strategy that treats all components on an equal footing in the energy estimates and allow the energy to grow with a slow $(1+t)^{C\epsilon}$ rate. At the same time, they applied an independent estimate, which is an extension of \eqref{Li.Li.intro} to the quasilinear setting, for the $L^\infty$ decay. This independent $L^\infty-L^\infty$ estimate exploits the reductive structure without commutation and gives the sharp $L^\infty$ decay rates. It is precisely because of this sharp $L^\infty$ decay estimates that it is possible to control the growth of the energy.

However, in the setting of our paper, as we have already seen in the model problem in Section~\ref{sec.weak.null}, it is important to capture the reductive structure also in $L^2$. Indeed, if we only capture the reductive structure at the $L^\infty$ level, then the energy grows as $(1+t)^C$ and we will not be able to close the bootstrap. \textbf{We therefore need also to capture the reductive structure also when proving the energy estimates.}

The main observation in this paper is that we can divide the spacetime into various regions\footnote{In the proof, we will also need to split into the regions $t\leq T$ and $t>T$ in a manner similar to Section \ref{sec.weak.null}. Let us suppress that at this moment to emphasize the decomposition in terms of the $(|x|-t)$-values.}. First, as mentioned in Section \ref{sec.proof.decay}, while the background solution does not have better than $O(\f{1}{1+t})$ uniform decay, it decays better in $|x|-t$ as $|x|-t\to\pm\infty$. Therefore, in the region where $||x|-t|$ is large, this is similar to the small data problem and it suffices to use the reductive structure in $L^\infty$ as in \cite{LR1, LR2}. In the remaining region, which has a finite $||x|-t|$-range, we show that the commutator of $\tBox_g$ with the projection to ${\bf E}^\mu$ is in fact controllable. More precisely, the most slowly-decaying term in the commutator only contains good derivatives\footnote{A similar observation for this commutation was made in \cite{Huneau} and was crucial for establishing the stability of Minkowski spacetime with $U(1)$-symmetry.}, and takes the form $\f{|\db\Gamma^I h|}{|x|}$. These terms can therefore be controlled using the good bulk term for $|\db\Gamma^I h|$ in the energy estimates (recall \eqref{EE.intro}). Notice that the weight $w'(|x|-t)$ in the good bulk term in \eqref{EE.intro} degenerates as $|x|-t\to\pm\infty$ - it is therefore important that we apply this estimate only in a region with some cutoff in the $||x|-t|$-length\footnote{Let us contrast this with the estimate in \eqref{EE.intro.1} in the model problem, where the error term with a $\db$ derivative takes the form $\f{|\db\Gamma^I\bar\phi|}{(1+t)^{1-\de}(1+||x|-t|)^\gamma}$. The crucial point is that there is extra decay in $(1+||x|-t|)^{-\gamma}$ to be exploited in that case, while such additional decay is not present in the term $\f{|\db\Gamma^I h|}{|x|}$ here.}.

The above discussion relies on the possibility to localize our estimates near spacelike and timelike infinities, as well as near the wave zone. However, there is another obstacle in order to carry out the localization of the estimates into different regions as outlined above. Even for small perturbations of the Minkowski spacetime in wave coordinates, the null hypersurfaces of the nonlinear spacetime diverge from that of the background spacetime logarithmically. In our setting, if such divergences occur, constant $(|x|-t)$-hypersurfaces will potentially\footnote{Notice that this does \underline{not} happen in small perturbations of the Minkowski spacetime since by the positive mass theorem, the ADM mass of the perturbation is no smaller than the background Minkowski spacetime. However, in the general case of stability of large dispersive spacetimes, it is of course desirable to allow perturbations both with larger and smaller ADM masses.} be timelike, which does not allow us to localize the energy estimates into regions as described above. As a consequence, we need to use a carefully chosen generalized wave coordinate condition such that the constant $(|x|-t)$-hypersurfaces approaches null as $t\to \infty$.

This is achieved in two parts: First, we need to choose a coordinate system for the background solution such that the outgoing null hypersurfaces are ``well-approximated'' by hypersurfaces with constant $|x|-t$ values. This is achieved by choosing the background gauge such that the components $|h_B|_{L\mathcal T}$ have improved decay (see discussions in Section~\ref{sec.proof.decay}). Second, we need to pick a gauge for the perturbed solution such that outgoing null hypersurfaces are again well-approximated constant $(|x|-t)$-hypersurfaces.

Dealing with the second point above is closely related to the ``problem of mass'', i.e.~the difficulties created by the long range effect of slow decay of the mass term when carrying out the estimates. In particular, the mass term gives infinite energy for an $L^2$ norm of the type that is used in \eqref{EE.intro} (see also the statement of Theorem \ref{main.thm}). In \cite{LR2}, this is dealt with by approximating the contribution of mass by a term $\chi(r)\chi(\f tr)\f{M}{r}\delta_{\mu\nu}$, where $\chi$ is an appropriate cutoff function. This choice, while sufficient for the purpose of \cite{LR1, LR2}, leads to a logarithmic divergence of the null hypersurfaces. Instead, we approximate the contribution of the mass by the metric $h_S$, to be defined in Definition~\ref{hS.def}. We then decompose the metric as $g=g_B+h_S+h$ so that $h$ has finite weighted energy and can be controlled using energy estimates. The key point of the choice $h_S$ is that the components $|h_S|_{L\mathcal T}=0$ and therefore the constant $(|x|-t)$-hypersurfaces approach null as $t\to \infty$, as long as we can show that $|h|_{L\mathcal T}$ is also sufficiently well-behaved. However, the issue now is that unlike $\f{M}{r}$, the components of $h_S$ are not solutions to the wave equation. We therefore need to modify the choice of our generalized wave gauge and to impose $\Box_g x^\mu = \mathcal G^\mu_B+\mathcal G^\mu_S$, where $\mathcal G^\mu_B$ is the gauge contribution from the background $g_B$ and $\mathcal G^\mu_S$ is chosen to cancel with the highest order contribution of $\Box_{m+h_S} x^\mu$ for large $|x|$ (see precise definitions in \eqref{wave.coord.2}).

Recall from our earlier discussions in Section \ref{sec.weak.null.2} that the wave coordinate condition is also used to handle some semilinear terms\footnote{It is also crucially for the quasilinear terms, see Section \ref{sec.proof.quasilinear}.} in \cite{LR1, LR2}. As we just discussed, this is replaced by a generalized wave condition involving $\mathcal G^\mu_B$ and $\mathcal G^\mu_S$ in our setting. When applying this condition to obtain improved estimates for the derivatives of the good components $|\rd h|_{L \mathcal T}$, there are extra terms coming from $\mathcal G^\mu_S$ of order $O(\f{\ep\log(2+t+|x|)}{1+t+|x|})$ (see for example \eqref{wave.con.2}), which is insufficient to close the estimates. Nevertheless, one finds a crucial cancellation in the terms $O(\f{\ep\log(2+t+|x|)}{1+t+|x|})$ so that \eqref{wave.coord.intro} still holds with some additional controllable error terms. This cancellation can be traced back to the fact that the approximate mass term $h_S$ is (at the highest order) chosen to be isometric to the Schwarzschild metric, which is itself a solution to the Einstein equations (see Proposition~\ref{gauge.est}).

\subsubsection{Localized energy estimates} \label{sec.localization}

Let us elaborate slightly further the localization procedure that we mentioned above. The key point is to use the fact that for every fixed $U\in \mathbb R$, there exists $T_U\geq 0$ such that the set $\mathcal B_U:=\{t-r-\frac{1}{(1+t)^{\f{\gamma}{4}}}=U \}$ is spacelike when restricted to to $t\geq T_U$. To see this, it suffices to show that along every fixed $\mathcal B_U$, $|g-m|_{L\mathcal T}(t,x)\ls_U \f{1}{(1+t)^{1+\f{\gamma}{2}}}$. This decay is achieved by a combination\footnote{Recall again the discussion from Section \ref{sec.mass} that $g$ is decomposed as $g-m=h_B+h_S+h$.} of the choice of the background gauge, the choice of $h_S$ (and the generalized wave coordinate condition) and also the decay for $|h|_{L\mathcal T}$. The decay for $|h_B|_{L\mathcal T}$ and $|h_S|_{L\mathcal T}$ has already been briefly discussed in previous subsections. The decay for $|h|_{L\mathcal T}$, on the other hand, is proven by using the generalized wave coordinate condition, which gives an analogue of \eqref{wave.coord.intro} and implies an estimate $|\rd h|_{L\mathcal T}(t,x)\ls \f{\ep (1+||x|-t|)^{\f{3\gamma}{8}}}{w(|x|-t)^{\f 12}(1+t+|x|)^{1+\f{\gamma}{2}}}$. This can then be integrated along radial constant $(t+|x|)$-curves to give $|h|_{L\mathcal T}(t,x)\ls \f{\ep (1+||x|-t|)^{\f 12+\gamma}}{w(|x|-t)^{\f 12}(1+t+|x|)^{1+\f{\gamma}{2}}}$ (see the proof of Proposition \ref{BS.close}). Therefore, on each fixed $\mathcal B_U$, we have the desired decay estimate.

Once we have this decay estimate for $|g-m|_{L\mathcal T}$, we can then prove the standard $(|x|-t)$-weighted energy estimates of the form \eqref{EE.intro} and note that the contributions on the set $\mathcal B_U\cap \{t\geq T\}$ have favorable signs (see the proof Proposition \ref{EE.2}). This then allows the estimates to be localized to the future or past of $\mathcal B_U\cap \{t\geq T\}$, as long as $T>0$ is sufficiently large.

\subsubsection{Commutators and higher order estimates}\label{sec.proof.quasilinear}

In this section and Section~\ref{sec.hierarchy} we further discuss some technical difficulties which are already present in \cite{LR1, LR2}, and can be treated with only minor modifications.

Up to this point, the discussions focused on the semilinear error terms, especially those that do not obey a classical null condition. In addition to those, there are also the quasilinear error terms. In particular, the most difficult error terms arise from the commutation of $\tBox_g$ and the Minkowskian vector fields $\Gamma$. It turns out that after choosing the gauge as described in the Section~\ref{sec.mass}, these error terms in the large data setting can also be treated in a similar manner as in \cite{LR1, LR2}. 

Writing $H^{\alp\bt}:=(g^{-1})^{\alp\bt}-m^{\alp\bt}$, where $m$ is the Minkowski metric, the commutators are given by $[\tBox_g, \Gamma]=[\Box_m,\Gamma]+(-\Gamma H^{\alp\bt}+H^{\alp\gamma}(\rd_\gamma \Gamma^\bt)+H^{\bt\gamma}(\rd_\gamma \Gamma^\alp))\rd_\alp\rd_\bt$. As pointed out in \cite{LR1, LR2}, for $\Gamma$ a Minkowskian commuting vector field\footnote{Recall that $L=\rd_t+\rd_r$.} $L^\sigma L^\gamma(m_{\alp\sigma}\rd_\gamma \Gamma^\alp)=0$ and therefore either there is a good derivative in the commutator term or we have to control\footnote{Here, the notation $|\cdot|_{L\mathcal T}$ is defined so that $H$ and $\Gamma H$ are understood as covariant $2$-tensors where the indices are lowered with respect to the Minkowski metric.} $|H|_{L\mathcal T}$ and $|\Gamma H|_{LL}$. To this end, we need improved decay\footnote{By this we mean faster than integrable decay along any fixed constant $(|x|-t)$-hypersurface. Similar to Section \ref{sec.mass}, while we need an improved decay in $t$, we can allow this bound to grow in $|t-|x||$. More precisely, we will prove a bound $\ls \frac{(1+|t-|x||)^{\frac 12+\gamma}}{(1+t+|x|)^{1+\frac{\gamma}{2}}w(|x|-t)^{\frac 12}}$ for some $\gamma>0$ to be introduced (see \eqref{inverse.4}).} for $|g-m|_{L\mathcal T}$ and $\sum_{|I|=1}|\Gamma^I(g-m)|_{LL}$ (see Proposition \ref{inverse}). For the zeroth order derivative, as we explained near the end of Section \ref{sec.mass}, this can be obtained precisely due to the choice of our gauge condition and the generalized wave coordinate condition. It turns out that similar ideas can be extended to control $|\Gamma(g-m)|_{LL}$ due to properties of the Minkowskian commuting vector fields.

Notice that however this improved decay no longer holds for $\sum_{|I|=2}|\Gamma^I h|_{LL}$ and $\sum_{|I|=1}|\Gamma^I h|_{L\mathcal T}$. As a consequence, for higher commutations, the good structure for the commutator terms only occurs at the top order and there are lower order terms which do not have a good structure. As we will explain further below, in order to deal with this issue, for every higher derivative that we take, we prove energy estimates that grow with a slightly higher power in $t$.

\subsubsection{Hierarchy of estimates}\label{sec.hierarchy}

In the discussions above, we saw that both the semilinear terms and the commutator terms have a good null-or-reductive structure. However, we have in fact only discussed this good structure at the top order of derivatives and there are in fact terms which are lower order in derivatives and do not have any good structure. We have already discussed one source of such terms near end of Section \ref{sec.proof.quasilinear}, which comes from the commutation of $\tBox_g$ and $\Gamma$.

More precisely, when considering the equation for $|I|$ $\Gamma$ derivatives, in addition to the top order terms which verify the structure we mentioned earlier, we have some additional terms which are lower order in the number of derivatives:
\begin{equation*}
\begin{split}
|\tBox_g \Gamma^I h|\ls &\: \mbox{Top order terms}+\sum_{|J_1|+|J_2|\leq |I|-1}\left(|\rd\Gamma^{J_1}h||\rd\Gamma^{J_2}h|+\f{|\Gamma^{J_1}h||\rd\Gamma^{J_2}h|}{1+||x|-t|}\right)\\
&+\sum_{|J|\leq |I|-1}\f{\log(2+t+|x|)|\rd\Gamma^J h|}{(1+t+|x|)}+\cdots, 
\end{split}
\end{equation*}
where $\dots$ denotes some lower order terms that can be treated similarly that we suppress for the exposition\footnote{In particular, we have suppressed terms involving the scalar field $\beta$}. Notice that the logarithm growth in the last term is due to the fact that our background metric only obeys the estimate\footnote{In the application, this is necessary in order to ensure that $|h_B|_{L\mathcal T}$ is well-behaved.} $|\Gamma^J h_B|\ls \f{\log(2+t+|x|)|\rd\Gamma^J h|}{(1+t+|x|)}$.

This difficulty was already present in \cite{LR1, LR2} and was handled by proving a hierarchy of estimates. In our setting, we will introduce a similar hierarchy. More precisely, we make use of the fact that the terms without a good structure are lower order in terms of derivatives and inductively prove estimates which are worse in terms of the time decay for every additional derivative that we take. We will choose small parameters $\de$ and $\de_0$ with $\de\ll \de_0$. We then prove energy estimates such that for the $k$-th derivative\footnote{For $k\leq N$, where $N\geq 11$.} the energy grows as $\ep(1+t)^{(2k+5)\de_0}$. For the decay estimate, we likewise allow some loss in $t$ for every derivative, but quantified with the smaller parameter $\de$. Namely, we prove for $k\leq \lfloor \f{N}{2}\rfloor +1$ that 
$$\sum_{|I|\leq k}\sup_{\substack{\tau\in [0,t]\\ x\in \mathbb R^3}}(1+\tau+|x|)(1+||x|-t|)^{\f 12-\f\gamma 4} w(|x|-t)^{\f 12}(|\rd\Gamma^Ih|+|\rd\Gamma^I\beta|)(\tau,x)\ls \ep(1+t)^{2^k\de}.$$
Notice that unlike in \cite{LR1, LR2}, our bad lower order error terms are no longer quadratically small, but nonetheless the hierarchy of estimates fit well with an induction argument in which when we consider an estimate with $k$ derivatives, either we have some additional smallness arising from the good structure of the equation at the top order, or we have error terms which depend on at most $k-1$ derivatives.

For such a scheme to work, we need to following two important facts: Firstly, the proof of the decay estimates is essentially independent of the loss in $t$ in the energy. This is achieved, as in \cite{LR1, LR2}, using an independent ODE argument to derive an estimate similar to \eqref{Li.Li.intro}. Secondly, while we allow the estimates to have a small loss in the powers of $t$ for both the energy and the pointwise estimates, in the lowest order it is important that we prove the sharp decay estimates $|\rd h|_{\mathcal T\mathcal U}+|\rd \beta|\ls \f{\ep}{1+t}$. These sharp estimates play a crucial role in recovering the energy estimates.

This concludes the discussion of the main difficulties and ideas in the proof of the main theorem.

\subsection{Outline of the paper}\label{sec.outline} We end the introduction with an outline of the remainder of the paper.
\begin{itemize}
\item We introduce our notations in \textbf{Section~\ref{sec.notation}}.
\item In \textbf{Section~\ref{sec.assumptions}}, we give a precise statement of our main theorem (see Theorem~\ref{main.thm.2}).
\item In \textbf{Section~\ref{sec.setup}}, we introduce the gauge condition and recast the equations as a system of quasilinear wave equations. In \textbf{Section~\ref{sec.third.thm}}, we once again rephrase our main theorem (see Theorem~\ref{main.thm}), but not in terms of the system of quasilinear wave equations. 
\item In \textbf{Section~\ref{sec.BA}}, we introduce the bootstrap assumptions; in \textbf{Section~\ref{sec.prelim.bounds}}, we then derive preliminary bounds which follow immediately from the bootstrap assumptions.
\item In \textbf{Section~\ref{sec.generalized.wave.coordinate}}, we further analyze our gauge condition. Using in particular results in Sections~\ref{sec.prelim.bounds} and \ref{sec.generalized.wave.coordinate}, we give the first pointwise estimates for the RHS of the equations for $h$ and $\bt$ in \textbf{Section~\ref{sec.eqn}} and \textbf{Section~\ref{sec.eqn.scalar}} respectively. In particular it is here that we derive the weak null structure of the equations.
\item In the remaining sections prove the main estimates needed for the proof of the main theorem.
\begin{itemize}
\item In \textbf{Section~\ref{sec.linear.estimates}}, we collect linear estimates for the wave equation on curved background.
\item In \textbf{Section~\ref{def.regions}}, we divide the spacetime into 4 regions; we then prove the $L^2$-energy estimates in different regions of the spacetime. The finite-$t$ region is treated by Cauchy stability in \textbf{Section~\ref{sec.Cauchy.stability}}. In \textbf{Section~\ref{sec.EE}} we give some general estimates to be used in all the remaining regions. Then in \textbf{Section~\ref{sec.II}}, \textbf{Section~\ref{sec.III}} and \textbf{Section~\ref{sec.IV}}, we prove estimates in the region near spatial infinity, near null infinity and near timelike infinity respectively.
\item Finally, we improve the bootstrap assumptions by proving $L^\infty$ estimates in \textbf{Section~\ref{sec.recover.bootstrap}} and conclude the proof.
\end{itemize}
\end{itemize}

\subsection{Acknowledgements} We thank Igor Rodnianski for stimulating discussions. Part of this work was carried out when J.~Luk visited UC Berkeley and he thanks UC Berkeley for their hospitality. 

J.~Luk is supported by a Terman fellowship and the NSF Grants DMS-1709458 and DMS-2005435. S.-J.~Oh is supported by the Samsung Science and Technology Foundation under Project Number SSTF-BA1702-02, a Sloan Research Fellowship and a NSF CAREER Grant DMS-1945615.

\section{Notations}\label{sec.notation}
In this section, we define the necessary notations that we will use in this paper. In our setting, we have a coordinate system $(t,x^1,x^2,x^3)$ on the manifold-with-boundary $[0,\infty)\times \mathbb R^3$. We will frequently use $x^0$ and $t$ interchangeably. Moreover, $r$ will denote the function $r=\sqrt{\sum_{i=1}^3(x^i)^2}$. The lower case Latin indices $i,j\in\{1,2,3\}$ are reserved for the spatial coordinates while the lower case Greek indices $\alp,\bt\in\{0,1,2,3\}$ are used for all the spacetime coordinates.

First, we define the following:
\begin{definition}\label{def.Mink.vf}
Let the \emph{Minkowskian commuting vector fields} to be the set of vector fields
$$\{\rd_\mu,\,x^i\rd_j-x^j\rd_i,\,t\rd_i+x^i\rd_t,\,S:=t\rd_t+\sum_{i=1}^3x^i\rd_i\}$$
defined with respect to the coordinate system $(t,x^1,x^2,x^3)$.
\end{definition}
In the remainder of the paper, we will use $\Gamma$ to denote a general Minkowskian commuting vector field. For a multi-index\footnote{Notice that this is slightly different from the usual multi-index notation.} $I=(i_1,i_2,\dots,i_{|I|})$, $\Gamma^I$ will denote a product of $|I|$ Minkowskian commuting vector field. More precisely, order the $11$ distinct Minkowskian commuting vector fields above as $\Gamma_{(1)}$, $\Gamma_{(2)}$, \dots, $\Gamma_{(11)}$. Then, for $(i_1,i_2,\dots,i_{|I|})\in \{1,2,\dots, 11\}^{|I|}$, $\Gamma^I=\Gamma_{(i_1)}\Gamma_{(i_2)}\cdots\Gamma_{(i_{|I|})}$.  

We also define the vector fields $\{L,\Lb,E^1,E^2,E^3\}$ as follows:
\begin{definition}\label{def.proj.vf}
Let $\rd_r=\sum_{i=1}^3\frac{x^i}{r}\rd_i$. Define
$$L=\rd_t+\rd_r,\quad \Lb=\rd_t-\rd_r.$$
We will also define the vector fields $\{E^1,E^2,E^3\}:=\{\f{x^2}{r}\rd_3-\f{x^3}{r}\rd_2,\,\f{x^1}{r}\rd_3-\f{x^3}{r}\rd_1,\,\f{x^1}{r}\rd_2-\f{x^2}{r}\rd_1\}$ tangent\footnote{We remark the obvious facts that $\{E^1,E^2,E^3\}$ span the tangent space of the coordinate $2$-spheres, but are not linearly independent. Away from $\{x^1=x^2=0\}\cup\{x^2=x^3=0\}\cup\{x^3=x^1=0\}$, any two of $\{E^1,E^2,E^3\}$ form a basis to the tangent space of the coordinate $2$-sphere.} to the coordinate $2$-spheres given by constant $r$-value. 
We will use capital Latin indexed $E^A$ to denote an element of $\{E^1,E^2,E^3\}$ and use small Greek indexed and bold ${\bf E}^\mu$ to denote an element of the set $\{L,\Lb,E^1,E^2,E^3\}$.
\end{definition}
We will also use the coordinates $(s,q,\om):=(s,q,\theta,\varphi)$, where $\om=(\theta,\varphi)$ are the usual polar coordinates and $(s,q)$ are defined by 
\begin{definition}
$$s=t+r,\quad q=r-t.$$
As a consequence, we have\footnote{Here and in the remainder of the paper, $\rd_s$ and $\rd_q$ denote the coordinate vector fields in the $(s,q,\om)$ coordinate system.}
$$\rd_s=\f12(\rd_t+\rd_r),\quad \rd_q=\f12(\rd_r-\rd_t).$$
\end{definition}
\begin{remark}
Notice the different normalizations of the coordinate vector field $(\rd_s,\rd_q)$ and $(L,\Lb)$.
\end{remark}
\begin{remark}\label{rmk.q1}
For convenience, we will also use the notation say $q_1=r-t_1$ when $t_1$ is a chosen value of the coordinate function $t$.
\end{remark}
The introduction of these vector fields is important for two reasons. First, the solution has better decay properties when it is differentiated with respect to the ``good derivatives'' $(L, E^1, E^2, E^3)$. Second, when projected to these vector fields, some ``good components'' of the metric decays better than the others.

We use the following notations for derivative for a scalar function $\xi$:
\begin{definition}
\begin{enumerate}
\item \emph{General} derivatives are denoted by
$$|\rd\xi|^2 := (\rd_t\xi)^2 + \sum_{i=1}^3(\rd_i\xi)^2.$$
\item To capture the improved decay with respect to the ``\emph{good} derivatives'', we define
$$|\nabb\xi|^2:=\f12\sum_{i,j=1}^3(\f{x^i}{r}\rd_j\xi-\f{x^j}{r}\rd_i\xi)^2$$
and
$$|\db \xi|^2:= |\rd_s\xi|^2+|\nabb\xi|^2.$$
\item Finally, \emph{spatial} derivatives are denoted by
$$|\nab\xi|^2:=\sum_{i=1}^3|\rd_i\xi|^2.$$
We will also use the multi-index notation $\nab^I$ in a similar manner as that for $\Gamma^I$ defined above.
\end{enumerate}
\end{definition}

We next define the notation for projection to ${\bf E}^\mu\in \{L,\Lb,E^1,E^2,E^3\}$ that will be useful in capturing the improved decay for certain components of the metric. First, we introduce the convention that for any 2-tensors, indices are raised and lowered with respect to the Minkowski metric \eqref{Minkowski.metric}.
We make the following definition for the norms of tensors:
\begin{definition}\label{def.tensor.norm}
Given a $2$-tensor $p$, define
$$|p|^2=\sum_{0\leq \mu,\nu\leq 3}|p_{\mu\nu}|^2.$$
\end{definition}
We also make the following definitions:
\begin{definition}\label{def.proj}
Let $\mathcal T=\{L,E^1,E^2,E^3\}$ and $\mathcal U=\{L,\Lb,E^1,E^2,E^3\}$. We also abuse notation slightly to denote by $L$ the single element set $\{L\}$. For any two of these family $\mathcal V$ and $\mathcal W$ and any $2$-tensor $p$, define
$$|p|^2_{\mathcal V\mathcal W}:=\sum_{V\in\mathcal V,W\in\mathcal W}|p_{\alp\bt}V^{\alp}W^{\bt}|^2,$$
$$|\rd p|^2_{\mathcal V\mathcal W}:=\sum_{U\in\mathcal U,V\in\mathcal V,W\in\mathcal W}|(\rd_\gamma p_{\alp\bt})V^{\alp}W^{\bt}U^\gamma|^2,$$
$$|\db p|^2_{\mathcal V\mathcal W}:=\sum_{T\in\mathcal T,V\in\mathcal V,W\in\mathcal W}|(\rd_\gamma p_{\alp\bt})V^{\alp}W^{\bt}T^\gamma|^2.$$
Most importantly, notice that the vector fields contracted outside the differentiation (in a way similar to \cite{LR1},\cite{LR2}).
\end{definition}

\begin{remark}
We will in particular use the notation introduced in Definition \ref{def.proj} for $\Gamma^I h$. Here, we view $(\Gamma^I h)_{\mu\nu}$ as a 2-tensor where each component with respect to the coordinate system $(t,x^1,x^2,x^3)$ is simply given by the component-wise derivative by $\Gamma^I$.
\end{remark}

\begin{remark}
We will use the convention that indices are raised and lowered using the Minkowski metric $m$. In particular, Definitions \ref{def.tensor.norm} and \ref{def.proj} apply to contravariant $2$-tensors as well as covariant $2$-tensors.
\end{remark}

We now define the some notations for the subsets of the spacetime that we consider.
$$\Sigma_{\tau}:=\{(t,x): t=\tau\},$$
$$D_{\tau}:=\{(t,x):t=\tau,\,\f t2\leq |x|\leq 2t\},$$
$$\mathcal R_1:= \{(t,x):t\leq T\},$$
$$\mathcal R_2:=\{(t,x):t\geq T,\,t-|x|-\f{1}{(1+t)^{\f{\gamma}{4}}}\leq U_2 \},\,\mathcal R_{2,\tau}:=\mathcal R_2\cap\Sigma_\tau,$$
$$\mathcal R_3:=\{(t,x):t\geq T,\,U_2\leq t-|x|-\f{1}{(1+t)^{\f{\gamma}{4}}}\leq U_3 \},\,\mathcal R_{3,\tau}:=\mathcal R_3\cap\Sigma_\tau,$$
$$\mathcal R_4:=\{(t,x):t\geq T,\,t-|x|-\f{1}{(1+t)^{\f{\gamma}{4}}}\geq U_3 \},\,\mathcal R_{4,\tau}:=\mathcal R_4\cap\Sigma_\tau,$$
$$\mathcal B_U:=\{(t,x):t-|x|-\f{1}{(1+t)^{\f{\gamma}{4}}}=U\},$$
where $U_1,\,U_2,\,U_3\in \mathbb R$ are constants. We will also use the notation $\mathcal R$ to denote either one of the regions $\mathcal R_2$, $\mathcal R_3$ or $\mathcal R_4$.

Let us collect here a few more pieces of notations that will be used. For metrics, $m$ denotes the Minkowski metric (see \eqref{Minkowski.metric}); $g_B$ denotes the background metric we perturb against; $h_B:=g_B-m$; $g$ denotes the unknown metric; $h_S$ is given in Definition \ref{hS.def}; $g_S:=m+h_S$ and $h:=g-m-h_B-h_S$. We will also use the following conventions for the inverse metrics: $H^{\alp\bt}:=(g^{-1})^{\alp\bt}-m^{\alp\bt}$ and $H_B^{\alp\bt}:=(g_B^{-1})^{\alp\bt}-m^{\alp\bt}$.

For the scalar field, $\phi$ denotes the unknown scalar field in the spacetime; $\phi_B$ is the background scalar field; $\beta:=\phi-\phi_B$.

We next introduce our conventions for integration. On $\Sigma_\tau$ or its subsets (e.g., $\mathcal R_{2,\tau}$, $\mathcal R_{3,\tau}$ and $\mathcal R_{4,\tau}$), unless otherwise stated, we integrate with respect to $dx:= dx^1dx^2dx^3=r^2\sin\theta\, d\theta\, d\varphi$. In a spacetime region, unless otherwise stated, we integrate with respect to $dx\, dt$. On $\mathcal B_U\cap\{t_1\leq t\leq t_2\}$, we will integrate with respect to the measure $dx$, which is defined as follows: For $\Phi(t):=t-\f{1}{(1+t)^{\f\gamma 4}}$, $\int_{\mathcal B_U\cap \{t_1\leq t\leq t_2\}} F\,dx:=\int_{\mathcal B_U\cap \{t_1\leq t\leq t_2\}} F(t=\Phi^{-1}(|x|+U),x)\,dx^1\, dx^2\, dx^3$, where $F(t=\Phi^{-1}(|x|+U),x)$ is considered as a function of $x^1$, $x^2$ and $x^3$. Frequently, when there is no danger of confusion, we suppress the explicit dependence of the integrand on the variables of integration. For instance, when $F$ is a function of spacetime, $\int_T^t \int_{\Sigma_\tau} F\,dx\, d\tau$ always implicitly means $\int_T^t \int_{\Sigma_\tau} F(\tau,x)\,dx\, d\tau$.

Finally, we introduce the convention that {\bf in the remainder of the paper, $x \ls y$ denotes the inequality $x\leq By$ for some constant $B$. This constant $B$ will eventually be allowed to depend only on the constants $C$, $\gamma_0$, $N$ in Definition \ref{def.dispersivespt}, the constants $\gamma$ in Definition \ref{def.pert} and also the constants $\de_0$, $\de$ that we will introduce in the proof.}

\section{Assumptions on the background solution and second version of main theorem}\label{sec.assumptions}

With the above definitions, we can describe the class of background metrics that we study. We consider a background Lorentzian metric $g_B$ on a manifold-with-boundary diffeomorphic to $[0,\infty)\times \mathbb R^3$ settling to Minkowski spacetime with a precise rate. On this manifold-with-boundary, there is also a real valued function $\phi$ which decays to $0$ with a rate. The metric $g_B$ and the scalar field $\phi$ together satisfy the Einstein--scalar field system. More precisely, we define
\begin{definition}\label{def.dispersivespt}
Let $\gamma_0>0$ be a real number and $N\geq 11$ be an integer. A spacetime $(\mathcal M=[0,\infty)\times \mathbb R^3,g_B)$ with a scalar field $\phi_B:\mathcal M\to\mathbb R$ is a \emph{dispersive spacetime solution} of size $(C,\gamma_0,N)$ if 
\begin{enumerate}
\item (Solution to the Einstein scalar field system) The triple $(\mathcal M,g_B,\phi_B)$ is a solution to the Einstein--scalar field system.
\item (Limiting to Minkowski space) There exists a global system of coordinate $(t,x^1,x^2,x^3)$ such that with respect to this coordinate system, the metric takes the form
$$g_B-m = h_B,$$
where 
$$m=-dt^2+\sum_{i=1}^3(dx^i)^2$$ 
is the Minkowski metric and $h_B$ obeys the bound
$$|\Gamma^I h_B|\leq \frac{C\log(2+s)}{1+s}$$
for $|I|\leq N+1$, where $\Gamma$'s are the Minkowskian commuting vector fields.
\item (Decay for derivatives of metric) For $|I|\leq N+1$, we have 
$$|\rd \Gamma^I h_B|\leq \frac{C}{(1+s)(1+|q|)^{\gamma_0}}$$
for any combinations of Minkowskian commuting vector fields $\Gamma$.
\item (Improved decay for ``good derivatives'' of metric) For $|I|\leq N+1$, we have 
$$|\bar\rd \Gamma^I h_B|\leq \frac{C}{(1+s)^{1+\gamma_0}}$$
for any combinations of Minkowskian commuting vector fields $\Gamma$.
\item (Improved decay for ``good components'' of the metric) For $|I|\leq 1$, the following components satisfy better bounds:
$$\sum_{|I|\leq 1}|\Gamma^I h_B|_{LL}+|h_B|_{L\mathcal T}\leq \f{C}{(1+s)^{1+\gamma_0}}$$
for any Minkowskian commuting vector field $\Gamma$.
\item (Decay for the scalar field) For $|I| \leq N+1$, we have 
$$|\rd\Gamma^I \phi_B|\leq \f{C}{(1+s)(1+|q|)^{\gamma_0}},\quad |\bar\rd\Gamma^I \phi_B|\leq \f{C}{(1+s)^{1+\gamma_0}}$$
for any combinations of Minkowskian commuting vector fields $\Gamma$.
\item (Uniform Lorentzian assumption of $g_B$) 
The metric $g_B$ is everywhere Lorentzian with uniformly bounded inverse:
\begin{equation}\label{inverse.bd.assumption}
|g_B^{-1}|\leq C.
\end{equation}
Let $(\hat{g}_B)_{ij}$ be the restriction of the metric $g_B$ on the tangent space to the constant $t$-hypersurfaces (where $i,j=1,2,3$). $(\hat{g}_B)_{ij}$ satisfies the condition that for any $\xi_i$,
\begin{equation}\label{spatial.eigen.bd.assumption}
C^{-1}|\xi|^2\leq \sum_{i,j=1}^3 (\hat{g}_B^{-1})^{ij}\xi_i\xi_j\leq C|\xi|^2,
\end{equation}
where 
$$|\xi|^2:=(\xi_1)^2+(\xi_2)^2+(\xi_3)^2.$$
Also, the spacetime gradient of $t$ is timelike and satisfies
\begin{equation}\label{g00.bd.assumption}
(g_B^{-1})^{00}=(g_B^{-1})^{\alp\bt}\rd_\alp t\rd_\bt t\leq -C^{-1}<0.
\end{equation}
\item (Almost wave coordinate condition) For $|I|\leq N+1$, the global coordinate functions satisfy the estimate\footnote{Recall again our notation that $t$ and $x^0$ are used interchangeably.}
$$|\Gamma^I(\Box_{g_B} x^{\mu})|\leq \f{C\log(2+s)}{(1+s)^2}.$$
Here, $\Box_{g_B}$ is the Laplace-Beltrami operator associated to the metric $g_B$, i.e.~$$\Box_{g_B}:= \f{1}{\sqrt{-\det g_B}}\rd_\alp((g_B^{-1})^{\alp\bt}\sqrt{-\det g_B}\rd_\bt \cdot).$$
\end{enumerate}
\end{definition}

For a fixed dispersive spacetime $(\mathcal M,g_B,\phi_B)$, we will define a class of admissible perturbations for which we will show that their maximal globally hyperbolic future developments are future causally geodesically complete and such that the metrics (resp.~scalar fields) are globally close to $g_B$ (resp. $\phi_B$). Recall that an initial data set to the Einstein scalar field system consists of a quintuplet $(\Sigma,\hat{g},\hat{k},\hat{\phi},\hat{\psi})$, where $(\Sigma,\hat{g})$ is a Riemannian 3-manifold, $\hat{k}$ is a symmetric 2-tensor and $\hat{\phi}$ and $\hat{\psi}$ are real valued functions on $\Sigma$.
Moreover, for ${\bf N}$ being the unit future-directed normal to $\Sigma$ in $\mathcal M$, the following \emph{constraint equations} are satisfied:
\begin{equation*}
\begin{split}
R_{\hat{g}}+(\mbox{tr}_{\hat g} \hat k)^2-|\hat k|_{\hat{g}}^2=&4 \mathbb T({\bf N},{\bf N}),\\
\mbox{div}_{\hat{g}} \hat k-\hat\nab (\mbox{tr}_{\hat g} \hat k)=&2\mathbb T({\bf N},\cdot),
\end{split}
\end{equation*}
where $\hat\nab$ is defined as the Levi-Civita connection induced by $\hat g$. The celebrated theorems\footnote{Notice that while the results in \cite{CB} and \cite{CBG} are originally proved for the Einstein vacuum equations, they can be generalized to the Einstein--scalar field system. See \cite{CBIP}.} of Choquet-Bruhat \cite{CB} and Choquet-Bruhat--Geroch \cite{CBG} show that there exists a unique maximal globally hyperbolic future development $(\mathcal M,g,\phi)$ to the initial data which solves the Einstein--scalar field system such that $\Sigma$ is an embedded hypersurface in $\mathcal M$ with $\hat{g}$ and $\hat{k}$ being the induced first and second fundamental forms. Moreover, $\phi \restriction_{\Sigma}=\hat{\phi}$ and ${\bf N}\phi \restriction_{\Sigma}=\hat{\psi}$.

Before we proceed to define the class of admissible perturbations, we need to first introduce another piece of notation.
\begin{definition}\label{hS.def}
Let $\tilde{h}_S$ be defined by
$$(\tilde{h}_S)_{00}=\frac{2M}{r},\quad (\tilde{h}_S)_{0i}=0,$$
$$(\tilde{h}_S)_{ij}=-\frac{4M\log(r-2M)}{r}\delta_{ij}-\frac{x_i x_j}{r^2}(\f{2M}{r}-\f{4M\log (r-2M)}{r})$$
and define $h_S$ by
$$h_S=\chi(r)\chi(\frac rt)\tilde{h}_S,$$
where $\chi(s)$ is a smooth cutoff function such that it takes value $1$ when $s\geq \frac 34$ and equals to $0$ when $s\leq \frac 12$.
\end{definition}
\begin{remark}
Notice that for $r$ sufficiently large, $m+h_S$ is just the Schwarzschild metric (with mass $M$) written in a non-standard coordinate system up to error terms of order $\f{\log r}{r^2}$. $h_S$ is introduced so as to capture the behavior of mass at infinity. 
\end{remark}

\begin{remark}\label{hS.LT.bd}
Notice that for $h_S$ defined above, we have 
$$|\Gamma^I h_S|\ls \f{|M| \log (2+s)}{1+s},\quad |\rd\Gamma^I h_S|\ls \f{|M|\log (2+s)}{(1+s)^2}$$
for\footnote{Of course the bound in fact holds for all $|I|$ with a constant depending on $|I|$. We will not need this below.} $|I|\leq N+1$.
Moreover,
$$(h_S)_{LL}=(h_S)_{LE^A}=0.$$
Hence $h_S$ satisfies
$$|h_S|_{L\mathcal T} =0 .$$
Moreover, a direct computation shows that the $LL$-component of the $\Gamma$ derivative also vanishes:
 $$|\Gamma h_S|_{LL} =0.$$
Notice however that this property fails for the general $L\mathcal T$-component or for higher $\Gamma$ derivatives.
\end{remark}

We now define the class of admissible perturbations:
\begin{definition}\label{def.pert}
Let $\ep>0$ and $\gamma>0$ be real numbers and $N\geq 11$ be an integer. An initial data set $(\Sigma,\hat{g},\hat{k},\hat{\phi},\hat{\psi})$ is an $(\ep,\gamma,N)$-admissible perturbation to a dispersive spacetime $(\mathcal M,g_B,\phi_B)$ if
\begin{itemize}
\item $\Sigma$ is diffeomorphic to $\mathbb R^3$.
\item There is a coordinate system $(x^1,x^2,x^3)$ on $\Sigma$ such that with respect to this coordinate system, we have
\begin{equation}\label{gh.decomp}
\hat{g}=g_B\restriction_{\Sigma}+h_S\restriction_{\Sigma}+\hat h,
\end{equation}
where the parameter $M$ in the definition of $h_S$ satisfies
$$|M|\leq \ep,$$
and $\hat h$ obeys the estimates
\begin{equation}\label{h.data.bd}
\sum_{|I|\leq N}\|(1+r)^{\f12+\gamma+|I|}\nabla\nabla^I \hat{h}\|_{L^2(\Sigma)}\leq \ep.
\end{equation}
\item In the coordinate system $(x^1,x^2,x^3)$, the second fundamental form verifies the estimates
\begin{equation}\label{k.data.bd}
\sum_{|I|\leq N-1}\|(1+r)^{\f12+\gamma+|I|}\nabla^I (\hat{k}-(k_B)\restriction_{\{t=0\}})\|_{L^2(\Sigma)}\leq \ep.
\end{equation}
\item In the coordinate system $(x^1,x^2,x^3)$, $\hat{\phi}$ and $\hat{\psi}$ obey the bounds
\begin{equation}\label{phi.data.bd}
\begin{split}
&\sum_{|I|\leq N}\|(1+r)^{\f12+\gamma+|I|}\nab\nab^I (\hat\phi-(\phi_B)\restriction_{\{t=0\}})\|_{L^2(\Sigma)}\\
&\quad+\sum_{|I|\leq N-1}\|(1+r)^{\f12+\gamma+|I|}\nab^I (\hat\psi-(\rd_t\phi_B)\restriction_{\{t=0\}})\|_{L^2(\Sigma)}\leq \ep.
\end{split}
\end{equation}
\end{itemize}
\end{definition}
Our main result can be summarized in the following theorem:
\begin{theorem}[Main theorem, second version]\label{main.thm.2}
Let $N\geq 11$ and $0<\gamma,\gamma_0\leq \f 18$. For every dispersive spacetime solution $(\mathcal M,g_B,\phi_B)$ of size $(C,\gamma_0,N)$, there exists $\ep=\ep(C,\gamma,\gamma_0,N)>0$ sufficiently small such that for all $(\ep,\gamma,N)$-admissible perturbations of $(\mathcal M_0,g_B,\phi_B)$, the maximal globally hyperbolic future development is future causally geodesically complete and the spacetime remains close to $(\mathcal M_0,g_B,\phi_B)$ in a suitable\footnote{The precise sense in which the solution remains close will be formulated in terms of a generalized wave gauge. See Theorem \ref{main.thm}.} sense.
\end{theorem}
\begin{remark}
Without loss of generality, we can assume that $\gamma=\gamma_0$. {\bf We make this assumption from now on.}
\end{remark}
\begin{remark}[Strongly asymptotically flat spacetimes (cf.~\cite{CK})]\label{SAF} 
A smooth initial data set $(\Sigma,\hat{g},\hat{k},\hat{\phi},\hat{\psi})$ is said to be \emph{strongly asymptotically flat} (with one end) if there exists a compact set $K\subset \Sigma$ such that $\Sigma\setminus K$ is diffeomorphic to\footnote{Here, $\overline{B(0,1)}$ denotes the closed unit ball in $\mathbb R^3$ centered at the origin.} $\mathbb R^3\setminus \overline{B(0,1)}$ and there exists a system of coordinates $(x^1, x^2, x^3)$ in a neighborhood of infinity such that as $r:=\sqrt{\sum_{i=1}^3 (x^i)^2}\to \infty$, we have\footnote{Here, we use the notation that a function $f$ is $o_{\infty}(r^{-\alp})$ if for every $(\ell_1,\ell_2,\ell_3)\in (\mathbb N\cup \{0\})^3$, $r^{\alp+\ell_1+\ell_2+\ell_3}|\rd_{x^1}^{\ell_1}\rd_{x^2}^{\ell_2}\rd_{x^3}^{\ell_3} f|\to 0$ as $r\to \infty$.}, for $\alp=\f 32$ and for some $M>0$,
\begin{equation}\label{SAF.g}
\hat{g}_{ij}=(1+\f{M}{r})\de_{ij}+o_{\infty}(r^{-\alp}),\quad \hat{k}_{ij}=o_{\infty}(r^{-\alp-1}),
\end{equation}
$$\hat{\phi}=o_{\infty}(r^{-\alp}),\quad \hat{\psi}=o_{\infty}(r^{-\alp-1}).$$
Theorem \ref{main.thm.2} only requires weaker asymptotics at spatial infinity than strong asymptotic flatness for both the background solution and the perturbation, with $\alp$ above replaced by any real number greater than $1$. Notice that the metric in Theorem \ref{main.thm.2} does not have the form as in \eqref{SAF.g}. However, as we will see in the next remark below, the different forms of the metric are equivalent after a change of coordinates.
\end{remark}
\begin{remark}[Relation of $h_S$ to the mass term]
We show that the term $h_S$ is related to the mass term in \eqref{SAF.g}. More precisely, suppose we have a coordinate system $(\tilde{x}^1,\tilde{x}^2,\tilde{x}^3)$ on $\Sigma=\mathbb R^3$ such that with respect to this coordinate system, the intrinsic metric $\hat{g}$ on $\Sigma$ takes the form
\begin{equation*}
\hat{g}_{\tilde i\tilde j}=\de_{\tilde i\tilde j}+\chi(\tilde r)\f{M}{\tilde r}\de_{\tilde i\tilde j}+\tilde{h}_{\tilde i\tilde j},
\end{equation*}
where $\tilde{h}=o_{\infty}(r^{-\alp})$.
Then, for
$$\tilde{r}^2:=\sum_{i=1}^3 (\tilde{x}^i)^2,$$
we can introduce the change of coordinates 
$$x^i=\chi(\tilde{r})\f{\tilde{x}^i}{\tilde{r}}\big(\tilde{r}+2M\log(\tilde r-2M)\big)$$
so that the metric takes the form of 
$$\hat{g}_{ij}=\de_{ij}+(h_S)_{ij}+h$$
where $h=o_{\infty}(r^{-\alp}\log(2+r))$.
\end{remark}
\section{Basic setup and gauge condition}\label{sec.setup}
In this section, we write the Einstein scalar field system in terms of a system of quasilinear wave equations. To this end, we introduce a generalized wave gauge. As mentioned before, we cannot use the standard wave gauge for our problem at hand but will need to carefully design a generalized wave gauge to capture the mass at infinity. In particular, this gauge will allow us to localize near the ``wave zone'' in order to capture the reductive structure of the Einstein scalar field system at the $L^2$ level.

We now define the generalized wave gauge that we will work with: we impose the following condition
\begin{equation}\label{wave.coord}
\Box_{g}x^{\mu}=\mathcal G^\mu=\mathcal G^\mu_B+\mathcal G^\mu_S
\end{equation}
where\footnote{Recall the definition of the cutoff function $\chi$ in Definition \ref{hS.def}.} 
\begin{equation}\label{wave.coord.2}
\mathcal G^\mu_B=\Box_{g_B}x^{\mu},\quad \mathcal G^0_S=0,\quad \mathcal G^i_S=-\chi(r)\chi(\frac rt)\left(\frac {4M\log (r-2M)}{r^2}\right)\f{x^i}{r}.
\end{equation}
One easily checks that the condition \eqref{wave.coord} can also be rephrased as
\begin{equation}\label{wave.coord.det}
\rd_{\alp}((g^{-1})^{\alpha\beta}\sqrt{|\det g|})=\sqrt{|\det g|}\mathcal G^{\beta},
\end{equation}
\begin{equation}\label{wave.coord.alt}
(g^{-1})^{\alpha\beta}\rd_{\alp}g_{\bt\mu}-\frac 12 (g^{-1})^{\alp\bt}\rd_{\mu}g_{\alp\bt}=-g_{\mu\nu}\mathcal G^{\nu},
\end{equation}
or
$$\rd_{\alp}(g^{-1})^{\alpha\nu}-\frac 12g_{\alpha\beta}(g^{-1})^{\mu\nu}\rd_{\mu}(g^{-1})^{\alpha\beta}=\mathcal G^{\nu}. $$

We can write down the Einstein--scalar field equations under the generalized wave coordinate condition. We introduce the reduced wave operator
$$\tBox_g:=(g^{-1})^{\alp\bt}\rd^2_{\alp\bt}.$$
The components of the metric $g$ satisfy a system of wave equations whose principal part is the reduced wave operator. More precisely,
\begin{proposition}\label{Einstein.eqn.g}
Let $(g,\phi)$ be a solution to the Einstein--scalar field system together with the generalized wave coordinate condition 
$$\Box_g x^{\mu}=\mathcal G^\mu.$$ 
Then $g_{\mu\nu}$ and $\phi$ solve the following system of equations:
$$\tBox_g g_{\mu\nu}:=(g^{-1})^{\alp\bt}\rd_{\alp\bt}^2 g_{\mu\nu}=P(\rd_\mu g,\rd_\nu g)+Q_{\mu\nu}(\rd g,\rd g)+T_{\mu\nu}(\rd\phi,\rd\phi)+G_{\mu\nu}(g,\mathcal G),$$
and
$$\Box_g \phi=0,$$
where $P$ denotes the term 
$$P(\rd_\mu g,\rd_\nu g)=\frac 14 (g^{-1})^{\alp\alp'}\rd_\mu g_{\alp \alp'}(g^{-1})^{\bt \bt'}\rd_\nu g_{\bt \bt'}-\frac 12 (g^{-1})^{\alp\alp'}\rd_\mu g_{\alp \bt}(g^{-1})^{\bt \bt'}\rd_\nu g_{\alp' \bt'},$$
$Q$ is given by
\begin{equation*}
\begin{split}
Q_{\mu\nu}(\rd g,\rd g)=&\rd_{\alp}g_{\bt\mu} (g^{-1})^{\alp\alp'} (g^{-1})^{\bt\bt'}\rd_{\alp'}g_{\bt'\nu}-(g^{-1})^{\alp\alp'}(g^{-1})^{\bt\bt'}(\rd_{\alp}g_{\bt\mu}\rd_{\bt'}g_{\alp'\nu}-\rd_{\bt'}g_{\bt\mu}\rd_\alp g_{\alp'\nu})\\
&+(g^{-1})^{\alp\alp'}(g^{-1})^{\bt\bt'}(\rd_\mu g_{\alp'\bt'}\rd_\alp g_{\bt\nu}-\rd_\alp g_{\alp'\bt'}\rd_\mu g_{\bt\nu})\\
&+(g^{-1})^{\alp\alp'}(g^{-1})^{\bt\bt'}(\rd_\nu g_{\alp'\bt'}\rd_\alp g_{\bt\mu}-\rd_\alp g_{\alp'\bt'}\rd_\nu g_{\bt\mu})\\
&+\frac 12(g^{-1})^{\alp\alp'}(g^{-1})^{\bt\bt'}(\rd_{\bt'} g_{\alp\alp'}\rd_\mu g_{\bt\nu}-\rd_\mu g_{\alp\alp'}\rd_{\bt'} g_{\bt\nu})\\
&+\frac 12(g^{-1})^{\alp\alp'}(g^{-1})^{\bt\bt'}(\rd_{\bt'} g_{\alp\alp'}\rd_\nu g_{\bt\mu}-\rd_\nu g_{\alp\alp'}\rd_{\bt'}g_{\bt\mu}),
\end{split}
\end{equation*}
$T$ is the following term involving the scalar field
$$T_{\mu\nu}(\rd\phi,\rd\phi)=-4\rd_{\mu}\phi\rd_\nu\phi,$$
and $G$ is the following term which arises from the choice of gauge condition:
\begin{equation}\label{G.def}
\begin{split}
G_{\mu\nu}(g,\mathcal G)=&-\rd_\mu(g_{\nu\lambda} \mathcal G^{\lambda})-\rd_\nu(g_{\mu\lambda} \mathcal G^{\lambda})-\mathcal G^{\alp}\rd_\alp g_{\mu \nu}- g_{\alp\nu} g_{\bt\mu} \mathcal G^\alp \mathcal G^{\bt}.
\end{split}
\end{equation}
\end{proposition}
\begin{proof}
First, notice that the Einstein--scalar field system \eqref{Einstein.scalar.field} is equivalent to 
\begin{equation}\label{Einstein.scalar.field.2}
Ric_{\mu\nu}=2 \rd_\mu\phi\rd_\nu\phi.
\end{equation}
To see this, we take the trace of \eqref{Einstein.scalar.field} to get
$$-R=2\rd^\alp\phi\rd_\alp\phi-4 \rd^{\alp}\phi\rd_{\alp}\phi=-2 \rd^{\alp}\phi\rd_{\alp}\phi.$$
Hence,
$$Ric_{\mu\nu}- g_{\mu\nu} \rd^{\alp}\phi\rd_{\alp}\phi=2 \mathbb T_{\mu\nu}=2\rd_\mu\phi\rd_\nu\phi-g_{\mu\nu} \rd^{\alp}\phi\rd_{\alp}\phi,$$
which implies the conclusion.

Now, in order to derive the equation, we simply need to write out the Ricci curvature in terms of the metric and then insert the generalized wave coordinate condition. This is similar to the derivation in \cite{LR1}. Since in our setting we have a different gauge condition, we include the proof for completeness.

Define\footnote{We use boldface  ${\bf \Gamma}$ to denote Christoffel symbols throughout and reserve the notation $\Gamma$ for the Minkowskian commuting vector fields (see Definition \ref{def.Mink.vf}).}
${\bf \Gamma}^\lambda_{\mu\nu}:= \frac{1}{2} (g^{-1})^{\lambda\sigma}(\rd_{\mu}g_{\nu\sigma}+\rd_{\nu}g_{\mu\sigma}-\rd_\sigma g_{\mu\nu})$. The Ricci curvature is given by
\begin{equation}\label{Ric.def}
Ric_{\mu\nu}=\underbrace{\rd_{\alp}{\bf \Gamma}^\alp_{\mu\nu}-\rd_{\nu}{\bf \Gamma}^\alp_{\alp\mu}}_{=:I}+\underbrace{{\bf \Gamma}^{\alp}_{\mu\nu}{\bf \Gamma}^{\bt}_{\alp\bt} -{\bf \Gamma}^{\alp}_{\mu\bt} {\bf \Gamma}^{\bt}_{\nu\alp}}_{=:II}.
\end{equation}
By \eqref{wave.coord.alt}, we have
$$\rd_\mu\left((g^{-1})^{\alp\bt}\rd_{\alp} g_{\bt\nu}-\f 12 (g^{-1})^{\alp\bt}\rd_{\nu} g_{\alp\bt}\right)=-\rd_\mu(g_{\nu\alp}\mathcal G^{\alp}).$$
Therefore, using the identity $\rd (g^{-1})^{\alp\bt}=-(g^{-1})^{\alp\alp'}(g^{-1})^{\bt\bt'}\rd g_{\alp'\bt'}$, we have
\begin{equation*}
\begin{split}
&-\rd_\mu(g_{\nu\lambda}\mathcal G^{\lambda})\\
=&(g^{-1})^{\alp\bt}\rd^2_{\alp\mu} g_{\bt\nu}-\f 12 (g^{-1})^{\alp\bt}\rd^2_{\mu\nu} g_{\alp\bt}-(g^{-1})^{\alp\alp'}(g^{-1})^{\bt\bt'}\rd_\mu g_{\alp'\bt'}(\rd_{\alp} g_{\bt\nu}-\f 12\rd_{\nu} g_{\alp\bt}).
\end{split}
\end{equation*}
Expanding term $I$ in \eqref{Ric.def} using the above identity (and also its analogue with $\mu$ and $\nu$ switched), we have
\begin{equation}\label{Einstein.eqn.g.1}
\begin{split}
&\rd_{\alp}{\bf \Gamma}^\alp_{\mu\nu}-\rd_{\nu}{\bf \Gamma}^\alp_{\alp\mu}\\
=&\f 12(g^{-1})^{\alp\bt}(\rd^2_{\alp\mu} g_{\bt\nu}+\rd^2_{\alp\nu} g_{\bt\mu}-\rd^2_{\alp\bt} g_{\mu\nu})-\f 12 (g^{-1})^{\alp\bt}(\rd^2_{\nu\alp}g_{\bt\mu}+\rd^2_{\nu\mu}g_{\alp\bt}-\rd^2_{\nu\bt}g_{\alp\mu})\\
&-\f 12(g^{-1})^{\alp\alp'}(g^{-1})^{\bt\bt'}\rd_{\alp} g_{\alp'\bt'}(\rd_{\mu} g_{\bt\nu}+\rd_{\nu} g_{\bt\mu}-\rd_{\bt} g_{\mu\nu})\\
&+\f 12(g^{-1})^{\alp\alp'}(g^{-1})^{\bt\bt'}\rd_{\nu} g_{\alp'\bt'}(\rd_{\alp}g_{\bt\mu}+\rd_{\mu}g_{\alp\bt}-\rd_{\bt}g_{\alp\mu})\\
=&-\f 12(g^{-1})^{\alp\bt}\rd^2_{\alp\bt} g_{\mu\nu}-\f 12\rd_{\mu}(g_{\nu\lambda}\mathcal G^{\lambda})-\f 12\rd_{\nu}(g_{\mu\lambda}\mathcal G^{\lambda})\\
&+\f 12(g^{-1})^{\alp\alp'}(g^{-1})^{\bt\bt'}\rd_\mu g_{\alp'\bt'}(\rd_{\alp} g_{\bt\nu}-\f 12\rd_{\nu} g_{\alp\bt})+\f 12(g^{-1})^{\alp\alp'}(g^{-1})^{\bt\bt'}\rd_\nu g_{\alp'\bt'}(\rd_{\alp} g_{\bt\mu}-\f 12\rd_{\mu} g_{\alp\bt})\\
&-\f 12(g^{-1})^{\alp\alp'}(g^{-1})^{\bt\bt'}\rd_{\alp} g_{\alp'\bt'}(\rd_{\mu} g_{\bt\nu}+\rd_{\nu} g_{\bt\mu}-\rd_{\bt} g_{\mu\nu})\\
&+\f 12(g^{-1})^{\alp\alp'}(g^{-1})^{\bt\bt'}\rd_{\nu} g_{\alp'\bt'}(\rd_{\alp}g_{\bt\mu}+\rd_{\mu}g_{\alp\bt}-\rd_{\bt}g_{\alp\mu}).
\end{split}
\end{equation}
Collecting the quadratic terms in the derivatives of $g$ in \eqref{Einstein.eqn.g.1}, we get
\begin{equation}\label{Einstein.eqn.g.2}
\begin{split}
&\rd_{\alp}{\bf \Gamma}^\alp_{\mu\nu}-\rd_{\nu}{\bf \Gamma}^\alp{ }_{\alp\mu}\\
=&-\f 12(g^{-1})^{\alp\bt}\rd^2_{\alp\bt} g_{\mu\nu}-\f 12\rd_{\mu}(g_{\nu\lambda}\mathcal G^{\lambda})-\f 12\rd_{\nu}(g_{\mu\lambda}\mathcal G^{\lambda})+\f 12(g^{-1})^{\alp\alp'}(g^{-1})^{\bt\bt'}\rd_\mu g_{\alp'\bt'}\rd_{\alp} g_{\bt\nu}\\
&+\f 12(g^{-1})^{\alp\alp'}(g^{-1})^{\bt\bt'}\rd_\nu g_{\alp'\bt'}(2\rd_{\alp}g_{\bt\mu}-\rd_{\bt}g_{\alp\mu})\\
&-\f 12(g^{-1})^{\alp\alp'}(g^{-1})^{\bt\bt'}\rd_{\alp} g_{\alp'\bt'}(\rd_{\mu} g_{\bt\nu}+\rd_{\nu} g_{\bt\mu}-\rd_{\bt} g_{\mu\nu})\\
=&-\f 12(g^{-1})^{\alp\bt}\rd^2_{\alp\bt} g_{\mu\nu}-\f 12\rd_{\mu}(g_{\nu\lambda}\mathcal G^{\lambda})-\f 12\rd_{\nu}(g_{\mu\lambda}\mathcal G^{\lambda})\\
&+\f 12(g^{-1})^{\alp\alp'}(g^{-1})^{\bt\bt'}(\rd_\mu g_{\alp'\bt'}\rd_{\alp} g_{\bt\nu}-\rd_\alp g_{\alp'\bt'}\rd_{\mu} g_{\bt\nu})\\
&+\f 12(g^{-1})^{\alp\alp'}(g^{-1})^{\bt\bt'}(\rd_\nu g_{\alp'\bt'}\rd_{\alp}g_{\bt\mu}-\rd_\alp g_{\alp'\bt'}\rd_{\nu}g_{\bt\mu})\\
&+\underbrace{\f 12(g^{-1})^{\alp\alp'}(g^{-1})^{\bt\bt'}\rd_\alp g_{\alp'\bt'}\rd_{\bt} g_{\mu\nu}}_{=:III}.
\end{split}
\end{equation}
Expanding term $II$ \eqref{Ric.def}, we get
\begin{equation}\label{Einstein.eqn.g.3}
\begin{split}
&{\bf \Gamma}^{\alp}_{\mu\nu}{\bf \Gamma}^{\bt}_{\alp\bt} -{\bf \Gamma}^{\alp}_{\mu\bt} {\bf \Gamma}^{\bt}_{\alp\nu}\\
=&\f 14(g^{-1})^{\alp\alp'}(g^{-1})^{\bt\bt'}(\rd_{\mu}g_{\alp'\nu}+\rd_\nu g_{\alp'\mu}-\rd_{\alp'}g_{\mu\nu})\rd_{\alp}g_{\bt\bt'}\\
&-\f 14(g^{-1})^{\alp\alp'}(g^{-1})^{\bt\bt'}(\rd_{\mu}g_{\alp'\bt}+\rd_\bt g_{\alp'\mu}-\rd_{\alp'}g_{\mu\bt})(\rd_{\alp}g_{\nu\bt'}+\rd_\nu g_{\alp\bt'}-\rd_{\bt'} g_{\alp\nu})\\
=&\f 14(g^{-1})^{\alp\alp'}(g^{-1})^{\bt\bt'}(\rd_{\mu}g_{\alp'\nu}\rd_{\alp}g_{\bt\bt'}+\rd_\nu g_{\alp'\mu}\rd_{\alp}g_{\bt\bt'}-\rd_{\alp'}g_{\mu\nu}\rd_{\alp}g_{\bt\bt'})\\
&-\f 14(g^{-1})^{\alp\alp'}(g^{-1})^{\bt\bt'}(\rd_{\mu}g_{\alp'\bt}\rd_{\alp}g_{\nu\bt'}+\rd_\bt g_{\alp'\mu}\rd_{\alp}g_{\nu\bt'}-\rd_{\alp'}g_{\mu\bt}\rd_{\alp}g_{\nu\bt'})\\
&-\f 14(g^{-1})^{\alp\alp'}(g^{-1})^{\bt\bt'}(\rd_{\mu}g_{\alp'\bt}\rd_\nu g_{\alp\bt'}+\rd_\bt g_{\alp'\mu}\rd_\nu g_{\alp\bt'}-\rd_{\alp'}g_{\mu\bt}\rd_\nu g_{\alp\bt'})\\
&+\f 14(g^{-1})^{\alp\alp'}(g^{-1})^{\bt\bt'}(\rd_{\mu}g_{\alp'\bt}\rd_{\bt'} g_{\alp\nu}+\rd_\bt g_{\alp'\mu}\rd_{\bt'} g_{\alp\nu}-\rd_{\alp'}g_{\mu\bt}\rd_{\bt'} g_{\alp\nu})\\
=&\underbrace{\f 14(g^{-1})^{\alp\alp'}(g^{-1})^{\bt\bt'}(\rd_{\mu}g_{\alp'\nu}\rd_{\alp}g_{\bt\bt'}+\rd_\nu g_{\alp'\mu}\rd_{\alp}g_{\bt\bt'}-\rd_{\alp'}g_{\mu\nu}\rd_{\alp}g_{\bt\bt'})}_{=:IV}\\
&\underbrace{-\f 12(g^{-1})^{\alp\alp'}(g^{-1})^{\bt\bt'}\rd_{\alp'}g_{\mu\bt}\rd_{\bt'} g_{\alp\nu}}_{=:V}+\f 12(g^{-1})^{\alp\alp'}(g^{-1})^{\bt\bt'}\rd_{\alp'}g_{\mu\bt}\rd_{\alp}g_{\nu\bt'}\\
&-\f 14(g^{-1})^{\alp\alp'}(g^{-1})^{\bt\bt'}\rd_{\mu}g_{\alp'\bt}\rd_\nu g_{\alp\bt'}.
\end{split}
\end{equation}
Further expanding $III$ in \eqref{Einstein.eqn.g.2} using \eqref{wave.coord.alt}, we get
\begin{equation}\label{Einstein.eqn.g.4}
\begin{split}
III=&\f 12(g^{-1})^{\alp\alp'}(g^{-1})^{\bt\bt'}\rd_\alp g_{\alp'\bt'}\rd_{\bt} g_{\mu\nu}=\f 12(g^{-1})^{\bt\bt'}(\f 12(g^{-1})^{\alp\alp'}\rd_{\bt'} g_{\alp\alp'}-g_{\bt'\alp}\mathcal G^{\alp})\rd_{\bt} g_{\mu\nu}\\
=&\f 14(g^{-1})^{\alp\alp'}(g^{-1})^{\bt\bt'}\rd_{\bt'} g_{\alp\alp'}\rd_{\bt} g_{\mu\nu}-\f 12\mathcal G^{\bt}\rd_{\bt} g_{\mu\nu}.
\end{split}
\end{equation}
Using \eqref{wave.coord.alt}, $IV$ in \eqref{Einstein.eqn.g.3} can be expressed as
\begin{equation}\label{Einstein.eqn.g.5}
\begin{split}
IV=&\f 14(g^{-1})^{\alp\alp'}(g^{-1})^{\bt\bt'}(\rd_{\mu}g_{\alp'\nu}\rd_{\alp}g_{\bt\bt'}+\rd_\nu g_{\alp'\mu}\rd_{\alp}g_{\bt\bt'}-\rd_{\alp'}g_{\mu\nu}\rd_{\alp}g_{\bt\bt'})\\
=&\f 14(g^{-1})^{\alp\alp'}(g^{-1})^{\bt\bt'}(\rd_{\mu}g_{\alp'\nu}\rd_{\alp}g_{\bt\bt'}-\rd_{\alp}g_{\alp'\nu}\rd_{\mu}g_{\bt\bt'})\\
&+\f 14(g^{-1})^{\alp\alp'}(g^{-1})^{\bt\bt'}(\rd_\nu g_{\alp'\mu}\rd_{\alp}g_{\bt\bt'}-\rd_\alp g_{\alp'\mu}\rd_{\nu}g_{\bt\bt'})\\
&+\f 14(g^{-1})^{\alp\alp'}(g^{-1})^{\bt\bt'}(\rd_{\alp}g_{\alp'\nu}\rd_{\mu}g_{\bt\bt'}+\rd_\alp g_{\alp'\mu}\rd_{\nu}g_{\bt\bt'}-\rd_{\alp'}g_{\mu\nu}\rd_{\alp}g_{\bt\bt'})\\
=&\f 14(g^{-1})^{\alp\alp'}(g^{-1})^{\bt\bt'}(\rd_{\mu}g_{\alp'\nu}\rd_{\alp}g_{\bt\bt'}-\rd_{\alp}g_{\alp'\nu}\rd_{\mu}g_{\bt\bt'})\\
&+\f 14(g^{-1})^{\alp\alp'}(g^{-1})^{\bt\bt'}(\rd_\nu g_{\alp'\mu}\rd_{\alp}g_{\bt\bt'}-\rd_\alp g_{\alp'\mu}\rd_{\nu}g_{\bt\bt'})+\f 14(g^{-1})^{\alp\alp'}(g^{-1})^{\bt\bt'}\rd_\nu g_{\alp\alp'}\rd_{\mu}g_{\bt\bt'}\\
&+\f 14(g^{-1})^{\bt\bt'}(-g_{\nu\alp}\mathcal G^\alp \rd_{\mu}g_{\bt\bt'}-g_{\mu\alp}\mathcal G^\alp\rd_{\nu}g_{\bt\bt'})-\f 14(g^{-1})^{\alp\alp'}(g^{-1})^{\bt\bt'}\rd_{\alp'}g_{\mu\nu}\rd_{\alp}g_{\bt\bt'}.
\end{split}
\end{equation}
We then expand the term $V$ in \eqref{Einstein.eqn.g.3} using \eqref{wave.coord.alt}:
\begin{equation}\label{Einstein.eqn.g.6}
\begin{split}
V=&-\f 12(g^{-1})^{\alp\alp'}(g^{-1})^{\bt\bt'}\rd_{\alp'}g_{\mu\bt}\rd_{\bt'} g_{\alp\nu}\\
=&-\f 12(g^{-1})^{\alp\alp'}(g^{-1})^{\bt\bt'}(\rd_{\alp'}g_{\mu\bt}\rd_{\bt'} g_{\alp\nu}-\rd_{\bt'}g_{\mu\bt}\rd_{\alp'} g_{\alp\nu}) \\
&-\f 12(\f 12(g^{-1})^{\bt\bt'}\rd_{\mu}g_{\bt\bt'}-g_{\mu\bt}\mathcal G^{\bt})(\f 12(g^{-1})^{\alp\alp'}\rd_{\nu} g_{\alp\alp'}-g_{\nu\alp}\mathcal G^{\alp})\\
=&-\f 12(g^{-1})^{\alp\alp'}(g^{-1})^{\bt\bt'}(\rd_{\alp'}g_{\mu\bt}\rd_{\bt'} g_{\alp\nu}-\rd_{\bt'}g_{\mu\bt}\rd_{\alp'} g_{\alp\nu}) -\f 18(g^{-1})^{\alp\alp'}(g^{-1})^{\bt\bt'}\rd_{\mu}g_{\bt\bt'}\rd_{\nu} g_{\alp\alp'}\\
&+\f 14(g^{-1})^{\bt\bt'}\rd_{\mu}g_{\bt\bt'}g_{\nu\alp}\mathcal G^{\alp}+\f 14(g^{-1})^{\alp\alp'}\rd_{\nu}g_{\alp\alp'}g_{\mu\bt}\mathcal G^{\bt}-\f 12g_{\nu\alp} g_{\mu\bt} \mathcal G^{\bt} \mathcal G^{\alp}.
\end{split}
\end{equation}
Adding \eqref{Einstein.eqn.g.4}, \eqref{Einstein.eqn.g.5} and \eqref{Einstein.eqn.g.6} yields
\begin{equation}\label{Einstein.eqn.g.7}
\begin{split}
&III+IV+V\\
=&-\f 12\mathcal G^{\bt}\rd_{\bt} g_{\mu\nu}+\f 18(g^{-1})^{\alp\alp'}(g^{-1})^{\bt\bt'}\rd_\nu g_{\alp\alp'}\rd_{\mu}g_{\bt\bt'}\\
&+\f 14(g^{-1})^{\alp\alp'}(g^{-1})^{\bt\bt'}(\rd_{\mu}g_{\alp'\nu}\rd_{\alp}g_{\bt\bt'}-\rd_{\alp}g_{\alp'\nu}\rd_{\mu}g_{\bt\bt'})\\
&+\f 14(g^{-1})^{\alp\alp'}(g^{-1})^{\bt\bt'}(\rd_\nu g_{\alp'\mu}\rd_{\alp}g_{\bt\bt'}-\rd_\alp g_{\alp'\mu}\rd_{\nu}g_{\bt\bt'})\\
&-\f 12(g^{-1})^{\alp\alp'}(g^{-1})^{\bt\bt'}(\rd_{\alp'}g_{\mu\bt}\rd_{\bt'} g_{\alp\nu}-\rd_{\bt'}g_{\mu\bt}\rd_{\alp'} g_{\alp\nu})-\f 12g_{\nu\alp} g_{\mu\bt} \mathcal G^{\bt} \mathcal G^{\alp}. 
\end{split}
\end{equation}
Combining \eqref{Einstein.scalar.field.2}, \eqref{Ric.def}, \eqref{Einstein.eqn.g.2}, \eqref{Einstein.eqn.g.3} and \eqref{Einstein.eqn.g.7}, multiplying by $2$ and re-arranging yield the desired result.
\end{proof}

Given an $(\ep,\gamma,N)$-admissible perturbation of a dispersive spacetime solution $(\mathcal M,g_B,\phi_B)$ of size $(C,\gamma, N)$, we show that we can impose the condition \eqref{wave.coord} on the initial data.
\begin{proposition}\label{wave.coord.impose}
Given a dispersive spacetime solution $(\mathcal M,g_B,\phi_B)$ of size $(C,\gamma,N)$, there exist $\ep=\ep(C,\gamma,N)$ such that for every $(\ep,\gamma,N)$-admissible perturbation, one can prescribe the spacetime metric $g$ and its time derivative $\rd_t g$ on the hypersurface $\{t=0\}$ such that the following hold:
\begin{enumerate}
\item $g$ can be decomposed as
$$g=m+h_B+h_S+h,$$
\item the restriction of $g$ on the tangent space of $\{t=0\}$ coincides with $\hat{g}$;
\item the second fundamental form of the hypersurface $\{t=0\}$ coincides with $\hat{k}$;
\item $g$ satisfies the generalized wave coordinate condition \eqref{wave.coord} initially on $\{t=0\}$;
\item $h$ obeys the following estimates on the initial hypersurface $\{t=0\}$:
\begin{equation}\label{h.data.bd.1}
\|(1+r)^{\f12+|I|+\gamma}\rd\rd^I h\|_{L^2(\{t=0\})}\ls \ep.
\end{equation}
\end{enumerate}
\end{proposition}
\begin{proof}
Recall our convention that $i,j,m,n=1,2,3$. The initial data for $g_{ij}$ is given by the condition 
$$g_{ij}\restriction_{\{t=0\}}=\hat{g}_{ij}.$$ 
We thus need to impose $g_{00}$, $g_{0i}$ and $\rd_t g_{\mu\nu}$ in a way that the generalized wave condition is verified and that the second fundamental form of the solution $g_{\mu\nu}$ coincides with the prescribed $\hat{k}$. To this end, we first impose\footnote{Recall that $(h_S)_{0i}=0$}
\begin{equation}\label{g0.def}
g_{00}\restriction_{\{t=0\}}=(g_B)_{00}\restriction_{\{t=0\}}+(h_S)_{00},\quad g_{0i}\restriction_{\{t=0\}}=(g_B)_{0i}\restriction_{\{t=0\}}.
\end{equation}
We then impose the condition\footnote{The second fundamental form is given by $\hat{k}_{ij}:=g(D_{\rd_i} {\bf N}, \rd_j)=\f 12 (\mathcal L_{\bf N} g)_{ij}$. Moreover, using the expression for ${\bf N}$ below, $2 \hat{k}_{ij}$ can also be computed by $\f{1}{\sqrt{-g_{00}+(\hat{g}^{-1})^{mn}g_{0m}g_{0n}}}(\mathcal L_{(\rd_t-g_{0m}(\hat{g}^{-1})^{mn}\rd_n)} g)_{ij}$.}
\begin{equation}\label{k.cond}
\begin{split}
2\hat{k}_{ij}=&{\bf N} g_{ij}\restriction_{\{t=0\}}-\f{g_{i0,j}+g_{j0,i}}{\sqrt{-g_{00}+(\hat{g}^{-1})^{mn}g_{0m}g_{0n}}}\restriction_{\{t=0\}}
+ \f{(\hat{g}^{-1})^{kl}g_{0l}(g_{ik,j}+g_{jk,i})}{\sqrt{-g_{00}+(\hat{g}^{-1})^{mn}g_{0m}g_{0n}}}\restriction_{\{t=0\}},
\end{split}
\end{equation}
where the future-directed unit normal ${\bf N}$ is given by 
$${\bf N}=\f{1}{\sqrt{-g_{00}+(\hat{g}^{-1})^{mn}g_{0m}g_{0n}}}(\rd_t-g_{0j}(\hat{g}^{-1})^{ij}\rd_i).$$
Notice that by \eqref{g0.def} and the fact that $\hat{g}$ is Riemannian, we have $-g_{00}+(\hat{g}^{-1})^{ij}g_{0i}g_{0j}>0$ and thus ${\bf N}$ is well-defined. In particular, \eqref{k.cond} determines $\rd_tg_{ij}\restriction_{\{t=0\}}$ (since all remaining terms are previous prescribed).
Finally, the $\rd_t$ derivatives of $g_{0\nu}$ can then given by \eqref{wave.coord.alt}. More precisely, by taking $\mu=i$ in \eqref{wave.coord.alt}, we show that $\rd_t g_{0i}$ can be defined by\footnote{Recall that $\mathcal G^\nu$ is defined in \eqref{wave.coord} and \eqref{wave.coord.2}.}
\begin{equation}\label{dtg0i.def}
\begin{split}
&(g^{-1})^{00}\rd_t g_{0 i}\\
=&-(g^{-1})^{0j}\rd_t g_{ij}-(g^{-1})^{0j}\rd_j g_{0i}-(g^{-1})^{jk}\rd_j g_{ki}+\frac 12 (g^{-1})^{\alp\bt}\rd_i g_{\alp\bt}-g_{i\nu}\mathcal G^{\nu}.
\end{split}
\end{equation}
By taking $\mu=0$ in \eqref{wave.coord.alt}, $\rd_t g_{00}$ can be defined by
\begin{equation}\label{dtg00.def}
\begin{split}
&\f12 (g^{-1})^{00}\rd_t g_{00}\\
=&-(g^{-1})^{0i}\rd_t g_{i0}- (g^{-1})^{0i}\rd_i g_{00}- (g^{-1})^{ij}\rd_ig_{j0}+ \f12(g^{-1})^{ij}\rd_t g_{ij}+ (g^{-1})^{0i}\rd_t g_{0i}-g_{0\nu}\mathcal G^{\nu}.
\end{split}
\end{equation}
By construction, the conditions (1)-(4) in the statement of the proposition hold. We now check that $h$ defined above satisfies the bound \eqref{h.data.bd.1}. In the case that all derivatives falling on $h$ are spatial derivatives, the bound for $h_{ij}$ follows from \eqref{h.data.bd} and that for $h_{0\mu}$ follows from \eqref{g0.def}.

In the case that exactly one time derivative falls on $h$, we take \eqref{k.cond}, \eqref{dtg0i.def}, \eqref{dtg00.def} and subtract off the corresponding equation for the background metric $g_B=m+h_B$. (Recall that $\mathcal G^\mu=\mathcal G_B^\mu+\mathcal G_S^{\mu}$, where $\mathcal G_B^\mu$ is given exactly by the generalized wave coordinate condition for the background metric $g_B$.) For $\rd_th_{ij}$ and $\rd_t h_{00}$, the desired estimate then follows from \eqref{h.data.bd}, \eqref{k.data.bd}, the estimate $|\rd^I\mathcal G_S|\ls \f{\ep\chi(r)\log r}{(1+r)^{2+|I|}}$ for $\mathcal G^i_S$ and the fact that $\mathcal G^0_S=0$.

For $\rd_t h_{0i}$, however, we note that each of the terms in \eqref{dtg0i.def} does not have sufficient decay in $r$. Nevertheless, by the definition of $h_S$ and $\mathcal G_S$, we have the exact cancellation that allows us to conclude\footnote{This fact will be proven in Proposition \ref{prop.S.gauge.cancel}.}
\begin{equation}\label{hS.gauge.cancel}
\left|-(m^{-1})^{jk}\rd_j (h_S)_{ki}+\frac 12 (m^{-1})^{\alp\bt}\rd_i (h_S)_{\alp\bt}-m_{i\nu}\mathcal G_S^{\nu}\right|\ls \f{\ep\log (2+r)}{(1+r)^3}.
\end{equation}
Similar cancellations occur for higher spatial derivatives, see Proposition \ref{prop.S.gauge.cancel}. This allows us to obtain the desired estimate for $\rd_t h_{0i}$ and its weighted spatial derivatives.

Finally, in the cases where there are at least two time derivatives falling on $h$, we need to use the equation in Proposition \ref{Einstein.eqn.g} to express time derivatives in terms of spatial derivatives. The calculations are largely similar to that in Section \ref{sec.eqn} and will be omitted. It is important to note, however, that as above, we need a crucial cancellation between the most slowly decaying terms (see Proposition \ref{gS.prelim}).
\end{proof}

\section{Third version of main theorem}\label{sec.third.thm}

In Proposition \ref{Einstein.eqn.g}, we have shown that under the generalized wave coordinate condition \eqref{wave.coord}, the Einstein scalar field system reduces to a system of quasilinear wave equations. As is well-known, if we solve the system of quasilinear wave equations in Proposition \ref{Einstein.eqn.g} with initial data satisfying the constraint equations and the generalized wave coordinate conditions, then the constraint equations and the generalized wave coordinate conditions are propagated and the solution is therefore a genuine solution to the Einstein scalar field system.

In view of this, we can now rephrase our main theorem (Theorem \ref{main.thm.2}) as a result on the global existence of solutions to a system of quasilinear wave equations. Before we proceed, we need one more notation:
\begin{definition}\label{w.def}
Let\footnote{Notice that $w$ as defined is Lipschitz.}
$$w(q)=\begin{cases}
1+(1+|q|)^{1+2\gamma} &\mbox{if }q\geq 0\\
1+(1+|q|)^{-\f{\gamma}{2}} &\mbox{if }q< 0.
\end{cases}$$ 
\end{definition}
Our main theorem can be stated as follows\footnote{For more details on how Theorem \ref{main.thm.2} implies Theorem \ref{main.thm}, see Remark \ref{2.implies.3} below.}:
\begin{theorem}[Main theorem, third version]\label{main.thm}
Let $N\geq 11$, $0<\gamma\leq \f 18$ and $(\mathcal M,g_B,\phi_B)$ be a dispersive spacetime solution of size $(C,\gamma,N)$. Then there exists $\ep_0=\ep_0(C,\gamma,N)$ such that if the initial data to the system of equations in Proposition \ref{Einstein.eqn.g} together with the wave equation 
$$\Box_g \phi=0$$ 
satisfy the constraint equations, the generalized wave coordinate condition \eqref{wave.coord} and the smallness assumptions 
$$\sum_{|I|\leq N}\|(1+r)^{\f12+|I|+\gamma}\rd\nab^I h\|_{L^2(\{t=0\})}\leq \ep$$
and
\begin{equation*}
\begin{split}
&\sum_{|I|\leq N}\|(1+r)^{\f12+\gamma+|I|}\nab\nab^I (\hat\phi-(\phi_B)\restriction_{\{t=0\}})\|_{L^2(\{t=0\})}\\
&\quad+\sum_{|I|\leq N-1}\|(1+r)^{\f12+\gamma+|I|}\nab^I (\hat\psi-(\rd_t\phi_B)\restriction_{\{t=0\}})\|_{L^2(\{t=0\})}\leq \ep
\end{split}
\end{equation*}
with $|x|^{1+\gamma}|h|,\,|x|^{1+\gamma}|\hat\phi|\to 0$ as $|x|\to \infty$ and $\ep<\ep_0$,
then the unique solution to this system of equations is global in time. Moreover, for $\beta:=\phi-\phi_B$, the solution obeys the following estimates:
\begin{equation}\label{main.thm.bound}
\sum_{|I|\leq N}\big(\|w^{\f12}\rd\Gamma^I h\|_{L^2(\{t=\tau\})}+\|w^{\f12}\rd\Gamma^I \beta\|_{L^2(\{t=\tau\})}\big)\ls \ep (1+\tau)^{\de_1},
\end{equation}
where $\de_1>0$ is a small constant which can be chosen to be arbitrary small as long as $\ep$ is also chosen to be accordingly small. The implicit constant in the above estimate depends on $C$, $\gamma$ and $\de_1$.
\end{theorem} 

\begin{remark}[Initial bounds in $L^2$ for the $\Gamma^I$ derivatives]
By using the equation \eqref{Einstein.eqn.g} for $g_{\mu\nu}$ and the equation $\Box_g \phi=0$ for $\phi$, it is easy to obtain also the $L^2$ bound for the $\Gamma$ derivatives,
\begin{equation*}
\begin{split}
\sum_{|I|\leq N}\|(1+r)^{\f12+\gamma}\rd\Gamma^I h\|_{L^2(\{t=0\})}
&+\sum_{|I|\leq N}\|(1+r)^{\f12+\gamma}\rd\Gamma^I \beta\|_{L^2(\{t=0\})}\ls \ep,
\end{split}
\end{equation*}
where the implicit constant is independent of $\ep$ as long as $\ep\leq \tilde{\ep}_0$, with $\tilde{\ep}_0$ depending only on $C$, $\gamma$ and $N$. In order to derive this bound, we need in particular to use the properties of $h_B$ (at $\{t=0\}$) given in Definition \ref{def.dispersivespt} and the properties of $h_S$ proven in Proposition \ref{gS.prelim}.
\end{remark}

\begin{remark}[Initial bounds in $L^\infty$]\label{rmk.initial.pointwise}
As an immediate consequence of the assumptions above and standard Gagliardo--Nirenberg theorems, we have
$$\sum_{|I|\leq N-2}\left(\|(1+r)^{1+\gamma}\Gamma^I h\|_{L^\infty(\{t=0\})}+\|(1+r)^{1+\gamma}\Gamma^I \beta\|_{L^\infty(\{t=0\})}\right) \lesssim \ep.$$
\end{remark}

\begin{remark}[Theorem \ref{main.thm} implies Theorem \ref{main.thm.2}]\label{2.implies.3}

Given an $(\ep,\gamma,N)$-admissible perturbation in the sense of Definition \ref{def.pert}, by Proposition \ref{wave.coord.impose}, the initial data for $h$ can be imposed to satisfy the generalized wave coordinate condition \eqref{wave.coord} such that the bounds in the assumptions of Theorem \ref{main.thm} are verified (after potentially changing $\ep$ by a constant factor). The definition of the initial data set also guarantees that the constraint equations are satisfied.

By Theorem \ref{main.thm}, we therefore have a solution which is global-in-$t$-time in the generalized wave coordinate system $(t,x^1,x^2,x^3)$. Finally, it remains to show that the spacetime is indeed future causally geodesically complete. However, this step is standard given the estimates that are established in Theorem \ref{main.thm}. We omit the details and refer the readers for example to \cite{LOY}.
\end{remark}

\section{Main bootstrap assumptions}\label{sec.BA}
We begin the proof of Theorem \ref{main.thm}. The proof proceeds via a bootstrap argument. By standard methods (see also Theorem \ref{Cauchy.stability}), we know that that there exists a local solution and it suffices to prove a priori estimates for $h_{\alp\bt}:=g_{\alp\bt}-m_{\alp\bt}-(h_B)_{\alp\bt}-(h_S)_{\alp\bt}$ and $\beta:=\phi-\phi_B$. 

{\bf Fix $\de_0>0$ be a small constant satisfying}\footnote{Notice in particular that $\de_0\leq \f{1}{88000}$. As we will show below, $\de_0$ can be taken to be arbitrarily small, as long as $\ep$ (and other parameters in the problem that we will introduce later) is chosen to be smaller accordingly.}
\begin{equation}\label{de_0.def}
0<\de_0\leq \f{\gamma}{1000 N}. 
\end{equation}
We will make bootstrap assumptions on the $L^\infty$ norms of $h$ and its derivatives. First, we assume the following for the derivatives of $h$:
\begin{equation}\label{BA1}
\sup_{t,x}\sum_{|I|\leq \lfloor\frac{N}{2}\rfloor} (1+s)^{1-\delta_0}(1+|q|)^{\frac 12-\f\gamma 4}w(q)^{\frac 12}|\rd \Gamma^I h (t,x)|\leq \ep^{\frac 12}
\end{equation}
and
\begin{equation}\label{BA2}
\sup_{t,x}\sum_{|I|\leq \lfloor\frac{N}{2}\rfloor} (1+s)^{2-\delta_0}(1+|q|)^{-\frac 12-\f\gamma 4}w(q)^{\frac 12}|\bar\rd \Gamma^I h (t,x)|\leq \ep^{\frac 12},
\end{equation}
where $w(q)$ is defined in Definition \ref{w.def}. For the first derivative of $h$, we further assume a refined estimate for all but one component. More precisely, we make the bootstrap assumption\footnote{Here, we recall the definition in \eqref{def.proj}.}
\begin{equation}\label{BA3}
\sup_{t,x}(1+s)|\rd h(t,x)|_{\mathcal T\mathcal U}\leq \ep^{\frac 12}.
\end{equation}
We then introduce the following bootstrap assumptions on the $L^\infty$ norm of $\Gamma^I h$:
\begin{equation}\label{BA4}
\sup_{t,x}\sum_{|I|\leq \lfloor\frac{N}{2}\rfloor} (1+s)^{1-\delta_0}(1+|q|)^{-\frac 12-\f\gamma 4}w(q)^{\frac 12}|\Gamma^I h (t,x)|\leq \ep^{\frac 12}.
\end{equation}
Finally, we introduce\footnote{Again, we recall \eqref{def.proj} for definition of the notations.} the bootstrap assumptions for $h_{L\mathcal T}$ and $(\Gamma h)_{LL}$:
\begin{equation}\label{BA5}
\sup_{t,x}(1+s)^{1+\frac{\gamma}{2}}(1+|q|)^{-\frac 12-\gamma}w(q)^{\frac 12}\left(|h (t,x)|_{L\mathcal T}+\sum_{|I|\leq 1}|\Gamma^I h(t,x)|_{LL}\right)\leq \ep^{\frac 12}.
\end{equation}
In addition to the bootstrap assumptions for the metric components, we also need bootstrap assumptions for the scalar field. More precisely, for $\beta:=\phi-\phi_B$, we make the following bootstrap assumptions, which can be thought of as analogues of \eqref{BA1}, \eqref{BA2} and \eqref{BA3}: 
\begin{equation}\label{BASF1}
\sup_{t,x}\sum_{|I|\leq \lfloor\frac{N}{2}\rfloor} (1+s)^{1-\delta_0}(1+|q|)^{\frac 12-\f\gamma 4}w(q)^{\frac 12}|\rd \Gamma^I \beta (t,x)|\leq \ep^{\frac 12}
\end{equation}
and
\begin{equation}\label{BASF2}
\sup_{t,x}\sum_{|I|\leq \lfloor\frac{N}{2}\rfloor} (1+s)^{2-\delta_0}(1+|q|)^{-\frac 12-\f\gamma 4}w(q)^{\frac 12}|\bar\rd \Gamma^I \beta (t,x)|\leq \ep^{\frac 12}
\end{equation}
and 
\begin{equation}\label{BASF3}
\sup_{t,x}(1+s)|\rd\beta(t,x)|\leq \ep^{\f12},
\end{equation}
where $w(q)$ is defined in Definition \ref{w.def}.

We will improve all of the above bootstrap assumptions, i.e.~we show that \eqref{BA1}-\eqref{BASF3} all hold with a better constant (see Proposition \ref{BS.close} at the very end of the proof of Theorem \ref{main.thm}). 

\begin{remark}[Choice of parameters $\gamma$, $\de_0$, $\de$, $T$, $U_2$, $U_3$, $\ep$]\label{rmk.smallness}
We now discuss the choice of various smallness parameters in the problem, some of which have already been introduced above, and the order in which they are chosen. $\gamma$ is given in Theorem \ref{main.thm} and is required to obey $0<\gamma\leq \f 18$. $\de_0$, which appears in the bootstrap assumptions above and is used to measure the ``loss'' in the decay rate, is chosen to satisfy\footnote{$\de_0$ can in fact be chosen arbitrarily small, as long as the constants $U_2$, $U_3$, $T$ and $\ep$ are then chosen accordingly.} \eqref{de_0.def}. It will be considered a fixed constant in the course of the proof. After $\de_0$ is fixed, we choose $U_2$ and $U_3$, which will appear in the definition of the partition of the spacetime (see Section \ref{def.regions}), so that $|U_2|$ and $|U_3|$ are large. After that we choose $T$ to be large. We then choose $\de$, which is used for the improved pointwise estimates (see Proposition \ref{imp.decay}), to be small. Finally, we choose $\ep$, which measures the size of the data, to be small.
\end{remark}

\section{Preliminary bounds}\label{sec.prelim.bounds}

In this section, we collect some preliminary bounds. After stating some standard facts regarding the Minkowskian commuting vector fields (Propositions \ref{decay.weights} and \ref{Gamma.commute}), we then turn to some estimates for $g^{-1}$ and its derivatives using the assumptions on the background metric (Definition \ref{def.dispersivespt}) and the bootstrap assumptions.

We now turn to the details. First, we have the following proposition regarding the decay that one can obtain using the Minkowskian commuting vector fields:
\begin{proposition}\label{decay.weights}
For every scalar function $\xi$, 
$$(1+t+|q|)|\bar{\rd}\xi|(t,x)+(1+|q|)|\rd\xi|(t,x)\ls \sum_{|I|=1}|\Gamma^I\xi|(t,x).$$
\end{proposition}
\begin{proof}
See Lemma 5.1 in \cite{LR2}.
\end{proof}
Next, we have a proposition regarding the commutation of $\rd$ and $\Gamma$:
\begin{proposition}\label{Gamma.commute}
$$[\rd_\mu,\Gamma]=^{(\Gamma)}c_{\mu}{ }^{\nu} \rd_\nu,$$ 
where $|{ }^{(\Gamma)}c|\ls 1$ and $^{(\Gamma)}c_{LL}=0$ for every Minkowskian commuting vector field $\Gamma$.
\end{proposition}
\begin{proof}
This is a direct computation. 
\end{proof}

To conclude this section, we prove the following bounds on the inverses of $g$ and $g_B$:
\begin{proposition}\label{inverse}
Define\footnote{We remark that our conventions for $h$ and $H$ are slightly different: While $h$ is defined by subtracting $m+h_B+h_S$ from the metric $g$, $H$ is defined by only subtracting $m$ from $g^{-1}$.} $H^{\alp\bt}=(g^{-1})^{\alp\bt}-m^{\alp\bt}$ and $H_B^{\alp\bt}=(g_B^{-1})^{\alp\bt}-m^{\alp\bt}$. Define also $H_{\alp\bt}=m_{\alp\mu}m_{\bt\nu}H^{\mu\nu}$. Then, for $|I|\leq N$, $H$, $H_B$ and $H-H_B$ obey the following estimates:
\begin{equation}\label{inverse.1}
|\Gamma^I H|(t,x)\ls \frac{\log(2+s)}{1+s}+\sum_{|J|\leq |I|}|\Gamma^J h|(t,x), 
\end{equation}
\begin{equation}\label{inverse.2}
|\Gamma^I H_B|(t,x)\ls \frac{\log(2+s)}{1+s}, 
\end{equation}
and 
\begin{equation}\label{inverse.3}
|\Gamma^I (H-H_B)|(t,x)\ls \frac{\ep\log (2+s)}{1+s}+\sum_{|J|\leq |I|}|\Gamma^J h|(t,x). 
\end{equation}
Moreover, we have the following improved estimates for the ``good components'' of $H$:
\begin{equation}\label{inverse.4}
|H|_{L\mathcal T}(t,x)+\sum_{|I|\leq 1}|\Gamma^I H|_{LL}(t,x)\ls \frac{(1+|q|)^{\frac 12+\gamma}}{(1+s)^{1+\frac{\gamma}{2}}w(q)^{\frac 12}},
\end{equation}
\begin{equation}\label{inverse.5}
\begin{split}
|\Gamma^I H|_{LL}(t,x)\ls &\f{\log(2+s)}{1+s}+|\Gamma^I h|_{LL}(t,x)+\sum_{|J|\leq |I|}\f{(1+|q|)^{\f12+\f{\gamma}{4}}|\Gamma^J h|(t,x)}{(1+s)^{1-\de_0}w(q)^{\f12}},
\end{split}
\end{equation}
as well as
\begin{equation}\label{inverse.6}
|\Gamma^I (H-H_B)|_{LL}(t,x)\ls \f{\ep\log(2+s)}{1+s}+|\Gamma^I h|_{LL}(t,x)+\sum_{|J|\leq |I|} \f{(1+|q|)^{\f12+\f{\gamma}{4}}|\Gamma^J h|(t,x)}{(1+s)^{1-\de_0}w(q)^{\f12}}. 
\end{equation}
\end{proposition}
\begin{proof}
If $s$ is small, the estimates in the proposition are much easier to prove based on \eqref{inverse.bd.assumption}. We will therefore only treat the case where $s$ is large.

\noindent{\bf Proof of \eqref{inverse.1}-\eqref{inverse.3} for $|I|=0$}

We begin with \eqref{inverse.1}-\eqref{inverse.3}, first starting with the $|I|=0$ case. We use the following easy fact: Suppose $A$ is a matrix such that $A^{-1}$ has bounded Frobenius norm, i.e.~$\|A^{-1}\|\leq C$ and $B$ is a matrix with Frobenius norm $\|B\|\leq a$. Then, for every constant $C$, there exists $a_0$ sufficiently small such that whenever $a\leq a_0$, we have
\begin{equation}\label{inverse.lemma}
\|(A+B)^{-1}-(A^{-1}-A^{-1}BA^{-1})\|\ls \|B\|^2.
\end{equation}
Taking $A=m$ and $B=h_B+h_S+h$, we obtain the bound
$$|(g^{-1})^{\alp\bt}-m^{\alp\bt}+m^{\alp\mu}(h_B+h_S+h)_{\mu\nu}m^{\bt\nu}|\ls (|h_B|+|h_S|+|h|)^2$$
for $s$ sufficiently large,
which implies\footnote{We recall here that $H^{\alp\bt}:=(g^{-1})^{\alp\bt}-m^{\alp\bt}$ and that the indices of $H$ are lowered using $m$.}
\begin{equation}\label{H.formula.0}
|H_{\alp\bt}-h_{\alp\bt}|\ls \frac{\log^2(2+s)}{(1+s)^2}+|(h_S)_{\alp\bt}|+|(h_B)_{\alp\bt}|+|h|^2\ls \f{\log(2+s)}{1+s}+|h|^2,
\end{equation}
where we have used $|h_S|+|h_B|\ls \f{\log(2+s)}{1+s}.$ A similar argument using \eqref{inverse.lemma} with $A=m$ and $B=h_B$ shows that 
\begin{equation}\label{HB.0.est}
|H_B|\ls \frac{\log(2+s)}{1+s}.
\end{equation}
To estimate $H-H_B$, we again return to \eqref{inverse.lemma} and this time let $A=g_B$ and $B=h_S+h$. Since $h_S$ and $h$ both have small $L^\infty$ norms, we can ignore to terms that are quadratic or higher and obtain
\begin{equation}\label{h.diff}
|H-H_B|\ls |h_S|+|h|\ls \frac{\ep\log (2+s)}{1+s}+|h|.
\end{equation}
By \eqref{H.formula.0}, \eqref{HB.0.est} and \eqref{h.diff}, we have thus obtained \eqref{inverse.1}, \eqref{inverse.2} and \eqref{inverse.3} in the $|I|=0$ case. 

\noindent{\bf Proof of \eqref{inverse.1}-\eqref{inverse.3} for general $|I|\leq N$}

In order to estimate the derivatives of the inverse of $g$ and $g_B$, we iterate the following formula
$$\rd A^{-1}=-A^{-1}(\rd A) A^{-1}$$
to obtain the following expression:
\begin{equation}\label{H.formula}
\begin{split}
&\Gamma^I H=\Gamma^I g^{-1}\\
=&-g^{-1}(\Gamma^I (h_S+h_B+h)) g^{-1}+\sum_{\substack{J_1+J_2=I\\J_1,J_2\neq 0}}g^{-1}(\Gamma^{J_1} (h_S+h_B+h)) g^{-1}(\Gamma^{J_2} (h_S+h_B+h)) g^{-1}\\
&-\sum_{\substack{J_1+J_2+J_3=I\\J_1,J_2,J_3\neq 0}}g^{-1}(\Gamma^{J_1} (h_S+h_B+h)) g^{-1}(\Gamma^{J_2} (h_S+h_B+h)) g^{-1}(\Gamma^{J_3} (h_S+h_B+h)) g^{-1}+\dots
\end{split}
\end{equation}
Here, for a given multi-index $I=(i_1,i_2,\dots,i_N)$, $\{J_1+J_2=I,\,J_1,J_2\neq 0\}$ denotes the set of all $J_1$, $J_2$ such that $J_1$ is an ordered $N_1$-sub-tuple of $I$ (for some $0<N_1<N$) and $J_2$ is an ordered $(N-N_1)$-tuple given by removing $J_1$ from $I$. The set $\{J_1+J_2+J_3=I,\,J_1,J_2,J_3\neq 0\}$ (and the higher order analogues) is defined in a similar manner.

Recall from the bootstrap assumption \eqref{BA4} that $\Gamma^J h$ are bounded for $|J|\leq \lfloor\f{N}{2}\rfloor$. This, together with the bounds for $g^{-1}$ derived from \eqref{H.formula.0}, allows us to bound all terms that are quadratic or higher in $h$ both those which are linear. We thus have
$$|\Gamma^I H|\ls \sum_{|J|\leq |I|}(|\Gamma^J h_S|+|\Gamma^J h_B|+|\Gamma^J h|),$$
from which \eqref{inverse.1} follows. Similarly, we prove \eqref{inverse.2} by
$$|\Gamma^I H_B|\ls \sum_{|J|\leq |I|}|\Gamma^J h_B|\ls \frac{\log(2+s)}{1+s}.$$
We now turn to the difference $H-H_B$, i.e.~the proof of \eqref{inverse.3}. Using \eqref{H.formula} for both $H$ and $H_B$ and taking the difference, we have
\begin{equation*}
\begin{split}
|\Gamma^I(H-H_B)|\ls &\sum_{|J|\leq |I|}\left(|\Gamma^J h_S|+|\Gamma^J h|+|H-H_B||\Gamma^J h_B|\right)\\
\ls &\f{\ep\log(2+s)}{1+s}+\sum_{|J|\leq |I|}|\Gamma^J h|,
\end{split}
\end{equation*}
where in the last line we have used \eqref{h.diff} and the bounds for $|\Gamma^J h_S|$ and $|\Gamma^J h_B|$. This thus gives \eqref{inverse.3}.

\noindent{\bf Proof of \eqref{inverse.4}}

We now turn to the proof of the improved estimates for certain components for $\Gamma^I H$. When $|I|=0$, we use \eqref{H.formula.0} and the triangle inequality to obtain the following bound for $|H|_{L\mathcal T}$:
\begin{equation}\label{H.imp.1}
\begin{split}
|H|_{L\mathcal T}\ls &|h|_{L\mathcal T}+\frac{\log^2(2+s)}{(1+s)^2}+|h_S|_{L\mathcal T}+|h_B|_{L\mathcal T}+|h|^2\\
\ls & \f{(1+|q|)^{\f12+\gamma}}{w(q)^{\f12}(1+s)^{1+\f{\gamma}{2}}}.
\end{split}
\end{equation}
Here, we have used the bounds for $|h|_{L\mathcal T}$, $|h_B|_{L\mathcal T}$ and $|h_S|_{L\mathcal T}$ from the bootstrap assumption \eqref{BA5}, Definition \ref{def.dispersivespt} and Remark \ref{hS.LT.bd} respectively. We have also used the bound for $|h|$ in the bootstrap assumption \eqref{BA4}.

We also show that the $LL$ component of $\Gamma^I H$ is better behaved. Using \eqref{H.formula}, we get
\begin{equation}\label{HLL.1}
\begin{split}
|\Gamma^I H|_{LL}\ls &|\Gamma^I h_S|_{LL}+|\Gamma^I h_B|_{LL}+|\Gamma^I h|_{LL}+|H|(|\Gamma^I h_S|+|\Gamma^I h_B|+|\Gamma^I h|)\\
&+\sum_{|J_1|+|J_2|\leq |I|}(|\Gamma^{J_1} h_S|+|\Gamma^{J_1} h_B|+|\Gamma^{J_1} h|)(|\Gamma^{J_2} h_S|+|\Gamma^{J_2} h_B|+|\Gamma^{J_2} h|).
\end{split}
\end{equation}
For $|I|\leq 1$, we have
$$|\Gamma^I h_S|_{LL}+|\Gamma^I h_B|_{LL}\ls \f{1}{(1+s)^{1+\f\gamma 2}}.$$
Using this together with the bootstrap assumptions \eqref{BA4} and \eqref{BA5}, we have
\begin{equation}\label{H.imp.2}
\sum_{|I|\leq 1}|\Gamma^I H|_{LL}\ls \f{(1+|q|)^{\f12+\gamma}}{w(q)^{\f12}(1+s)^{1+\f{\gamma}{2}}}.
\end{equation}
\eqref{H.imp.1} and \eqref{H.imp.2} together imply \eqref{inverse.4}.

\noindent{\bf Proof of \eqref{inverse.5}}

It suffices to consider $|I|> 1$. Returning again to \eqref{HLL.1}, we first use the weaker bounds for $|\Gamma^I h_S|_{LL}$ and $|\Gamma^I h_B|_{LL}$ for $|I|> 1$:
$$|\Gamma^I h_S|_{LL}+|\Gamma^I h_B|_{LL}\ls \f{\log(2+s)}{1+s}.$$ 
Using also the bootstrap assumption \eqref{BA4}, we get \eqref{inverse.5}.

\noindent{\bf Proof of \eqref{inverse.6}}

Finally, we prove the estimates for $|\Gamma^I(H-H_B)|_{LL}$. We again use \eqref{H.formula} and subtract from it the corresponding equation for $g_B$. Using \eqref{inverse.1}, \eqref{inverse.3} and the bootstrap assumption \eqref{BA4}, we get
\begin{equation*}
\begin{split}
&|\Gamma^I(H-H_B)|_{LL}\\
\ls &|\Gamma^I h|_{LL}+|\Gamma^I h_S|_{LL}\\
&+\sum_{|J_1|+|J_2|\leq |I|}(|\Gamma^{J_1} h|+|\Gamma^{J_1} h_S|+|H-H_B|)(|H|+|\Gamma^{J_2} h_S|+|\Gamma^{J_2} h_B|+|\Gamma^{J_2} h|)\\
\ls &\f{\ep\log(2+s)}{1+s}+|\Gamma^I h|_{LL}+\sum_{|J|\leq |I|} \f{(1+|q|)^{\f12+\f{\gamma}{4}}|\Gamma^J h|}{(1+s)^{1-\de_0}w(q)^{\f12}}.
\end{split}
\end{equation*}
This proves \eqref{inverse.6} and thus concludes the proof of the proposition.
\end{proof}

\section{Generalized wave coordinate condition}\label{sec.generalized.wave.coordinate}
Recall that the metric $g$ satisfies the condition
\begin{equation}\label{wave.con.2}
(g^{-1})^{\alpha\beta}\frd{x^\alpha}g_{\bt\mu}-\frac 12 (g^{-1})^{\alp\bt}\frd{x^\mu}g_{\alp\bt}=-g_{\mu\nu}\mathcal G^{\nu},
\end{equation}
where $\mathcal G^\mu=\mathcal G_B^\mu+\mathcal G_S^\mu$, $\mathcal G_B^\mu=\Box_{g_B} x^\mu$ is given by the choice of coordinates for the background solution and $\mathcal G_S$ is given explicitly as in \eqref{wave.coord.2} by
$$\mathcal G^0_S=0,\quad \mathcal G^i_S=-\chi(r)\chi(\frac rt)\left(\frac {4M\log (r-2M)}{r^2}\right)\frac{x^{i}}{r}$$
In this section, we show that \eqref{wave.con.2} allows us to rewrite $|\rd\Gamma^I h|_{L\mathcal T}$ in terms of better behaved quantities which either have a ``good derivative'' $\bar{\rd}$ or are lower order in terms of derivatives. This follows closely the ideas in \cite{LR1, LR2}. We show that while we use generalized wave coordinates instead of wave coordinates in our setting, the most slowly decaying terms cancel and the methods for dealing with the ``bad derivatives'' of the ``good components'' still apply.

As a preliminary step, we need a calculation about $h_S$ and $\mathcal G_S$, which is a more general version of the estimate \eqref{hS.gauge.cancel}. More precisely, we have
\begin{proposition}\label{prop.S.gauge.cancel}
The following estimate holds for all $I$ with a constant that may depend on $|I|$:
$$\left|\Gamma^I\left(m^{\alpha\beta}\frd{x^\alpha}(h_S)_{\bt\mu}-\frac 12 m^{\alp\bt}\frd{x^\mu}(h_S)_{\alp\bt}+m_{\mu\nu} \mathcal G_S^{\nu}\right)\right|(t,x)\ls 
\begin{cases}
\f{\ep\log(2+s)}{(1+s)^3} &\mbox{if }q\geq 0\\
\f{\ep(1+|q|)\log(2+s)}{(1+s)^3} &\mbox{if }q< 0.
\end{cases}
$$
\end{proposition}
\begin{proof}
Recall from Definition \ref{hS.def} that $\tilde{h}_S$ is defined by
$$(\tilde{h}_S)_{00}=\frac{2M}{r},\quad (\tilde{h}_S)_{0i}=0,$$
$$(\tilde{h}_S)_{ij}=-\frac{4M\log(r-2M)}{r}\delta_{ij}-\frac{x_i x_j}{r^2}(\f{2M}{r}-\f{4M\log (r-2M)}{r})$$
and $h_S$ is defined by
$$h_S=\chi(r)\chi(\frac rt)\tilde{h}_S.$$
Recalling also the definition in \eqref{wave.coord.2}, we can write $\mathcal G_S$ as
\begin{equation}\label{tildeGS.def}
\mathcal G_S^\mu=\chi(r)\chi(\f rt) \tilde{\mathcal G}_S^\mu,\quad \mbox{where}\quad\tilde{\mathcal G}_S^0=0,\quad \tilde{\mathcal G}_S^i=-\left(\frac {4M\log (r-2M)}{r^2}\right)\f{x^i}{r}.
\end{equation}
Given the above definitions, we have the basic estimate that $|\rd \Gamma^I \tilde{h}_S|+|\Gamma^I\tilde{\mathcal G}_S|\ls \f{\log(2+s)}{(1+s)^2}$ for all\footnote{with constants depending on $|I|$.} $I$. This is of course not sufficient to conclude the proposition and indeed we will need the cancellation between various terms.

To proceed, first, it is easy to observe that the $\Gamma^I$ derivatives of $\chi(r)\chi(\frac rt)$ are bounded and supported in the region $\{r\geq \f 12\}\cap \{\f 12\leq \f{r}{t}\leq \f 34\}$. In particular, this region is a subset of $\{q<0\}$ and we also have $1\ls \f{(1+|q|)}{(1+s)}$. Therefore, using the basic estimate above, all the terms with at least one $\Gamma$ differentiating $\chi(r)\chi(\frac rt)$ obey the desired estimate.

It thus remains to control 
\begin{equation}\label{GS.gauge.main.term}
\left|\Gamma^I\left(m^{\alpha\beta}\frd{x^\alpha}(\tilde h_S)_{\bt\mu}-\frac 12 m^{\alp\bt}\frd{x^\mu}(\tilde h_S)_{\alp\bt}+m_{\mu\nu} \tilde{\mathcal G}_S^{\nu}\right)\right|(t,x)
\end{equation}
in the region $\{r\geq \f 12\}\cap\{r\geq \f t2\}$. In the case where $\mu=0$, the above term vanishes identically. We now compute \eqref{GS.gauge.main.term} in the case where $i=1,2,3$. We will use the following simple facts:
\begin{equation}\label{omi.facts}
\frd{x^i} r=\f{x^i}{r},\quad \frd{x^i}\left(\f{x^j}{r}\right)=\left(\f{\de_{ij}}{r}-\f{x^i x^j}{r^3}\right),\quad \sum_{i=1}^3\frd{x^i}\left(\f{x^i}{r}\right)=\f 2r,\quad \sum_{j=1}^3\left(\f{x^j}{r}\right)\frd{x^j}\left(\f{x^i}{r}\right)=0.
\end{equation}
For $i=1,2,3$, we have
\begin{equation}\label{gauge.S.1}
\begin{split}
&m^{\alpha\beta}\frd{x^\alpha}(\tilde h_S)_{\bt i}\\
=&\de^{jk}\frd{x^j} \left(-\frac{4M\log(r-2M)}{r}\delta_{ik}-\frac{x^i x^k}{r^2}(\f{2M}{r}-\f{4M\log (r-2M)}{r})\right)\\
=&\left(-\f{4M}{r(r-2M)}+\f{4M\log(r-2M)}{r^2}\right)\f{x^i}{r}\\
&-\frac{x^i}{r}\left(-\f{2M}{r^2}-\f{4M}{r(r-2M)}+\f{4M\log(r-2M)}{r^2}\right)\\
&-\f{2 x^i}{r^2}(\f{2M}{r}-\f{4M\log (r-2M)}{r})\\
=&\f{8M x^i\log(r-2M)}{r^3}-\f{2Mx^i}{r^3}
\end{split}
\end{equation}
and
\begin{equation}\label{gauge.S.2}
\begin{split}
&\frac 12 m^{\alp\bt}\frd{x^i}(\tilde h_S)_{\alp\bt}\\
=&-\frac 12 \frd{x^i}(\tilde h_S)_{00}+\frac 12 \de^{jk}\frd{x^i}(\tilde h_S)_{jk}\\
=&-\f 12 \frd{x^i}\frac{2M}{r}+\f 12\de^{jk}\frd{x^i}\left(-\frac{4M\log(r-2M)}{r}\delta_{jk}-\frac{x^j x^k}{r^2}(\f{2M}{r}-\f{4M\log (r-2M)}{r})\right)\\
=&\f{M x^i}{r^3}+\f 32 \left(-\f{4M}{r(r-2M)}+\f{4M\log(r-2M)}{r^2}\right)\f{x^i}{r}\\
&-\f 12\left(-\f{2M}{r^2}-\f{4M}{r(r-2M)}+\f{4M\log(r-2M)}{r^2}\right)\f{x^i}{r}\\
=&\f{2 M x^i}{r^3}-\f{4Mx^i}{r^2(r-2M)}+\f{4Mx^i\log(r-2M)}{r^3}.
\end{split}
\end{equation}
Subtracting \eqref{gauge.S.2} from \eqref{gauge.S.1}, we get
\begin{equation}\label{gauge.S.3}
\begin{split}
m^{\alpha\beta}\frd{x^\alpha}(\tilde h_S)_{\bt i}-\frac 12 m^{\alp\bt}\frd{x^i}(\tilde h_S)_{\alp\bt}
=&\f{4M x^i\log(r-2M)}{r^3}-\f{4Mx^i}{r^3}+\f{4Mx^i}{r^2(r-2M)}.
\end{split}
\end{equation}
On the other hand,
\begin{equation}\label{gauge.S.4}
-\f{4Mx^i}{r^3}+\f{4Mx^i}{r^2(r-2M)}=\f{4Mx^i(r-(r-2M))}{r^3(r-2M)}=\f{8M^2 x^i}{r^3(r-2M)}.
\end{equation}
Recalling the definition of $\tilde{\mathcal G}_S$ in \eqref{tildeGS.def} and using \eqref{gauge.S.3} and \eqref{gauge.S.4}, we thus obtain
$$m^{\alpha\beta}\frd{x^\alpha}(\tilde h_S)_{\bt i}-\frac 12 m^{\alp\bt}\frd{x^i}(\tilde h_S)_{\alp\bt}+m_{i\nu} \tilde{\mathcal G}_S^{\nu}=\f{8M^2 x^i}{r^3(r-2M)},$$
which immediately implies (in the region $r\geq \f t2$) that
$$\left|\Gamma^I\left(m^{\alpha\beta}\frd{x^\alpha}(\tilde h_S)_{\bt i}-\frac 12 m^{\alp\bt}\frd{x^i}(\tilde h_S)_{\alp\bt}+m_{i\nu} \tilde{\mathcal G}_S^{\nu}\right)\right|(t,x)\ls \f{\ep}{(1+s)^3}.$$
We have thus estimated \eqref{GS.gauge.main.term}. This concludes the proof of the proposition.
\end{proof}

We now move on to use the generalized wave coordinate condition to bound the ``good components'' of the metric. First, we prove an estimate for $|\rd h|_{L\mathcal T}$:
\begin{proposition}\label{wave.con.lower}
$|\rd h|_{L\mathcal T}$ satisfies the estimate
\begin{equation*}
\begin{split}
&|\rd h|_{L\mathcal T}(t,x)\\
\ls &|\bar\rd h|(t,x)+\f{\ep \log (2+s)}{(1+s)^2 w(q)^{\f{\gamma}{1+2\gamma}}}+\f{\log (2+s)}{1+s}|\rd h|(t,x)+\f{1}{(1+s)(1+|q|)^{\gamma}}|h|(t,x)+|h||\rd h|(t,x).
\end{split}
\end{equation*}
\end{proposition}
\begin{proof}
As a first step, we start from the wave coordinate condition \eqref{wave.con.2} and subtract off the contributions from the corresponding wave coordinate conditions for $h_B$ and $h_S$. More precisely, we now use $g^{-1}=(g_B)^{-1}+(H-H_B)$ to rewrite \eqref{wave.con.2} as follows:
\begin{equation}\label{wave.con.3}
(g_B^{-1})^{\alpha\beta}\frd{x^\alpha}(h_B+h_S+h)_{\bt\mu}-\frac 12 (g_B^{-1})^{\alp\bt}\frd{x^\mu}(h_B+h_S+h)_{\alp\bt}+g_{\mu\nu}\mathcal G^{\nu}=O((H-H_B)\cdot\rd (h_B+h_S+h)). 
\end{equation}
By definition of $\mathcal G_B$, we have
\begin{equation}\label{wave.con.B}
(g_B^{-1})^{\alpha\beta}\frd{x^\alpha}(h_B)_{\bt\mu}-\frac 12 (g_B^{-1})^{\alp\bt}\frd{x^\mu}(h_B)_{\alp\bt}+(g_B)_{\mu\nu}\mathcal G_B^{\nu}=0. 
\end{equation}
Therefore, subtracting \eqref{wave.con.B} from \eqref{wave.con.3}, we obtain
\begin{equation}\label{wave.con.4}
\begin{split}
&(g_B^{-1})^{\alpha\beta}\frd{x^\alpha}(h_S+h)_{\bt\mu}-\frac 12 (g_B^{-1})^{\alp\bt}\frd{x^\mu}(h_S+h)_{\alp\bt}+(h_S)_{\mu\nu}\mathcal G_B^{\nu}+h_{\mu\nu}\mathcal G_B^{\nu}+g_{\mu\nu}\mathcal G_S^{\nu}\\
=&O((H-H_B)\cdot\rd (h_B+h_S+h)). 
\end{split}
\end{equation}
Next, we subtract 
$$m^{\alpha\beta}\frd{x^\alpha}(h_S)_{\bt\mu}-\frac 12 m^{\alp\bt}\frd{x^\mu}(h_S)_{\alp\bt}+m_{\mu\nu} \mathcal G_S^{\nu}$$
from the left hand side of \eqref{wave.con.4} and use Proposition \ref{prop.S.gauge.cancel} to conclude
\begin{equation}\label{wave.con.4.5}
\begin{split}
&(g_B^{-1})^{\alpha\beta}\frd{x^\alpha}h_{\bt\mu}-\frac 12 (g_B^{-1})^{\alp\bt}\frd{x^\mu}h_{\alp\bt}+(h_S)_{\mu\nu}\mathcal G_B^{\nu}+h_{\mu\nu}\mathcal G_B^{\nu}\\
=&O((H-H_B)\cdot\rd (h_B+h_S+h))+O(H_B \cdot\rd h_S)\\
&+O((h_B+h_S+h)\cdot \mathcal G_S)+O(\f{\ep(1+|q|)\log(2+s)}{(1+s)^3w(q)^{\f{1}{1+2\gamma}}}). 
\end{split}
\end{equation}
We now rewrite \eqref{wave.con.4.5}, viewing $m^{\alpha\beta}\frd{x^\alpha}h_{\bt\mu}-\frac 12 m^{\alp\bt}\frd{x^\mu}h_{\alp\bt}$ as the main term. More precisely, we write \eqref{wave.con.4.5} as follows:
\begin{equation}\label{wave.con.decomp}
Main+Error_1+Error_2=Error_3+Error_4+Error_5+O\left(\f{\ep(1+|q|)\log(2+s)}{(1+s)^3w(q)^{\f{1}{1+2\gamma}}}\right),
\end{equation}
where 
\begin{equation}\label{wave.con.main.def}
Main:=m^{\alpha\beta}\frd{x^\alpha}h_{\bt\mu}-\frac 12 m^{\alp\bt}\frd{x^\mu}h_{\alp\bt},
\end{equation}
\begin{equation}\label{wave.con.error.1.def}
Error_1:=O(h_B\cdot \rd h),
\end{equation}
\begin{equation}\label{wave.con.error.2.def}
Error_2:=O((h+h_S)\cdot \mathcal G_B),
\end{equation}
\begin{equation}\label{wave.con.error.3.def}
Error_3:=O((H-H_B)\cdot\rd (h_B+h_S+h))
\end{equation}
\begin{equation}\label{wave.con.error.4.def}
Error_4:=O(H_B \cdot\rd h_S),
\end{equation}
\begin{equation}\label{wave.con.error.5.def}
Error_5:=O((h_B+h_S+h)\cdot \mathcal G_S).
\end{equation}
The term $O\left(\f{\ep(1+|q|)\log(2+s)}{(1+s)^3w(q)^{\f{1}{1+2\gamma}}}\right)$ in \eqref{wave.con.decomp} is clearly acceptable. We now show that the main term indeed gives the desired control up to some acceptable error terms and that all the error terms are controllable. We first estimate the terms in \eqref{wave.con.error.1.def}-\eqref{wave.con.error.5.def}. To that end, recall from Definitions \ref{def.dispersivespt} and \ref{hS.def} and the estimate \eqref{inverse.2} that
$$|h_S|\ls \f{\ep\log(2+s)}{1+s},\quad |h_B|\ls \f{\log(2+s)}{1+s},\quad |H_B|\ls \f{\log(2+s)}{1+s},$$
$$|\rd h_S|\ls \f{\ep\log(2+s)}{(1+s)^2},\quad |\rd h_B|\ls \f{1}{(1+s)(1+|q|)^{\gamma}};$$
and by (8) in Definition \ref{def.dispersivespt} and the definition in \eqref{wave.con.2}, we have
$$|\mathcal G_B|\ls \f{\log(2+s)}{(1+s)^2},\quad |\mathcal G_S|\ls \f{\ep\log(2+s)}{(1+s)^2}.$$
Therefore,
\begin{equation}\label{wave.con.error.12}
Error_1\ls \f{\log(2+s)}{1+s}|\rd h|,\quad Error_2\ls \f{\log(2+s)}{(1+s)^2}|h|+\f{\ep\log^2(2+s)}{(1+s)^3},
\end{equation}
\begin{equation}\label{wave.con.error.45}
Error_4\ls \f{\ep\log^2(2+s)}{(1+s)^3},\quad Error_5\ls \f{\ep\log^2(2+s)}{(1+s)^3}+\f{\ep\log(2+s)}{(1+s)^2}|h|.
\end{equation}
For the term $Error_3$, recall from \eqref{inverse.3} in Proposition \ref{inverse} that we have
$$|H-H_B|\ls \f{\ep\log (2+s)}{1+s}+| h|.$$
Therefore,
\begin{equation}\label{wave.con.error.3}
Error_3\ls \f{\ep\log (2+s)}{(1+s)^2(1+|q|)^{\gamma}}+\f{\log(2+s)}{1+s}|\rd h|+\f{1}{(1+s)(1+|q|)^{\gamma}}|h|+|h||\rd h|.
\end{equation}
Combining the estimates from \eqref{wave.con.decomp}, \eqref{wave.con.error.12}, \eqref{wave.con.error.45} and \eqref{wave.con.error.3}, we obtain
\begin{equation}\label{wave.con.5}
\begin{split}
|Main|=&\left|m^{\alpha\beta}\frd{x^\alpha}h_{\bt\mu}-\frac 12 m^{\alp\bt}\frd{x^\mu}h_{\alp\bt}\right|\\
\ls &\f{\ep \log (2+s)}{(1+s)^2(1+|q|)^\gamma}+\f{\log (2+s)}{1+s}|\rd h|+\f{1}{(1+s)(1+|q|)^{\gamma}}|h|+|h||\rd h|. 
\end{split}
\end{equation}
We now contract the index $\mu$ in \eqref{wave.con.5} with $L$ and $E^A$. The term $|\frac 12 L^{\mu} m^{\alp\bt}\frd{x^\mu}h_{\alp\bt}|$ and $|\frac 12 (E^A)^{\mu} m^{\alp\bt}\frd{x^\mu}h_{\alp\bt}|$ can be controlled by a good derivative, i.e.~
$$|\frac 12 L^{\mu} m^{\alp\bt}\frd{x^\mu}h_{\alp\bt}|,\,|\frac 12 (E^A)^{\mu} m^{\alp\bt}\frd{x^\mu}h_{\alp\bt}|\ls |\bar\rd h|.$$
We now consider $m^{\alpha\beta}\frd{x^\alpha}h_{\bt\mu}$ contracted with $L^\mu$ or $(E^A)^\mu$. Writing $m^{\alp\bt}=-L^{(\alp}\Lb^{\bt)}+\sum_{A=1}^3 (E^A)^{\alp} (E^A)^{\bt}$, we notice that in each case there is exactly one term on the left hand side with a bad derivative, i.e.~the term $L^\alpha L^\beta \rd_q h_{\alp\bt}$ or $L^\alpha (E^A)^\beta \rd_q h_{\alp\bt}$. These are exactly the $|\rd_q h|_{L\mathcal T}$ terms that we want to control. As a result, we get
$$|\rd_q h|_{L\mathcal T}\ls |\bar\rd h|+\f{\ep \log (2+s)}{(1+s)^2(1+|q|)^\gamma}+\f{\log (2+s)}{1+s}|\rd h|+\f{1}{(1+s)(1+|q|)^{\gamma}}|h|+|h||\rd h|.$$
This implies the desired conclusion.
\end{proof}

We can also use the condition \eqref{wave.con.2} to obtain an estimate for the bad derivative of the $LL$ component of higher derivatives of $h$. However, in this case, there are commutator terms containing bad derivative of bad components of the lower derivatives. More precisely, we have
\begin{proposition}\label{wave.con.higher}
The following estimate holds for the components $(\Gamma^I h)_{LL}$ for all $|I|\leq N$:
\begin{equation*}
\begin{split}
&|\rd\Gamma^I h|_{LL}(t,x)\\
\ls &\f{\ep\log(2+s)}{(1+s)^2 w(q)^{\f{\gamma}{1+2\gamma}}}+\f{\log(2+s)}{(1+s)(1+|q|)^\gamma}\sum_{|J|\leq |I|}|\Gamma^J h|(t,x)\\
&+\f{\log(2+s)}{1+s}\sum_{|J|\leq |I|}|\rd \Gamma^J h|(t,x)+\sum_{|J_1|+|J_2|\leq |I|}|\Gamma^{J_1}h||\rd \Gamma^{J_2}h|(t,x)\\
&+\sum_{|J|\leq |I|}|\db\Gamma^J h|(t,x)+\sum_{|J|\leq |I|-2}|\rd \Gamma^J h|(t,x).
\end{split}
\end{equation*}
\end{proposition}
\begin{proof}
Using \eqref{wave.con.2} and arguing as in the proof of Proposition \ref{wave.con.lower}, we obtain the following analogue of \eqref{wave.con.decomp}, except that now the error terms also have up to $|I|$ derivatives:
\begin{equation}\label{wave.con.higher.decomp}
\begin{split}
&\left|\Gamma^I\big(m^{\alpha\beta}\rd_{\alp}h_{\bt\mu}-\frac 12 m^{\alp\bt}\rd_{\mu}h_{\alp\bt}\big)\right|\\
\ls &\f{\ep \log (2+s)}{(1+s)^2 w(q)^{\f{\gamma}{1+2\gamma}}}+Error_1+Error_2+Error_3+Error_4,
\end{split}
\end{equation}
where
\begin{equation}\label{wave.con.higher.error.1.def}
Error_1:=\sum_{|J_1|+|J_2|\leq |I|}|\Gamma^{J_1}h_B||\rd\Gamma^{J_2} h|,
\end{equation}
\begin{equation}\label{wave.con.higher.error.2.def}
Error_2:=\sum_{|J_1|+|J_2|\leq |I|}|\Gamma^{J_1}h||\Gamma^{J_2} \mathcal G_B|,
\end{equation}
\begin{equation}\label{wave.con.higher.error.3.def}
Error_3:=\sum_{|J_1|+|J_2|\leq |I|}|\Gamma^{J_1}(H-H_B)||\rd\Gamma^{J_2} (h_B+h_S+h)|
\end{equation}
\begin{equation}\label{wave.con.higher.error.4.def}
Error_4:=\sum_{|J_1|+|J_2|\leq |I|}|\Gamma^{J_1}h||\Gamma^{J_2} \mathcal G_S|.
\end{equation}
Notice that compared to \eqref{wave.con.decomp}, we have not write explicitly the terms that are products of ``explicit quantities'', i.e.~$h_S$, $h_B$, $\mathcal G_B$, $\mathcal G_S$: These terms can clearly be dominated up a constant by $\f{\ep \log (2+s)}{(1+s)^2 w(q)^{\f{\gamma}{1+2\gamma}}}$ in exactly the same manner as in the proof of Proposition \ref{wave.con.lower}.

By Definitions \ref{def.dispersivespt} and \ref{hS.def}; \eqref{wave.coord.2} and \eqref{inverse.2}, we have
$$\sum_{|J|\leq |I|}|\Gamma^J h_S|\ls \f{\ep\log(2+s)}{1+s},\quad \sum_{|J|\leq |I|}|\Gamma^J h_B|\ls \f{\log(2+s)}{1+s},\quad \sum_{|J|\leq |I|}|\Gamma^J H_B|\ls \f{\log(2+s)}{1+s},$$
$$\sum_{|J|\leq |I|}|\rd \Gamma^J h_S|\ls \f{\ep\log(2+s)}{(1+s)^2},\quad \sum_{|J|\leq |I|}|\rd \Gamma^J h_B|\ls \f{1}{(1+s)(1+|q|)^{\gamma}};$$
$$\sum_{|J|\leq |I|}|\Gamma^J \mathcal G_B|\ls \f{\log(2+s)}{(1+s)^2},\quad \sum_{|J|\leq |I|}|\Gamma^J \mathcal G_S|\ls \f{\ep\log(2+s)}{(1+s)^2}.$$
Therefore,
\begin{equation}\label{wave.con.higher.error.est.1}
Error_1\ls \f{\log(2+s)}{1+s}\sum_{|J|\leq |I|}|\rd\Gamma^J h|,\quad Error_2\ls \f{\log(2+s)}{(1+s)^2}\sum_{|J|\leq |I|}|\Gamma^J h|,
\end{equation}
\begin{equation}\label{wave.con.higher.error.est.2}
Error_4\ls \f{\ep\log(2+s)}{(1+s)^2}\sum_{|J|\leq |I|}|\Gamma^{J}h|.
\end{equation}
For $Error_3$, we apply \eqref{inverse.3} and the above estimates to obtain
\begin{equation}\label{wave.con.higher.error.est.3}
\begin{split}
Error_3\ls &\sum_{|J_1|+|J_2|\leq |I|}\left(\f{\ep\log(2+s)}{1+s}+|\Gamma^{J_1} h|\right)\left(\f{1}{(1+s)(1+|q|)^{\gamma}}+|\rd\Gamma^{J_2}h|\right)\\
\ls &\f{\ep\log(2+s)}{(1+s)^2(1+|q|)^{\gamma}}+\f{\ep\log(2+s)}{1+s}\sum_{|J|\leq |I|}|\rd\Gamma^{J} h|+\f{1}{(1+s)(1+|q|)^{\gamma}}\sum_{|J|\leq |I|}|\Gamma^{J} h|\\
&+\sum_{|J_1|+|J_2|\leq |I|}|\Gamma^{J_1} h||\rd\Gamma^{J_2}h|.
\end{split}
\end{equation}
Combining \eqref{wave.con.higher.decomp}, \eqref{wave.con.higher.error.est.1}, \eqref{wave.con.higher.error.est.2} and \eqref{wave.con.higher.error.est.3}, we thus obtain
\begin{equation}\label{wave.con.6}
\begin{split}
&\left|\Gamma^I\left(\underbrace{m^{\alpha\beta}\rd_{\alp}h_{\bt\mu}}_{=:A}-\underbrace{\frac 12 m^{\alp\bt}\rd_{\mu}h_{\alp\bt}}_{=:B}\right)\right|\\
\ls &\f{\ep \log (2+s)}{(1+s)^2 w(q)^{\f{\gamma}{1+2\gamma}}}+\f{\log(2+s)}{(1+s)(1+|q|)^\gamma}\sum_{|J|\leq |I|}|\Gamma^J h|+\f{\log(2+s)}{1+s}\sum_{|J|\leq |I|}|\rd \Gamma^J h|\\
&+\sum_{|J_1|+|J_2|\leq |I|}|\Gamma^{J_1}h||\rd \Gamma^{J_2}h|. 
\end{split}
\end{equation}
We now contract the left hand side of \eqref{wave.con.6} with $L^\mu$ and study the resulting expression. In particular, we want to keep track of the structure of the terms after commuting $\Gamma^I$ with $\rd$. We first control the contraction the term $A$ in \eqref{wave.con.6} with $L^\mu$. Given\footnote{We refer the readers back to Section \ref{sec.notation} to recall our use of the multi-index notation.} $I=(i_1,\dots,i_{|I|})$, using the notation in Proposition \ref{Gamma.commute}, we have
\begin{equation}\label{wave.con.A.1}
\begin{split}
&|L^\mu \Gamma^I (m^{\alp\bt}\rd_{\alp} h_{\bt\mu})-L^\mu( m^{\alp\bt} \rd_\alp (\Gamma^Ih)_{\bt\mu})|\\
\ls &\sum_{n=1}^{|I|}|L^\mu m^{\alp\bt} ({ }^{\Gamma_{(i_n)}}c_{\alp}{ }^\nu)\rd_\nu(\Gamma_{(i_1)}\cdots\Gamma_{(i_{n-1})}\Gamma_{(i_{n+1})}\cdots\Gamma_{(i_{|I|})}h)_{\bt\mu}|+\sum_{|J|\leq |I|-2}|\rd\Gamma^J h|.
\end{split}
\end{equation}
By Proposition \ref{Gamma.commute}, $c_{LL}=0$. Using also $m^{\alp\bt}=-L^{(\alp} \Lb^{\bt)}+\sum_{A=1}^3 (E^A)^\alp (E^A)^\bt$, we therefore have
\begin{equation}\label{wave.con.A.2}
|L^\mu m^{\alp\bt} ({ }^{\Gamma_{(i_n)}}c_{\alp}{ }^\nu)\rd_\nu(\Gamma_{(i_1)}\cdots\Gamma_{(i_{n-1})}\Gamma_{(i_{n+1})}\cdots\Gamma_{(i_{|I|})}h)_{\bt\mu}| \ls \sum_{|J|\leq |I|-1} \left(|\db \Gamma^J h|+|\rd\Gamma^J h|_{L\mathcal T}\right).
\end{equation}
On the other hand, using again $m^{\alp\bt}=-L^{(\alp} \Lb^{\bt)}+\sum_{A=1}^3 (E^A)^\alp (E^A)^\bt$, we get
\begin{equation}\label{wave.con.A.3}
|L^\mu( m^{\alp\bt} \rd_\alp (\Gamma^Ih)_{\bt\mu})+\f 12 L^\mu\Lb^{\alp} L^\bt \rd_\alp(\Gamma^I h)_{\bt\mu}|\ls |\db\Gamma^I h|.
\end{equation}
Combining \eqref{wave.con.A.1}, \eqref{wave.con.A.2} and \eqref{wave.con.A.3}, we therefore obtain
\begin{equation}\label{wave.con.A.4}
\begin{split}
&|L^\mu \Gamma^I (m^{\alp\bt}\rd_{\alp} h_{\bt\mu})+\f 12 L^\mu\Lb^{\alp} L^\bt \rd_\alp(\Gamma^I h)_{\bt\mu}|\\
\ls & \sum_{|J|\leq |I|} |\db \Gamma^J h|+\sum_{|J|\leq |I|-1}|\rd\Gamma^J h|_{L\mathcal T}+\sum_{|J|\leq |I|-2}|\rd\Gamma^J h|.
\end{split}
\end{equation}
We now turn to the contraction of the term $B$ in \eqref{wave.con.6} with $L^\mu$. Using Proposition \ref{Gamma.commute}, we have
\begin{equation}\label{wave.con.B.1}
\left|L^\mu\left(\frac 12 m^{\alp\bt}\Gamma^I\rd_{\mu}h_{\alp\bt}-\frac 12 m^{\alp\bt}\rd_{\mu}(\Gamma^Ih)_{\alp\bt}\right)\right|\ls  \sum_{|J|\leq |I|-1}|\db\Gamma^J h|+\sum_{|J|\leq |I|-2}|\rd\Gamma^J h|,
\end{equation}
which then implies
\begin{equation}\label{wave.con.B.2}
|L^\mu m^{\alp\bt}\Gamma^I\rd_{\mu}h_{\alp\bt}|\ls  \sum_{|J|\leq |I|}|\db\Gamma^J h|+\sum_{|J|\leq |I|-2}|\rd\Gamma^J h|.
\end{equation}
Combining \eqref{wave.con.6}, \eqref{wave.con.A.4} and \eqref{wave.con.B.2}, we thus obtain
\begin{equation}\label{wave.con.higher.1}
\begin{split}
&|\rd \Gamma^I h|_{LL}\ls |L^\mu\Lb^{\alp} L^\bt \rd_\alp(\Gamma^I h)_{\bt\mu}|+|\db \Gamma^I h|\\
\ls &\f{\ep \log (2+s)}{(1+s)^2 w(q)^{\f{\gamma}{1+2\gamma}}}+\f{\log(2+s)}{(1+s)(1+|q|)^\gamma}\sum_{|J|\leq |I|}|\Gamma^J h|+\f{\log(2+s)}{1+s}\sum_{|J|\leq |I|}|\rd \Gamma^J h|\\
&+\sum_{|J_1|+|J_2|\leq |I|}|\Gamma^{J_1}h||\rd \Gamma^{J_2}h|+\sum_{|J|\leq |I|}|\db\Gamma^J h|+\sum_{|J|\leq |I|-1} |\rd\Gamma^J h|_{L\mathcal T}+\sum_{|J|\leq |I|-2}|\rd \Gamma^J h|.
\end{split}
\end{equation}
To proceed, we need an estimate for $\sum_{|J|\leq |I|-1} |\rd\Gamma^J h|_{L\mathcal T}$. Clearly, $\sum_{|J|\leq |I|-1} |\rd\Gamma^J h|_{LL}$ can be controlled in an identical manner as in \eqref{wave.con.higher.1}, with $I$ replaced by $J$ for some $|J|\leq |I|-1$. It thus remains to control $\sum_{|J|\leq |I|-1} |\rd\Gamma^J h|_{LA}$, for which we have used the convention
$$|\rd\Gamma^J h|_{LA}:=\sum_{B=1}^3 \sum_{U\in \mathcal U}|L^\alp (E^B)^\bt U^\gamma \rd_\gamma (\Gamma^J h)_{\alp\bt}|.$$
To estimate this term, we first use $m^{\alp\bt}=-L^{(\alp} \Lb^{\bt)}+\sum_{A=1}^3 (E^A)^\alp (E^A)^\bt$ to get
\begin{equation}\label{wave.con.LT.higher.1}
\begin{split}
&\left|\f 12(E^B)^\mu \Lb^{\alp} L^{\bt}\rd_\alp (\Gamma^J h)_{\bt\mu} +(E^B)^\mu\Gamma^J\left(m^{\alpha\beta}\rd_{\alp}h_{\bt\mu}-\frac 12 m^{\alp\bt}\rd_{\mu}h_{\alp\bt}\right)\right|\\
\ls &\sum_{|J'|\leq |J|}|\db \Gamma^{J'} h|+\sum_{|J'|\leq |J|-1}|\rd\Gamma^{J'} h|.
\end{split}
\end{equation}
Then, contracting the left hand side of \eqref{wave.con.6} with $E^B$ (with $I$ replaced by $J$) and using \eqref{wave.con.LT.higher.1}, we obtain
\begin{equation*}
\begin{split}
&|\rd \Gamma^J h|_{LA} \ls \sum_{B=1}^3|(E^B)^\mu\Lb^{\alp} L^{\bt}\rd_\alp (\Gamma^J h)_{\bt\mu}|+\sum_{|J'|\leq |J|}|\db\Gamma^{J'} h|\\
\ls &\sum_{B=1}^3|(E^B)^\mu \left(\Gamma^J\big(m^{\alpha\beta}\rd_{\alpha}h_{\bt\mu}-\frac 12 m^{\alp\bt}\rd_{\mu}h_{\alp\bt}\big)\right)|+\sum_{|J'|\leq |J|}|\db\Gamma^{J'} h|+\sum_{|J'|\leq |J|-1}|\rd \Gamma^{J'} h|\\
\ls &\f{\ep \log (2+s)}{(1+s)^2 w(q)^{\f{\gamma}{1+2\gamma}}}+\f{\log(2+s)}{(1+s)(1+|q|)^\gamma}\sum_{|J'|\leq |J|}|\Gamma^{J'} h|+\f{\log(2+s)}{1+s}\sum_{|J'|\leq |J|}|\rd \Gamma^{J'} h|\\
&+\sum_{|J_1|+|J_2|\leq |J|}|\Gamma^{J_1}h||\rd \Gamma^{J_2}h|+\sum_{|J'|\leq |J|}|\db\Gamma^{J'} h|+\sum_{|J'|\leq |J|-1} |\rd\Gamma^{J'} h|.
\end{split}
\end{equation*}
This implies
\begin{equation}\label{wave.con.higher.2}
\begin{split}
&\sum_{|J|\leq |I|-1}|\rd\Gamma^J h|_{LA} \\
\ls &\f{\ep \log (2+s)}{(1+s)^2 w(q)^{\f{\gamma}{1+2\gamma}}}+\f{\log(2+s)}{(1+s)(1+|q|)^\gamma}\sum_{|J|\leq |I|-1}|\Gamma^J h|+\f{\log(2+s)}{1+s}\sum_{|J|\leq |I|-1}|\rd \Gamma^J h|\\
&+\sum_{|J_1|+|J_2|\leq |I|-1}|\Gamma^{J_1}h||\rd \Gamma^{J_2}h|+\sum_{|J|\leq |I|-1}|\db\Gamma^J h|+\sum_{|J|\leq |I|-2} |\rd\Gamma^J h|.
\end{split}
\end{equation}
Combining \eqref{wave.con.higher.1} and \eqref{wave.con.higher.2}, we therefore obtain
\begin{equation*}
\begin{split}
&|\rd\Gamma^I h|_{LL}\\
\ls &\f{\ep\log(2+s)}{(1+s)^2 w(q)^{\f{\gamma}{1+2\gamma}}}+\f{\log(2+s)}{(1+s)(1+|q|)^\gamma}\sum_{|J|\leq |I|}|\Gamma^J h|\\
&+\f{\log(2+s)}{1+s}\sum_{|J|\leq |I|}|\rd \Gamma^J h|+\sum_{|J_1|+|J_2|\leq |I|}|\Gamma^{J_1}h||\rd \Gamma^{J_2}h|+\sum_{|J|\leq |I|}|\db\Gamma^J h|+\sum_{|J|\leq |I|-2}|\rd \Gamma^J h|.
\end{split}
\end{equation*}
\end{proof}

\section{Equations for $h$}\label{sec.eqn}

In order to obtain estimates for the metric $g$, we will bound $h:=g-m-h_B-h_S$. To this end, we need to subtract the equations for metric $g_B$ of the background solution from the metric $g$ of the unknown spacetime to derive a wave equation for $h$. Our goal in this section is to obtain a form of all the terms that appear in the equations for $\tBox_g \Gamma^I h$ (see Propositions \ref{schematic.eqn} and \ref{schematic.eqn.improved}). We will then use these equations to control $\Gamma^Ih$ in the remainder of the paper. We begin with a preliminary proposition:

\begin{proposition}\label{basic.eqn}
The inhomogeneous terms in the equation for $\tBox_g \Gamma^I h$ contain the following terms:
\begin{equation}\label{gS.gauge}
-\Gamma^I(\tBox_g (h_S)_{\mu\nu}-G_{\mu\nu}(g,\mathcal G_S)),
\end{equation}
\begin{equation}\label{gB.gauge}
\Gamma^I(G_{\mu\nu}(g,\mathcal G_B)-G_{\mu\nu}(g_B,\mathcal G_B)),
\end{equation}
\begin{equation}\label{mixed.gauge}
-\Gamma^I\left(2g_{\alp\nu}g_{\bt\mu}\mathcal G_B^{(\alp} \mathcal G_S^{\bt)}\right),
\end{equation}
\begin{equation}\label{commute}
[\tBox_g,\Gamma^I] h,
\end{equation}
\begin{equation}\label{quasilinear}
-\Gamma^I (\tBox_g h_B-\tBox_{g_B} h_B),
\end{equation}
\begin{equation}\label{Q.diff}
\Gamma^I\left(Q_{\mu\nu}(g,g;\rd g,\rd g)-Q_{\mu\nu}(g_B,g_B;\rd g_B,\rd g_B)\right),
\end{equation}
\begin{equation}\label{P.diff}
\Gamma^I\left(P_{\mu\nu}(g,g;\rd g,\rd g)-P_{\mu\nu}(g_B,g_B;\rd g_B,\rd g_B)\right),
\end{equation}
\begin{equation}\label{T.diff}
\Gamma^I\left(T_{\mu\nu}(\rd \phi,\rd \phi)-T_{\mu\nu}(\rd \phi_B,\rd \phi_B)\right).
\end{equation}
\end{proposition}
\begin{proof}
$g_{\mu\nu}$ and $(g_B)_{\mu\nu}$ both satisfy equations of the form as given in Proposition \ref{Einstein.eqn.g}. Taking the difference of these equations, we obtain
\begin{equation}\label{basic.eqn.1}
\Gamma^I(\tBox_g g_{\mu\nu}-\tBox_{g_B} (g_B)_{\mu\nu})=\mbox{\eqref{Q.diff}}+\mbox{\eqref{P.diff}}+\mbox{\eqref{T.diff}}+\Gamma^I\left(G_{\mu\nu}(g,\mathcal G)-G_{\mu\nu}(g_B,\mathcal G_B)\right).
\end{equation}
We first expand the terms on the left hand side of \eqref{basic.eqn.1}:
\begin{equation}\label{basic.eqn.2}
\begin{split}
&\Gamma^I(\tBox_g g_{\mu\nu}-\tBox_{g_B} (g_B)_{\mu\nu})\\
=&\Gamma^I(\tBox_g h_{\mu\nu})+\Gamma^I(\tBox_g (h_B)_{\mu\nu}-\tBox_{g_B} (h_B)_{\mu\nu})+\Gamma^I(\tBox_g (h_S)_{\mu\nu})\\
=&\underbrace{\tBox_g (\Gamma^Ih_{\mu\nu})}_{=:I}-\underbrace{([\tBox_g, \Gamma^I] h_{\mu\nu})}_{=:II}+\underbrace{\Gamma^I(\tBox_g (h_B)_{\mu\nu}-\tBox_{g_B} (h_B)_{\mu\nu})}_{=:III}+\underbrace{\Gamma^I(\tBox_g (h_S)_{\mu\nu})}_{=:IV}.
\end{split}
\end{equation}
Notice now that $I$ is the main term, $II$ is the term \eqref{commute} and $III$ is the term \eqref{quasilinear}. The term $IV$ will be taken into account later (after combining with what will be called $VII$ in \eqref{basic.eqn.3} below).

To compute the last term on the right hand side of \eqref{basic.eqn.1}, we first observe from \eqref{G.def} that $G_{\mu\nu}(g,\mathcal G)$ is linear in $\mathcal G$ except for the term $-g_{\alp\nu}g_{\bt\mu}\mathcal G^{\alp}\mathcal G^{\bt}$. Therefore, 
$$G_{\mu\nu}(g,\mathcal G)=G_{\mu\nu}(g,\mathcal G_B)+G_{\mu\nu}(g,\mathcal G_S)-2g_{\alp\nu}g_{\bt\mu}\mathcal G_B^{(\alp} \mathcal G_S^{\bt)}.$$
This implies
\begin{equation}\label{basic.eqn.3}
\begin{split}
&\Gamma^I\left(G_{\mu\nu}(g,\mathcal G)-G_{\mu\nu}(g_B,\mathcal G_B)\right)\\
=&\underbrace{\Gamma^I\left(G_{\mu\nu}(g,\mathcal G_B)-G_{\mu\nu}(g_B,\mathcal G_B)\right)}_{=:V}-\underbrace{2\Gamma^I\left(g_{\alp\nu}g_{\bt\mu}\mathcal G_B^{(\alp} \mathcal G_S^{\bt)}\right)}_{=:VI}+\underbrace{\Gamma^I \left(G_{\mu\nu}(g,\mathcal G_S)\right)}_{=:VII}.
\end{split}
\end{equation}
$V$ gives the term \eqref{gB.gauge} and $VI$ gives the term \eqref{mixed.gauge}. Finally, combining $IV$ from \eqref{basic.eqn.2} and $VII$ from \eqref{basic.eqn.3} gives \eqref{gS.gauge}.
\end{proof}

Given Proposition \ref{basic.eqn}, our goal in the remainder of this section is therefore to further estimate each of the terms \eqref{gS.gauge}-\eqref{T.diff}. In the process of estimating these terms, we will be using the bootstrap assumptions in Section \ref{sec.BA}. We will control these terms in the order that they appeared in Proposition \ref{basic.eqn}. For the convenience of the readers, let us mention the proposition in which each of these terms will be estimated: \eqref{gS.gauge} will be estimated in Proposition \ref{gauge.est}; \eqref{gB.gauge} will be bounded in Proposition \ref{gB.est}; \eqref{mixed.gauge} will be estimated in Proposition \ref{mixed.gauge.est}; \eqref{commute} will be controlled in Proposition \ref{commute.est}; \eqref{quasilinear} will be bounded in Proposition \ref{quasilinear.est}; \eqref{Q.diff} will be estimated in Proposition \ref{Q.est}. The terms \eqref{P.diff} and \eqref{T.diff} will be bounded in separately for the general case (in which one has a bad term with ``insufficient decay'') and for the $\mathcal T\mathcal U$ components. For \eqref{P.diff}, they will be carried out in Propositions \ref{P.est.1} and \ref{P.est.2} respectively; while for \eqref{T.diff}, they will be carried out in Propositions \ref{T.est.1} and \ref{T.est.2} respectively.

We now control the contribution from the $h_S$ term, i.e.~\eqref{gS.gauge}. First, we have the following preliminary bound:
\begin{proposition}\label{gS.prelim}
For every quadruple of non-negative integers $(k_0,k_1,k_2,k_3)$, $g_S:=m+h_S$ obeys the following estimates\footnote{where the implicit constants depend on $k_0, k_1, k_2, k_3$.}:
\begin{equation}\label{gS.prelim.0.1}
|\rd_t^{k_0}\rd_{x_1}^{k_1}\rd_{x_2}^{k_2}\rd_{x_3}^{k_3}\big(\Box_m (g_S)_{\mu\nu}+m_{\mu\lambda}\rd_\nu \mathcal G^{\lambda}_S+m_{\nu\lambda}\rd_\mu \mathcal G^{\lambda}_S\big)|\ls \f{\ep}{(1+s)^{4+k_0+k_1+k_2+k_3}},\quad\mbox{for }q\geq 0;
\end{equation}
and
$$|\rd_t^{k_0}\rd_{x_1}^{k_1}\rd_{x_2}^{k_2}\rd_{x_3}^{k_3}\big(\Box_m (g_S)_{\mu\nu}+m_{\mu\lambda}\rd_\nu \mathcal G^{\lambda}_S+m_{\nu\lambda}\rd_\mu \mathcal G^{\lambda}_S\big)|\ls \f{\ep\log(2+s)}{(1+s)^{3+k_0+k_1+k_2+k_3}},\quad\mbox{for }q< 0.$$
\end{proposition}
\begin{proof}
We recall from Definition \ref{hS.def} that $h_S$ is given\footnote{For the proof of this proposition, it is important to recall that Greek indices run through $0$, $1$, $2$, $3$ while Latin indices only run through $1$, $2$, $3$.} by $h_S:=\chi(r)\chi(\frac rt)\tilde{h}_S,$ with $(\tilde{h}_S)_{00}:=\frac{2M}{r},\quad (\tilde{h}_S)_{0i}:=0$, $(\tilde{h}_S)_{ij}:=-\frac{4M\log(r-2M)}{r}\delta_{ij}-\frac{x_i x_j}{r^2}(\f{2M}{r}-\f{4M\log (r-2M)}{r})$.
Recall moreover from \eqref{wave.coord.2} that $\mathcal G^\mu_S:=\chi(r)\chi(\frac rt)\tilde{\mathcal G}^\mu_S$, where $\tilde{\mathcal G}^0_S:=0$ and $\tilde{\mathcal G}^i_S:=-\left(\frac {4M\log (r-2M)}{r^2}\right)\f{x^i}{r}$.

We first prove the desired estimate for $q\geq 0$. In this region, the cutoff function $\chi(\frac rt)$ is identically $1$. Moreover, if any of the derivatives fall on $\chi(r)$, the resulting term is compactly supported in spacetime and clearly obeys the desired estimates. We can therefore carry out the computations suppressing the cutoff functions, i.e.~we only need to estimate $\tilde{h}_S$ and $\tilde{\mathcal G}^i_S$.

We first deal with the simple cases where at least one of the indices $\mu$ or $\nu$ is $0$. For $\mu=\nu=0$, the desired inequality is trivial as $\Box_m \big(\frac 1r\big)=0$ and $\mathcal G_S^0=0$. For $\mu=0$ and $\nu=i$, we have $(g_{S})_{0i}=\mathcal G_S^0=\rd_0 \tilde{\mathcal G}_S^i=0$. We therefore also have the desired conclusion in this case.

It thus remains to check the case $\mu=i$ and $\nu=j$. 
We first have\footnote{Note that while we have suppressed the cutoff function, this computation is only used when $r\geq \f 12$. In particular, all terms are well-defined for $M$ sufficiently small.} 
\begin{equation}\label{Box.hS.compute}
\begin{split}
&\Box_m (-\f{4M\log(r-2M)}{r}\de_{ij})\\
=&-\de_{ij}(\rd_r^2+\f 2r\rd_r)\f{4M\log(r-2M)}{r}\\
=&-\de_{ij}\left(\rd_r(\f{4M}{r(r-2M)}-\f{4M\log(r-2M)}{r^2})+\f 2r(\f{4M}{r(r-2M)}-\f{4M\log(r-2M)}{r^2})\right)\\
=&\f{4M}{r(r-2M)^2}\de_{ij}.
\end{split}
\end{equation}
To compute $\Box_m \big(\f{x^i x^j}{r^2}(\f{4M\log(r-2M)}{r}-\f{2M}{r})\big)$, notice that in terms of the polar coordinates $(r,\theta,\varphi)$, $\f{x^i x^j}{r^2}$ is a function of the angular variables alone and $\f{4M\log(r-2M)}{r}-\f{2M}{r}$ is a function of $r$ alone. We then compute the $\Box_m$ of the angular part and the radial part separately to get
\begin{equation}\label{Box.hS.compute.1}
\begin{split}
\Box_m \f{x^i x^j}{r^2}=&\sum_{k=1}^3(\f{x^j}{r})\rd_k^2(\f{x^i}{r})+\sum_{k=1}^3(\f{x^i}{r})\rd_k^2(\f{x^j}{r})+2\sum_{k=1}^3(\rd_k \f{x^i}{r})(\rd_k \f{x^j}{r})\\
=&\sum_{k=1}^3\f{x^j}{r}\rd_k(\f{\de_{ik}}{r}-\f{x^ix^k}{r^3})+\sum_{k=1}^3\f{x^i}{r}\rd_k(\f{\de_{jk}}{r}-\f{x^jx^k}{r^3})+2(\f{\de_{ik}}{r}-\f{x^ix^k}{r^3})(\f{\de_{jk}}{r}-\f{x^jx^k}{r^3})\\
=&-\f{2 x^ix^j}{r^4}-\f{4x^ix^j}{r^4}+\f{2x^ix^j}{r^4}+\f{2\de_{ij}}{r^2}-\f{2x^ix^j}{r^4}=\f 2{r^2}\de_{ij}-\f{6 x^i x^j}{r^4}.
\end{split}
\end{equation}
and
\begin{equation}\label{Box.hS.compute.2}
\begin{split}
\Box_m (\f{4M\log(r-2M)}{r}-\f{2M}{r})=-\f{4M}{r(r-2M)^2}.
\end{split}
\end{equation}
Notice that the term \eqref{Box.hS.compute.1} was computed using \eqref{omi.facts} and the term \eqref{Box.hS.compute.2} was calculated in a similar manner as \eqref{Box.hS.compute}. Therefore, by combining \eqref{Box.hS.compute.1} and \eqref{Box.hS.compute.2}, we have
\begin{equation}\label{Box.hS.compute.3}
\begin{split}
&\Box_m \big(-\f{x^i x^j}{r^2}(\f{2M}{r}-\f{4M\log(r-2M)}{r})\big)=\Box_m \big(\f{x^i x^j}{r^2}(\f{4M\log(r-2M)}{r}-\f{2M}{r})\big)\\
=&(\Box_m \f{x^i x^j}{r^2})(\f{4M\log(r-2M)}{r}-\f{2M}{r})+\f{x^i x^j}{r}\big(\Box_m(\f{4M\log(r-2M)}{r}-\f{2M}{r})\big)\\
=&(\f 2{r^2}\de_{ij}-\f{6 x^i x^j}{r^4})(\f{4M\log(r-2M)}{r}-\f{2M}{r})-\f{4Mx^i x^j}{r^3(r-2M)^2}\\
=& (-\f{4M}{r^3}+\f{8 M\log (r-2M)}{r^3})\de_{ij}+\f{8 M x^i x^j}{r^5}\\
&-\f{24 M x^i x^j\log(r-2M)}{r^5}-\f{16M^2 x^i x^j (r-M)}{r^5(r-2M)^2}.
\end{split}
\end{equation}
Adding \eqref{Box.hS.compute} and \eqref{Box.hS.compute.3}, and recalling the definition of $(\tilde{h}_S)_{ij}$, we have
\begin{equation}\label{gS.prelim.1}
\begin{split}
&\Box_m (\tilde{h}_S)_{ij}\\
=& \f{8 M\log (r-2M)}{r^3}\de_{ij}+\f{8 M x^i x^j}{r^5}-\f{24 M x^i x^j\log(r-2M)}{r^5}+\f{16 M^2(r-M)}{r^3(r-2M)^2}(\de_{ij}-\f{x^i x^j}{r^2}).
\end{split}
\end{equation}
On the other hand, by the definition of $\tilde{\mathcal G}_S$, we have
\begin{equation}\label{gS.prelim.2}
\begin{split}
&m_{i\lambda}\rd_j \tilde{\mathcal G}_S^\lambda=m_{j\lambda}\rd_i \tilde{\mathcal G}_S^\lambda\\
=&-\f{4 M\log (r-2M)}{r^3}\de_{ij}-\f{x^i x^j}{r^4}\left(\f{4M}{r-2M}-\f{12 M\log(r-2M)}{r}\right)\\
=&-\f{4 M\log (r-2M)}{r^3}\de_{ij}-\f{x^i x^j}{r^5}(4M-12 M\log(r-2M))-\f{8 M^2 x^i x^j}{r^5 (r-2M)}.
\end{split}
\end{equation}
Combining \eqref{gS.prelim.1} and \eqref{gS.prelim.2}, we obtain
\begin{equation*}
\begin{split}
\Box_m (\tilde{h}_S)_{ij}+m_{i\lambda}\rd_j \tilde{\mathcal G}^{\lambda}_S+m_{j\lambda}\rd_i \mathcal G^{\lambda}_S
=&\f{16 M^2(r-M)}{r^3(r-2M)^2}(\de_{ij}-\f{x^i x^j}{r^2})-\f{16 M^2 x^i x^j}{r^5 (r-2M)}.
\end{split}
\end{equation*}
Using the fact that $t\leq r$ in the region $q\geq 0$, it is clear that all derivatives of the right hand side obeys bounds as in the right hand side of \eqref{gS.prelim.0.1}. We thus obtain \eqref{gS.prelim.0.1} for $(\mu,\nu)=(i,j)$.

Finally, we consider the case $q<0$. For this estimate, we simply need to note that on the support of $\chi(\frac rt)$, using the notation $\rd^k=\rd_t^{k_0}\rd_{x_1}^{k_1}\rd_{x_2}^{k_2}\rd_{x_3}^{k_3}$ and $|k|=k_0+k_1+k_2+k_3$, we have $|\rd^k \chi(\frac rt)|\ls \f 1{(1+s)^{|k|}}$, $|\rd^k \tilde{h}_S|\ls \f{\ep\log(2+s)}{(1+s)^{|k|+1}}$ and $|\rd^k \mathcal{\tilde{G}}^{\lambda}_S|\ls \frac{\ep\log (2+s)}{(1+s)^{|k|+2}}$. Using these estimates to bound each of the terms\footnote{Strictly speaking, there also additional terms where the derivatives act on $\chi(r)$. These terms are compactly supported in spacetime and can be handled easily.}, we obtain the desired result.
\end{proof}
Using this, we obtain the following bounds for the term \eqref{gS.gauge}:

\begin{proposition}\label{gauge.est}
For $|I|\leq N$, $g_S=m+h_S$ obeys the following estimates. In the region $q\geq  0$, we have
\begin{equation*}
\begin{split}
&|\Gamma^I(\tBox_g (g_S)_{\mu\nu}-G_{\mu\nu}(g,\mathcal G_S))|\\
\ls &\frac{\ep\log^2 (2+s)}{(1+s)^3(1+|q|)^{\gamma}}+\sum_{|J|\leq |I|}\left(\frac{\ep\log (2+s)|\rd\Gamma^J h|}{(1+s)^2}+\frac{\ep\log (2+s)|\Gamma^J h|}{(1+s)^3}\right).
\end{split}
\end{equation*}
In the region $q< 0$, we have
\begin{equation*}
\begin{split}
&|\Gamma^I(\tBox_g (g_S)_{\mu\nu}-G_{\mu\nu}(g,\mathcal G_S))|\\
\ls &\frac{\ep\log^2 (2+s)}{(1+s)^{3}}+\sum_{|J|\leq |I|}\left(\frac{\ep\log (2+s)|\rd\Gamma^J h|}{(1+s)^2}+\frac{\ep\log (2+s)|\Gamma^J h|}{(1+s)^3}\right).
\end{split}
\end{equation*}
\end{proposition}
\begin{proof}
We recall the definition of the term $G_{\mu\nu}(g,\mathcal G_S)$:
\begin{equation*}
\begin{split}
G_{\mu\nu}(g,\mathcal G_S)=&-\rd_\mu(g_{\nu\lambda} \mathcal G_S^{\lambda})-\rd_\nu(g_{\mu\lambda} \mathcal G_S^{\lambda})-\mathcal G_S^{\alp}\rd_\alp g_{\mu \nu}- g_{\alp\nu} g_{\bt\mu} \mathcal G_S^\alp \mathcal G_S^{\bt}.
\end{split}
\end{equation*}
By the bound
$$\sum_{|J|\leq |I|}|\rd\Gamma^J g|\ls \f{\log(2+s)}{(1+s)(1+|q|)^{\gamma}}+\sum_{|J|\leq |I|}|\rd \Gamma^J h|$$
and the fact that 
\begin{equation}\label{gauge.est.1}
\sum_{|J|\leq |I|}|\Gamma^J\mathcal G_S^\lambda|\ls \frac{\ep\log (2+s)}{(1+s)^2},
\end{equation}
we have
$$|\Gamma^I(\mathcal G^{\lambda}_S\rd_\mu g_{\nu\lambda})|+|\Gamma^I(\mathcal G^{\alp}_S\rd_\alp g_{\mu\nu})|\ls \frac{\ep\log^2 (2+s)}{(1+s)^{3}(1+|q|)^{\gamma}}+\sum_{|J|\leq |I|}\frac{\ep\log (2+s)|\rd\Gamma^J h|}{(1+s)^2}.$$
On the other hand, using \eqref{gauge.est.1} again, we have
$$|\Gamma^I(g_{\alp\nu} g_{\bt\mu} \mathcal G_S^\alp \mathcal G_S^{\bt})|\ls \f{\ep^2\log(2+s)}{(1+s)^4}.$$
Therefore, we get
\begin{equation}\label{gS.gauge.1}
\begin{split}
&|\Gamma^I(G_{\mu\nu}(g,\mathcal G_S)+g_{\mu\lambda}\rd_\nu \mathcal G^{\lambda}_S+g_{\nu\lambda}\rd_\mu \mathcal G^{\lambda}_S)|\\
\ls &\frac{\ep\log^2 (2+s)}{(1+s)^{3}(1+|q|)^{\gamma}}+\sum_{|J|\leq |I|}\frac{\ep\log (2+s)|\rd\Gamma^J h|}{(1+s)^2}.
\end{split}
\end{equation}
We recall that $|\Gamma^J(g_{\mu\nu}-m_{\mu\nu})|\ls \frac{\log(2+s)}{1+s}+\sum_{|J'|\leq |J|}|\Gamma^{J'} h|$. Moreover, we have $|\Gamma^J(\rd_{\sigma}\mathcal G^{\lambda}_S)|\ls \frac{\ep\log (2+s)}{(1+s)^3}$ using the definition of $\mathcal G_S$. They imply
\begin{equation}\label{gS.gauge.2}
\begin{split}
&|\Gamma^I(m_{\mu\lambda}\rd_\nu \mathcal G^{\lambda}_S+m_{\nu\lambda}\rd_\mu \mathcal G^{\lambda}_S-g_{\mu\lambda}\rd_\nu \mathcal G^{\lambda}_S-g_{\nu\lambda}\rd_\mu \mathcal G^{\lambda}_S)|\\
\ls &\frac{\ep\log^2 (2+s)}{(1+s)^{4}}+\sum_{|J|\leq |I|}\frac{\ep\log (2+s)|\Gamma^J h|}{(1+s)^3}.
\end{split}
\end{equation}
We now apply Proposition \ref{gS.prelim} to obtain
\begin{equation}\label{gS.gauge.3}
|\Gamma^I\big(\Box_m (g_S)_{\mu\nu}+m_{\mu\lambda}\rd_\nu \mathcal G^{\lambda}_S+m_{\nu\lambda}\rd_\mu \mathcal G^{\lambda}_S\big)|\ls
\begin{cases}
\f{\ep }{(1+s)^4} &\mbox{if }q\geq 0\\
\f{\ep \log (2+s)}{(1+s)^3} &\mbox{if }q< 0.
\end{cases}
\end{equation}
Finally, using \eqref{inverse.1} in Proposition \ref{inverse}, we have 
\begin{equation}\label{gS.gauge.4}
|\Gamma^I(\tBox_g (g_S)_{\mu\nu}-\tBox_m (g_S)_{\mu\nu})|\ls \frac{\ep\log^2 (2+s)}{(1+s)^4}+\sum_{|J|\leq |I|}\frac{\ep\log (2+s)|\Gamma^J h|}{(1+s)^3}.
\end{equation}
Combining \eqref{gS.gauge.1}, \eqref{gS.gauge.2}, \eqref{gS.gauge.3} and \eqref{gS.gauge.4} gives the desired conclusion of the proposition.
\end{proof}

We now consider the other term where the gauge condition enters, i.e.~the term \eqref{gB.gauge} which involves the term $\mathcal G_B$:
\begin{proposition}\label{gB.est}
For $|I|\leq N$, we have the following bound for \eqref{gB.gauge}:
\begin{equation*}
\begin{split}
&|\Gamma^I(G_{\mu\nu}(g,\mathcal G_B)-G_{\mu\nu}(g_B,\mathcal G_B))|\\
\ls &\frac{\ep \log^2 (2+s)}{(1+s)^3(1+|q|)}+\sum_{|J|\leq |I|}\left(\frac{\log(2+s)|\rd \Gamma^J h|}{(1+s)^2}+\frac{\log(2+s)| \Gamma^J h|}{(1+s)^2(1+|q|)}\right).
\end{split}
\end{equation*}
\end{proposition}
\begin{proof}
Recall that
\begin{equation*}
\begin{split}
G_{\mu\nu}(g,\mathcal G_B)=&-\rd_\mu(g_{\nu\lambda} \mathcal G_B^{\lambda})-\rd_\nu(g_{\mu\lambda} \mathcal G_B^{\lambda})-\mathcal G_B^{\alp}\rd_\alp g_{\mu \nu}- g_{\alp\nu} g_{\bt\mu} \mathcal G_B^\alp \mathcal G_B^{\bt}.
\end{split}
\end{equation*}
We therefore have to control the following terms:
$$|\Gamma^I ((\rd g-\rd g_B) \mathcal G_B)|,\quad |\Gamma^I ((g-g_B) \rd \mathcal G_B)|,$$
$$|\Gamma^I ((g-g_B) g_B\mathcal G_B\mathcal G_B)|,\quad |\Gamma^I ((g-g_B)(g-g_B)\mathcal G_B\mathcal G_B)|.$$
By the definition of dispersive spacetimes ((8) of Definition \ref{def.dispersivespt}), we have
$$\sum_{|J|\leq|I|}|\Gamma^J\mathcal G_B|\ls \frac{\log(2+s)}{(1+s)^2}.$$
Using Proposition \ref{decay.weights}, this implies
$$\sum_{|J|\leq|I|}|\Gamma^J\rd\mathcal G_B|\ls \frac{\log(2+s)}{(1+s)^2(1+|q|)}.$$
Moreover, since $g=g_B+h_S+h$, we have
$$\sum_{|J|\leq|I|}|\Gamma^J(g-g_B)|\ls \sum_{|J|\leq |I|}|\Gamma^J h|+\frac{\ep\log (2+s)}{1+s}$$
and
$$\sum_{|J|\leq|I|}|\Gamma^J\rd (g-g_B)|\ls \sum_{|J|\leq |I|}|\rd \Gamma^J h|+\frac{\ep\log (2+s)}{(1+s)^2}.$$
The conclusion follows after combining these estimates and using the bootstrap assumption \eqref{BA4} for the quadratic terms in $h$.
\end{proof}
The next term to be controlled is \eqref{mixed.gauge}, which can be bounded as follows:
\begin{proposition}\label{mixed.gauge.est}
For $|I|\leq N$, the following estimate holds:
\begin{equation*}
\begin{split}
|\Gamma^I(g_{\alp\nu}g_{\bt\mu}\mathcal G_B^{(\alp}\mathcal G_S^{\bt)})|\ls \f{\ep\log^2(2+s)}{(1+s)^4}+\sum_{|J|\leq |I|}\f{\ep\log^2(2+s)}{(1+s)^4}|\Gamma^J h|.
\end{split}
\end{equation*}
\end{proposition}
\begin{proof}
By (8) in Definition \ref{def.dispersivespt} and the formula in \eqref{wave.coord.2}, we have the pointwise estimates
$$\sum_{|J|\leq |I|}|\Gamma^J \mathcal G_B|\ls \f{\log(2+s)}{(1+s)^2},\quad\sum_{|J|\leq |I|}|\Gamma^J \mathcal G_S|\ls\f{\ep\log(2+s)}{(1+s)^2}.$$
The conclusion is therefore implied by the simple bound $\sum_{|J|\leq |I|}|\Gamma^J g|\ls 1+\sum_{|J|\leq |I|}|\Gamma^J h|$ together with the bootstrap assumption \eqref{BA4}.
\end{proof}

We now turn to the commutator term \eqref{commute}. To estimate this term, we need the following estimate in Corollary 12.3 in Appendix A in \cite{LR2} for the commutator between $\tBox_g$ and $\Gamma$. Observe from the formula that one of the following three scenarios occur: either there is a good ($LL$- or $L\mathcal T$-) component of $H$ or there is extra $\f{1}{1+s}$ decay or the term is lower order in terms of derivatives.
\begin{proposition}\label{commutation}
The following commutation formula holds for any scalar function $\xi$:
\begin{equation*}
\begin{split}
&|\tBox_g \Gamma^I\xi-\hat{\Gamma}^I \tBox_g \xi|\\
\ls &\frac{1}{1+s}\sum_{|J_2|\leq |I|,\,|J_1|+(|J_2|-1)_+\leq |I|} |\Gamma^{J_1} H| |\rd \Gamma^{J_2}\xi|\\
&+\frac{1}{1+|q|}\sum_{|J_2|\leq |I|}\big(\sum_{|J_1|+(|J_2|-1)_+\leq |I|}|\Gamma^{J_1} H|_{LL}\\
&+\sum_{|J_1'|+(|J_2|-1)_+\leq |I|-1}|\Gamma^{J_1'} H|_{L\mathcal T}
+\sum_{|J_1''|+(|J_2|-1)_+\leq |I|-2}|\Gamma^{J_1''} H|\big)|\rd \Gamma^{J_2} \xi|,
\end{split}
\end{equation*}
where\footnote{Recall from Definition \ref{def.Mink.vf} that $S:=t\rd_t+\sum_{i=1}^3x^i\rd_i$.} $\hat{\Gamma}=\begin{cases}
\Gamma+2 & \mbox{if } \Gamma=S\\
\Gamma & \mbox{otherwise}
\end{cases}
$. Here, $(|K|-1)_+=|K|-1$ if $|K|\geq 1$ and $(|K|-1)_+=0$ if $|K|=0$.
\end{proposition}
Applying this commutation formula for each component $h_{\mu\nu}$, we obtain the following bound for \eqref{commute}:
\begin{proposition}\label{commute.est}
For $|I|\leq N$, the following estimates hold:
\begin{equation*}
\begin{split}
|[\tBox_g,\Gamma^I] h|\ls &\sum_{|J|\leq |I|-1}|\Gamma^J\tBox_g  h|+\sum_{|J|\leq |I|}\frac{|\rd\Gamma^J h|}{(1+s)^{1+\frac{\gamma}{2}}}+\sum_{|J|\leq |I|}\f{|\Gamma^J h|}{(1+s)^{2-2\de_0}(1+|q|)^{\f 12-\f \gamma 4}w(q)^{\f 12}}\\
&+\sum_{|J|\leq |I|-1}\frac{\log(2+s)|\rd \Gamma^J h|}{(1+s)(1+|q|)}+\sum_{|J|\leq |I|-1}\f{|\Gamma^J h|}{(1+s)(1+|q|)^{1+\gamma}}\\
&+\sum_{|J_2|\leq |I|,\,|J_1|+(|J_2|-1)_+\leq |I|}\frac{|\Gamma^{J_1} h|_{LL}|\rd\Gamma^{J_2} h|}{1+|q|}+\sum_{\substack{|J_1|+|J_2|\leq |I|\\\max\{|J_1|,|J_2|\}\leq |I|-1}}\frac{|\Gamma^{J_1} h||\rd \Gamma^{J_2} h|}{1+|q|}.
\end{split}
\end{equation*}
\end{proposition}
\begin{proof}
First, notice that 
\begin{equation}\label{commute.0}
|[\tBox_g,\Gamma^I] h|\ls |\tBox_g \Gamma^I h-\hat{\Gamma}^I \tBox_g h|+\sum_{|J|\leq |I|-1}|\Gamma^J\tBox_g  h|.
\end{equation}
It therefore suffices to control $|\tBox_g \Gamma^I h-\hat{\Gamma}^I \tBox_g h|$, for which we can apply Proposition \ref{commutation} to each component of $h_{\mu\nu}$. We now control each of the four terms from Proposition \ref{commutation} using Proposition \ref{inverse}. We have the following estimate for the first term:
\begin{equation}\label{commute.1}
\begin{split}
&\frac{1}{1+s}\sum_{|J_2|\leq |I|,\,|J_1|+(|J_2|-1)_+\leq |I|} |\Gamma^{J_1} H| |\rd \Gamma^{J_2} h|\\
\ls &\sum_{|J|\leq |I|}\frac{\log(2+s)|\rd \Gamma^J h|}{(1+s)^2}+\sum_{|J_2|\leq |I|,\,|J_1|+(|J_2|-1)_+\leq |I|}\frac{|\Gamma^{J_1} h||\rd \Gamma^{J_2} h|}{1+s}\\
\ls &\sum_{|J|\leq |I|}\frac{\log(2+s)|\rd \Gamma^J h|}{(1+s)^2}+\sum_{|J|\leq |I|}\f{(1+|q|)^{\f12+\f\gamma 4}|\rd\Gamma^{J} h|}{(1+s)^{2-\de_0}w(q)^{\f12}}\\
&+\sum_{|J|\leq |I|}\f{|\Gamma^J h|}{(1+s)^{2-\de_0}(1+|q|)^{\f12-\f \gamma 4}w(q)^{\f12}}.
\end{split}
\end{equation}
where in the first inequality we have used \eqref{inverse.1} in Proposition \ref{inverse} and in the last inequality we have used the bootstrap assumptions \eqref{BA1} and \eqref{BA4}. Here, notice in particular that at least one of the terms have at most $\lfloor \f{N}{2}\rfloor$ $\Gamma$'s, which allows us to apply the bootstrap assumptions.

For the second term, we use \eqref{inverse.4} and \eqref{inverse.5} in Proposition \ref{inverse} to get
\begin{equation}\label{commute.2}
\begin{split}
&\frac{1}{1+|q|}\sum_{|J_2|\leq |I|,\,|J_1|+(|J_2|-1)_+\leq |I|}|\Gamma^{J_1} H|_{LL}|\rd \Gamma^{J_2} h|\\
\ls &\sum_{|J_2|\leq |I|,\,|J_1|+(|J_2|-1)_+\leq |I|}\frac{|\Gamma^{J_1} h|_{LL}|\rd\Gamma^{J_2} h|}{1+|q|}+\sum_{|J|= |I|}\f{|\rd \Gamma^J h|}{(1+s)^{1+\f{\gamma}{2}}}\\
&+\sum_{|J|\leq |I|-1}\f{\log(2+s)|\rd \Gamma^J h|}{(1+s)(1+|q|)}\\
&+\sum_{|J_2|\leq |I|,\,|J_1|+(|J_2|-1)_+\leq |I|}\f{|\Gamma^{J_1}h||\rd\Gamma^{J_2} h|}{(1+s)^{1-\de_0}(1+|q|)^{\f12-\f{\gamma}{4}}w(q)^{\f12}}\\
\ls &\sum_{|J_2|\leq |I|,\,|J_1|+(|J_2|-1)_+\leq |I|}\frac{|\Gamma^{J_1} h|_{LL}|\rd\Gamma^{J_2} h|}{1+|q|}+\sum_{|J|= |I|}\f{|\rd \Gamma^J h|}{(1+s)^{1+\f{\gamma}{2}}}\\
&+\sum_{|J|\leq |I|-1}\f{\log(2+s)|\rd \Gamma^J h|}{(1+s)(1+|q|)}+\sum_{|J|\leq |I|}\f{|\Gamma^J h|}{(1+s)^{2-2\de_0}(1+|q|)^{1-\f{\gamma}{2}}w(q)}.
\end{split}
\end{equation}
In the last step above we have used the bootstrap assumptions \eqref{BA1} and \eqref{BA4}.

We now turn to the third term. Using \eqref{inverse.1} and \eqref{inverse.4} in Proposition \ref{inverse} to control $|H|_{L\mathcal T}$, we obtain
\begin{equation}\label{commute.3}
\begin{split}
&\frac{1}{1+|q|}\sum_{|J_2|\leq |I|,\,|J_1|+(|J_2|-1)_+\leq |I|-1}|\Gamma^{J_1} H|_{L\mathcal T}|\rd \Gamma^{J_2} h|\\
\ls &\sum_{|J|=|I|}\frac{|H|_{L\mathcal T}|\rd\Gamma^I h|}{1+|q|}+\sum_{\substack{|J_1|+|J_2|\leq |I|\\\max\{|J_1|,|J_2|\}\leq |I|-1}}\frac{|\Gamma^{J_1} H||\rd \Gamma^{J_2} h|}{1+|q|}\\
\ls &\sum_{|J|=|I|}\frac{|\rd\Gamma^J h|}{(1+s)^{1+\f{\gamma}{2}}(1+|q|)^{\f12-\gamma}w(q)^{\f12}}+\sum_{|J|\leq |I|-1}\frac{\log(2+s)|\rd \Gamma^J h|}{(1+s)(1+|q|)}\\
&+\sum_{\substack{|J_1|+|J_2|\leq |I|\\\max\{|J_1|,|J_2|\}\leq |I|-1}}\frac{|\Gamma^{J_1} h||\rd \Gamma^{J_2} h|}{1+|q|}.
\end{split}
\end{equation}
Finally, notice that the last term in Proposition \ref{commutation} does not contain highest order contribution, i.e.~we have $|J_1|,|J_2|\leq |I|-1$. We use \eqref{inverse.1} in Proposition \ref{inverse} to get
\begin{equation}\label{commute.4}
\begin{split}
&\frac{1}{1+|q|}\sum_{|J_2|\leq |I|,\,|J_1|+(|J_2|-1)_+\leq |I|-2}|\Gamma^{J_1} H||\rd \Gamma^{J_2} h|\\
\ls &\sum_{|J|\leq |I|-1}\frac{\log(2+s)|\rd \Gamma^J h|}{(1+s)(1+|q|)}+\sum_{|J_1|+|J_2|\leq |I|-1}\frac{|\Gamma^{J_1} h||\rd \Gamma^{J_2} h|}{1+|q|}.
\end{split}
\end{equation}
The conclusion follows after combining \eqref{commute.0}, \eqref{commute.1}, \eqref{commute.2}, \eqref{commute.3} and \eqref{commute.4}.
\end{proof}

We turn to the estimates for the term \eqref{quasilinear}:
\begin{proposition}\label{quasilinear.est}
The following estimates hold for \eqref{quasilinear} for $|I|\leq N$, :
\begin{equation*}
\begin{split}
&|\Gamma^I(\tBox_g  h_B-\tBox_{g_B} h_B)|\\
\ls &\frac{\ep \log (2+s)}{(1+s)^2(1+|q|)^{1+\gamma}}+\frac{|\Gamma^I h|_{LL}}{(1+s)(1+|q|)^{1+\gamma}}\\
&+\sum_{|J|\leq |I|}\f{|\Gamma^J h|}{(1+s)^{2-\de_0}(1+|q|)^{\f12+\f{3\gamma}{4}}w(q)^{\f12}}+\sum_{|J|\leq |I|-1}\frac{|\Gamma^J h|}{(1+s)(1+|q|)^{1+\gamma}}.
\end{split}
\end{equation*}
Moreover, we have the following improved bound for the case $|I|=0$:
\begin{equation}\label{quasilinear.improved}
|\tBox_g  h_B-\tBox_{g_B} h_B|\ls  \frac{\ep}{(1+s)^{2+\f{\gamma}{2}}(1+|q|)^{\f12}w(q)^{\f12}}.
\end{equation}
\end{proposition}
\begin{proof}
In order to control the difference of $\tBox_g$ and $\tBox_{g_B}$, we will apply the estimates for $H-H_B$ given by \eqref{inverse.3} and \eqref{inverse.6} in Proposition \ref{inverse}. Notice that in order to control this difference, we need a total of $N+2$ derivatives of $h_B$ and this is precisely the reason that we assume in Definition \ref{def.dispersivespt} that the background spacetime $(\mathcal M, g_B,\phi_B)$ has one extra degree of differentiability compared to that of $(\mathcal M, g, \phi)$.

We now turn to the proof of the estimates. First, we have the bound
\begin{equation}\label{quasilinear.1}
\begin{split}
&|\Gamma^I(\tBox_g  h_B-\tBox_{g_B} h_B)|\\
\ls &\sum_{\substack{|J_1|+|J_2|\leq |I|\\|J_1|\leq |I|-1}}|\Gamma^{J_1} (H-H_B)||\rd^2\Gamma^{J_2} h_B|+|\Gamma^I (H-H_B)|_{LL}|\rd^2 h_B|+|\Gamma^I (H-H_B)||\bar\rd \rd h_B|.
\end{split}
\end{equation}
Using Proposition \ref{decay.weights} and the assumptions on $h_B$ in Definition \ref{def.dispersivespt}, we have
\begin{equation}\label{quasilinear.2}
\sum_{|J|\leq |I|}|\rd^2\Gamma^J h_B|\ls\sum_{|J|\leq |I|+1}\frac{|\rd\Gamma^J h_B|}{1+|q|}\ls \frac{1}{(1+s)(1+|q|)^{1+\gamma}}
\end{equation}
and
\begin{equation}\label{quasilinear.3}
\sum_{|J|\leq |I|}|\bar\rd\rd\Gamma^J h_B|\ls \sum_{|J|\leq |I|+1}\frac{|\rd\Gamma^J h_B|}{1+s}\ls \frac{1}{(1+s)^2(1+|q|)^{\gamma}}.
\end{equation}
On the other hand, \eqref{inverse.3} in Proposition \ref{inverse} implies that 
$$|\Gamma^J (H-H_B)|\ls \frac{\ep\log (2+s)}{1+s}+\sum_{|J'|\leq |J|}|\Gamma^{J'} h| $$
and \eqref{inverse.5} in Proposition \ref{inverse} implies that
$$|\Gamma^I (H-H_B)|_{LL}\ls \frac{\ep\log(2+s)}{1+s}+|\Gamma^I h|_{LL}+\sum_{|J|\leq |I|}\f{(1+|q|)^{\f12+\f{\gamma}{4}}|\Gamma^J h|}{(1+s)^{1-\de_0}w(q)^{\f12}}. $$
Combining these estimates and substituting into \eqref{quasilinear.1} yield the first conclusion of the proposition. 

Finally, for the improved estimate in the case $|I|=0$, notice that for $|I|=0$, the first term on the right hand side of \eqref{quasilinear.1} is absent, i.e.
\begin{equation}\label{quasilinear.4}
\begin{split}
|\tBox_g  h_B-\tBox_{g_B} h_B|
\ls &|H-H_B|_{LL}|\rd^2 h_B|+|H-H_B||\bar\rd \rd h_B|.
\end{split}
\end{equation}
Using the estimates \eqref{inverse.4} and \eqref{inverse.6} in Proposition \ref{inverse} and the bootstrap assumption \eqref{BA5}, we have
\begin{equation}\label{quasilinear.5}
\begin{split}
|H-H_B|_{LL}\ls \frac{\ep(1+|q|)^{\f12+\gamma}}{(1+s)^{1+\f{\gamma}{2}}w(q)^{\f12}}
\end{split}
\end{equation}
and
\begin{equation}\label{quasilinear.6}
\begin{split}
|H-H_B|\ls \frac{\ep(1+|q|)^{\f12+\f{\gamma}{4}}}{(1+s)^{1-\de_0}w(q)^{\f12}}.
\end{split}
\end{equation}
We now substitute the bounds \eqref{quasilinear.2}, \eqref{quasilinear.3}, \eqref{quasilinear.5}, \eqref{quasilinear.6} into \eqref{quasilinear.4} to get
\begin{equation*}
|\tBox_g  h_B-\tBox_{g_B} h_B|\ls  \frac{\ep}{(1+s)^{2+\f{\gamma}{2}}(1+|q|)^{\f12}w(q)^{\f12}}+\frac{\ep(1+|q|)^{\f12+\f{\gamma}{4}}}{(1+s)^{3-\de_0}(1+|q|)^{\gamma}w(q)^{\f12}}.
\end{equation*}
The claimed estimate for the $|I|=0$ case thus follows after noting that for $\de_0$ satisfying \eqref{de_0.def}, we have $1-\de_0-\f{\gamma}{2}>1-\f{3\gamma}{4}$.
\end{proof}

We now turn to the quadratic terms in $\rd g$, i.e.~\eqref{Q.diff} and \eqref{P.diff}. First, we deal with the terms $Q$ for which a null structure is present, i.e.~we control the term \eqref{Q.diff}. Before we deal with these terms, we first need a discussion on some standard facts about the classical null forms on Minkowski spacetime. 
\begin{definition}\label{def.null.form}
We say $\widetilde{Q}(\xi,\eta)=A^{\alp\bt}\rd_\alp\xi\rd_\bt\eta$ is a classical null form if $A^{\alp\bt}$ are constants satisfying $A^{\alp\bt}X_{\alp}X_{\bt}$ whenever $X_0^2=X_1^2+X_2^2+X_3^2$.
\end{definition}
It is a standard easy fact that every classical null forms can be controlled by a product where at least one of the derivatives is a good derivative $\db$:
\begin{lemma}\label{null.form.1}
If $\widetilde{Q}(\xi,\eta)$ is a classical null form, then
$$|\widetilde{Q}(\xi,\eta)|\ls |\rd\xi||\db\eta|+|\db\xi||\rd\eta|.$$
\end{lemma}
The final fact that we need about classical null forms is that they commute well with the Minkowskian vector fields:
\begin{lemma}\label{null.form.2}
Let $\widetilde{Q}(\xi,\eta)$ be a classical null form and $\Gamma$ be a Minkowskian commuting vector field. Then 
$$\Gamma \widetilde{Q}(\xi,\eta)=\widetilde{Q}(\Gamma\xi,\eta)+\widetilde{Q}(\xi,\Gamma\eta)+\widetilde{Q}'(\xi,\eta),$$
where $\widetilde{Q}'(\xi,\eta)$ is also a classical null form.
\end{lemma}
Using Lemmas \ref{null.form.1} and \ref{null.form.2}, we can now proceed to estimate the term \eqref{Q.diff}. More precisely, we have
\begin{proposition}\label{Q.est}
For $|I|\leq N$, we have
\begin{equation*}
\begin{split}
&|\Gamma^I(Q_{\mu\nu}(g,g;\rd g,\rd g)-Q_{\mu\nu}(g_B,g_B;\rd g_B,\rd g_B))|\\
\ls &\frac{\ep\log (2+s)}{(1+s)^3(1+|q|)^{\gamma}}+\sum_{|J|\leq |I|}\left(\frac{|\rd \Gamma^J h|}{(1+s)^{1+\gamma}}\right.\\
&\left.+\frac{|\bar\rd \Gamma^J h|}{(1+s)^{1-\delta_0}(1+|q|)^{\gamma+\de_0}}+\frac{|\Gamma^J h|}{(1+s)^{2-2\de_0}(1+|q|)^{2\gamma+2\de_0}}\right).
\end{split}
\end{equation*}
\end{proposition}
\begin{proof}
By the triangle inequality and using the bootstrap assumptions \eqref{BA4} together with Proposition \ref{inverse} to bound the higher order terms, it suffices to control the following terms:
\begin{equation}\label{Q.est.1}
\Gamma^I(Q_{\mu\nu}(m,m;\rd h_*,\rd h)),
\end{equation}
\begin{equation}\label{Q.est.2}
\sum_{|J_1|+|J_2|+|J_3|\leq |I|}(1+|\Gamma^{J_1} H|)|\rd\Gamma^{J_2} h_*||\rd\Gamma^{J_3} h_S|,
\end{equation}
\begin{equation}\label{Q.est.3}
\sum_{|J_1|+|J_2|+|J_3|\leq |I|}|\Gamma^{J_1} H||\rd\Gamma^{J_2} h_*||\rd\Gamma^{J_3} h|,
\end{equation}
\begin{equation}\label{Q.est.4}
\sum_{|J_1|+|J_2|+|J_3|\leq |I|}|\Gamma^{J_1}(H-H_B)||\rd\Gamma^{J_2} h_*||\rd \Gamma^{J_3}h_*|,
\end{equation}
where we have used the notation $h_*\in \{h_S,h_B,h\}$. We briefly explain the estimates for these terms before turning to the details. First, notice that since we are taking the difference of a $g$ term and a $g_B$ term, every term in the resulting expression must have at least one factor of $\rd\Gamma^J h$, $\rd\Gamma^J h_S$ or $\Gamma^J (H-H_B)$. Now, in order to estimate these terms, we observe that as long as the term is cubic (i.e.~the terms \eqref{Q.est.3}, \eqref{Q.est.4}), there is enough decay to guarantee that it obeys the desired estimates. Turning to the quadratic terms, if one of the factors is $h_S$ (i.e.~the terms in \eqref{Q.est.2} arising from the $1$ in the first pair of brackets or the terms in \eqref{Q.est.1} where $h_*=h_S$), then we can use the better decay properties of $\rd\Gamma^J h_S$ to show that these terms are also acceptable. The main term is therefore the quadratic terms \eqref{Q.est.1} where $h_*\in\{h_B,h\}$. For these terms we use the fact that $Q_{\mu\nu}(m,m;\cdot,\cdot)$ is a classical null form and can be controlled using Lemmas \ref{null.form.1} and \ref{null.form.2}.

We now turn to the details of the estimates of these terms:

\noindent{\bf Estimates for \eqref{Q.est.1}}

Again, it is easy to check that $Q_{\mu\nu}(m,m; \cdot,\cdot)$ is a classical null form. Using Lemmas \ref{null.form.1} and \ref{null.form.2} on classical null forms, we have
\begin{equation}\label{Q.est.1.2}
\Gamma^I(Q_{\mu\nu}(m,m;\rd h_*,\rd h))\ls \sum_{|J_1|+|J_2|\leq|I|} (|\bar\rd \Gamma^{J_1} h_*||\rd \Gamma^{J_2} h|+|\rd \Gamma^{J_1} h_*||\bar\rd \Gamma^{J_2} h|).
\end{equation}
First, if $h_*\in\{h_S,h_B\}$, we can simply use the bound 
\begin{equation}\label{Q.est.1.3}
\sum_{|J|\leq N}|\bar\rd\Gamma^J h_*|\ls \frac{1}{(1+s)^{1+\gamma}},
\end{equation}
\begin{equation}\label{Q.est.1.4}
\sum_{|J|\leq N}|\rd\Gamma^J h_*|\ls \frac{1}{(1+s)(1+|q|)^{\gamma}}.
\end{equation}
Combining \eqref{Q.est.1.2}, \eqref{Q.est.1.3} and \eqref{Q.est.1.4}, we have thus shown that in the case $h_*\in \{h_S,h_B\}$, \eqref{Q.est.1} obeys bounds as stated in the proposition. Now, turning to the case $h_*=h$, notice that since $|I|\leq N$, we have $\min\{J_1,J_2\}\leq \lfloor \f{N}{2}\rfloor$. Therefore, by the bootstrap assumption \eqref{BA1}, \eqref{BA2}
\begin{equation}\label{Q.est.1.3.1}
\sum_{|J|\leq N}|\bar\rd\Gamma^{\min\{J_1,J_2\}} h|\ls \frac{1}{(1+s)^{1+\gamma}},
\end{equation}
\begin{equation}\label{Q.est.1.4.1}
\sum_{|J|\leq N}|\rd\Gamma^{\min\{J_1,J_2\}} h|\ls \frac{1}{(1+s)^{1-\de_0}(1+|q|)^{\f12-\f{\gamma}{4}}w(q)^{\f12}}.
\end{equation}
\eqref{Q.est.1.2}, \eqref{Q.est.1.3.1} and \eqref{Q.est.1.4.1} clearly imply the desired bounds in the case $h_*=h$.

\noindent{\bf Estimates for \eqref{Q.est.2}}

We first consider the case where $h_*=h$. Using the bound
$$|\rd\Gamma^{J_3}h_S|\ls \frac{\ep\log (2+s)}{(1+s)^2},$$
and \eqref{inverse.1} in Proposition \ref{inverse}, we have
$$\sum_{|J_1|+|J_2|+|J_3|\leq |I|}(1+|\Gamma^{J_1} H|)|\rd\Gamma^{J_2}h||\rd\Gamma^{J_3}h_S|\ls \frac{\ep\log (2+s)}{(1+s)^2}\sum_{|J_1|+|J_2|\leq |I|}(1+ |\Gamma^{J_1} h|)|\rd\Gamma^{J_2} h|.$$
If $|J_1|\leq |J_2|$, this can be controlled using the bootstrap assumption \eqref{BA4} by
$$\ls\frac{\ep\log (2+s)}{(1+s)^2}\sum_{|J_2|\leq |I|}|\rd\Gamma^{J_2} h|,$$
which is acceptable.
On the other hand, if $|J_1|\geq |J_2|$, we bound the above expression using the bootstrap assumption \eqref{BA1} by
$$\ls\frac{\ep\log (2+s)}{(1+s)^2}\sum_{|J_2|\leq |I|}|\rd\Gamma^{J_2} h|+\frac{\ep\log (2+s)}{(1+s)^{3-\de_0}(1+|q|)^{\f12-\f{\gamma}{4}}w(q)^{\f12}}\sum_{|J_1|\leq |I|} |\Gamma^{J_1} h|,$$
which is also acceptable.

We now turn to the case $h_*\in\{h_B,h_S\}$. For these terms, we can apply the $L^\infty$ bound to $|\rd\Gamma^{J_2}h_*||\rd\Gamma^{J_3}h_S|$ and to obtain
$$\sum_{|J_2|+|J_3|\leq N}|\rd\Gamma^{J_2}h_*||\rd\Gamma^{J_3}h_S|\ls \frac{\ep\log (2+s)}{(1+s)^3(1+|q|)^{\gamma}}.$$
On the other hand, by \eqref{inverse.1} in Proposition \ref{inverse}, we have
$$\sum_{|J_1|\leq |I|}(1+|\Gamma^{J_1} H|)\ls 1+\sum_{|J|\leq |I|} |\Gamma^J h|.$$
Combining these estimates, we obtain
$$\sum_{|J_1|+|J_2|+|J_3|\leq |I|}(1+|\Gamma^{J_1} H|)|\rd\Gamma^{J_2}h_*||\rd\Gamma^{J_3}h_S|\ls \frac{\ep\log (2+s)}{(1+s)^3(1+|q|)^{\gamma}}(1+\sum_{|J|\leq |I|} |\Gamma^J h|).$$
This clearly obeys the bounds stated in the proposition.
Notice that here we do not need to use any structure of the quadratic form $Q$.

\noindent{\bf Estimates for \eqref{Q.est.3}}

Since this is a cubic term, we do not need to exploit any structure of the nonlinearity. We can assume without loss of generality that in the case $h_*=h$, we have $|J_2|\leq |J_3|$. Therefore, we have the pointwise bound
$$|\rd\Gamma^{J_2}h_*|\ls \f{1}{(1+s)^{1-\de_0}(1+|q|)^{\gamma+\de_0}}.$$
Therefore, when combining this estimate with \eqref{inverse.1} in Proposition \ref{inverse}, we get 
\begin{equation}\label{q.est.3.1}
\begin{split}
 &\sum_{|J_1|+|J_2|+|J_3|\leq |I|}|\Gamma^{J_1} H||\rd\Gamma^{J_2}h_*||\rd\Gamma^{J_3} h|\\
\ls &\sum_{|J_3|\leq |I|}\f{\log(2+s)}{1+s}\f{|\rd\Gamma^{J_3}h|}{(1+s)^{1-\de_0}(1+|q|)^{\gamma+\de_0}}+\sum_{|J_1|+|J_3|\leq |I|}\f{|\Gamma^{J_1}h||\rd\Gamma^{J_3}h|}{(1+s)^{1-\de_0}(1+|q|)^{\gamma+\de_0}}.
\end{split}
\end{equation}
The first term in \eqref{q.est.3.1} is clearly acceptable. For the second term, we consider the cases $|J_1|\leq |J_3|$ and $|J_1|>|J_3|$ separately. In the case $|J_1|\leq |J_3|$, by the bootstrap assumption \eqref{BA4}, the second term on the right hand side of \eqref{q.est.3.1} is bounded by
$$\ls \sum_{|J|\leq |I|}\f{(1+|q|)^{\f12+\f{\gamma}{4}-\gamma-\de_0}|\rd\Gamma^{J}h|}{(1+s)^{2-2\de_0}w(q)^{\f12}}.$$
This is acceptable since $1-\gamma-2\de_0>\f 12+\f{\gamma}{4}-\gamma-\de_0$ for $\de_0$ satisfying \eqref{de_0.def}. In the case $|J_1|> |J_3|$, by the bootstrap assumption \eqref{BA1}, the second term on the right hand side of \eqref{q.est.3.1} is instead bounded by
$$\ls \sum_{|J|\leq |I|}\f{|\Gamma^{J}h|}{(1+s)^{2-2\de_0}(1+|q|)^{\f12-\f{\gamma}{4}+\gamma+\de_0}w(q)^{\f12}},$$
which is also acceptable.

\noindent{\bf Estimates for \eqref{Q.est.4}}

Finally, in order to control \eqref{Q.est.4}, we can assume that both instance of $h_*$ are in fact $h_B$ for otherwise, this can be bounded in a similar manner as \eqref{Q.est.2} and \eqref{Q.est.3}. We can therefore use the pointwise bound (see Definition \ref{def.dispersivespt})
$$\sum_{|J_2|+|J_3|\leq N}|\rd\Gamma^{J_2}h_B||\rd\Gamma^{J_3}h_B|\ls \f{1}{(1+s)^2(1+|q|)^{2\gamma}}.$$
Then, using the bound for $\Gamma^{J_1}(H-H_B)$ in \eqref{inverse.3} in Proposition \ref{inverse}, we get
\begin{equation*}
\begin{split}
&\sum_{|J_1|+|J_2|+|J_3|\leq |I|}|\Gamma^{J_1}(H-H_B)||\rd\Gamma^{J_2} h_*||\rd \Gamma^{J_3}h_*|\\
\ls &\f{\ep\log (2+s)}{(1+s)^3(1+|q|)^{2\gamma}}+\sum_{|J|\leq |I|}\f{|\Gamma^J h|}{(1+s)^2(1+|q|)^{2\gamma}},
\end{split}
\end{equation*}
which is acceptable.
\end{proof}
We then turn to the remaining quadratic terms $P$ in $\rd g$ for which the classical null condition is violated. While the classical null condition is violated, as observed by Lindblad-Rodnianski, there is a weak null structure which can be exploited. Here, we in particular need to make use of the generalized wave coordinate condition and Proposition \ref{wave.con.higher}.
\begin{proposition}\label{P.est.1}
For $|I|\leq N$, we have
\begin{equation*}
\begin{split}
&|\Gamma^I(P(g,g;\rd_\mu g,\rd_{\nu} g)-P(g_B,g_B;\rd_\mu g_B,\rd_{\nu} g_B))|\\
\ls &\f{\ep \log (2+s)}{(1+s)^{3-\de_0}(1+|q|)^{\de_0} w(q)^{\f{\gamma}{1+2\gamma}}}+|\rd \Gamma^I h|_{\mathcal T\mathcal U}|\rd h|_{\mathcal T\mathcal U}+\f{|\rd \Gamma^I h|_{\mathcal T\mathcal U}}{(1+s)(1+|q|)^\gamma}\\
&+\sum_{|J|\leq |I|}\f{|\rd \Gamma^J h|}{(1+s)^{1+\gamma}}+\sum_{|J|\leq |I|-1}\f{|\rd \Gamma^J h|}{(1+s)(1+|q|)^{\gamma}}\\
&+\sum_{|J|\leq |I|}\f{|\db\Gamma^J h|}{(1+s)^{1-\de_0}(1+|q|)^{\gamma+\de_0}}+\sum_{|J|\leq |I|}\f{\log(2+s)|\Gamma^J h|}{(1+s)^{2-2\de_0}(1+|q|)^{\gamma+2\de_0}}\\
&+\sum_{\substack{|J_1|+|J_2|\leq |I|\\\max\{|J_1|,|J_2|\}\leq |I|-1}}|\rd \Gamma^{J_1} h||\rd \Gamma^{J_2} h|.
\end{split}
\end{equation*}
\end{proposition}
\begin{proof}
Notice that in the proof of Proposition \ref{Q.est}, among the terms \eqref{Q.est.1}, \eqref{Q.est.2}, \eqref{Q.est.3} and \eqref{Q.est.4}, we have only used the null structure of $Q_{\mu\nu}$ in the bounds for the term \eqref{Q.est.1}:
$$\Gamma^I(Q(m,m;\rd_\mu h_*,\rd_{\nu} h)).$$
Therefore, we can now focus our attention to the term\footnote{Here, we recall the notation from the proof of Proposition \ref{Q.est} that $h_*\in \{h,h_S,h_B\}$.}
$$\Gamma^I(P(m,m;\rd_\mu h_*,\rd_{\nu} h)),$$
since all the remaining term can be bounded analogously as in Proposition \ref{Q.est}. Moreover, we can make a further reduction and assume that $h_*\in\{h,h_B\}$, since if $h_*=h_S$, we have better decay for $|\rd\Gamma^J h_S|$ and the term can be controlled in the same way as the first term in \eqref{Q.est.2}. Recalling the structure of the term $P$, we therefore have to bound the two terms
\begin{equation}\label{P.est.1.0.1}
\left|\Gamma^I\left(m^{\alp\alp'}\rd_\mu (h_*)_{\alp \alp'}m^{\bt \bt'}\rd_\nu h_{\bt \bt'}\right)\right|
\end{equation}
and
\begin{equation}\label{P.est.1.0.2}
\left|\Gamma^I\left(m^{\alp\alp'}\rd_\mu (h_*)_{\alp \bt}m^{\bt \bt'}\rd_\nu h_{\alp' \bt'}\right)\right|.
\end{equation}
First, notice that for the terms in \eqref{P.est.1.0.1} and \eqref{P.est.1.0.2} in which we do not have the highest derivatives $\rd\Gamma^I h$, we can simply estimate by naive bound
\begin{equation}\label{P.est.1.0.0}
\ls\sum_{|J|\leq |I|-1}\f{|\rd\Gamma^J h|}{(1+s)(1+|q|)^{\gamma}}+\sum_{\substack{|J_1|+|J_2|\leq |I|\\\max\{|J_1|,|J_2|\}\leq |I|-1}}|\rd \Gamma^{J_1} h||\rd \Gamma^{J_2} h|,
\end{equation}
which is acceptable. In the first term above, we have used the bounds in Definition \ref{def.dispersivespt} to control $\rd\Gamma^J h_B$.

It therefore remains to consider the highest order terms in $h$, i.e.~when we have $\rd\Gamma^I h$. Noticing that $m^{\Lb\Lb}=m^{\Lb A}=0$, we must have one of the following two scenarios: either we do not have the $\Lb\Lb$ component in either of the factors, i.e.
\begin{equation}\label{P.est.1.0}
|\rd h_*|_{\mathcal T\mathcal U}|\rd \Gamma^I h|_{\mathcal T\mathcal U}
\end{equation}
or the $\Lb\Lb$ component is coupled with a ``good'' $LL$ component, i.e.
\begin{equation}\label{P.est.1.1}
|\rd h_*|_{LL}|\rd \Gamma^I h|_{\Lb\Lb}+|\rd h_*|_{\Lb\Lb}|\rd \Gamma^I h|_{LL}.
\end{equation}
In the first case, i.e.~\eqref{P.est.1.0}, we have the bound
\begin{equation}\label{P.est.1.2}
|\rd h_*|_{\mathcal T\mathcal U}|\rd \Gamma^I h|_{\mathcal T\mathcal U}\ls |\rd \Gamma^I h|_{\mathcal T\mathcal U}(|\rd h|_{\mathcal T\mathcal U}+\f{1}{(1+s)(1+|q|)^\gamma}).
\end{equation}
In the second case, i.e.~\eqref{P.est.1.1}, we first note that by Definition \ref{def.dispersivespt}, we have
$$|h_B|_{LL}+|\Gamma h_B|_{LL}\ls \f{1}{(1+s)^{1+\gamma}}.$$
By Proposition \ref{decay.weights}, we have
$$|\rd ((h_B)_{LL})|\ls \f{1}{(1+s)^{1+\gamma}(1+|q|)}.$$
Now, notice that $\rd_q$ commutes with the projection onto $L$, therefore we have
\begin{equation}\label{P.est.1.3}
|\rd h_B|_{LL}\ls \f{1}{(1+s)^{1+\gamma}(1+|q|)}+|\db h_B|\ls \f{1}{(1+s)^{1+\gamma}},
\end{equation}
using again Definition \ref{def.dispersivespt}. 

Now, we use Propositions \ref{wave.con.lower} and \ref{wave.con.higher} to control the terms $|\rd h|_{LL}$ and $|\rd \Gamma^I h|_{LL}$. More precisely, by Proposition \ref{wave.con.lower} and the bootstrap assumptions \eqref{BA1}, \eqref{BA2} and \eqref{BA4}, we have
\begin{equation}\label{P.est.1.4}
\begin{split}
|\rd h|_{LL}
\ls &\f{\ep^{\f12} }{(1+s)^{1+\gamma}}
\end{split}
\end{equation}
Therefore, by \eqref{P.est.1.3} and \eqref{P.est.1.4}, the first term in \eqref{P.est.1.1} can be controlled by
\begin{equation}\label{P.est.1.6}
|\rd h_*|_{LL}|\rd \Gamma^I h|_{\Lb\Lb}\ls \f{|\rd \Gamma^I h|}{(1+s)^{1+\gamma}},
\end{equation}
which is acceptable.

It now remains to control the second term in \eqref{P.est.1.1}. By Proposition \ref{wave.con.higher}, we have
\begin{equation}\label{P.est.1.5}
\begin{split}
&|\rd\Gamma^I h|_{LL}\\
\ls &\underbrace{\f{\ep\log(2+s)}{(1+s)^2 w(q)^{\f{\gamma}{1+2\gamma}}}}_{=:I}+\underbrace{\f{\log(2+s)}{(1+s)(1+|q|)^\gamma}\sum_{|J|\leq |I|}|\Gamma^J h|}_{=:II}\\
&+\underbrace{\f{\log(2+s)}{1+s}\sum_{|J|\leq |I|}|\rd \Gamma^J h|}_{=:III}+\underbrace{\sum_{|J_1|+|J_2|\leq |I|}|\Gamma^{J_1}h||\rd \Gamma^{J_2}h|}_{=:IV}+\underbrace{\sum_{|J|\leq |I|}|\db\Gamma^J h|}_{=:V}+\underbrace{\sum_{|J|\leq |I|-2}|\rd \Gamma^J h|}_{=:VI}.
\end{split}
\end{equation}
On the other hand, we have the naive bound
\begin{equation}\label{P.est.naive}
|\rd h_*|_{\Lb\Lb}\ls \f{1}{(1+s)^{1-\de_0}(1+|q|)^{\gamma+\de_0}}.
\end{equation}
This is already sufficient to control the terms $I$, $II$, $III$ and $V$ in \eqref{P.est.1.5} since
\begin{equation}\label{P.est.1.7}
\begin{split}
&\f{\left(I+II+III+V\right)}{(1+s)^{1-\de_0}(1+|q|)^{\gamma+\de_0}}\\
\ls &\f{\ep\log(2+s)}{(1+s)^{3-\de_0}(1+|q|)^{\gamma+\de_0} w(q)^{\f{\gamma}{1+2\gamma}}}+\sum_{|J|\leq |I|}\f{\log(2+s)|\Gamma^J h|}{(1+s)^{2-2\de_0}(1+|q|)^{\gamma+2\de_0}}\\
&+\sum_{|J|\leq |I|}\f{|\rd\Gamma^J h|}{(1+s)^{1+\gamma}}+\sum_{|J|\leq |I|} \f{|\db\Gamma^J h|}{(1+s)^{1-\de_0}(1+|q|)^{\gamma+\de_0}},
\end{split}
\end{equation}
which is acceptable. 

For the term $IV$, since either $|J_1|\leq \lfloor \f N2 \rfloor$ or $|J_2|\leq \lfloor \f N2 \rfloor$, we have
\begin{equation}\label{P.est.1.8}
\begin{split}
&|\rd h_*|_{\Lb\Lb}\times (IV)\\
\ls &|\rd h_*|_{\Lb\Lb} \left(\f{\ep^{\f 12}(1+|q|)^{\f 12+\f{\gamma}{4}}}{(1+s)^{1-\de_0}w(q)^{\f 12}}\sum_{|J|\leq |I|}|\rd\Gamma^J h|+\f{\ep^{\f 12}}{(1+s)^{1-\de_0}(1+|q|)^{\f 12-\f{\gamma}{4}}w(q)^{\f 12}}\sum_{|J|\leq |I|}|\Gamma^J h|\right)\\
\ls &\f{\ep^{\f 12}(1+|q|)^{\f 12-\f{3\gamma}{4}-\de_0}}{(1+s)^{2-2\de_0}w(q)^{\f 12}}\sum_{|J|\leq |I|}|\rd\Gamma^J h|+\f{\ep^{\f 12}}{(1+s)^{2-2\de_0}(1+|q|)^{\f 12+\f{3\gamma}{4}+\de_0}w(q)^{\f 12}}\sum_{|J|\leq |I|}|\Gamma^J h|,
\end{split}
\end{equation}
where we have used the bootstrap assumptions \eqref{BA1} and \eqref{BA4} as well as \eqref{P.est.naive}. This bound is acceptable since $1-2\de_0-\gamma>\f 12+\f{3\gamma}{4}+\de_0$ and $\f 12+\f{3\gamma}{4}+\de_0 >\gamma+2\de_0$ for $\de_0$ satisfying \eqref{de_0.def}.

Finally, the term $VI$ in \eqref{P.est.1.5} requires the application of the slightly more refined estimate separating the contributions from $|\rd h|$ and $|\rd h_B|$:
$$|\rd h_*|_{\Lb\Lb}\ls \f{1}{(1+s)(1+|q|)^{\gamma}}+|\rd h|$$
from which we obtain
\begin{equation}\label{P.est.1.9}
\begin{split}
&|\rd h_*|_{\Lb\Lb}\times (VI)\\
\ls &\f{1}{(1+s)(1+|q|)^{\gamma}}\sum_{|J|\leq |I|-2}|\rd \Gamma^J h|+\sum_{|J|\leq |I|-2}|\rd h||\rd \Gamma^J h|,
\end{split}
\end{equation}
which is acceptable.

Combining \eqref{P.est.1.5}, \eqref{P.est.1.7}, \eqref{P.est.1.8} and \eqref{P.est.1.9}, we have thus shown that $|\rd h_*|_{\Lb\Lb}|\rd \Gamma^I h|_{LL}$ is acceptable. Combining this with \eqref{P.est.1.0}, \eqref{P.est.1.2} and \eqref{P.est.1.6}, we have therefore proven that the terms \eqref{P.est.1.0.1} and \eqref{P.est.1.0.2} can be dominated by terms on the right hand side of the statement of the proposition. 

As mentioned in the beginning of the proof of the proposition, the terms other than \eqref{P.est.1.0.1} and \eqref{P.est.1.0.2} are either cubic or contain a factor of $\rd \Gamma^J h_S$. They can therefore be controlled in an identical manner as \eqref{Q.est.2}, \eqref{Q.est.3} and \eqref{Q.est.4} in the proof of Proposition \ref{Q.est}. Therefore, we have
\begin{equation*}
\begin{split}
&|\mbox{Cubic terms}|+|\mbox{Quadratic terms containing }\rd\Gamma^Jh_S|\\
\ls &\f{\ep \log (2+s)}{(1+s)^3 (1+|q|)^{\gamma}}+\sum_{|J|\leq |I|}\left(\frac{|\rd \Gamma^J h|}{(1+s)^{1+\gamma}}\right.\\
&\left.+\frac{|\bar\rd \Gamma^J h|}{(1+s)^{1-\delta_0}(1+|q|)^{\gamma+\de_0}}+\frac{|\Gamma^J h|}{(1+s)^{2-2\de_0}(1+|q|)^{2\gamma+2\de_0}}\right).
\end{split}
\end{equation*}
This concludes the proof of the proposition.

\end{proof}

Finally, we bound the last term in Proposition \ref{basic.eqn}, namely, \eqref{T.diff}. Since this term is only quadratic and there are no contributions from the metric terms, it is easy to see that we have the following bound.
\begin{proposition}\label{T.est.1}
Denote $\beta:=\phi-\phi_B$. For $|I|\leq N$, the following estimate holds:
\begin{equation*}
\begin{split}
&\left|\Gamma^I\left(T_{\mu\nu}(\rd \phi,\rd \phi)-T_{\mu\nu}(\rd \phi_B,\rd \phi_B)\right)\right|\\
\ls &\sum_{|J|\leq |I|}\f{|\rd\Gamma^J\beta|}{(1+s)(1+|q|)^{\gamma}}+\sum_{|J_1|+|J_2|\leq |I|}|\rd\Gamma^{J_1}\beta||\rd\Gamma^{J_2}\beta|.
\end{split}
\end{equation*}
\end{proposition}
\begin{proof}
After using (6) in Definition \ref{def.dispersivespt}, this is straightforward.
\end{proof}

Using the above propositions, we obtain
\begin{proposition}\label{schematic.eqn}
For $|I|\leq N$, the right hand side of the equation for $\Gamma^I h$ can be decomposed into the following terms
$$|\tBox_g\Gamma^I h| \ls \mathfrak I_I + \mathfrak G_I + \mathfrak B_I + \mathfrak T_I + \mathfrak L_I + \mathfrak W_I + \mathfrak N_I,$$
where each of these terms is defined as follows
$$\mathfrak I_I(t,x):= \frac{\ep \log^2 (2+s)}{(1+s)^2(1+|q|)w(q)^{\f{\gamma}{1+2\gamma}}},$$
$$\mathfrak G_I(t,x):= \sum_{|J|\leq |I|}\frac{|\Gamma^J h|_{LL}}{(1+s)(1+|q|)^{1+\gamma}}+\sum_{|J|\leq |I|}\frac{\left(|\bar{\rd}\Gamma^J h|+|\bar{\rd}\Gamma^J \beta|\right)}{(1+s)^{1-\de_0}(1+|q|)^{\gamma+\de_0}},$$
$$\mathfrak B_I(t,x):= |\rd\Gamma^I h|_{\mathcal T\mathcal U}|\rd h|_{\mathcal T\mathcal U}+\frac{\left(|\rd\Gamma^I h|_{\mathcal T\mathcal U}+|\rd\Gamma^I\beta|\right)}{(1+s)(1+|q|)^{\gamma}}+|\rd\Gamma^I \beta||\rd\beta|,$$
$$\mathfrak T_I(t,x):= \sum_{|J|\leq |I|}\frac{\left(|\rd\Gamma^J h|+|\rd\Gamma^J\beta|\right)}{(1+s)^{1+\frac{\gamma}{2}}},$$
$$\mathfrak L_I(t,x):= \sum_{|J|\leq |I|-1}\left(\frac{\log(2+s)\left(|\rd\Gamma^J h|+|\rd\Gamma^J\beta|\right)}{(1+s)(1+|q|)^{\gamma}}+\frac{|\Gamma^J h|}{(1+s)(1+|q|)^{1+\gamma}}\right),$$
$$\mathfrak W_I(t,x):= \sum_{|J|\leq |I|}\frac{\log(2+s)|\Gamma^J h|}{(1+s)^{2-2\de_0}(1+|q|)^{\gamma+2\de_0}},$$
and
\begin{equation*}
\begin{split}
\mathfrak N_I(t,x):= &\sum_{\substack{|J_1|+|J_2|\leq |I|\\ \max\{|J_1|,|J_2|\}\leq |I|-1}}\left(|\rd\Gamma^{J_1}h||\rd\Gamma^{J_2}h|+|\rd\Gamma^{J_1}\beta||\rd\Gamma^{J_2}\beta|+\frac{|\Gamma^{J_1}h||\rd\Gamma^{J_2}h|}{1+|q|}\right)\\
&+\sum_{|J_2|\leq |I|,\,|J_1|+(|J_2|-1)_+\leq |I|}\frac{|\Gamma^{J_1} h|_{LL}|\rd\Gamma^{J_2} h|}{1+|q|}.
\end{split}
\end{equation*}
We will call these terms the inhomogeneous term, the good term, the bad term, the top order term, the lower order term, the potential term and the nonlinear term respectively.
Moreover, when $|I|=0$, in addition to the bounds above, we also have:
$$\mathfrak I_0(t,x)\ls \frac{\ep \log^2(2+s)}{(1+s)^{2+\f{\gamma}{2}}(1+|q|)^{\f 12-\f\gamma 2-\de_0}w(q)^{\f12}}.$$
\end{proposition}
\begin{proof}
It suffices to show that the right hand side of Propositions \ref{gauge.est}, \ref{gB.est}, \ref{mixed.gauge.est}, \ref{commute.est}, \ref{quasilinear.est}, \ref{Q.est}, \ref{P.est.1} and \ref{T.est.1} can be controlled by the terms as stated in this proposition. We will briefly indicate how to bound the terms from each of these propositions:

\noindent{\bf Terms from Proposition \ref{gauge.est}}

Note that the first terms are different for $q\geq 0$ and $q<0$ while the second and third terms are the same. Combining the estimates for the first terms for $q\geq 0$ and $q<0$, we have a term
$$\f{\ep\log^2(2+s)}{(1+s)^3w(q)^{\f{\gamma}{1+2\gamma}}},$$
which can be dominated by $\mathfrak I_I$. The second term can be bounded by $\mathfrak T_I$ while the third term can be estimated by $\mathfrak W_I$.

\noindent{\bf Terms from Proposition \ref{gB.est}}

Similar to terms from Proposition \ref{gauge.est}, the first term can be controlled by $\mathfrak I_I$; the second term by $\mathfrak T_I$; and the third term by $\mathfrak W_I$.

\noindent{\bf Terms from Proposition \ref{mixed.gauge.est}}

The first term can be bounded by $\mathfrak I_I$; the second term can be controlled by $\mathfrak W_I$.

\noindent{\bf Terms from Proposition \ref{commute.est}}

The first term on the right hand side of Proposition \ref{commute.est}, i.e.~$\sum_{|J|\leq |I|-1} |\Gamma^J \tBox_g h|$ contains all the terms in Proposition \ref{basic.eqn} (with $I $ replaced by $J$ for $|J|\leq |I|-1$) except for \eqref{commute}. All these terms are controlled in the rest of the proof of the present proposition. We now move to the remaining terms. The second and third terms can be bounded by $\mathfrak T_I$, $\mathfrak W_I$ respectively. The fourth and fifth terms can be estimated by $\mathfrak L_I$. Finally, the sixth and seventh terms can be controlled by $\mathfrak N_I$.

\noindent{\bf Terms from Proposition \ref{quasilinear.est}} 

The first to fourth terms are controlled by $\mathfrak I_I$, $\mathfrak G_I$, $\mathfrak W_I$ and $\mathfrak L_I$ respectively.

\noindent{\bf Terms from Proposition \ref{Q.est}}

The four terms in Proposition \ref{Q.est} can be controlled by $\mathfrak I_I$, $\mathfrak T_I$, $\mathfrak G_I$ and $\mathfrak W_I$ respectively.

\noindent{\bf Terms from Proposition \ref{P.est.1}}

The first term can be controlled by $\mathfrak I_I$. The second and third terms are bounded by $\mathfrak B_I$. The fourth, fifth, sixth, seventh and eighth terms can be estimated by $\mathfrak T_I$, $\mathfrak L_I$, $\mathfrak G_I$, $\mathfrak W_I$ and $\mathfrak N_I$ respectively.

\noindent{\bf Terms from Proposition \ref{T.est.1}}

The first term can be estimated by $\mathfrak B_I$ while the second term can be controlled by $\mathfrak B_I$ and $\mathfrak N_I$.

\noindent{\bf The case $|I|=0$}

Finally, we show the improved estimate for the $\mathfrak I_0$ term. Indeed, we check that for most of the contributions to $\mathfrak I_I$, we have better decay in $s$. More precisely, the contributions to $\mathfrak I_I$ from Proposition \ref{gauge.est}, \ref{gB.est}, \ref{mixed.gauge.est}, \ref{Q.est}, \ref{P.est.1} can be bounded above by
$$\f{\ep\log^2(1+s)}{(1+s)^{3-\de_0}}.$$
The only contribution to $\mathfrak I_I$ for which we do not have such good estimates is therefore the term in Proposition \ref{quasilinear.est}. On the other hand, by Proposition \ref{quasilinear.est}, in the case $|I|=0$, we have the improved estimate \eqref{quasilinear.improved}. The conclusion hence follows. \qedhere

\end{proof}

Notice that if we naively apply Gr\"onwall's inequality, the bad term $\mathfrak B_I$ would in particular force the the energy to grow like $(1+t)^C$ for some large constant $C$ except in the regions where $|q|$ is large. We therefore need to further exploit the structure of the Einstein equations to get better bounds in the region where $|q|$ is finite. To this end, we recall that the only contributions to the term $\mathfrak B_I$ are from \eqref{P.diff} and \eqref{T.diff} and we therefore need more refined estimate compared to Propositions \ref{P.est.1} and \ref{T.est.1}, which will be carried out in Propositions \ref{P.est.2} and \ref{T.est.2} below:
\begin{proposition}\label{P.est.2}
Projecting to the vector fields $\{L,\Lb, E^1, E^2, E^3\}$, we have the following bounds for the term \eqref{P.diff}: If ${\bf E}^\mu, {\bf E}^\nu\in \{L,\Lb, E^1, E^2, E^3\}$ such that ${\bf E}^\mu\neq \Lb^\mu$ or ${\bf E}^\nu\neq \Lb^\nu$,
$$|{\bf E}^\mu {\bf E}^\nu\Gamma^I(P_{\mu\nu}(g,g;\rd_{\mu} g,\rd_{\nu} g)-P_{\mu\nu}(g_B;g_B;\rd_{\mu} g_B,\rd_{\nu} g_B))|
\ls 
\mathfrak I_I+\mathfrak G_I+\mathfrak T_I+\mathfrak L_I+\mathfrak W_I+\mathfrak N_I.$$
\end{proposition}
\begin{proof}
We revisit the proof of Proposition \ref{P.est.1}. Arguing as in the proof of Proposition \ref{P.est.1} it suffices to control the terms \eqref{P.est.1.0.1} and \eqref{P.est.1.0.2} where $h_*\in \{h,h_B\}$ and all the $\Gamma$ derivatives fall on $h$, i.e.~we have $\rd\Gamma^I h$. For these terms, if ${\bf E}^\mu\neq \Lb^\mu$ or ${\bf E}^\nu\neq \Lb^\nu$, then at least one of the derivatives is a good derivatives, and we can therefore bound them by
$$\ls |\bar{\rd}\Gamma^I h|(|\rd h|+|\rd h_B|)+|\rd\Gamma^I h|(|\bar{\rd} h|+|\bar{\rd} h_B|)\ls \mathfrak G_I+\mathfrak T_I.$$
\end{proof}

\begin{proposition}\label{T.est.2}
Projecting to the vector fields $\{L,\Lb, E^1, E^2, E^3\}$, we have the following bounds for the term \eqref{T.diff}: If ${\bf E}^\mu, {\bf E}^\nu\in \{L,\Lb, E^1, E^2, E^3\}$ such that ${\bf E}^\mu\neq \Lb^\mu$ or ${\bf E}^\nu\neq \Lb^\nu$,
$$|{\bf E}^\mu {\bf E}^\nu\Gamma^I(T_{\mu\nu}(\rd\phi,\rd\phi)-T_{\mu\nu}(\rd\phi_B,\rd\phi_B))|
\ls 
\mathfrak I_I+\mathfrak G_I+\mathfrak T_I+\mathfrak L_I+\mathfrak W_I+\mathfrak N_I.$$
\end{proposition}
\begin{proof}
We revisit the proof of Proposition \ref{T.est.1}, keeping track more carefully the terms: 
\begin{equation*}
\begin{split}
&-\Gamma^I(T_{\mu\nu}(\rd\phi,\rd\phi)-T_{\mu\nu}(\rd\phi_B,\rd\phi_B))\\
=&4\Gamma^I\left(\rd_\mu\phi_B\rd_\nu\beta+\rd_\mu\beta\rd_\nu\phi_B+\rd_\mu\beta\rd_\nu\beta\right)\\
=&4\left(\underbrace{\left(\rd_\mu\phi_B\right)\left(\rd_\nu\Gamma^I\beta\right)+\left(\rd_\mu\Gamma^I\beta\right)\left(\rd_\nu\phi_B\right)+\left(\rd_\mu\Gamma^I\beta\right)\left(\rd_\nu\beta\right)+\left(\rd_\mu\beta\right)\left(\rd_\nu\Gamma^I\beta\right)}_{\mbox{Main term}}\right)\\
&+\underbrace{O\left(\sum_{|J|\leq |I|-1}\f{|\rd\Gamma^{J}\beta|}{(1+s)(1+|q|)^{\gamma}}\right)}_{Error_1}+\underbrace{O\left(\sum_{\substack{|J_1|+|J_2|\leq |I|\\\max\{|J_1|,\,|J_2|\}\leq |I|-1}}|\rd\Gamma^{J_1}\beta||\rd\Gamma^{J_2}\beta|\right)}_{Error_2}.
\end{split}
\end{equation*}
$Error_1$ and $Error_2$ can be controlled by $\mathfrak L_I$ and $\mathfrak N_I$ respectively. If we are contracting with ${\bf E}^\mu$ and ${\bf E}^\nu$ such that at least one of them is in $\{L,E^1,E^2,E^3\}$, then the main term has at least one good derivative, i.e.~it can be bounded by
\begin{equation*}
\begin{split}
&|\db\Gamma^I\beta||\rd\phi_B|+|\rd\Gamma^I\beta||\db\phi_B|+|\rd\Gamma^I\beta||\db\beta|+|\db\Gamma^I\beta||\rd\beta|\\
\ls &\f{|\db\Gamma^I\beta|}{(1+s)^{1-\de_0}(1+|q|)^{\gamma+\de_0}}+\f{|\rd\Gamma^I\beta|}{(1+s)^{1+\gamma}},
\end{split}
\end{equation*}
where we have used the bootstrap assumptions \eqref{BASF1} and \eqref{BASF2} together with (6) in Definition \ref{def.dispersivespt}. It is then easy to check that these two terms can be dominated by $\mathfrak G_I$ and $\mathfrak T_I$ respectively. 
\end{proof}

As a consequence of Propositions \ref{P.est.2} and \ref{T.est.2}, we thus obtain that the following components of $\tBox_g\Gamma^I h$ are better behaved in the sense that the bad term $\mathfrak B_I$ is absent:
\begin{proposition}\label{schematic.eqn.improved}
Projecting to the vector fields $\{L,\Lb, E^1, E^2, E^3\}$, we have the following bounds for $|\tBox_g \Gamma^I h |$: If ${\bf E}^\mu, {\bf E}^\nu\in \{L,\Lb, E^1, E^2, E^3\}$ such that ${\bf E}^\mu\neq \Lb^\mu$ or ${\bf E}^\nu\neq \Lb^\nu$,
$$|(\tBox_g \Gamma^I h)_{\mu\nu} {\bf E}^{\mu} {\bf E}^{\nu}|\ls 
\mathfrak I_I+\mathfrak G_I+\mathfrak T_I+\mathfrak L_I+\mathfrak W_I+\mathfrak N_I.$$
In other words, compared to Proposition \ref{schematic.eqn}, the term $\mathfrak B_I$ is absent. 
\end{proposition}
\begin{proof}
Returning to the proof of Proposition \ref{schematic.eqn}, one sees that the only contributions for the term $\mathfrak B_I$ come from the terms \eqref{P.diff} and \eqref{T.diff}. The conclusion thus follows from Propositions~\ref{P.est.2} and \ref{T.est.2}.
\end{proof}

\section{Equation for the scalar field}\label{sec.eqn.scalar}

In order to close the estimates for the Einstein scalar field system, we need to control the scalar field in addition to the metric. In this section, we derive an analogue of Proposition \ref{schematic.eqn} for $\Gamma^I\beta$, in which we estimate $|\tBox_g\Gamma^I \beta|$. Since the scalar wave equation is considerable simpler, the derivation of this analogous result is also simpler. We note that the terms in Proposition \ref{SF.schematic.eqn} below are similar to those in Proposition \ref{schematic.eqn} --- indeed most terms are subsets of those in Proposition \ref{schematic.eqn}. Most importantly, however, notice that there are no analogue of the term $\mathfrak B_I$ in Proposition \ref{SF.schematic.eqn}.

\begin{proposition}\label{SF.schematic.eqn}
For $|I|\leq N$ and $\beta:=\phi-\phi_B$, the right hand side of the equation for $\Gamma^I\beta$ can be decomposed into the following terms:
$$|\tBox_g \Gamma^I \beta|\ls \mathfrak I_I^{(\phi)}+\mathfrak G_I^{(\phi)}+\mathfrak T_I^{(\phi)}+\mathfrak L_I^{(\phi)}+\mathfrak W_I^{(\phi)}+\mathfrak N_I^{(\phi)},$$
where
$$\mathfrak I_I^{(\phi)}(t,x):= \frac{\ep \log^2 (2+s)}{(1+s)^2(1+|q|)w(q)^{\f{\gamma}{1+2\gamma}}},$$
$$\mathfrak G_I^{(\phi)}(t,x):= \sum_{|J|\leq |I|}\frac{|\Gamma^J h|_{LL}}{(1+s)(1+|q|)^{1+\gamma}},$$
$$\mathfrak T_I^{(\phi)}(t,x):= \sum_{|J|\leq |I|}\frac{|\rd\Gamma^J\beta|}{(1+s)^{1+\frac{\gamma}{2}}},$$
$$\mathfrak L_I^{(\phi)}(t,x):= \sum_{|J|\leq |I|-1}\left(\frac{\log(2+s)|\rd\Gamma^J\beta|}{(1+s)(1+|q|)^{\gamma}}+\frac{|\Gamma^J h|}{(1+s)(1+|q|)^{1+\gamma}}\right),$$
$$\mathfrak W_I^{(\phi)}(t,x):= \sum_{|J|\leq |I|}\frac{\log(2+s)|\Gamma^J h|}{(1+s)^{2-2\de_0}(1+|q|)^{\gamma+2\de_0}},$$
and
\begin{equation*}
\begin{split}
\mathfrak N_I^{(\phi)}(t,x):= &\sum_{\substack{|J_1|+|J_2|\leq |I|\\ \max\{|J_1|,|J_2|\}\leq |I|-1}}\frac{|\Gamma^{J_1}h||\rd\Gamma^{J_2}\beta|}{1+|q|}+\sum_{|J_2|\leq |I|,\,|J_1|+(|J_2|-1)_+\leq |I|}\frac{|\Gamma^{J_1} h|_{LL}|\rd\Gamma^{J_2} \beta|}{1+|q|}.
\end{split}
\end{equation*}
As in Proposition \ref{schematic.eqn}, we will call these terms the inhomogeneous term, the good term, the top order term, the lower order term, the potential term and the nonlinear term respectively. Moreover, when $|I|=0$, in addition to the bounds above, we also have:
$$\mathfrak I_0^{(\phi)}(t,x)\ls \frac{\ep \log^2(2+s)}{(1+s)^{2+\f{\gamma}{2}}(1+|q|)^{\f 12-\f\gamma 2-\de_0}w(q)^{\f12}}.$$
Importantly, notice that for the equation of the scalar field, there are no bad terms, i.e.~there are no analogue of the term $\mathfrak B_I$ in Proposition \ref{schematic.eqn}.
\end{proposition}
\begin{proof}
Subtracting the equation $\Box_{g_B}\phi_B=0$ from $\Box_g\phi=0$ and rewriting\footnote{We rewrite $\Box$ into $\tBox$ so that the terms $II$ and $IV$ below can be handled using the calculations in Propositions \ref{commutation}, \ref{commute.est} and \ref{quasilinear.est}.} $\Box_g=\tBox_g+(\Box_g-\tBox_g)$ (and similarly for $\Box_{g_B}$), we obtain
\begin{equation}\label{SF.schematic.eqn.main}
\begin{split}
0=&\Gamma^I\left(\Box_g\phi-\Box_{g_B}\phi_B\right)\\
=&\underbrace{ \tBox_g \Gamma^I\beta}_{=:I}-\underbrace{ [\tBox_g,\Gamma^I] \beta}_{=:II}+\underbrace{(\Gamma^I\Box_g\beta-\Gamma^I\tBox_g\beta)}_{=:III}\\
&+\underbrace{\Gamma^I(\tBox_g-\tBox_{g_B})\phi_B}_{=:IV}+\underbrace{\Gamma^I\left(\left((\Box_g-\tBox_g)-(\Box_{g_B}-\tBox_{g_B})\right)\phi_B\right)}_{=:V}.
\end{split}
\end{equation}
$I$ is the main term. We thus need to control the remaining terms. For $II$, we apply Proposition \ref{commutation}, use the bootstrap assumptions \eqref{BA1}-\eqref{BA5} and \eqref{BASF1}-\eqref{BASF2} and argue as in Proposition \ref{commute.est} to obtain 
\begin{equation}\label{SF.schematic.eqn.main.II}
\begin{split}
|II|\ls &\sum_{|J|\leq |I|-1}|\Gamma^J\tBox_g  \beta|+\sum_{|J|\leq |I|}\frac{|\rd\Gamma^J \beta|}{(1+s)^{1+\frac{\gamma}{2}}}+\sum_{|J|\leq |I|}\f{|\Gamma^J h|}{(1+s)^{2-2\de_0}(1+|q|)^{\f 12-\f \gamma 4}w(q)^{\f 12}}\\
&+\sum_{|J|\leq |I|-1}\frac{\log(2+s)|\rd \Gamma^J \beta|}{(1+s)(1+|q|)}+\sum_{|J|\leq |I|-1}\f{|\Gamma^J h|}{(1+s)(1+|q|)^{1+\gamma}}\\
&+\sum_{|J_2|\leq |I|,\,|J_1|+(|J_2|-1)_+\leq |I|}\frac{|\Gamma^{J_1} h|_{LL}|\rd\Gamma^{J_2} \beta|}{1+|q|}+\sum_{\substack{|J_1|+|J_2|\leq |I|\\\max\{|J_1|,|J_2|\}\leq |I|-1}}\frac{|\Gamma^{J_1} h||\rd \Gamma^{J_2} \beta|}{1+|q|}.
\end{split}
\end{equation}
For the term $IV$, we use the bootstrap assumptions \eqref{BA1}-\eqref{BA5} and \eqref{BASF1}-\eqref{BASF2} and argue as in Proposition \ref{quasilinear.est} to obtain
\begin{equation}\label{SF.schematic.eqn.main.IV}
\begin{split}
|IV|
\ls &\frac{\ep \log (2+s)}{(1+s)^2(1+|q|)^{1+\gamma}}+\frac{|\Gamma^I h|_{LL}}{(1+s)(1+|q|)^{1+\gamma}}\\
&+\sum_{|J|\leq |I|}\f{|\Gamma^J h|}{(1+s)^{2-\de_0}(1+|q|)^{\f12+\f{3\gamma}{4}}w(q)^{\f12}}+\sum_{|J|\leq |I|-1}\frac{|\Gamma^J h|}{(1+s)(1+|q|)^{1+\gamma}}.
\end{split}
\end{equation}
For the terms $III$ and $V$ in \eqref{SF.schematic.eqn.main}, notice that by \eqref{wave.coord.det}, for any scalar function $\xi$, we have
\begin{equation}\label{SF.wave.coord.det}
(\Box_g-\tBox_g)\xi=(\mathcal G_S^\mu+\mathcal G_B^\mu)\rd_\mu\xi,\quad (\Box_{g_B}-\tBox_{g_B})\xi=\mathcal G_B^\mu\rd_\mu\xi.
\end{equation}
Therefore, by \eqref{gauge.est.1} and (8) in Definition \ref{def.dispersivespt}, we have
\begin{equation}\label{SF.schematic.eqn.main.III}
\begin{split}
|III|\ls \sum_{|J_1|+|J_2|\leq |I|} (|\Gamma^{J_1}\mathcal G_S|+|\Gamma^{J_1}\mathcal G_B|)|\rd\Gamma^{J_2}\beta|\ls \f{\log(2+s)}{(1+s)^2}\sum_{|J|\leq |I|}|\rd\Gamma^J\beta|.
\end{split}
\end{equation}
Finally, for the term $V$, we have by \eqref{SF.wave.coord.det} that
$$\left((\Box_g-\tBox_g)-(\Box_{g_B}-\tBox_{g_B})\right)\phi_B=\mathcal G_S^{\mu}\rd_\mu\phi_B.$$
Therefore, by \eqref{gauge.est.1} and (6) in Definition \ref{def.dispersivespt},
\begin{equation}\label{SF.schematic.eqn.main.V}
|V|\ls \sum_{|J_1|+|J_2|\leq |I|}|\Gamma^{J_1}\mathcal G_S^{\mu}||\rd\Gamma^{J_2}\phi_B|\ls \f{\ep\log(2+s)}{(1+s)^3(1+|q|)^{\gamma}}.
\end{equation}
Combining \eqref{SF.schematic.eqn.main}, \eqref{SF.schematic.eqn.main.II}, \eqref{SF.schematic.eqn.main.IV}, \eqref{SF.schematic.eqn.main.III} and \eqref{SF.schematic.eqn.main.V}, we thus obtain the desired conclusion.
\end{proof}

\section{Linear estimates}\label{sec.linear.estimates}
In this section, we prove some linear estimates for the wave equation on the curved background $(\mathcal M,g)$ where $g$ satisfies the bootstrap assumptions \eqref{BA1}-\eqref{BA5}. These include the energy estimates, the Hardy inequalities and various pointwise decay estimates. Many of these estimates are already present in the works of Lindblad--Rodnianski \cite{LR1, LR2}, but since we need various refinements and localized versions in our setting, we include the proofs for completeness.

\subsection{Energy estimates}
In this subsection, we derive the energy estimates. We first recall the energy estimates (Lemma 6.1 in \cite{LR2}).
\begin{proposition}\label{EE.0}
Suppose $H^{\alp\bt}=(g^{-1})^{\alp\bt}-m^{\alp\bt}$ satisfies $|H|\leq \frac 12$ in $\{t_1\leq t\leq t_2\}\times \mathbb R^3$.
Then for every solution $\xi$ to
$$\tBox_g\xi=F $$
with $\xi$ decaying sufficiently fast as $|x|\to\infty$, we have the following estimate
\begin{equation*}
\begin{split}
&\int_{\Sigma_{t_2}} (|\rd_t\xi|^2+|\nab\xi|^2) w(q_2)\,dx+2\int_{t_1}^{t_2}\int_{\Sigma_t}|\db\xi|^2(t,x) w'(q)\,dx\,dt\\
\leq &4\int_{\Sigma_{t_1}} (|\rd_t\xi|^2+|\nab\xi|^2) w(q_1)\,dx\\
&+2\int_{t_1}^{t_2}\int_{\Sigma_t}|2(\rd_{\alpha}H^{\alp\bt})\rd_{\bt}\xi\rd_t\xi-(\rd_tH^{\alp\bt})\rd_\alp\xi\rd_{\bt}\xi+2F\rd_t\xi|w(q)\,dx\,dt\\
&+2\int_{t_1}^{t_2}\int_{\Sigma_t}|H^{\alp\bt}\rd_{\alp}\xi\rd_{\bt}\xi+2(\f{x_i}{r}H^{i\bt}-H^{0\bt})\rd_\bt\xi\rd_t\xi|w'(q)\,dx\,dt.
\end{split}
\end{equation*}
Here, we recall from Remark \ref{rmk.q1} that we have used the notation $q_1:=r-t_1$, $q_2:=r-t_2$
\end{proposition}
Under the bootstrap assumptions \eqref{BA1}-\eqref{BA5}, we show that the estimate in Proposition \ref{EE.0} above implies the energy estimates in Proposition \ref{EE.1} below. This proposition can be thought of as the analogue of Proposition 6.2 in \cite{LR2}, which uses the structure of the terms on the right hand side of the estimate in Proposition \ref{EE.0}. In the scenario of the present paper, it is in particular important that we require $|h_B|_{L\mathcal T}$ to have better decay to avoid the term
$$\int_0^t\int_{\Sigma_t} \f{|\rd\xi|^2}{1+t} \,dx\,dt.$$
\begin{proposition}\label{EE.1}
Suppose $g$ satisfies the bootstrap assumptions \eqref{BA1}-\eqref{BA5}.
Then there exists $T>0$ sufficiently large such that for $T\leq t_1<t_2$ and for every solution $\xi$ to
$$\tBox_g\xi=F $$
with $\xi$ decaying sufficiently fast as $|x|\to\infty$, we have the following estimate:
\begin{equation*}
\begin{split}
&\left(\int_{\Sigma_{t_2}} (|\rd_t\xi|^2+|\nab\xi|^2)w(q_2)\,dx\right)^{\f12}+\left(\int_{t_1}^{t_2}\int_{\Sigma_t}|\db\xi|^2w'(q)\,dx\,dt\right)^{\f12}\\
\ls &\left(\int_{\Sigma_{t_1}} (|\rd_t\xi|^2+|\nab\xi|^2)w(q_1)\,dx\right)^{\f12}+\int_{t_1}^{t_2}\left(\int_{\Sigma_t}|F|^2 w(q)\,dx \right)^{\f12}\,dt.
\end{split}
\end{equation*}
\end{proposition}
\begin{proof}
By choosing $T$ sufficiently large, Definition \ref{def.dispersivespt}, \eqref{inverse.1} and \eqref{BA4} imply that $|H|\leq \f 12$ and hence Proposition \ref{EE.0} apply. The main point is that for each of the terms $H\rd\xi\rd\xi$ and $\rd H\rd\xi \rd\xi$ on the right hand side of the energy estimate in Proposition \ref{EE.0}, we can show that one of the following three possibilities holds: either we have a good derivative on $\xi$, i.e.~$\db\xi$; or we have a good derivative on $H$, i.e.~$\db H$; or we have a good component of $H$, i.e.~$H_{LL}$ or $(\rd H)_{LL}$. This observation is, of course, already present in \cite{LR1, LR2}. Here, we show that this structure together with our bootstrap assumptions give the claimed energy estimates. In particular, our gauge choice guarantees that the term $\int_0^t\int_{\Sigma_t} \f{|\rd\xi|^2}{1+t} \,dx\,dt$ is absent.

We now turn to the details. We first show that the terms $H\rd\xi\rd\xi$ and $\rd H\rd\xi \rd\xi$ have the structure that we mentioned above. More precisely, we have
\begin{align}
|(\rd_{\alpha}H^{\alp\bt})\rd_{\bt}\xi\rd_t\xi|\ls &(|\rd H|_{LL}+|\db H|)|\rd \xi|^2+|\rd H||\db\xi||\rd\xi|,\label{EE.1.0.1}\\
|(\rd_tH^{\alp\bt})\rd_\alp\xi\rd_{\bt}\xi|\ls &|\rd H|_{LL}|\rd \xi|^2+|\rd H||\db\xi||\rd\xi|,\label{EE.1.0.2}\\
|H^{\alp\bt}\rd_{\alp}\xi\rd_{\bt}\xi|\ls &|H|_{LL}|\rd \xi|^2+|H||\db\xi||\rd\xi|,\label{EE.1.0.3}\\
|(\f{x_i}{r}H^{i\bt}-H^{0\bt})\rd_\bt\xi\rd_t\xi|\ls &|H|_{LL}|\rd \xi|^2+|H||\db\xi||\rd\xi|.\label{EE.1.0.4}
\end{align}
\eqref{EE.1.0.1}-\eqref{EE.1.0.3} can be proven in a similar manner by writing $m^{\alp\alp'}m^{\bt\bt'}H_{\alp'\bt'}$ and using $m^{\alp\bt}=-L^{(\alp}\Lb^{\bt)}+\sum_{A=1}^3(E^A)^\alp(E^A)^\bt$; we omit the details. For \eqref{EE.1.0.4}, notice that $\f{x_i}{r}H^{i\bt}-H^{0\bt}=m^{\bt\bt'}L^\alp H_{\alp\bt'}$.

Now, we apply the estimates for $\sum_{|I|\leq 1} |\Gamma^I h|$ from Proposition \ref{inverse} together with the bootstrap assumptions \eqref{BA1}-\eqref{BA5} and Proposition \ref{decay.weights} to get\footnote{The following estimates hold even with $t$ replaced by $s$, but this will not be necessary.} 
$$|H|\ls \f{(1+|q|)^{\f12+\f{\gamma}{4}}}{(1+t)^{1-\de_0}},\quad |\rd H|\ls \f{1}{(1+t)^{1-\de_0}(1+|q|)^{\gamma+\de_0}},\quad |\bar\rd H|\ls \f{1}{(1+t)^{1+\f{\gamma}{2}}},$$ 
$$|H|_{LL}\ls \f{(1+|q|)^{\f12+\gamma}}{(1+t)^{1+\f\gamma 2}},\quad |\rd H|_{LL}\ls \f{1}{(1+t)^{1+\f\gamma 2}}.$$
Therefore, we have 
$$|(\rd_{\alpha}H^{\alp\bt})\rd_{\bt}\xi\rd_t\xi|+|(\rd_tH^{\alp\bt})\rd_\alp\xi\rd_{\bt}\xi|\ls \f{|\rd\xi|^2}{(1+t)^{1+\f{\gamma}{2}}}+\f{|\db\xi||\rd\xi|}{(1+t)^{1-\de_0}(1+|q|)^{\gamma+\de_0}}$$
and
$$|H^{\alp\bt}\rd_{\alp}\xi\rd_{\bt}\xi|+|(\f{x_i}{r}H^{i\bt}-H^{0\bt})\rd_\bt\xi\rd_t\xi|\ls \f{(1+|q|)^{\f12+\gamma}|\rd\xi|^2}{(1+t)^{1+\f{\gamma}{2}}}+\f{(1+|q|)^{\f 12+\f{\gamma}{4}}|\db\xi||\rd\xi|}{(1+t)^{1-\de_0}}.$$
We now plug these estimates into the energy estimates in Proposition \ref{EE.0}. Since $w'(q)\ls \f{w(q)}{1+|q|}$, we have
\begin{equation}\label{EE.1.1}
\begin{split}
&\int_{\Sigma_{t_2}} (|\rd_t\xi|^2+|\nab\xi|^2)w(q_2)\,dx+\int_{t_1}^{t_2}\int_{\Sigma_t}|\db\xi|^2w'(q)\,dx\,dt\\
\ls &\int_{\Sigma_{t_1}} (|\rd_t\xi|^2+|\nab\xi|^2)w(q_1)\,dx+\int_{t_1}^{t_2}\int_{\Sigma_t}|F||\rd\xi|w(q)\,dx\,dt\\
&+\int_{t_1}^{t_2}\int_{\Sigma_t}\big(\f{|\rd\xi|^2}{(1+t)^{1+\f{\gamma}{2}}}+\f{|\db\xi||\rd\xi|}{(1+t)^{1-\de_0}(1+|q|)^{\gamma+\de_0}}\big)w(q)\,dx\,dt.
\end{split}
\end{equation}
To control the last term, we use the bound $w\ls w'(1+|q|)^{1+\f{\gamma}{2}}$, which implies after using the H\"older's inequality
\begin{equation*}
\begin{split}
&\int_{t_1}^{t_2}\int_{\Sigma_{t}}\f{|\db\xi||\rd\xi|}{(1+t)^{1-\de_0}(1+|q|)^{\gamma+\de_0}}w(q)\,dx\,dt\\
\ls &\left(\int_{t_1}^{t_2}\int_{\Sigma_{t}}\f{|\rd\xi|^2 w(q)(1+|q|)^{1+\f{\gamma}{2}}}{(1+t)^{2-2\de_0}(1+|q|)^{2\gamma+2\de_0}}\,dx\,dt\right)^{\f12}\left(\int_{t_1}^{t_2}\int_{\Sigma_{t}}\f{|\db\xi|^2 w(q)}{(1+|q|)^{1+\f{\gamma}{2}}}\,dx\,dt\right)^{\f12}\\
\ls &\big(\int_{t_1}^{t_2}\f{dt}{(1+t)^{1+\f{3\gamma}{4}}}\big)^{\f12}\big(\sup_{t_1\leq t\leq t_2}\int_{\Sigma_{t}} |\rd\xi|^2 w(q)\,dx \big)^{\f12}\big(\int_{t_1}^{t_2}\int_{\Sigma_t}|\db\xi|^2w'(q)\,dx\,dt\big)^{\f12}.
\end{split}
\end{equation*}
Notice that by taking $T$ large and $T\leq t_1\leq t_2$, the quantity $\big(\int_{t_1}^{t_2}\f{dt}{(1+t)^{1+\f{3\gamma}{4}}}\big)^{\f12}$
is bounded by a small constant and we can absorb this term to the left hand side of \eqref{EE.1.1}.
Therefore, we have
\begin{equation}\label{EE.1.2}
\begin{split}
&\int_{\Sigma_{t_2}} (|\rd_t\xi|^2+|\nab\xi|^2)w(q_2)\,dx+\int_{t_1}^{t_2}\int_{\Sigma_t}|\db\xi|^2w'(q)\,dx\,dt\\
\ls &\int_{\Sigma_{t_1}} (|\rd_t\xi|^2+|\nab\xi|^2)w(q_1)\,dx+\int_{t_1}^{t_2}\int_{\Sigma_t}|F||\rd\xi|w(q)\,dx\,dt\\
&+\int_{t_1}^{t_2}\int_{\Sigma_t}\f{|\rd\xi|^2}{(1+t)^{1+\f{\gamma}{2}}}w(q)\,dx\,dt.
\end{split}
\end{equation}
Applying the Gr\"onwall's inequality to \eqref{EE.1.2}, we get
\begin{equation}\label{EE.1.3}
\begin{split}
&\int_{\Sigma_{t_2}} (|\rd_t\xi|^2+|\nab\xi|^2)w(q_2)\,dx+\int_{t_1}^{t_2}\int_{\Sigma_t}|\db\xi|^2w'(q)\,dx\,dt\\
\ls &\int_{\Sigma_{t_1}} (|\rd_t\xi|^2+|\nab\xi|^2)w(q_1)\,dx+\int_{t_1}^{t_2}\int_{\Sigma_t}|F||\rd\xi|w(q)\,dx\,dt.
\end{split}
\end{equation}
Finally, applying H\"older's inequality to the last term and absorbing $\sup_{t\in[t_1,t_2]}\int_{\Sigma_{t}} |\rd\xi|^2w(q)\,dx$ to the left hand side, we obtain the desired conclusion.
\end{proof}

Unlike \cite{LR2}, we will also need energy estimates that are localized in various regions of the spacetime. 
To describe this localization, we introduce the hypersurface $\mathcal B_U$ defined by
$$\mathcal B_U=\{t-r-\frac{1}{(1+t)^{\f{\gamma}{4}}}=U \}.$$
Under the bootstrap assumptions \eqref{BA1}-\eqref{BA5} for the metric $g$, we can show that for every fixed $U$, there exists $T$ sufficiently large such that the restriction of $\mathcal B_U$ to $t\geq T$ is spacelike. In particular, this allows us to prove energy estimates in a region localized to the future of the $\{t=T\}$ hypersurface and to the past of $\mathcal B_U$. More precisely, we have
\begin{proposition}\label{EE.2}
For every fixed $U\in \mathbb R$, there exists $T>0$ sufficiently large such that if 
$$\tBox_g\xi=F$$
for $|\xi|$ decaying sufficiently fast in $r$ initially and $g$ obeying the bootstrap assumptions \eqref{BA1}-\eqref{BA5}, then 
\begin{equation*}
\begin{split}
&\left(\int_{\Sigma_{t_2}\cap\{t_2-r-\frac{1}{(1+t_2)^{\f{\gamma}{4}}}\leq U\}} |\rd\xi|^2 w(q_2)\,dx\right)^{\f12}+\left(\int_{t_1}^{t_2}\int_{\Sigma_t\cap\{t-r-\frac{1}{(1+t)^{\f{\gamma}{4}}}\leq U\}}|\db\xi|^2w'(q)\,dx\,dt\right)^{\f12}\\
&+\left(\int_{\mathcal B_U\cap\{t_1\leq t\leq t_2\}}(|\db\xi|^2+\f{|\rd\xi|^2}{(1+t)^{\f{\gamma}{4}+1}}) w(q)\,dx\right)^{\f12}\\
\ls &\left(\int_{\Sigma_{t_1}\cap\{t_1-r-\frac{1}{(1+t_1)^{\f{\gamma}{4}}}\leq U\}} |\rd\xi|^2 w(q_1)\,dx\right)^{\f12}+\int_{t_1}^{t_2}\left(\int_{\Sigma_t\cap\{t-r-\frac{1}{(1+t)^{\f{\gamma}{4}}}\leq U\}}|F|^2w(q)\,dx\right)^{\f12}\,dt
\end{split}
\end{equation*}
for $T\leq t_1< t_2$.
\end{proposition}
\begin{proof}
We first compute that
\begin{align*}
	&\tBox_g\xi \,w(q) \rd_t\xi =(g^{-1})^{\alp \bt} \rd_{\alp} \rd_{\bt} \xi \, w(q) \rd_{t} \xi \\
	= &  \rd_{j} ((g^{-1})^{j \bt} \rd_{\bt} \xi \rd_{t} \xi w(q) )	
	+ \rd_{t} \bb( (g^{-1})^{0 \bt} \rd_{\bt} \xi \rd_{t} \xi w(q) - \frac{1}{2} (g^{-1})^{\alp \bt}  (\rd_{\alp} \xi  \rd_{\bt} \xi) w(q) \bb)	\\
	& - \rd_{\alp} (g^{-1})^{\alp \bt} \rd_{\bt} \xi \rd_{t} \xi w(q)
	+ \frac{1}{2} (\rd_{t}(g^{-1})^{\alp \bt})  (\rd_{\alp} \xi  \rd_{\bt} \xi) w(q) \\
	& + \frac{1}{2} (g^{-1})^{\alp \bt } \rd_{\alp} \xi \rd_{\bt} \xi \rd_{t} w(q)
	- (g^{-1})^{\alp \bt} \rd_{\bt} \xi \rd_{t}  \xi \rd_{\alp} w(q). 
\end{align*}
Take this identity and integrate by parts in the region $\{(t,x^1,x^2,x^3): t_1\leq t\leq t_2,\,t-r-\frac{1}{(1+t)^{\f{\gamma}{4}}}\leq U\}$, we obtain
\begin{equation}\label{EE.2.1}
\begin{split}
&\int_{\Sigma_{t_2}\cap\{t_2-r-\frac{1}{(1+t_2)^{\f{\gamma}{4}}}\leq U\}} \big(-(g^{-1})^{0 \bt} \rd_{\bt} \xi \rd_{t} \xi  + \frac{1}{2} (g^{-1})^{\alp \bt}  (\rd_{\alp} \xi  \rd_{\bt} \xi)\big) w(q_2) \,dx\\
&+\int_{\mathcal B_U\cap\{t_1\leq t\leq t_2\}} \big(-(g^{-1})^{0\bt}\rd_\bt\xi\rd_t\xi+\f12 (g^{-1})^{\alp\bt}\rd_{\alp}\xi\rd_{\bt}\xi+\f{x_j (g^{-1})^{j\bt}\rd_{\bt}\xi\rd_t\xi}{r(1+\f{\gamma}{4(1+t)^{\f{\gamma}{4}+1}})}\big) w(q)\, dx\\
=&\int_{\Sigma_{t_1}\cap\{t_1-r-\frac{1}{(1+t_1)^{\f{\gamma}{4}}}\leq U\}} \big(-(g^{-1})^{0 \bt} \rd_{\bt} \xi \rd_{t} \xi  + \frac{1}{2} (g^{-1})^{\alp \bt}  (\rd_{\alp} \xi  \rd_{\bt} \xi)\big) w(q_1)\, dx\\
&+\int_{t_1}^{t_2}\int_{\Sigma_t\cap\{t-r-\frac{1}{(1+t)^{\f{\gamma}{4}}}\leq U\}} \big(- \rd_{\alp} (g^{-1})^{\alp \bt} \rd_{\bt} \xi \rd_{t} \xi w(q)
	+ \frac{1}{2} (\rd_{t}(g^{-1})^{\alp \bt})  (\rd_{\alp} \xi  \rd_{\bt} \xi) -\rd_t\xi F\big)w(q) \, dx\,dt\\
&	+ \int_{t_1}^{t_2}\int_{\Sigma_t\cap\{t-r-\frac{1}{(1+t)^{\f{\gamma}{4}}}\leq U\}}\big(\frac{1}{2} (g^{-1})^{\alp \bt } \rd_{\alp} \xi \rd_{\bt} \xi \rd_{t} w(q)
	- (g^{-1})^{\alp \bt} \rd_{\bt} \xi \rd_{t}  \xi \rd_{\alp} w(q)\big)\,dx\,dt.
\end{split}
\end{equation}
We now show that the boundary term on $B_U$ has a sign. First, we expand this term
\begin{equation}\label{bdry.term.+}
\begin{split}
&\int_{\mathcal B_U\cap\{t_1\leq t\leq t_2\}} \big(-(g^{-1})^{0\bt}\rd_\bt\xi\rd_t\xi+\f12 (g^{-1})^{\alp\bt}\rd_{\alp}\xi\rd_{\bt}\xi+\f{x_j (g^{-1})^{j\bt}\rd_{\bt}\xi\rd_t\xi}{r(1+\f{\gamma}{4(1+t)^{\f{\gamma}{4}+1}})}\big) w(q)\, dx\\
=&\int_{\mathcal B_U\cap\{t_1\leq t\leq t_2\}} \big( \f12(\rd_t\xi)^2+\f12(\rd_r\xi)^2+\f12|\nabb\xi|^2+\f{(\rd_r\xi\rd_t\xi)}{1+\f{\gamma}{4(1+t)^{\f{\gamma}{4}+1}}} \big) w(q)\, dx\\
&+\int_{\mathcal B_U\cap\{t_1\leq t\leq t_2\}} \big(-H^{0\bt}\rd_\bt\xi\rd_t\xi+\f12 H^{\alp\bt}\rd_{\alp}\xi\rd_{\bt}\xi+\f{x_j H^{j\bt}\rd_{\bt}\xi\rd_t\xi}{r(1+\f{\gamma}{4(1+t)^{\f{\gamma}{4}+1}})}\big) w(q)\, dx\\
=&\int_{\mathcal B_U\cap\{t_1\leq t\leq t_2\}} \big( (1+\f{1}{1+\f{\gamma}{4(1+t)^{\f{\gamma}{4}+1}}})(\rd_s\xi)^2+\f12|\nabb\xi|^2+\f{\gamma}{\gamma+4(1+t)^{\f{\gamma}{4}+1}}(\rd_q\xi)^2 \big) w(q)\, dx\\
&+\int_{\mathcal B_U\cap\{t_1\leq t\leq t_2\}} \big(-H^{0\bt}\rd_\bt\xi\rd_t\xi+\f12 H^{\alp\bt}\rd_{\alp}\xi\rd_{\bt}\xi+\f{x_j H^{j\bt}\rd_{\bt}\xi\rd_t\xi}{r(1+\f{\gamma}{4(1+t)^{\f{\gamma}{4}+1}})}\big) w(q)\, dx.
\end{split}
\end{equation}
Now, by \eqref{inverse.1} and \eqref{inverse.4}, the bootstrap assumption \eqref{BA4}, and \eqref{EE.1.0.3} and \eqref{EE.1.0.4},
\begin{equation}\label{bdry.term.+.1}
\begin{split}
&|-H^{0\bt}\rd_\bt\xi\rd_t\xi+\f12 H^{\alp\bt}\rd_{\alp}\xi\rd_{\bt}\xi+\f{x_j H^{j\bt}\rd_{\bt}\xi\rd_t\xi}{r(1+\f{\gamma}{4(1+t)^{\f{\gamma}{4}+1}})}|\\
\ls &|H^{\alp\bt}\rd_{\alp}\xi\rd_{\bt}\xi|+|(H^{0\bt}-\f{x_j}{r} H^{j\bt})\rd_{\bt}\xi\rd_t\xi|+\f{|H|}{(1+t)^{\f{\gamma}{4}+1}(1+\f{\gamma}{4(1+t)^{\f{\gamma}{4}+1}})}|\rd\xi|^2\\
\ls &(|H|_{LL}+\f{|H|}{2\gamma+8(1+t)^{\f{\gamma}{4}+1}})|\rd\xi|^2+|H||\db\xi||\rd\xi|\\
\ls &\f{(1+|U|)^{\f12+\gamma}}{(1+t)^{1+\f{\gamma}{2}}}|\rd\xi|^2+\f{|\db\xi|^2}{(1+t)^{\f{\gamma}{2}}}\\
\ll & \f{|\rd \xi|^2}{(1+t)^{\f{\gamma}{4}+1}}+|\db\xi|^2
\end{split}
\end{equation}
if $T$ is chosen to be sufficiently large depending on $U$. Therefore, returning to \eqref{bdry.term.+}, we obtain
\begin{equation}\label{bdry.term.+.2}
\begin{split}
&\int_{B_U\cap\{t_1\leq t\leq t_2\}} \big(-(g^{-1})^{0\bt}\rd_\bt\xi\rd_t\xi+\f12 (g^{-1})^{\alp\bt}\rd_{\alp}\xi\rd_{\bt}\xi+\f{x_j (g^{-1})^{j\bt}\rd_{\bt}\xi\rd_t\xi}{r(1+\f{\gamma}{4(1+t)^{\f{\gamma}{4}+1}})}\big) w(q)\, dx\\
\gtrsim &\int_{B_U\cap\{t_1\leq t\leq t_2\}} \big( |\db\xi|^2+\f{\gamma |\rd_q\xi|^2}{\gamma+4(1+t)^{\f{\gamma}{4}+1}} \big) w(q)\, dx,
\end{split}
\end{equation}
for $T$ sufficiently large. The other boundary terms on $\Sigma_{t_1}$ and $\Sigma_{t_2}$ in \eqref{EE.2.1} can be easily be controlled since $|H|\leq \f 12$ (which holds for $T$ sufficiently large by \eqref{inverse.1} and \eqref{BA4}) implies
\begin{equation}\label{bdry.t}
\f12 ((\rd_t\xi)^2+|\nab\xi|^2)\leq -(g^{-1})^{00}(\rd_t\xi)^2+(g^{-1})^{ij}(\rd_i\xi)(\rd_j\xi)\leq 2((\rd_t\xi)^2+|\nab\xi|^2).
\end{equation}
Next, we write $(g^{-1})^{\alp\bt}=m^{\alp\bt}+H^{\alp\bt}$ and consider the contributions from the Minkowski metric $m^{\alp\bt}$ in the last line of \eqref{EE.2.1}:
\begin{equation}\label{EE.2.2}
\begin{split}
&-\int_{t_1}^{t_2}\int_{\Sigma_t\cap\{t-r-\frac{1}{(1+t)^{\f{\gamma}{4}}}\leq U\}}\big(\frac{1}{2} m^{\alp \bt } \rd_{\alp} \xi \rd_{\bt} \xi \rd_{t} w(q)
	- m^{\alp \bt} \rd_{\bt} \xi \rd_{t}  \xi \rd_{\alp} w(q)\big)\,dx\,dt\\
\gtrsim &\int_{t_1}^{t_2}\int_{\Sigma_t\cap\{t-r-\frac{1}{(1+t)^{\f{\gamma}{4}}}\leq U\}} |\db\xi|^2 w'(q)\,dx\,dt.
\end{split}
\end{equation}
Therefore, substituting \eqref{bdry.term.+.2}, \eqref{bdry.t} and \eqref{EE.2.2} into \eqref{EE.2.1}, we obtain
\begin{equation*}
\begin{split}
&\int_{\Sigma_{t_2}\cap\{t_2-r-\frac{1}{(1+t_2)^{\f{\gamma}{4}}}\leq U\}} (|\rd_t\xi|^2+|\nab\xi|^2)w(q_2)\,dx+\int_{t_1}^{t_2}\int_{\Sigma_t\cap\{t-r-\frac{1}{(1+t)^{\f{\gamma}{4}}}\leq U\}}|\db\xi|^2w'(q)\,dx\,dt\\
&+\int_{B_U\cap\{t_1\leq t\leq t_2\}} \big( |\db\xi|^2+\f{\gamma |\rd_q\xi|^2}{\gamma+4(1+t)^{\f{\gamma}{4}+1}} \big) w(q)\, dx\\
\ls &\int_{\Sigma_{t_1}\cap\{t_1-r-\frac{1}{(1+t_1)^{\f{\gamma}{4}}}\leq U\}} (|\rd_t\xi|^2+|\nab\xi|^2)w(q_1)\,dx\\
&+\int_{t_1}^{t_2}\int_{\Sigma_t\cap\{t-r-\frac{1}{(1+t)^{\f{\gamma}{4}}}\leq U\}}|2(\rd_{\alpha}H^{\alp\bt})\rd_{\bt}\xi\rd_t\xi-(\rd_tH^{\alp\bt})\rd_\alp\xi\rd_{\bt}\xi+2F\rd_t\xi|w(q)\,dx\,dt\\
&+\int_{t_1}^{t_2}\int_{\Sigma_t\cap\{t-r-\frac{1}{(1+t)^{\f{\gamma}{4}}}\leq U\}}|H^{\alp\bt}\rd_{\alp}\xi\rd_{\bt}\xi+2(\f{x_i}{r}H^{i\bt}-H^{0\bt})\rd_\bt\xi\rd_t\xi|w'(q)\,dx\,dt.
\end{split}
\end{equation*}
Finally, we need to control the terms on the last two lines which are quadratic in $\rd\xi$. Of course, these terms are the same as those in Proposition \ref{EE.0} and as shown in the proof of Proposition \ref{EE.1}, they have a favorable structure. We can then control them in an identical manner as in the proof of Proposition \ref{EE.1} after choosing $T$ to be larger if necessary.

\end{proof}

There are obvious variations of Propositions \ref{EE.1} and \ref{EE.2} which allows us to also localize to the future of $\mathcal B_{U'}$. We summarize them below. Their proofs are completely analogous to Propositions \ref{EE.1} and \ref{EE.2}.

\begin{proposition}\label{EE.3}
For every fixed $U'<U$, there exists $T>0$ sufficiently large such that if 
$$\tBox_g\xi=F$$
for $|\xi|$ decaying sufficiently fast in $r$ initially and $g$ obeying the bootstrap assumptions \eqref{BA1}-\eqref{BA5}, then 
\begin{equation*}
\begin{split}
&(\int_{\Sigma_{t_2}\cap\{U'\leq t_2-r-\frac{1}{(1+t_2)^{\f{\gamma}{4}}}\leq U\}} |\rd\xi|^2 w(q_2)\,dx)^{\f12}+(\int_{t_1}^{t_2}\int_{\Sigma_t\cap\{U'\leq t-r-\frac{1}{(1+t)^{\f{\gamma}{4}}}\leq U\}}|\db\xi|^2w'(q)\,dx\,dt)^{\f12}\\
&+(\int_{\mathcal B_U\cap\{t_1\leq t\leq t_2\}}(|\db\xi|^2+\f{|\rd\xi|^2}{(1+t)^{\f{\gamma}{4}+1}}) w(q)\,dx)^{\f12}\\
\ls &(\int_{\Sigma_{t_1}\cap\{U'\leq t_1-r-\frac{1}{(1+t_1)^{\f{\gamma}{4}}}\leq U\}} |\rd\xi|^2 w(q_1)\,dx)^{\f12}+(\int_{\mathcal B_{U'}\cap\{t_1\leq t\leq t_2\}}(|\db\xi|^2+\f{|\rd\xi|^2}{(1+t)^{\f{\gamma}{4}+1}}) w(q)\,dx)^{\f12}\\
&+\int_{t_1}^{t_2}(\int_{\Sigma_t\cap\{U'\leq t-r-\frac{1}{(1+t)^{\f{\gamma}{4}}}\leq U\}}|F|^2 w(q)\,dx)^{\f12}\,dt
\end{split}
\end{equation*}
for $T\leq t_1\leq t_2$.
\end{proposition}
\begin{proof}
The proof is identical to that of Proposition \ref{EE.2} and will be omitted. We only note that by a similar argument as \eqref{bdry.term.+} and \eqref{bdry.term.+.1}, we can bound the term on $\mathcal B_{U'}$ by
$$\int_{\mathcal B_{U'}\cap\{t_1\leq t\leq t_2\}}(|\db\xi|^2+\f{|\rd\xi|^2}{(1+t)^{\f{\gamma}{4}+1}}) \,dx$$
after choosing $T$ to be sufficiently large.
\end{proof}

Similar, an analogous theorem holds if we only have a boundary $\mathcal B_{U'}$ in the past of the region:
\begin{proposition}\label{EE.4}
For every $U'\in \mathbb R$, there exists $T>0$ sufficiently large such that if 
$$\tBox_g\xi=F$$
for $|\xi|$ decaying sufficiently fast in $r$ initially and $g$ obeying the bootstrap assumptions \eqref{BA1}-\eqref{BA5}, then 
\begin{equation*}
\begin{split}
&(\int_{\Sigma_{t_2}\cap\{U'\leq t_2-r-\frac{1}{(1+t_2)^{\f{\gamma}{4}}}\}} |\rd\xi|^2 w(q_2)\,dx)^{\f12}+(\int_{t_1}^{t_2}\int_{\Sigma_t\cap\{U'\leq t-r-\frac{1}{(1+t)^{\f{\gamma}{4}}}\}}|\db\xi|^2w'(q)\,dx\,dt)^{\f12}\\
\ls &(\int_{\Sigma_{t_1}\cap\{U'\leq t_1-r-\frac{1}{(1+t_1)^{\f{\gamma}{4}}}\}} |\rd\xi|^2 w(q_1)\,dx)^{\f12}+(\int_{\mathcal B_{U'}\cap\{t_1\leq t\leq t_2\}}(|\db\xi|^2+\f{|\rd\xi|^2}{(1+t)^{\f{\gamma}{4}+1}}) w(q)\,dx)^{\f12}\\
&+\int_{t_1}^{t_2}(\int_{\Sigma_t\cap\{U'\leq t-r-\frac{1}{(1+t)^{\f{\gamma}{4}}}\}}|F|^2 w(q)\,dx)^{\f12}\,dt
\end{split}
\end{equation*}
for $T\leq t_1\leq t_2$.
\end{proposition}

\subsection{Hardy inequality}

In \cite{LR2}, a $|q|$-weighted Hardy inequality is proved: The main novelty is that the weights in $|q|$ are used instead of the $r$ weights in the classical Hardy inequality. This is useful in the setting of \cite{LR2} as there is ``insufficient $r$ decay'' near the wave zone (i.e.~when $t$ and $r$ are comparable). In our setting, we also need a similar version of the $|q|$-weighted Hardy inequality except that we also need to localize it to an annulus $\{R_1\leq r\leq R_2\}$, where $R_1$ and $R_2$ satisfy $0\leq R_1<R_2 \leq\infty$. More precisely, we have
\begin{proposition}\label{Hardy}
The following inequality holds for any $0\leq \alpha \leq 2$, $\mu_1>-1$, $\mu_2>0$ and for any scalar functions $\xi$ and for $R_1$, $R_2$ satisfying $0\leq R_1<R_2 \leq \infty$ (with an implicit constant depending on $\mu_1$ and $\mu_2$):
\begin{equation*}
\begin{split}
&\int\limits_{\{r=R_1\}} \f{r^2 (1_{\{q\geq 0\}}(1+|q|)^{\mu_2}+1_{\{q< 0\}}(1+|q|)^{-1-\mu_1})\xi^2}{(1+s)^\alpha} \sin\theta\,d\theta\,d\varphi\\
&+\int_{\min\{R_1,t\}}^{\min\{R_2,t\}}\int_{\mathbb S^2} \f{\xi^2}{(1+|q|)^{2+\mu_1}}\f{r^2 \sin\theta\,d\theta\,d\varphi \,dr}{(1+s)^\alpha}+\int_{\max\{R_1,t\}}^{\max\{R_2,t\}}\int_{\mathbb S^2} \f{\xi^2}{(1+|q|)^{1-\mu_2}}\f{r^2 \sin\theta\,d\theta\,d\varphi\,dr}{(1+s)^\alpha}\\
\ls &\int\limits_{\{r=R_2\}} \f{r^2 (1_{\{q\geq 0\}}(1+|q|)^{\mu_2}+1_{\{q< 0\}}(1+|q|)^{-1-\mu_1})\xi^2}{(1+s)^\alpha} \sin\theta\,d\theta\,d\varphi\\
&+\int_{\min\{R_1,t\}}^{\min\{R_2,t\}} \int_{\mathbb S^2}\f{|\rd_r\xi|^2}{(1+|q|)^{\mu_1}}\f{r^2 \sin\theta\,d\theta\,d\varphi\,dr}{(1+s)^\alpha}+\int_{\max\{R_1,t\}}^{\max\{R_2,t\}}\int_{\mathbb S^2} \f{|\rd_r\xi|^2 (1+|q|)^{1+\mu_2} r^2 \sin\theta\,d\theta\,d\varphi\,dr}{(1+s)^\alpha}.
\end{split}
\end{equation*}
Here, $1_{\{q<0\}}$ and $1_{\{q\geq 0\}}$ denote the indicator functions of the sets $\{q<0\}$ and $\{q\geq 0 \}$ respectively. Notice that one of the two integrals $\int_{\min\{R_1,t\}}^{\min\{R_2,t\}}$ and $\int_{\max\{R_1,t\}}^{\max\{R_2,t\}}$ can possibly be empty. Moreover, when $R_2=\infty$, we do not need the first integral on the right hand side.
\end{proposition}
\begin{proof}
As in \cite{LR2}, we consider the weight function\footnote{Notice that $n(q)$ is chosen to be continuous at $q=0$.}
$$n(q)=\begin{cases}(1+|q|)^{\mu_2} & \mbox{if }q\geq 0\\ (1+|q|)^{-1-\mu_1} & \mbox{if }q<0. \end{cases}$$
We then compute
$$n'(q)=\begin{cases}\mu_2(1+|q|)^{-1+\mu_2} & \mbox{if }q> 0\\ (1+\mu_1)(1+|q|)^{-2-\mu_1} & \mbox{if }q<0. \end{cases}$$
On the other hand, we have
$$\rd_r(\f{r^2 n(q)}{(1+s)^\alpha})=(\f 2r-\f{\alpha}{1+s}+\f{n'(q)}{n(q)})\f{r^2 n(q)}{(1+s)^\alpha}\geq \f{r^2 n'(q)}{(1+s)^\alpha},$$
which implies
$$\rd_r(\f{r^2 n(q)\xi^2}{(1+s)^\alpha})\geq \f{r^2 n'(q)\xi^2}{(1+s)^\alpha}+\f{2r^2 n(q)\xi\rd_r\xi}{(1+s)^\alpha}.$$
Integrating from $R_1$ to $R_2$ for every fixed $(\theta,\varphi)$, we get
\begin{equation*}
\begin{split}
& \f{r^2 n(q)\xi^2}{(1+s)^\alpha} \restriction_{r=R_1}+\int_{R_1}^{R_2} \f{r^2 n'(q)\xi^2}{(1+s)^\alpha} dr\\
\leq &\f{r^2 n(q)\xi^2}{(1+s)^\alpha} \restriction_{r=R_2}+\int_{R_1}^{R_2} \big|\f{2r^2 n(q)\xi\rd_r\xi}{(1+s)^\alpha}\big| dr.
\end{split}
\end{equation*}
Applying the Cauchy-Schwarz inequality, the last term can be controlled by 
$$\int_{R_1}^{R_2} \big|\f{2r^2 n(q)\xi\rd_r\xi}{(1+s)^\alpha}\big| dr\leq \int_{R_1}^{R_2} \f{r^2 n'(q)\xi^2}{(1+s)^\alpha} dr+\int_{R_1}^{R_2} \f{r^2 n^2(q)|\rd_r\xi|^2}{n'(q)(1+s)^\alpha} dr.$$
Notice that $|\f{n(q)}{n'(q)}|\ls (1+|q|)$. Therefore,  absorbing $\int_{R_1}^{R_2} \f{r^2 n'(q)\xi^2}{(1+s)^\alpha} dr$ to the left hand side and integrating over $\mathbb S^2$ with respect to $\sin\theta\,d\theta\,d\varphi$ yield the desired conclusion.
\end{proof}
\begin{remark}
In the case $R_1=0$ and $R_2=\infty$, we recover the Hardy inequality in \cite{LR2}.
\end{remark}

\subsection{Klainerman-Sobolev inequality}
 We record the following Klainerman-Sobolev inequality. Such global Sobolev inequalities first appeared in the work of Klainerman \cite{K}. The precise weighted version that we use can be found in Proposition 14.1 in \cite{LR2}.
\begin{proposition}\label{KS.ineq}
There exists a universal constant $C$ such that for any function\footnote{In applications, the functions $\xi$ that we consider will not be compactly supported, but the desired estimates nonetheless follow from a standard approximation argument.} $\xi\in C^\infty_c(\mathbb R^{3+1})$, the following estimate holds for $t\geq 0$: 
$$\sup_{x}|\xi(t,x)|(1+s)(1+|q|)^{\frac 12} w(q)^{\frac 12}\leq C\sum_{|I|\leq 3}\|w^{\frac 12}(|\cdot|-t)\Gamma^I\xi(t,\cdot) \|_{L^2(\mathbb R^3)} .$$
\end{proposition}

\subsection{Decay estimates}
While the Klainerman-Sobolev inequality gives pointwise decay of a function in terms of the weighted $L^2$ norms of its higher derivatives, in view of the fact that the bounds on the energies we obtain grow with time (see Propositions \ref{EE.R2}, \ref{EE.R3} and \ref{EE.R4}), it does not give the sharp pointwise decay near the wave zone. We thus need to complement this inequality with an ODE argument near the wave zone. This argument was first used in \cite{L1} for the constant coefficient wave equation and in \cite{LR1, LR2} for the variable coefficient wave equation. The proof of the following proposition is modified from \cite{LR2}. Notice in particular that in the gauge used in this paper, we have integrable decay for $H_{LL}$, which results in a slightly stronger proposition. 

First, we need a computation\footnote{Lemma 5.2 in \cite{LR2} gives slightly more information, but in order not to introduce additional notations, we only need the following consequence of it.} from \cite{LR2}:
\begin{proposition}[Lemma 5.2, \cite{LR2}]\label{ODE.lemma}
Suppose $\xi$ and $F$ are functions such that $\tBox_g\xi=F$. Then at points such that $H$ satisfies $|H|\leq \f14$, there exists a function $f(t,x,g)$ with $|f(t,x,g)|\ls |H|_{LL}$ such that the following estimate holds:
\begin{equation*}
\begin{split}
\left|\left(\rd_s +f\rd_q\right)\rd_q(r\xi)\right|
\ls &(1+\f{r|H|_{L\mathcal T}}{1+|q|}+|H|)r^{-1}\sum_{|I|\leq 2}|\Gamma^I\xi|+r^{-1}|H||\rd_q(r\xi)|+r|F|. 
\end{split}
\end{equation*}
\end{proposition}

Using Proposition \ref{ODE.lemma}, we obtain the following decay estimates. 
\begin{proposition}\label{decay.est}
Suppose $\tBox_g\xi=F$, where $\xi$ is either a scalar or a $2$-tensor and $g$ satisfies the bootstrap assumptions \eqref{BA1}-\eqref{BA5}. Then, for $\varpi(q):=(1+|q|)^{\f 12-\f{\gamma}{4}}w(q)^{\f12}$, the following decay estimate holds:
\begin{equation*}
\begin{split}
&\sup_x(1+s)\varpi(q)|\rd\xi(t,x)|\\
\ls &\sup_{0\leq \tau\leq t}\sum_{|I|\leq 1}\|\varpi(|\cdot|-\tau)\Gamma^I\xi(\tau,\cdot)\|_{L^\infty(\Sigma_\tau)}\\
&+\int_0^t \left((1+\tau)\|\varpi(|\cdot|-\tau)F(\tau,\cdot)\|_{L^\infty(D_\tau)}+\sum_{|I|\leq 2}(1+\tau)^{-1}\|\varpi(|\cdot|-\tau)\Gamma^I\xi(\tau,\cdot)\|_{L^\infty(D_\tau)}\right)d\tau.
\end{split}
\end{equation*}
Here, $D_t$ is defined to be the region $D_t:=\{x:\f t2\leq |x|\leq 2t\}$.
\end{proposition}
\begin{proof}
We will only prove the proposition where $\xi$ is a scalar --- the case where $\xi$ is a $2$-tensor can be proven analogously by considering separately every component with respect to the Minkowskian coordinates.

We now carry out some easy reductions. First, we can assume $t\geq T$ for any finite $T>0$, since as long as we allow the implicit constant to depend on $T$, the desired estimate when $t<T$ follows from Proposition \ref{decay.weights}. (In fact, the estimate holds even only with the first term on the right hand side.) We can therefore choose $T$ sufficiently large such that according to \eqref{inverse.1} and \eqref{BA4}, $|H|\leq \f 14$ holds for $t\geq T$.

Second, by Proposition \ref{decay.weights}, we only need to prove the desired estimate in the region $\f t2+\f12 \leq |x|\leq 2t-1$, since in the complement of this region, we have $(1+|q|)^{-1}\ls (1+s)^{-1}$. (As in the $t<T$ case, the estimate holds even only with the first term on the right hand side.)

By the choice of $T$ above, if $t\geq T$, then $|H|\leq \f 14$ and we can apply Proposition \ref{ODE.lemma}, which implies
\begin{equation}\label{ODE.lemma.2}
\begin{split}
&\left|\left(\rd_s +f\rd_q\right)(\varpi(q)\rd_q(r\xi))\right|\\
\ls &(1+\f{r|H|_{L\mathcal T}}{1+|q|}+|H|)r^{-1}\sum_{|I|\leq 2}\varpi(q)|\Gamma^I\xi|+(r^{-1}|H|\varpi(q)+|H|_{LL}\varpi'(q))|\rd_q(r\xi)|+r\varpi(q)|F|. 
\end{split}
\end{equation}
Let $(t,x)$ be such that $t\geq T$ and $\f t2+\f12 \leq |x|\leq 2t-1$. Consider the integral curves of the vector field $\rd_s+f\rd_q$ through the point $(t,x)$ restricted to the region $\cup_{\tau\geq T} D_\tau$. Since $|f|\ls |H|_{LL}\ls \f{(1+|q|)^{\f 12+\f{\gamma}{4}}}{(1+s)^{1+\f{\gamma}{2}}w(q)^{\f 12}}$ by \eqref{inverse.4}, the integral curve intersects the boundary of $\cup_{\tau\geq T} D_\tau$ at a point with $t$-value comparable to $t$. Integrating \eqref{ODE.lemma.2} along such an integral curve, using the Gr\"onwall's inequality and noting that $\int_1^{\infty} \sup_{x\in D_t}(r^{-1}|H|+|H|_{LL}\f{\varpi'(q)}{\varpi(q)})\, dt\ls \ep^{\f12}$ (by \eqref{inverse.1}, \eqref{inverse.5}, \eqref{BA4} and \eqref{BA5}) thus give the desired conclusion.
\end{proof}
As in \cite{LR1, LR2}, we note that the estimate in Proposition \ref{decay.est} still holds after projecting to the vector fields in $\mathcal T$ and/or $\mathcal U$. 
\begin{proposition}\label{decay.est.2}
Suppose $\tBox_g\xi=F$, where $\xi$ is a $2$-tensor and $g$ satisfies the bootstrap assumptions \eqref{BA1}-\eqref{BA5}. Then, for $\varpi(q):=(1+|q|)^{\f 12-\f{\gamma}{4}}w(q)^{\f12}$, the following decay estimate holds:
\begin{equation*}
\begin{split}
&\sup_x(1+s)\varpi(q)|\rd\xi(t,x)|_{\mathcal T\mathcal U}\\
\ls &\sup_{0\leq \tau\leq t}\sum_{|I|\leq 1}\|\varpi(|\cdot|-\tau)\Gamma^I\xi(\tau,\cdot)\|_{L^\infty(\Sigma_\tau)}\\
&+\int_0^t \left((1+\tau)\|\varpi(|\cdot|-\tau)|F(\tau,\cdot)|_{\mathcal T\mathcal U}\|_{L^\infty(D_\tau)}+\sum_{|I|\leq 2}(1+\tau)^{-1}\|\varpi(|\cdot|-\tau)\Gamma^I\xi(\tau,\cdot)\|_{L^\infty(D_\tau)}\right)d\tau.
\end{split}
\end{equation*}
As before, $D_t$ is defined to be the region $D_t:=\{x:\f t2\leq |x|\leq 2t\}$.
\end{proposition}
\begin{proof}
After noting that $\rd_s$ and $\rd_q$ commute with the projection to vector fields in $\mathcal T$ and $\mathcal U$, the proof is completely analogous to that in Proposition \ref{decay.est} and will be omitted.
\end{proof}

\section{Definition of the spacetime regions}\label{def.regions}
Starting from this section and until Section~\ref{sec.IV}, our goal is to prove energy estimates for $\Gmm^{I} h$ and $\Gmm^{I} \bt$ for $\abs{I} \leq N$ under the bootstrap assumptions made in Section~\ref{sec.BA} (for the end result, see Theorem~\ref{thm:energy.bound.combined}). Our argument depends crucially on the decomposition of the spacetime into the regions $\calR_{1}, \ldots, \calR_{4}$ introduced in Section~\ref{sec.notation}, which we recall now: Given parameters $T > 0$, $U_{2} < 0$ and $U_{3} > 0$ to be fixed below, we define
\begin{align*}
\mathcal R_1&:= \{(t,x):t\leq T\},& 
\mathcal R_2&:=\{(t,x):t\geq T,\,t-|x|-\f{1}{(1+t)^{\f{\gamma}{4}}}\leq U_2 \},\\
\mathcal R_3&:=\{(t,x):t\geq T,\,U_2\leq t-|x|-\f{1}{(1+t)^{\f{\gamma}{4}}}\leq U_3 \},&
\mathcal R_4&:=\{(t,x):t\geq T,\,t-|x|-\f{1}{(1+t)^{\f{\gamma}{4}}}\geq U_3 \}.
\end{align*}
Recall also the notation $\mathcal R_{j, \tau} := \mathcal R_{j} \cap \set{t = \tau}$ and $\calB_{U} := \set{(t,x): t - \abs{x} - \frac{1}{(1+t)^{\frac{\gmm}{4}}} = U}$. See Figure~\ref{fig:regions}.
\begin{figure}[h]
\begin{center}
\def\svgwidth{300px}
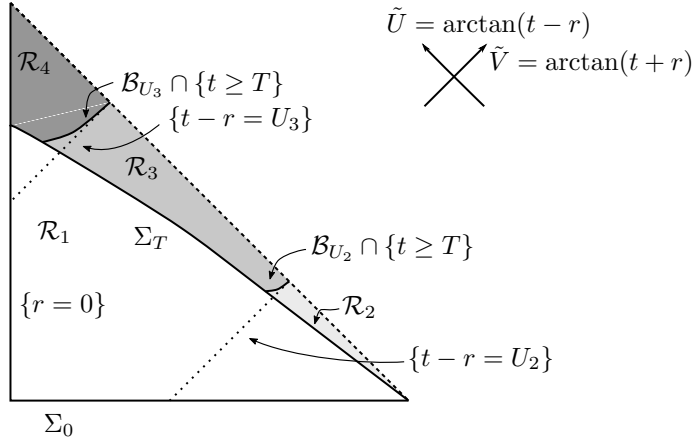 
\caption{Decomposition of the spacetime employed in the proof. We have drawn the standard Penrose diagram with respect to the background Minkowski metric $m$, i.e.~a fixed-$(\tht, \varphi)$ plane in the coordinates $(\tilde{U} = \arctan(t-r), \tilde{V} = \arctan(t + r), \tht, \varphi)$, where the future-pointing null $\tilde{U}$- and $\tilde{V}$-axes are drawn at $45^{\circ}$ with respect to the vertical axis. For radial vectors, the causality properties read off from the diagram coincide with those with respect to $m$. Abusing the terminology a bit, we refer to the (dashed) idealized boundary curve $\{ \tilde{V} = \frac{\pi}{2} \}$ as \emph{null infinity}, and its past and future endpoints as \emph{spacelike} and \emph{timelike infinities}, respectively. By Proposition~\ref{inverse}, $\set{t-r = U_{j}}$ $(j=2, 3)$, which is evidently null with respect to $m$, is asymptotically null with respect to $g$ as $t$ increases. Accordingly, the hypersurface $\calB_{U_{j}} \cap \set{t \geq T} = \set{(t, x) : t - \abs{x} + \frac{1}{(1+t)^{\frac{\gmm}{4}}} = U_{j}, \, t \geq T}$, which asymptotes to $\set{t-r=U_{j}}$ as $t \to \infty$, is spacelike with respect to $g$ for $T$ large enough, as we have implicitly observed in Propositions~\ref{EE.2}--\ref{EE.4}. 
} \label{fig:regions}
\end{center}
\end{figure}

Energy estimates for $h$ and $\bt$ are proved sequentially in the regions $\calR_{1}$, $\calR_{2}$, $\calR_{3}$ and $\calR_{4}$, in Sections~\ref{sec.Cauchy.stability}, \ref{sec.II}, \ref{sec.III} and \ref{sec.IV}, respectively. The parameters $U_{2}$, $U_{3}$ and $T$ are fixed at the end of Section~\ref{sec.IV}; we refer to Remark~\ref{rmk.smallness} for the precise order of the choices of the parameters involved in the proof.

\section{Cauchy stability up to large time}\label{sec.Cauchy.stability}

In this section, we prove the energy estimates in the region $\mathcal R_1$. This is a standard Cauchy stability argument, which we include for completeness. We begin with Lemma \ref{CS.lemma}. The main estimate will be proven in Proposition \ref{Cauchy.stability}.

\begin{lemma}\label{CS.lemma}
Suppose $\Box_g\xi=F$, where $\xi$ is a scalar function and $g$ satisfies the bootstrap assumptions \eqref{BA1}-\eqref{BA5}. Then for every $T>0$, there exists ${\tilde{\ep}}_1>0$ such that for $\ep<{\tilde{\ep}}_1$, the following estimate holds for all $T'\leq T$ with an implicit constant\footnote{In most places of this paper, it is important for the constants to be independent of $T$ (as long as $T$ is sufficiently large). This dependence on only allowed in this proposition and the next proposition and we therefore use the notation $\ls_T$ to emphasize this.} depending on $T$ (in addition to $C$, $\gamma$, $N$ and $\de_0$):
\begin{equation*}
\begin{split}
\sup_{\tau\in [0,T']}\int_{\Sigma_\tau} |\rd\xi|^2\,w(r-\tau)\, dx\ls_T &\int_{\Sigma_0} |\rd\xi|^2\,w(r)\, dx+\int_0^{T'}\int_{\Sigma_t} |\rd\xi||F|\,w(q)\, dx\, dt.
\end{split}
\end{equation*}
\end{lemma}

\begin{proof}
Let $T>0$ be as in the the statement of the lemma. In the proof of this lemma, we allow all implicit constants to depend on $T$ (in addition to $C$, $\gamma$, $N$ and $\de_0$).

It is convenient to proceed in a more geometric fashion for which we need some notations. Introduce the stress-energy-momentum tensor\footnote{This is not to be confused with $\mathbb T$, which is the notation for the stress-energy-momentum tensor in the Einstein equations (see for example \eqref{Einstein.scalar.field}).}
$${\bf T}_{\mu\nu}=\rd_\mu\xi\rd_\nu\xi-\f 12 g_{\mu\nu}(g^{-1})^{\alp\bt}\rd_\alp\xi\rd_\bt \xi.$$
Since $\Box_g\xi=F$, ${\bf T}_{\mu\nu}$ satisfies\footnote{Here, $D$ denotes the Levi-Civita connection with respect to the spacetime metric $g$.} $D^\mu{\bf T}_{\mu\nu}=(\rd_\nu\xi)F$. Contracting this with $-(Dt)^{\mu}=-(g^{-1})^{\mu\nu}\rd_\mu t$, integrating this in the region $\{0\leq t\leq \tau\}$ (for $\tau\in [0,T]$) with respect to the volume form $w(q)\,dVol:=w(q)\,\sqrt{-\det g}\,dt\, dx^1\, dx^2\,dx^3$, and applying the divergence theorem, we obtain
\begin{equation}\label{CS.EE}
\begin{split}
&-\int_{\Sigma_\tau} {\bf T}_{\mu\nu} {\bf N}^{\mu}(D t)^{\nu}w(r-\tau)\sqrt{\det\hat{g}}\, dx\\
\ls &-\int_{\Sigma_0} {\bf T}_{\mu\nu} {\bf N}^{\mu}(D t)^{\nu}w(r)\sqrt{\det\hat{g}} \, dx\\
&+\int_0^{\tau}\int_{\Sigma_t} \left(|\rd\xi||F|+\left(1+|D^2 t|\right)|\rd\xi|^2\right)\,w(q)\, dVol,
\end{split}
\end{equation}
where ${\bf N}^\mu$ is the future directed unit normal to the constant $t$ hypersurfaces. To derive \eqref{CS.EE}, we have used the upper bounds for $g$ and $g^{-1}$, which follow from (2) in Definition \ref{def.dispersivespt}, the bootstrap assumption \eqref{BA4} and the estimate \eqref{inverse.1}; and also $|\rd w|\ls w$.

We now show that the boundary term on the left hand side of \eqref{CS.EE} controls the derivatives of $\xi$. Since ${\bf N}^\mu=-\f{(g^{-1})^{\mu\nu}\rd_\nu t}{\sqrt{-(g^{-1})^{00}}}$, we apply\footnote{Notice in particular that \eqref{spatial.eigen.bd.assumption} and \eqref{BA4} imply that $1\ls \sqrt{\det\hat{g}}\ls 1$.} \eqref{spatial.eigen.bd.assumption}, \eqref{g00.bd.assumption} and the bootstrap assumption \eqref{BA4} to obtain
$$-\int_{\Sigma_\tau} {\bf T}_{\mu\nu} {\bf N}^{\mu}(D t)^{\nu}\,w(r-\tau)\sqrt{\det\hat{g}}\, dx^1\,dx^2\,dx^3\gtrsim \int_{\Sigma_\tau}\left(({\bf N}\xi)^2+\sum_{i=1}^3(\rd_i\xi)^2\right)\,w(r-\tau)\, dx.$$
Since we have\footnote{Notice that $\sqrt{-(g^{-1})^{00}}\gtrsim 1$ by \eqref{g00.bd.assumption} and \eqref{BA4}. Moreover, it is easy to check that for ${\bf N}$ defined as above
\begin{equation*}
\begin{split}
g({\bf N},{\bf N})=&-(g^{-1})^{00}g_{00}-2(g^{-1})^{0i}g_{0i}-\f{g_{ij}(g^{-1})^{0i}(g^{-1})^{0j}}{(g^{-1})^{00}}\\
=&-(g^{-1})^{00}g_{00}-2(g^{-1})^{0i}g_{0i}+\f{g_{0j}(g^{-1})^{00}(g^{-1})^{0j}}{(g^{-1})^{00}}=-1.
\end{split}
\end{equation*}
}
${\bf N}=\sqrt{-(g^{-1})^{00}}\rd_t-\f{(g^{-1})^{0i}}{\sqrt{-(g^{-1})^{00}}}\rd_i$, this then implies
\begin{equation}\label{CS.EE.2}
-\int_{\Sigma_\tau} {\bf T}_{\mu\nu} {\bf N}^{\mu}(D t)^{\nu}w(r-\tau)\sqrt{\det\hat{g}}\, dx^1\,dx^2\,dx^3\gtrsim \int_{\Sigma_\tau}|\rd\xi|^2\,w(r-\tau)\, dx.
\end{equation}
On the other hand, by \eqref{spatial.eigen.bd.assumption}, \eqref{g00.bd.assumption} and \eqref{BA4}, we have
\begin{equation}\label{CS.EE.3}
-\int_{\Sigma_0} {\bf T}_{\mu\nu} {\bf N}^{\mu}(D t)^{\nu}w(r)\sqrt{\det\hat{g}} \, dx\ls \int_{\Sigma_0} |\rd\xi|^2 w(r) \, dx.
\end{equation}
To proceed, note that\footnote{using expansion by minors and cofactors.} $(g^{-1})^{00}=\f{\det \hat{g}}{\det g}$, which implies upper and lower bounds for $\det g$. Therefore, \eqref{CS.EE},\eqref{CS.EE.2} and \eqref{CS.EE.3} together imply
\begin{equation*}
\begin{split}
\int_{\Sigma_\tau} |\rd\xi|^2\, w(r-\tau)\, dx
\ls &\int_{\Sigma_0} |\rd\xi|^2\,w(r) \, dx+\int_0^{\tau}\int_{\Sigma_t} \left(|\rd\xi||F|+\left(1+|D^2 t|\right)|\rd\xi|^2\right)\,w(q)\, dx\, dt,
\end{split}
\end{equation*}
Finally, by \eqref{BA4} and \eqref{inverse.1}, we have $|D^2 t|\ls 1$ and therefore the desired conclusion follows after an application of the Gr\"onwall's inequality. \qedhere

\end{proof}

We now apply Lemma \ref{CS.lemma} to obtain the following energy estimates in the region $\mathcal R_1$:
\begin{proposition}\label{Cauchy.stability}
There exists $\ep_1>0$ sufficiently small such that for every $T>0$, there exists a constant $C_T$ (depending on $T$ in addition to $C$, $\gamma$, $N$ and $\de_0$) such that the following estimate holds in $[0,T]\times\mathbb R^3$:
$$\sup_{0 \leq t\leq T} \left(\sum_{|I|\leq N}\int_{\Sigma_t}\left(|\rd\Gamma^I h|^2+|\rd\Gamma^I\beta|^2\right) w(q)\, dx\right)^{\frac 12}\leq C_T \ep$$
whenever $\ep<\ep_1$. 
\end{proposition}
\begin{proof}
For this argument in $\mathcal R_1$, we only need a much rougher form of the equation than that in Propositions \ref{schematic.eqn} and \ref{SF.schematic.eqn}. By Propositions \ref{schematic.eqn} and \ref{SF.schematic.eqn} and the bootstrap assumptions \eqref{BA1} and \eqref{BA4}, we easily obtain
\begin{equation}\label{CS.schematic.eqn}
\begin{split}
&|\tBox_g \Gamma^I h|+|\tBox_g \Gamma^I \beta|\\
\ls &\frac{\ep \log^2 (2+s)}{(1+s)^2(1+|q|)w(q)^{\f{\gamma}{1+2\gamma}}}+\sum_{|J|\leq |I|} \left(|\rd \Gamma^J h|+|\rd \Gamma^J \beta|+(1+|q|)^{-1}|\Gamma^J h|\right).
\end{split}
\end{equation}

By \eqref{SF.wave.coord.det}, \eqref{gauge.est.1} and (8) in Definition \ref{def.dispersivespt}, we can replace the $\tBox_g$ in \eqref{CS.schematic.eqn} by the scalar d'Alembertian $\Box_g$ for each component of $h$ and $\bt$, while retaining the same bounds on the right hand side. Therefore, by Lemma \ref{CS.lemma} (which is applicable if we choose $\ep_1<\tilde{\ep}_1$) and the Cauchy-Schwarz inequality, we have
\begin{equation*}
\begin{split}
&\sup_{0 \leq t\leq T'} \sum_{|I|\leq k}\int_{\Sigma_t}\left(|\rd\Gamma^I h|^2+|\rd\Gamma^I\beta|^2\right) w(q)\, dx\\
\ls_T & \sum_{|I|\leq k}\int_{\Sigma_0}\left(|\rd\Gamma^I h|^2+|\rd\Gamma^I\beta|^2\right) w(r)\, dx\\
&+\sum_{|I|\leq k}\int_0^{T'}\int_{\Sigma_t} \left(\frac{\ep^2 \log^4 (2+s)}{(1+s)^4(1+|q|)^2 w(q)^{\f{2\gamma}{1+2\gamma}}}+|\rd\Gamma^I h|^2+|\rd\Gamma^I\beta|^2+\f{|\Gamma^I h|^2}{(1+|q|)^2}\right)\,w(q)\, dx\, dt\\
\ls_T &\ep^2+\sum_{|I|\leq k}\int_0^{T'}\int_{\Sigma_t} \left(|\rd\Gamma^I h|^2+|\rd\Gamma^I\beta|^2\right)\,w(q)\, dx\, dt,
\end{split}
\end{equation*}
where in the last line we have used $\int_0^T\int_{\Sigma_t} \frac{\ep^2 \log^4 (2+s)}{(1+s)^4(1+|q|)^2 w(q)^{\f{2\gamma}{1+2\gamma}}}\,w(q)\, dx\, dt\ls \ep^2$ as well as the Hardy inequality in Proposition \ref{Hardy} with $\alp=\mu_1=0$ and $\mu_2=2\gamma$. The conclusion thus follows from Gr\"onwall's inequality.
\end{proof}

\section{Estimates for the error terms arising from the energy estimates}\label{sec.EE}

In this section and the next three sections, we prove energy estimates for the metric $h$ and the scalar field $\beta$. We first show that most of the terms in the equation for $\tBox_g\Gamma h$ in Proposition \ref{schematic.eqn} and all of the terms for $\tBox_g\Gamma^I\beta$ in Proposition \ref{SF.schematic.eqn} can in fact be controlled by the energy itself with a smallness constant (and allowing some growth in the lower order term). The only term in Proposition \ref{schematic.eqn} that cannot be controlled in such a manner is what we called the ``bad term'', i.e.~the term $\mathfrak B_I$. Indeed, this $\mathfrak B_I$ term has to be dealt with differently in each of the regions $\mathcal R_2$, $\mathcal R_3$ and $\mathcal R_4$ and will be treated in Sections \ref{sec.II}-\ref{sec.IV}.

In this section, we instead control all the other terms which are better behaved. The main results of this section are contained in Propositions \ref{EE.main} and \ref{SF.EE.main}. In order to show the estimates for all regions simultaneously, we introduce the following notation: We will use $\mathcal R$ to denote either the regions $\mathcal R_2$, $\mathcal R_3$ or $\mathcal R_4$. Also, $\mathcal R_\tau$ will denote the intersection of $\mathcal R$ and a constant $\{t=\tau\}$ hypersurface.

We first control the terms in $|\Box_g\Gamma^I h|$ except for $\mathfrak B_I$. According to Proposition \ref{schematic.eqn}, we need to control the terms $\mathfrak I_I$, $\mathfrak G_I$, $\mathfrak T_I$, $\mathfrak L_I$, $\mathfrak W_I$ and $\mathfrak N_I$. As a first step, we use the bootstrap assumptions \eqref{BA1} and \eqref{BA4} to further estimate the $\mathfrak N_I$ term:
\begin{equation*}
\begin{split}
&\mathfrak N_I\\
\ls &\sum_{\substack{|J_1|+|J_2|\leq |I|\\ \max\{|J_1|,|J_2|\}\leq |I|-1}}\left(|\rd\Gamma^{J_1}h||\rd\Gamma^{J_2}h|+|\rd\Gamma^{J_1}\beta||\rd\Gamma^{J_2}\beta|+\frac{|\Gamma^{J_1}h||\rd\Gamma^{J_2}h|}{1+|q|}\right)\\
&+\sum_{|J_2|\leq |I|,\,|J_1|+(|J_2|-1)_+\leq |I|}\frac{|\Gamma^{J_1} h|_{LL}|\rd\Gamma^{J_2} h|}{1+|q|}\\
\ls & \sum_{|J|\leq |I|-1} \f{\left(|\rd\Gamma^J h|+|\rd\Gamma^J \beta|\right)}{(1+s)^{1-\de_0}(1+|q|)^{\f12-\f{\gamma}{4}}w(q)^{\f12}}+\sum_{|J|\leq |I|-1} \f{|\Gamma^J h|}{(1+s)^{1-\de_0}(1+|q|)^{\f32-\f \gamma 4}w(q)^{\f12}}\\
&+\sum_{|J|\leq |I|} \f{|\Gamma^J h|_{LL}}{(1+s)^{1-\de_0}(1+|q|)^{\f32-\f\gamma 4}w(q)^{\f12}}
\end{split}
\end{equation*}
We will control the $\mathfrak N_I$ term together with the $\mathfrak L_I$, $\mathfrak G_I$ terms. More precisely, define 
\begin{equation}\label{tL.def}
\tilde{\mathfrak L}_I:= \sum_{|J|\leq |I|-1}\left(\frac{|\rd\Gamma^J h|+|\rd\Gamma^J \beta|}{(1+s)^{1-\de_0}}+\frac{|\Gamma^J h|}{(1+s)^{1-\de_0}(1+|q|)^{1+\gamma}}\right),
\end{equation}
and 
\begin{equation}\label{tG.def}
\tilde{\mathfrak G}_I := \sum_{|J|\leq |I|}\frac{|\Gamma^J h|_{LL}}{(1+s)^{1-\de_0}(1+|q|)^{1+\gamma}}+\sum_{|J|\leq |I|}\frac{|\bar{\rd}\Gamma^J h|+|\bar{\rd}\Gamma^J \beta|}{(1+s)^{1-\de_0}(1+|q|)^{\gamma+\de_0}}.
\end{equation}
It is easy to see that we have
\begin{proposition}
The following pointwise estimate holds everywhere in $[0,\infty)\times \mathbb R^3$:
$$\mathfrak G_I+\mathfrak L_I+\mathfrak N_I\ls \tilde{\mathfrak G}_I+\tilde{\mathfrak L}_I.$$
\end{proposition}
We will therefore control the terms $\mathfrak I_I$, $\mathfrak T_I$, $\mathfrak W_I$, $\tilde{\mathfrak L}_I$, $\tilde{\mathfrak G}_I$ in the $w(q)$-weighted $L^1_tL^2_x$ space according to the energy estimates (Propositions \ref{EE.2}, \ref{EE.3}, \ref{EE.4}). We note again that we do not bound the term $\mathfrak B_I$ in this section, but will estimate it in later sections.

Before we proceed, we introduce some more notations. In the following, we will frequently apply the Hardy inequality in Proposition \ref{Hardy} and will generate boundary terms on $\mathcal B_U$. To capture these boundary terms, we introduce the notation\footnote{Let us note explicitly that the term \eqref{bdry.L2.def} arises from the application of Proposition \ref{Hardy} with parameters $\alpha=\mu_1=0$, $\mu_2=2\gamma$.} that for every scalar function $\xi$, let
\begin{equation}\label{bdry.L2.def}
\|v(U)\xi \|_{L^2(\mathbb S^2(U,\tau))}^2:=v(U)\int_{\{t=\tau,\,t-r-\f{1}{(1+t)^{\f\gamma 4}}=U\}} |\xi|^2 r^2\sin\theta\,d\theta\,d\varphi,
\end{equation}
where 
$$v(U):=\begin{cases} (1+|U|)^{2\gamma} &\mbox{if }U\leq 0\\ (1+|U|)^{-1} &\mbox{if }U\geq 0 \end{cases}.$$
Moreover, we define $U'(\mathcal R)$ to be the $U$ value associated to the hypersurface $\mathcal B_{U'}$ to the past of the region $\mathcal R$ and $U(\mathcal R)$ to be the $U$ value associated to the hypersurface $\mathcal B_U$ to the future of the region $\mathcal R$. More precisely, for the region $\mathcal R_2$, we have $U'(\mathcal R_2)=-\infty$, $U(\mathcal R_2)=U_2$; for the region $\mathcal R_3$, we have $U'(\mathcal R_3)=U_2$, $U(\mathcal R_3)=U_3$; and for the region $\mathcal R_4$, we have $U'(\mathcal R_4)=U_3$, $U(\mathcal R_4)=\infty$. We will also use the convention that in the case $U'(\mathcal R_2)=-\infty$ and $U(\mathcal R_4)=\infty$, the term in \eqref{bdry.L2.def} is taken to be zero. 

We will also introduce the following convention for the $L^2(\mathcal R_\tau)$ norm. Given a function $F$ in the spacetime variables $(t,x)$, unless otherwise stated, the norm $\|F w^{\f12}\|_{L^2_x(\mathcal R_\tau)}$ will be understood such that $F$ is evaluated at $(\tau,x)$, while $w$ is evaluated at $|x|-\tau$, i.e.~
$$\|F w^{\f12}\|_{L^2_x(\mathcal R_\tau)}:=\left(\int_{\mathbb R^3} |F|^2(\tau,x) w(r-\tau)\, dx\right)^{\f 12}.$$

We now turn to the estimates. First we consider the term $\mathfrak I_I$.
\begin{proposition}\label{I.EE}
The following estimate holds for $|I|\leq N$ and $t>T>0$:
$$\int_T^t \|\mathfrak I_I w^{\f12}\|_{L^2_x(\mathcal R_\tau)} d\tau\ls \ep\log^4(2+t).$$
\end{proposition}
\begin{proof}
Recalling the definition of $\mathfrak I_I$ in Proposition \ref{schematic.eqn}, we get\footnote{Notice that we have just bounded the integral in $\sin\theta\, d\theta\,d\varphi$ by a constant factor.}
\begin{align}
&\int_T^t \|\mathfrak I_I w^{\f12}\|_{L^2_x(\mathcal R_\tau)} d\tau\notag\\
\ls &\ep \int_T^t \big(\int_0^{\infty}\f{\log^4(2+\tau+r) w(r-\tau)^{1-\f{2\gamma}{1+2\gamma}} r^2\,dr}{(1+\tau+r)^4(1+|r-\tau|)^2} \big)^{\f12}\, d\tau \notag\\
\ls &\ep \int_T^t \big(\int_0^{\infty}\f{\log^4(2+\tau+r) \,dr}{(1+\tau+r)^2(1+|r-\tau|)} \big)^{\f12}\, d\tau \label{I.EE.cancel weights}\\
\ls &\ep \int_T^t \f{\log^3(2+\tau)}{1+\tau}\big(\int_{-\infty}^{\infty}\f{dq}{(1+|q|)\log^2(2+|q|)} \big)^{\f12}\, d\tau\notag\\
\ls &\ep \int_T^t \f{\log^3(2+\tau)}{1+\tau} d\tau\ls \ep\log^4(2+t).\notag
\end{align}
Notice that in line \eqref{I.EE.cancel weights}, we have used $\f{r^2}{(1+\tau+r)^2}\ls 1$ and $\f{w(q)^{1-\f{2\gamma}{1+2\gamma}}}{1+|q|}\ls 1$.
\end{proof}

We now estimate the $\mathfrak T_I$ term, which can be controlled easily using the $\f{1}{(1+s)^{1+\f\gamma 2}}$ decay:
\begin{proposition}\label{T.EE}
The following estimate holds for $|I|\leq N$ and $t>T>0$:
$$\int_T^t \|\mathfrak T_I w^{\f12}\|_{L^2_x(\mathcal R_\tau)} d\tau\ls T^{-\f\gamma 2}\sum_{|J|\leq |I|} \sup_{\tau\in [T,t]}\|\left(|\rd\Gamma^J h|+|\rd\Gamma^J \beta|\right) w^{\f12}\|_{L^2_x(\mathcal R_\tau)}.$$
\end{proposition}
\begin{proof}
Recalling the definition of $\mathfrak T_I$ in Proposition \ref{schematic.eqn}, we get
\begin{equation*}
\begin{split}
&\sum_{|J|\leq |I|}\int_T^t \left(\int_{\mathcal R_\tau} |\mathfrak T_I|^2 w(r-\tau) dx\right)^{\f12}\,d\tau\\
\ls &\sum_{|J|\leq |I|}\int_T^t \left(\int_{\mathcal R_\tau} \frac{\left(|\rd\Gamma^J h|^2+|\rd\Gamma^J \beta|^2\right) w(r-\tau)}{(1+\tau+r)^{2+\gamma}} dx\right)^{\f12}\,d\tau\\
\ls &\sum_{|J|\leq |I|} \sup_{\tau\in [T,t]}\left(\int_{\mathcal R_\tau} \left(|\rd\Gamma^J h|^2+|\rd\Gamma^J \beta|^2\right) w(r-\tau) dx\right)^{\f12}\left(\int_T^t\f{d\tau}{(1+\tau)^{1+\f\gamma2}}\right)\\
\ls &T^{-\f\gamma 2}\sum_{|J|\leq |I|} \sup_{\tau\in [T,t]}\|\left(|\rd\Gamma^J h|+|\rd\Gamma^J \beta|\right) w^{\f12}\|_{L^2_x(\mathcal R_\tau)}.
\end{split}
\end{equation*}

\end{proof}

We now bound the term $\mathfrak W_I$. These estimates require the use of the Hardy inequality in Proposition \ref{Hardy} and therefore have a boundary term\footnote{Recall the notation for the boundary term from the discussions prior to Proposition \ref{I.EE}.} on $\mathcal B_{U(\mathcal R)}$.
\begin{proposition}\label{V.EE}
The following estimate holds for $|I|\leq N$ and $t>T>0$:
$$\int_T^t \|\mathfrak W_I w^{\f12}\|_{L^2_x(\mathcal R_\tau)} d\tau\ls T^{-\f\gamma 2}\sum_{|J|\leq |I|} \sup_{\tau\in [T,t]}\left(\||\rd\Gamma^J h| w^{\f12}\|_{L^2_x(\mathcal R_\tau)}+\|v(U'(\mathcal R))|\Gamma^J h|\|_{L^2(\mathbb S^2(U'(\mathcal R),\tau))}\right).$$
\end{proposition}
\begin{proof}
Recalling the definition of $\mathfrak W_I$ in Proposition \ref{schematic.eqn}, we get
\begin{align}
&\int_T^t \left(\int_{\mathcal R_\tau} |\mathfrak W_I|^2 w(r-\tau) dx\right)^{\f12}\,d\tau\notag\\
\ls &\sum_{|J|\leq |I|}\int_T^t \left(\int_{\mathcal R_\tau} \frac{\log^2(2+\tau+r)|\Gamma^J h|^2 w(r-\tau)}{(1+\tau+r)^{4-4\de_0}(1+|r-\tau|)^{2\gamma+4\de_0}} dx\right)^{\f12}\,d\tau\notag\\
\ls &\sum_{|J|\leq |I|}\sup_{\tau\in[T,t]} \left(\int_{\mathcal R_\tau} \frac{|\Gamma^J h|^2 w(r-\tau)}{(1+|r-\tau|)^2} dx\right)^{\f12}\left(\int_T^t\f{\log(2+\tau)\,d\tau}{(1+\tau)^{1+\gamma}}\right)\label{V.EE.3}\\
\ls &T^{-\f\gamma 2}\sum_{|J|\leq |I|}\sup_{\tau\in[T,t]} \left(\||\rd\Gamma^J h| w^{\f12}\|_{L^2_x(\mathcal R_\tau)}+\|v(U'(\mathcal R))|\Gamma^J h|\|_{L^2(\mathbb S^2(U'(\mathcal R),\tau))}\right)\label{V.EE.4},
\end{align}
where in \eqref{V.EE.3}, we used $\left(\f{\log^2(2+\tau+r)(1+|r-\tau|)^{2-2\gamma-4\de_0}}{(1+\tau+r)^{4-4\de_0}}\right)^{\f 12}\ls \f{\log(2+\tau)}{(1+\tau)^{1+\gamma}}; $
and in \eqref{V.EE.4} we have used the Hardy inequality in Proposition \ref{Hardy} with $\alpha=0$, $\mu_1=0$ and $\mu_2=2\gamma$.
\end{proof}

We now bound the $\tilde{\mathfrak L}_I$ term:
\begin{proposition}\label{L.EE}
The following estimate holds for $|I|\leq N$ and $t>T>0$:
\begin{equation*}
\begin{split}
&\int_T^t \|\tilde{\mathfrak L}_I w^{\f12}\|_{L^2_x(\mathcal R_\tau)} d\tau\\
\ls &(1+t)^{\de_0}\sum_{|J|\leq |I|-1} \sup_{\tau\in [T,t]}\left(\|\left(|\rd\Gamma^J h|+|\rd\Gamma^J \beta|\right) w^{\f12}\|_{L^2_x(\mathcal R_\tau)}+\|v(U'(\mathcal R))|\Gamma^J h|\|_{L^2(\mathbb S^2(U'(\mathcal R),\tau))}\right).
\end{split}
\end{equation*}
\end{proposition}
\begin{proof}
Recalling the definition of $\tilde{\mathfrak L}_I$ in \eqref{tL.def}, we get
\begin{align}
&\int_T^t \left(\int_{\mathcal R_\tau} |\tilde{\mathfrak L}_I|^2 w(r-\tau) dx\right)^{\f12}\,d\tau\notag\\
\ls &\sum_{|J|\leq |I|-1}\int_T^t \left(\int_{\mathcal R_\tau} \frac{\left(|\rd\Gamma^J h|^2+|\rd\Gamma^J \beta|^2\right) w(r-\tau)}{(1+\tau+r)^{2-2\de_0}} dx\right)^{\f12}\,d\tau\notag\\
&+\sum_{|J|\leq |I|-1}\int_T^t \left(\int_{\mathcal R_\tau} \frac{|\Gamma^J h|^2 w(q)}{(1+\tau+r)^{2-2\de_0}(1+|r-\tau|)^{2+2\gamma}} dx\right)^{\f12}\,d\tau\notag\\
\ls &\sum_{|J|\leq |I|-1}\sup_{\tau\in[T,t]} \left(\int_{\mathcal R_\tau}\left( |\rd\Gamma^J h|^2+ |\rd\Gamma^J \beta|^2\right) w(r-\tau) dx\right)^{\f12}\left(\int_T^t\f{d\tau}{(1+\tau)^{1-\de_0}}\right)\notag\\
&+\sum_{|J|\leq |I|-1}\sup_{\tau\in[T,t]} \left(\int_{\mathcal R_\tau} \frac{|\Gamma^J h|^2 w(r-\tau)}{(1+|r-\tau|)^2} dx\right)^{\f12}\left(\int_T^t\f{d\tau}{(1+\tau)^{1-\de_0}}\right)\notag\\
\ls &(1+t)^{\de_0}\sum_{|J|\leq |I|-1}\sup_{\tau\in[T,t]} \left(\|\left(|\rd\Gamma^J h|+|\rd\Gamma^J \beta|\right) w^{\f12}\|_{L^2_x(\mathcal R_\tau)}+\|v(U'(\mathcal R))|\Gamma^J h|\|_{L^2(\mathbb S^2(U'(\mathcal R),\tau))}\right)\label{L.EE.6},
\end{align}
where in \eqref{L.EE.6} we have used the Hardy inequality in Proposition \ref{Hardy} with $\alpha=0$, $\mu_1=0$ and $\mu_2=2\gamma$.
\end{proof}
Finally, we move to the good term $\tilde{\mathfrak G}_I$.
\begin{proposition}\label{G.EE}
The following bound\footnote{Notice that we in fact have a slightly stronger bound where boundary term on the right hand side can be replaced by
\begin{equation*}
\begin{split}
&\sum_{|J|\leq |I|}\sup_{\tau\in[T,t]}\left(T^{-\f{\gamma}{2}}\|v(U'(\mathcal R))|\Gamma^J h|\|_{L^2(\mathbb S^2(U'(\mathcal R),\tau))}+(1+\tau)^{\de_0}\|v(U'(\mathcal R))|\Gamma^J h|_{LL}\|_{L^2(\mathbb S^2(U'(\mathcal R),\tau))}\right)\\
&+\sum_{|J|\leq |I|-1}\sup_{\tau\in[T,t]}(1+\tau)^{\de_0}\|v(U'(\mathcal R))|\Gamma^J h|\|_{L^2(\mathbb S^2(U'(\mathcal R),\tau))}.
\end{split}
\end{equation*}
Using these better bounds can give a slightly improvement in the exponent of $(1+t)$ in the estimates for the energy. Since they are not necessary, we only state the weaker bounds below.
} holds for $|I|\leq N$ and $t>T>0$:
\begin{equation*}
\begin{split}
&\int_T^t \|\tilde{\mathfrak G}_I w^{\f12}\|_{L^2_x(\mathcal R_\tau)} d\tau\\
\ls &\ep\log^4(2+t)+T^{-\f\gamma 2}\sum_{|J|\leq |I|}\sup_{\tau\in[T,t]}\||\rd\Gamma^J h| w^{\f12}\|_{L^2_x(\mathcal R_\tau)}\\
&+T^{-\f\gamma 2}\sum_{|J|\leq |I|}\left(\int_T^t \int_{ \mathcal R_\tau} \left(|\bar{\rd}\Gamma^J h|^2+|\bar{\rd}\Gamma^J \beta|^2\right) w'(r-\tau) dx\, d\tau\right)^{\f12}\\
&+(1+t)^{\de_0}\sum_{|J|\leq |I|-1} \sup_{\tau\in [T,t]}\||\rd\Gamma^J h| w^{\f12}\|_{L^2_x(\mathcal R_\tau)}\\
&+\sum_{|J|\leq |I|}\sup_{\tau\in[T,t]}(1+\tau)^{\de_0}\|v(U'(\mathcal R))|\Gamma^J h|\|_{L^2(\mathbb S^2(U'(\mathcal R),\tau))}.
\end{split}
\end{equation*}
\end{proposition}
\begin{proof}
In order to control $\tilde{\mathfrak G}_I$, we recall from \eqref{tG.def} that
\begin{equation}\label{tg.terms}
\tilde{\mathfrak G}_I\ls \underbrace{\sum_{|J|\leq |I|}\frac{|\Gamma^J h|_{LL}}{(1+s)^{1-\de_0}(1+|q|)^{1+\gamma}}}_{=:I}+\underbrace{\sum_{|J|\leq |I|}\frac{|\bar{\rd}\Gamma^J h|+|\db\Gamma^J\beta|}{(1+s)^{1-\de_0}(1+|q|)^{\gamma+\de_0}}}_{=:II}.
\end{equation}

We first control the term $II$, which has a good derivative. As we will see later, it will be convenient to derive a slightly better estimate in which we bound
\begin{equation}\label{G.EE.II'.def}
II':=\sum_{|J|\leq |I|}\frac{|\bar{\rd}\Gamma^J h|+|\db\Gamma^J\beta|}{(1+s)^{1-\de_0}(1+|q|)^{\f{7\gamma}{8}}}
\end{equation}
instead of $II$.
As a consequence of Definition \ref{w.def}, $w'(q)$ is positive and in fact satisfies the lower bounds
$$w'(q)\gtrsim\begin{cases}
(1+|q|)^{2\gamma} &\mbox{if }q\geq 0\\
(1+|q|)^{-1-\f{\gamma}{2}} &\mbox{if }q< 0.
\end{cases}$$ 
Therefore, for $\de_0$ satisfying \eqref{de_0.def}, we can estimate $II'$ separately in the regions $r-\tau\geq 0$ and $r-\tau<0$ to obtain
\begin{equation}\label{G.gd.est}
\begin{split}
&\sum_{|J|\leq |I|}\int_T^t \left(\int_{\mathcal R_\tau} \frac{\left(|\bar{\rd}\Gamma^J h|^2+|\bar{\rd}\Gamma^J \beta|^2\right) w(r-\tau)}{(1+s)^{2-2\de_0}(1+|r-\tau|)^{\f{7\gamma}{4}}} dx\right)^{\f12}\,d\tau\\
\ls &\sum_{|J|\leq |I|} \left(\int_{\{\tau\in [T,t], r-\tau\geq 0\}\cap \mathcal R} \frac{\left(|\bar{\rd}\Gamma^J h|^2+|\bar{\rd}\Gamma^J \beta|^2\right) (1+|r-\tau|)^{1+\f{\gamma}{4}}}{(1+\tau+r)^{1-2\de_0-\gamma}} dx\,d\tau\right)^{\f12} \left(\int_T^t \f{d\tau}{(1+\tau)^{1+\gamma}}\right)^{\f12}\\
&+\sum_{|J|\leq |I|} \left(\int_{\{\tau\in [T,t], r-\tau < 0\}\cap \mathcal R} \frac{\left(|\bar{\rd}\Gamma^J h|^2+|\bar{\rd}\Gamma^J \beta|^2\right) }{(1+\tau+r)^{1-2\de_0-\gamma}(1+|r-\tau|)^{\f{7\gamma}{4}}} dx\,d\tau\right)^{\f12} \left(\int_T^t \f{d\tau}{(1+\tau)^{1+\gamma}}\right)^{\f12}\\
\ls &T^{-\f{\gamma}{2}}\sum_{|J|\leq |I|} \left(\int_{\{\tau\in [T,t], r-\tau\geq 0\}\cap \mathcal R} \left(|\bar{\rd}\Gamma^J h|^2+|\bar{\rd}\Gamma^J \beta|^2\right) (1+|r-\tau|)^{\f{5\gamma}{4}+2\de_0} dx\,d\tau\right)^{\f12} \\
&+T^{-\f{\gamma}{2}}\sum_{|J|\leq |I|} \left(\int_{\{\tau\in [T,t], r-\tau < 0\}\cap \mathcal R} \frac{|\bar{\rd}\Gamma^J h|^2+|\bar{\rd}\Gamma^J \beta|^2 }{(1+|r-\tau|)^{1+\f{3\gamma}{4}-2\de_0}} dx\,d\tau\right)^{\f12} \\
\ls &\sum_{|J|\leq |I|}T^{-\f\gamma 2}\left(\int_T^t \int_{ \mathcal R_\tau} \left(|\bar{\rd}\Gamma^J h|^2+|\bar{\rd}\Gamma^J \beta|^2\right) w'(r-\tau) dx\, d\tau\right)^{\f12}.
\end{split}
\end{equation}
We now turn to term I in \eqref{tg.terms}. Notice that $\rd_q$ commutes with the projection to $L$. Hence, 
$$|\rd_r ((\Gamma^J h)_{LL})|\ls |\rd\Gamma^J h|_{LL}+|\bar{\rd}\Gamma^J h|.$$
Therefore, we can apply the Hardy inequality in Proposition \ref{Hardy} with $\alpha=2-2\de_0$, $\mu_1=\f{7\gamma}4$ and $\mu_2=\f\gamma 4$ to obtain 
\begin{equation}\label{G.LL.est.1}
\begin{split}
&\sum_{|J|\leq |I|}\int_T^t \left(\int_{\mathcal R_\tau} \frac{|\Gamma^J h|_{LL}^2 w(r-\tau)}{(1+\tau+r)^{2-2\de_0}(1+|r-\tau|)^{2+2\gamma}} dx\right)^{\f12}\,d\tau\\
\ls &\sum_{|J|\leq |I|}\int_T^t \left(\int_{\mathcal R_\tau} \frac{|\Gamma^J h|_{LL}^2 w(r-\tau)}{(1+\tau+r)^{2-2\de_0}(1+|r-\tau|)^{2+\f{7\gamma}{4}}} dx\right)^{\f12}\,d\tau\\
\ls &\sum_{|J|\leq |I|}\int_T^t \f{\|v(U'(\mathcal R))|\Gamma^J h|_{LL}\|_{L^2(\mathbb S^2(U'(\mathcal R),\tau))}}{(1+\tau)^{1-\de_0}}d\tau\\
&+\sum_{|J|\leq |I|} \int_T^t \left(\int_{\mathcal R_\tau} \frac{(|\rd\Gamma^J h|_{LL}^2+|\bar{\rd}\Gamma^J h|^2)w(r-\tau)}{(1+\tau+r)^{2-2\de_0}(1+|r-\tau|)^{\f{7\gamma}{4}}} dx\right)^{\f12}d\tau\\
\ls &\sum_{|J|\leq |I|}\sup_{\tau\in[T,t]}(1+\tau)^{\de_0}\|v(U'(\mathcal R))|\Gamma^J h|_{LL}\|_{L^2(\mathbb S^2(U'(\mathcal R),\tau))}\\
&+\sum_{|J|\leq |I|} \int_T^t \left(\int_{\mathcal R_\tau} \frac{(|\rd\Gamma^J h|_{LL}^2+|\bar{\rd}\Gamma^J h|^2)w(r-\tau)}{(1+\tau+r)^{2-2\de_0}(1+|r-\tau|)^{\f{7\gamma}{4}}} dx\right)^{\f12}d\tau.
\end{split}
\end{equation}
Notice that in the first inequality above, we simply discard the extra $\f{1}{(1+|r-\tau|)^{\f\gamma 4}}$ decay. Similarly, in the second inequality, we discard the extra decay in $1+|r-\tau|$ in the boundary term.

Clearly the terms with $|\db\Gamma^J h|$ in the last line of \eqref{G.LL.est.1} can be handled by the estimate\footnote{Of course, it is exactly for handling this contribution from \eqref{G.LL.est.1} that we prove the slightly stronger estimate (for $II'$ instead of $II$) in \eqref{G.gd.est}.} \eqref{G.gd.est}. The other term in \eqref{G.LL.est.1}, i.e. the term with $|\rd\Gamma^I h|_{LL}$, can be handled using Proposition \ref{wave.con.higher}, which implies that
\begin{equation*}
\begin{split}
&|\rd\Gamma^I h|_{LL}\\
\ls &\f{\ep \log (2+s)}{(1+s)^2 w(q)^{\f{\gamma}{1+2\gamma}}}+\f{\log(2+s)}{(1+s)(1+|q|)^\gamma}\sum_{|J|\leq |I|}|\Gamma^J h|+\f{\log(2+s)}{1+s}\sum_{|J|\leq |I|}|\rd \Gamma^J h|\\
&+\sum_{|J_1|+|J_2|\leq |I|}|\Gamma^{J_1}h||\rd \Gamma^{J_2}h|+\sum_{|J|\leq |I|}|\db\Gamma^J h|+\sum_{|J|\leq |I|-2}|\rd \Gamma^J h|\\
\ls &\f{\ep \log (2+s)}{(1+s)^2 w(q)^{\f{\gamma}{1+2\gamma}}}+\f{\log(2+s)}{(1+s)^{1-\de_0}(1+|q|)^{\gamma+\de_0}}\sum_{|J|\leq |I|}|\Gamma^J h|\\
&+\f{\log(2+s)(1+|q|)^{\f12+\f\gamma 4}}{(1+s)^{1-\de_0}}\sum_{|J|\leq |I|}|\rd \Gamma^J h|+\sum_{|J|\leq |I|}|\db\Gamma^J h|+\sum_{|J|\leq |I|-2}|\rd \Gamma^J h|,
\end{split}
\end{equation*}
where in the last line we have used the bootstrap assumptions \eqref{BA1} and \eqref{BA4}. Therefore, for $\de_0$ satisfying \eqref{de_0.def}, we have\footnote{Recall again the definitions of the terms on the right hand side in Proposition \ref{schematic.eqn} and \eqref{tL.def}.}
\begin{equation*}
\begin{split}
&\f{|\rd\Gamma^I h|_{LL}}{(1+s)^{1-\de_0}(1+|q|)^{\frac{7\gamma}{8}}}\\
\ls &\f{\ep \log (2+s)}{(1+s)^{3-\de_0}(1+|q|)^{\frac{7\gamma}{8}}w(q)^{\f{\gamma}{1+2\gamma}}}+\f{\log(2+s)}{(1+s)^{2-2\de_0}(1+|q|)^{\frac{15\gamma}{8}+\de_0}}\sum_{|J|\leq |I|}|\Gamma^J h|\\
&+\f{1}{(1+s)^{1+\f{\gamma}{2}}}\sum_{|J|\leq |I|}|\rd \Gamma^J h|+\f 1{(1+s)^{1-\de_0}(1+|q|)^{\frac{7\gamma}{8}}}\sum_{|J|\leq |I|}|\db\Gamma^J h|\\
&+\f 1 {(1+s)^{1-\de_0}(1+|q|)^{\frac{7\gamma}{8}}}\sum_{|J|\leq |I|-2}|\rd \Gamma^J h|\\
\ls & \mathfrak I_I+\mathfrak W_I+\mathfrak T_I+\tilde{\mathfrak L}_I+\f 1{(1+s)^{1-\de_0}(1+|q|)^{\frac{7\gamma}{8}}}\sum_{|J|\leq |I|}|\db\Gamma^J h|.
\end{split}
\end{equation*}
Note that the last term, which can be bounded above by \eqref{G.EE.II'.def}, can be dealt with by \eqref{G.gd.est}. Therefore, using Propositions \ref{I.EE}, \ref{T.EE}, \ref{V.EE}, \ref{L.EE} and \eqref{G.gd.est}, we have\footnote{In Propositions \ref{T.EE} and \ref{L.EE}, there are also terms involving $\beta$, which are of course not present in the estimate here since we are only bounding the $h$ terms.}
\begin{equation*}
\begin{split}
&\sum_{|J|\leq |I|} \int_T^t \left(\int_{\mathcal R_\tau} \frac{|{\rd}\Gamma^J h|_{LL}^2 w(r-\tau)}{(1+\tau+r)^{2-2\de_0}(1+|r-\tau|)^{\frac{7 \gamma}{4}}} dx\right)^{\f12}d\tau\\
\ls &\ep\log^4(2+t)+T^{-\f\gamma 2}\sum_{|J|\leq |I|}\sup_{\tau\in[T,t]}\left(\||\rd\Gamma^J h| w^{\f12}\|_{L^2_x(\mathcal R_\tau)}+\|v(U'(\mathcal R))|\Gamma^J h|\|_{L^2(\mathbb S^2(U'(\mathcal R),\tau))}\right)\\
&+(1+t)^{\de_0}\sum_{|J|\leq |I|-1} \sup_{\tau\in [T,t]}\left(\||\rd\Gamma^J h| w^{\f12}\|_{L^2_x(\mathcal R_\tau)}+\|v(U'(\mathcal R))|\Gamma^J h|\|_{L^2(\mathbb S^2(U'(\mathcal R),\tau))}\right)\\
&+T^{-\f\gamma 2}\sum_{|J|\leq |I|}\left(\int_T^t \int_{ \mathcal R_\tau} |\bar{\rd}\Gamma^J h|^2 w'(r-\tau) dx\, d\tau\right)^{\f12}.
\end{split}
\end{equation*}
Combining this with \eqref{tg.terms}, \eqref{G.gd.est} and \eqref{G.LL.est.1} gives the desired bounds.
\end{proof}

We gather all the bounds derived so far in this section to obtain the following proposition:
\begin{proposition}\label{EE.main}
The following bounds hold for $|I|\leq N$ and $t>T>0$:
\begin{equation*}
\begin{split}
&\int_T^t \|(\mathfrak I_I+\mathfrak G_I+\mathfrak T_I+\mathfrak L_I+\mathfrak W_I+\mathfrak N_I) w^{\f12}\|_{L^2_x(\mathcal R_\tau)} d\tau\\
\ls &\ep\log^4(2+t)+(1+t)^{\de_0}\sum_{|J|\leq |I|-1} \sup_{\tau\in [T,t]}\|\left(|\rd\Gamma^J h|+|\rd\Gamma^J \beta|\right) w^{\f12}\|_{L^2_x(\mathcal R_\tau)}\\
&+T^{-\f\gamma 2}\sum_{|J|\leq |I|}\sup_{\tau\in[T,t]}\|\left(|\rd\Gamma^J h|+|\rd\Gamma^J \beta|\right) w^{\f12}\|_{L^2_x(\mathcal R_\tau)}\\
&+T^{-\f\gamma 2}\sum_{|J|\leq |I|}\left(\int_T^t \int_{ \mathcal R_\tau} \left(|\bar{\rd}\Gamma^J h|^2+|\bar{\rd}\Gamma^J \beta|^2\right) w'(r-\tau) dx\, d\tau\right)^{\f12}\\
&+\sum_{|J|\leq |I|}\sup_{\tau\in[T,t]}(1+\tau)^{\de_0}\|v(U'(\mathcal R))|\Gamma^J h|\|_{L^2(\mathbb S^2(U'(\mathcal R),\tau))}.
\end{split}
\end{equation*}
\end{proposition}

We end this section by proving an analogue of Proposition \ref{EE.main} for terms in the equation $\tBox_g(\Gamma^I \beta)$. Recall from Proposition \ref{SF.schematic.eqn} that all of the terms are analogous to those in Proposition \ref{schematic.eqn} with the important exception that there are no analogue of $\mathfrak B_I$ in Proposition \ref{SF.schematic.eqn}. As a consequence, using exactly the same argument as that which leads to Proposition \ref{EE.main}, we have the following estimate for $|\tBox_g(\Gamma^I \beta)|$.
\begin{proposition}\label{SF.EE.main}
The following bounds hold for $|I|\leq N$ and $t>T>0$:
\begin{equation*}
\begin{split}
&\int_T^t \||\tBox_g(\Gamma^I\beta)| w^{\f12}\|_{L^2_x(\mathcal R_\tau)} d\tau\\
\ls &\ep\log^4(2+t)+(1+t)^{\de_0}\sum_{|J|\leq |I|-1} \sup_{\tau\in [T,t]}\|\left(|\rd\Gamma^J h|+|\rd\Gamma^J \beta|\right) w^{\f12}\|_{L^2_x(\mathcal R_\tau)}\\
&+T^{-\f\gamma 2}\sum_{|J|\leq |I|}\sup_{\tau\in[T,t]}\|\left(|\rd\Gamma^J h|+|\rd\Gamma^J \beta|\right) w^{\f12}\|_{L^2_x(\mathcal R_\tau)}\\
&+T^{-\f\gamma 2}\sum_{|J|\leq |I|}\left(\int_T^t \int_{ \mathcal R_\tau} \left(|\bar{\rd}\Gamma^J h|^2+|\bar{\rd}\Gamma^J \beta|^2\right) w'(r-\tau) dx\, d\tau\right)^{\f12}\\
&+\sum_{|J|\leq |I|}\sup_{\tau\in[T,t]}(1+\tau)^{\de_0}\|v(U'(\mathcal R))|\Gamma^J h|\|_{L^2(\mathbb S^2(U'(\mathcal R),\tau))}.
\end{split}
\end{equation*}
\end{proposition}

\section{Region near spacelike infinity}\label{sec.II}
In this section, we prove the energy estimates in the region $\mathcal R_2$ near spacelike infinity (see Figure~\ref{fig:regions}); the main result of the section is Proposition \ref{EE.R2}. We first need to control the bad term $\mathfrak B_I$ in Proposition \ref{schematic.eqn}, which has not been estimated in Proposition \ref{EE.main}. The key observation, as we discussed in the introduction, is that we can obtain a smallness constant for $U_2$ sufficiently negative.
\begin{proposition}\label{B.EE.II}
In the region $\mathcal R_2$, for $U_2<0$, the following bounds for $\mathfrak B_I$ hold for $|I|\leq N$ and $t>T>0$:
\begin{equation*}
\begin{split}
\int_T^t \| \mathfrak B_I w^{\f12}\|_{L^2_x(\mathcal R_{2,\tau})} d\tau
\ls &\max\{\ep^{\f12},\f{1}{(1+|U_2|)^{\gamma}}\}\sum_{|J|= |I|}\int_T^t \f{\|\left(|\rd\Gamma^J h|+|\rd\Gamma^J \beta|\right) w^{\f12}\|_{L^2_x(\mathcal R_{2,\tau})}}{1+\tau} d\tau.
\end{split}
\end{equation*}
\end{proposition}
\begin{proof}
Recall from \eqref{schematic.eqn} that 
$$\mathfrak B_I= |\rd\Gamma^I h|_{\mathcal T\mathcal U}|\rd h|_{\mathcal T\mathcal U}+|\rd\Gamma^I\beta||\rd\beta|+\frac{|\rd\Gamma^I h|_{\mathcal T\mathcal U}+|\rd\Gamma^I\beta|}{(1+s)(1+|q|)^{\gamma}}.$$
By the bootstrap assumptions \eqref{BA3} and \eqref{BASF3}, we have $|\rd h|_{\mathcal T\mathcal U}+|\rd\beta|\leq \f{2\ep^\f12}{(1+s)}.$ Also, in the region $\mathcal R_2$, the bound $\f{1}{(1+s)(1+|q|)^{\gamma}}\leq \f{1}{(1+s)(1+|U_2|)^{\gamma}}$ holds. The conclusion follows from directly plugging in these bounds.
\end{proof}
Using the energy estimates in Proposition \ref{EE.2} and controlling the error terms using Propositions \ref{EE.main} and \ref{B.EE.II}, we get
\begin{proposition}\label{EE.R2}
Let $0\leq k\leq N$. There exists $\ep_2\in (0,\ep_1]$ sufficiently small, $U_2<0$ sufficiently negative and $T_2>0$ sufficiently large such that the following estimate\footnote{From now on, we use the notation that $C_T>0$ is a constant depending on $T$, which can be different from line to line.} holds in the region $\mathcal R_2$:
\begin{equation*}
\begin{split}
\sum_{|I|\leq k}\sup_{\tau\in [T,t]}\|\left(|\rd\Gamma^I h|+|\rd\Gamma^I \beta|\right) w^{\f12}\|_{L^2_x(\mathcal R_{2,\tau})}
\ls &C_T\,\ep(1+t)^{(2k+2)\de_0},
\end{split}
\end{equation*}
for $\ep<\ep_2$ and $t>T>T_2$. Moreover, on the boundary $\mathcal B_{U_2}$, the following estimates are verified:
\begin{equation}\label{II.bdry.est.1}
\begin{split}
\sum_{|I|\leq k}\left(\int_{\mathcal B_{U_2}\cap\{T\leq \tau\leq t\}}\left(|\db\Gamma^I h|^2+|\db\Gamma^I \beta|^2+\f{|\rd\Gamma^I h|^2+|\rd\Gamma^I \beta|^2}{(1+t)^{\f{\gamma}{4}+1}}\right) w(r-\tau)\,dx\right)^{\f12}\ls C_T\,\ep(1+t)^{(2k+2)\de_0}
\end{split}
\end{equation}
and
\begin{equation}\label{II.bdry.est.2}
\sum_{|I|\leq k}\|v(U_2)|\Gamma^I h|\|_{L^2(\mathbb S^2(U_2,t))}\ls C_T\,\ep(1+t)^{(2k+2)\de_0}.
\end{equation}
\end{proposition}
\begin{proof}
By Propositions \ref{EE.2}, \ref{Cauchy.stability} and \ref{EE.main}, for every choice of $U_2$, there exists $T'_2(U_2)$ sufficiently large such that the following estimate holds for $T>T'_2(U_2)$ with implicit constants in $\ls$ in particular independent\footnote{Notice that we nevertheless have a term depending on $C_T$, which arises from the application of Proposition \ref{Cauchy.stability}.} of $U_2$, $T$ and $T'_2(U_2)$:
\begin{equation}\label{EE.R2.1}
\begin{split}
&\sup_{\tau\in[T,t]}(\int_{\mathcal R_{2,\tau}} \left(|\rd\Gamma^I h|^2+|\rd\Gamma^I \beta|^2\right) w(r-\tau)\,dx)^{\f12}+(\int_{T}^{t}\int_{\mathcal R_{2,\tau}}\left(|\db\Gamma^I h|^2+|\db\Gamma^I \beta|^2\right)w'(r-\tau)\,dx\,d\tau)^{\f12}\\
&+(\int_{\mathcal B_{U_2}\cap\{T\leq \tau\leq t\}}(|\db\Gamma^I h|^2+|\db\Gamma^I \beta|^2+\f{|\rd\Gamma^I h|^2+|\rd\Gamma^I \beta|^2}{(1+t)^{\f{\gamma}{4}+1}}) w(r-\tau)\,dx)^{\f12}\\
\ls &(\int_{\mathcal R_{2,T}} \left(|\rd\Gamma^I h|^2+|\rd\Gamma^I \beta|^2\right) w(r-T)\,dx)^{\f12}+(\int_{T}^{t}\|\left(|\tBox_g(\Gamma^I h)|+|\tBox_g(\Gamma^I \beta)|\right)w^{\f12}\|_{L^2(\mathcal R_{2,\tau})}\,d\tau)^{\f12}\\
\ls &C_T\,\ep\log^4(2+t)+T^{-\f\gamma 2}\sum_{|J|\leq |I|}\sup_{\tau\in[T,t]}\|\left(|\rd\Gamma^J h|+|\rd\Gamma^J \beta|\right) w^{\f12}\|_{L^2_x(\mathcal R_{2,\tau})}\\
&+(1+t)^{\de_0}\sum_{|J|\leq |I|-1} \sup_{\tau\in [T,t]}\|\left(|\rd\Gamma^J h|+|\rd\Gamma^J \beta|\right) w^{\f12}\|_{L^2_x(\mathcal R_{2,\tau})}\\
&+T^{-\f\gamma 2}\sum_{|J|\leq |I|}(\int_T^t \int_{\mathcal R_{2,\tau}} \left(|\bar{\rd}\Gamma^J h|^2+|\bar{\rd}\Gamma^J \beta|^2\right) w'(r-\tau) dx\, d\tau)^{\f12}\\
&+\max\{\ep^{\f12},\f{1}{(1+|U_2|)^{\gamma}}\}\sum_{|J|= |I|}\int_T^t \f{\|\left(|\rd\Gamma^J h|+|\rd\Gamma^J \beta|\right) w^{\f12}\|_{L^2_x(\mathcal R_{2,\tau})}}{1+\tau} d\tau.
\end{split}
\end{equation}
Notice that we have applied Proposition~\ref{EE.main} with $U'(\mathcal R_2)=-\infty$, i.e.~there are no boundary terms on $\mathcal B_{U'(\mathcal R_2)}$ when applying Proposition \ref{EE.main}. We now sum \eqref{EE.R2.1} over all $|I|\leq k$. For every $U_2$, there exists a large $T''_2(U_2)>T'_2(U_2)$ such that for $T>T''_2(U_2)$, the second and fourth terms can be absorbed into the left hand side to get
\begin{equation}\label{EE.II.1}
\begin{split}
&\sum_{|I|\leq k}\left(\sup_{\tau\in[T,t]}\|\left(|\rd\Gamma^J h|+|\rd\Gamma^J \beta|\right) w^{\f12}\|_{L^2_x(\mathcal R_{2,\tau})}+(\int_{T}^{t}\int_{\mathcal R_{2,\tau}}\left(|\db\Gamma^I h|^2+|\db\Gamma^I \beta|^2\right) w'(r-\tau)\,dx\,d\tau)^{\f12}\right)\\
&+\sum_{|I|\leq k}\left(\int_{\mathcal B_{U_2}\cap\{T\leq \tau\leq t\}}\left(|\db\Gamma^I h|^2+|\db\Gamma^I \beta|^2+\f{|\rd\Gamma^I h|^2+|\rd\Gamma^I \beta|^2}{(1+t)^{\f{\gamma}{4}+1}}\right) w(r-\tau)\,dx\right)^{\f12}\\
\ls &C_T\,\ep\log^{4}(2+t)+(1+t)^{\de_0}\sum_{|J|\leq k-1} \sup_{\tau\in [T,t]}\|\left(|\rd\Gamma^J h|+|\rd\Gamma^J \beta|\right) w^{\f12}\|_{L^2_x(\mathcal R_{2,\tau})}\\
&+\max\{\ep^{\f12},\f{1}{(1+|U_2|)^{\gamma}}\}\sum_{|J|= k}\int_T^t \f{\|\left(|\rd\Gamma^J h|+\||\rd\Gamma^J \beta|\right) w^{\f12}\|_{L^2_x(\mathcal R_{2,\tau})}}{1+\tau} d\tau.
\end{split}
\end{equation}
We now proceed to show by induction that by choosing $U_2$ and $\ep_2$ appropriately, and then choosing $T_2>T''_2(U_2)$, we have
\begin{equation}\label{EE.II.2}
\sum_{|I|\leq k} \sup_{\tau\in[T,t]} \|\left(|\rd\Gamma^J h|+|\rd\Gamma^J \beta|\right) w^{\f12}\|_{L^2_x(\mathcal R_{2,\tau})}\ls C_T\,\ep(1+t)^{2(k+1)\de_0}
\end{equation}
for every $k\leq N$. To begin with the base case, notice that when $k=0$, the second term in \eqref{EE.II.1} is absent. Therefore, for $\ep$ sufficiently small and $U_2$ sufficiently negative, we have by the Gr\"onwall's inequality that\footnote{Here, and below, $C_2>0$ is some constant (which can be different from line to line) depending on $C$, $N$, $\gamma$ and $\de_0$.}
\begin{equation*}
\begin{split}
&\sum_{|I|\leq k}\sup_{\tau\in[T,t]}\sup_{\tau\in[T,t]}\|\left(|\rd\Gamma^J h|+|\rd\Gamma^J \beta|\right) w^{\f12}\|_{L^2_x(\mathcal R_{2,\tau})}\\
\ls &C_T\,\ep\log^4(2+t)\exp(C_2\max\{\ep^{\f12},\f{1}{(1+|U_2|)^{\gamma}}\}\int_T^t \f{d\tau}{1+\tau})\\
\ls &C_T\,\ep(1+t)^{\de_0}\log^4(2+t)\ls C_T\,\ep(1+t)^{2\de_0}.
\end{split}
\end{equation*}
We now continue with the induction step by assuming that \eqref{EE.II.2} holds with $k\leq k_0-1$ for some $1\leq k_0\leq N$. Then using \eqref{EE.II.1} again we have
\begin{equation*}
\begin{split}
&\sum_{|I|\leq k}\sup_{\tau\in[T,t]}\|\left(|\rd\Gamma^J h|+|\rd\Gamma^J \beta|\right) w^{\f12}\|_{L^2_x(\mathcal R_{2,\tau})}\\
\ls &C_T\,\ep(1+t)^{(2k_0+1)\de_0}+\max\{\ep^{\f12},\f{1}{(1+|U_2|)^{\gamma}}\}\sum_{|J|= k}\int_T^t \f{\|\left(|\rd\Gamma^J h|+|\rd\Gamma^J \beta|\right) w^{\f12}\|_{L^2_x(\mathcal R_{2,\tau})}}{1+\tau} d\tau.
\end{split}
\end{equation*}
By Gr\"onwall's inequality, we obtain that for $\ep$ sufficiently small and $U_2$ sufficiently negative
\begin{equation*}
\begin{split}
&\sum_{|I|\leq k}\sup_{\tau\in[T,t]}\|\left(|\rd\Gamma^J h|+|\rd\Gamma^J \beta|\right) w^{\f12}\|_{L^2_x(\mathcal R_{2,\tau})}\\
\ls &C_T\,\ep(1+t)^{(2k_0+1)\de_0}\exp\left(C_2\max\{\ep^{\f12},\f{1}{(1+|U_2|)^{\gamma}}\}\int_T^t \f{d\tau}{1+\tau}\right)
\ls C_T\,\ep(1+t)^{(2k_0+2)\de_0}.
\end{split}
\end{equation*}
This concludes the proof of \eqref{EE.II.2}. Returning to \eqref{EE.II.1}, we also get \eqref{II.bdry.est.1}. Finally, \eqref{II.bdry.est.2} follows from the estimate \eqref{EE.II.2} together with an application of the Hardy inequality in Proposition \ref{Hardy} with $\alp=0$, $\mu_2=2\gamma$, $R_1=t-\f{1}{(1+t)^{\f{\gamma}{4}}}-U_2$ and $R_2=\infty$.
\end{proof}

This concludes the estimates in the region $\mathcal R_2$. {\bf We now fix the parameter $U_2$ according to Proposition \ref{EE.R2}.}

\section{Region near null infinity}\label{sec.III}

In this section, we prove energy estimates in the region $\mathcal R_3$, which is the region near null infinity but away from spacelike and timelike infinities (see Figure~\ref{fig:regions}). In this region, the background quantities $|\rd\Gamma^I h_B|$ and $|\rd\Gamma^I\phi_B|$ are only of size $\f C{1+t}$, without any additional smallness as\footnote{In Sections \ref{sec.II} and \ref{sec.IV} where we deal with the regions $\mathcal R_2$ and $\mathcal R_4$, we have an extra smallness factor of $\f{1}{(1+|U_2|)^{\gamma}}$ or $\f{1}{(1+|U_3)^{\gamma}}$.} in the regions $\mathcal R_2$ and $\mathcal R_4$. The term $\f C{1+t}$ is barely non-integrable in time, and if one were to estimate the corresponding error term naively using Gr\"onwall's inequality, the energy would grow as $(1+t)^C$ and one will not be able to recover the bootstrap assumptions.

In order to handle these terms, we need to capture the reductive structure of the system of equations when proving energy estimates. In other words, we need to first prove energy estimates for the components of $\rd\Gamma^I h$ for which the right hand side does not have the bad terms alluded to above. We then use this estimates that we have already obtained to control the remaining components of $\rd\Gamma^I h$, so that the energy would still grow only with a slow rate. This is reminiscent of the reductive structure that was used in \cite{LR1, LR2}, although in \cite{LR1, LR2} it was only used for the $L^\infty$ estimate. In order to reveal this reductive structure for the energy estimates (as opposed to $L^\infty$ estimates), we need to commute $\tBox_g$ with the projections to ${\bf E}^\mu$. One of the key observations is that the $\frac{1}{r}|\rd \Gamma^I h|$ terms that we generate in this commutation in fact contain only good derivatives, i.e.~they are of the form $\frac 1r |\db \Gamma^I h|$ (see Proposition \ref{main.frame.proj}). As a consequence, all such terms obey sufficiently strong estimates (Proposition \ref{proj.est}). We can then use this and the reductive structure to prove the desired energy estimates (Propositions \ref{EE.R3.1} and \ref{EE.R3}).

We now turn to the details. First, we have the following proposition, which contains the crucial observation that the most slowly decaying terms in the commutation of $\tBox_g$ with the projection to ${\bf E}^\mu$ in fact have good $\db$ derivatives:
\begin{proposition}\label{main.frame.proj}
Given ${\bf E}^{\mu},{\bf E}^{\nu}\in\{L,\Lb,E^1,E^2,E^3\}$, we have for every $(t,x)\in [0,\infty)\times \mathbb R^3$ that
$$|(\tBox_g \Gamma^I h_{\mu\nu}){\bf E}^{\mu}{\bf E}^{\nu}-\tBox_g ((\Gamma^I h)_{\mu\nu}{\bf E}^{\mu}{\bf E}^{\nu})|\ls \frac 1r |\db \Gamma^I h|+\frac 1{r^2}|\Gamma^I h|+\f{(1+|q|)^{\f12+\f{\gamma}{4}}}{r(1+s)^{1-\de_0}}|\rd\Gamma^I h|.$$
\end{proposition}
\begin{proof}
A direct calculation shows that
\begin{equation}\label{frame.proj.1}
\begin{split}
&(\tBox_g \Gamma^I h_{\mu\nu}){\bf E}^{\mu}{\bf E}^{\nu}-\tBox_g (\Gamma^I h_{\mu\nu}{\bf E}^{\mu}{\bf E}^{\nu})\\
=&-2(g^{-1})^{\alp\bt}(\Gamma^I h_{\mu\nu})(\rd_\alp {\bf E}^\mu)(\rd_\bt {\bf E}^\nu)-\underbrace{4(g^{-1})^{\alp\bt}\rd_\alp (\Gamma^I h_{\mu\nu})(\rd_\bt {\bf E}^\mu) {\bf E}^\nu}_{=:main\mbox{ }term}-2(\Gamma^I h_{\mu\nu})(\tBox_g {\bf E}^\mu) {\bf E}^\nu.
\end{split}
\end{equation}
Expressing the ${\bf E}^\mu$ in terms of the coordinate vector fields $\rd_\alp$, we notice that the coefficients are either constant or take the form $\f{x^i}{r}$. Therefore, we have
$$|{\bf E}|\ls 1,\quad |\rd {\bf E}|\ls \f{1}{r},\quad |\tBox_g {\bf E}|\ls \f{1}{r^2}.$$
Therefore,
\begin{equation}\label{frame.proj.2}
|-2(g^{-1})^{\alp\bt}(\Gamma^I h_{\mu\nu})(\rd_\alp {\bf E}^\mu)(\rd_\bt {\bf E}^\nu)-2(\Gamma^I h_{\mu\nu})(\tBox_g {\bf E}^\mu) {\bf E}^\nu|\ls \f{1}{r^2}|\Gamma^I h|.
\end{equation}
It thus remains to control the main term in \eqref{frame.proj.1} above. We now write $(g^{-1})^{\alp\bt}=m^{\alp\bt}+H^{\alp\bt}$ so that in the main term, there is a contribution from $m^{\alp\bt}$ and one from $H^{\alp\bt}$. Recall that $m^{\alp\bt}=-L^{(\alp}\Lb^{\bt)}+\sum_{A=1}^3 E^A E^A$ and notice that we have $\rd_s {\bf E}^\mu=\rd_q {\bf E}^\mu=0$ for any ${\bf E}^\mu$. Therefore, in the contribution from $m^{\alp\bt}$, there are no $\f{1}{r}|\rd_q \Gamma^I h|$ terms! On the other hand, the contribution from $H^{\alp\bt}$ has more decay. More precisely,
\begin{equation}\label{frame.proj.3}
\begin{split}
&|-4(g^{-1})^{\alp\bt}\rd_\alp (\Gamma^I h_{\mu\nu})(\rd_\bt {\bf E}^\mu) {\bf E}^\nu|\\
=&|-4m^{\alp\bt}\rd_\alp (\Gamma^I h_{\mu\nu})(\rd_\bt {\bf E}^\mu) {\bf E}^\nu-4H^{\alp\bt}\rd_\alp (\Gamma^I h_{\mu\nu})(\rd_\bt {\bf E}^\mu) {\bf E}^\nu|\\
=&|-2\sum_{i,j=1}^3(\f{x^i}{r}\rd_j(\Gamma^I h_{\mu\nu})-\f{x^j}{r}\rd_i(\Gamma^I h_{\mu\nu}))(\f{x^i}{r}\rd_j{\bf E}^\mu-\f{x^j}{r}\rd_i{\bf E}^\mu){\bf E}^\nu-4H^{\alp\bt}\rd_\alp (\Gamma^I h_{\mu\nu})(\rd_\bt {\bf E}^\mu) {\bf E}^\nu|\\
\ls &\f{1}{r}|\bar{\rd}\Gamma^I h|+\f{(1+|q|)^{\f12+\f{\gamma}{4}}}{r(1+s)^{1-\de_0}}|\rd\Gamma^I h|,
\end{split}
\end{equation}
where in the last line we have also used the bounds for $|H|$ in Proposition \ref{inverse} together with bootstrap assumption \eqref{BA4}. Combining the \eqref{frame.proj.1}, \eqref{frame.proj.2} and \eqref{frame.proj.3} above gives the conclusion of the proposition.
\end{proof}

\begin{remark}
While Proposition \ref{main.frame.proj} shows that the commutation of $\tBox_g$ and the projection to the ${\bf E}^\mu$ is favorable in the sense that the $\f1r |\rd\Gamma^I h|$ terms are absent, it seems that this is only useful in the region under consideration, but not in the regions near spacelike infinity or timelike infinity. As we will see below, the crucial fact for $|\bar\rd\Gamma^I\phi|$ that we will use is that it obeys an $L^2_tL^2_x$ estimate. On the other hand, it is also important that we are only dealing with a region {\bf with finite $q$ range}, as otherwise the $|q|$-weights in the $L^2_tL^2_x$ estimate will not be sufficient to control this term.
\end{remark}

Our next goal is to show that all the error terms in Proposition \ref{main.frame.proj} arising from commutation with the projection to ${\bf E}^\mu$ can be controlled. In Region $\mathcal R_3$, if $T>2U_3$, then we have $t\ls r$ and therefore the term $\f{(1+|q|)^{\f12+\f{\gamma}{4}}}{r(1+s)^{1-\de_0}}|\rd\Gamma^I h|$ behaves better than the $\mathfrak T_I$ term. It can therefore be treated in the same manner as in Proposition \ref{T.EE}. Let us summarize this as follows:
\begin{proposition}\label{proj.est.0}
For every $U_3>0 > U_{2}$, if $T>2U_3$, then the following estimate holds for all $|I|\leq N$ and $t>T$:
\begin{equation*}
\begin{split}
\int_T^t \|\f{(1+|q|)^{\f12+\f{\gamma}{4}}}{r(1+s)^{1-\de_0}}|\rd\Gamma^I h|w^{\f12}\|_{L^2_x(\mathcal R_{3,\tau})}\,d\tau\ls T^{-\f\gamma 2}\sup_{\tau\in [T,t]}\| |\rd\Gamma^J h| w^{\f 12}\|_{L^2_x(\mathcal R_{3,\tau})}.
\end{split}
\end{equation*}
\end{proposition}
\begin{proof}
This follows from noting that the term can be dominated by $\mathfrak T_I$ and using the estimates in Proposition \ref{T.EE}.
\end{proof}
On the other hand, the other two error terms in Proposition \ref{main.frame.proj} are not as good as the terms estimated in Section \ref{sec.EE} and we need to crucially use the fact that we are localized in a bounded $q$ region\footnote{Since we are in a bounded $q$ region, we have the bounds in terms of the (large) parameters $|U_2|$ and $|U_3|$ below. We will later choose $T$ sufficiently large so that $T^{-\f 12}\max\{|U_2|^{\f 12},|U_3|^{\f 12+\f\gamma 4}\}$ and $T^{-1}\max\{|U_2|,|U_3|\}$ are small.}. More precisely, we have the following bounds for those terms: 
\begin{proposition}\label{proj.est}
For every $U_3>0 > U_2$, if $T>2U_3$, then the following estimate holds for all $|I|\leq N$ and $t>T$:
\begin{equation*}
\begin{split}
&\int_T^t \left(\underbrace{\|\frac 1r |\db \Gamma^I h|w^{\f12}\|_{L^2_x(\mathcal R_{3,\tau})}}_{=:I}+\underbrace{\|\frac 1{r^2}|\Gamma^I h|w^{\f12}\|_{L^2_x(\mathcal R_{3,\tau})}}_{=:II}\right)\,d\tau\\
\ls &T^{-\f 12}\max\{|U_2|^{\f 12},|U_3|^{\f 12+\f\gamma 4}\}\left(\int_T^t\int_{\mathcal R_{3,\tau}} |\db\Gamma^I h|^2 w'(r-\tau) dx\,d\tau\right)^{\f12}\\
&+T^{-1}\max\{|U_2|,|U_3|\}\sup_{\tau\in [T,t]} \left(\||\rd\Gamma^I h| w^{\f12}\|_{L^2_x(\mathcal R_{3,\tau})}+\|v(U_2) |\Gamma^I h|\|_{L^2_x(\mathbb S^2(U_2,\tau))}\right).
\end{split}
\end{equation*}

\end{proposition}
\begin{proof}
In the following, we will frequently use the easy observation that if $T>2U_3$, then we have $1+s\ls r$ in region $\mathcal R_3$. We first control the term $I$, which involves a good derivative of $\Gamma^I h$:
\begin{equation*}
\begin{split}
&\int_T^t \|\frac 1r |\db \Gamma^I h|w^{\f12}\|_{L^2_x(\mathcal R_{3,\tau})}d\tau\\
\ls &\max\{|U_2|^{\f12},|U_3|^{\f 12+\f \gamma 4}\}\left(\int_T^t \f{d\tau}{(1+s)^2}\right)^{\f12} \left(\int_T^t\int_{\mathcal R_{3,\tau}} |\db\Gamma^I h|^2 w'(r-\tau) dx\,d\tau\right)^{\f12}\\
\ls &T^{-\f 12}\max\{|U_2|^{\f 12},|U_3|^{\f 12+\f \gamma 4}\}\left(\int_T^t\int_{\mathcal R_{3,\tau}} |\db\Gamma^I h|^2 w'(r-\tau) dx\,d\tau\right)^{\f12}.
\end{split}
\end{equation*}
Here, we have used that fact that $\f{w(q)}{w'(q)}\ls 
\begin{cases}
1+|q|,\quad\mbox{if }q\geq 0\\
(1+|q|)^{1+\f{\gamma}{2}},\quad\mbox{if }q< 0
\end{cases}$.

The term $II$ in the statement of the proposition can be controlled using the Hardy inequality (Proposition \ref{Hardy}) after losing appropriate powers of $|U_2|$ and $|U_3|$. More precisely, we apply Proposition \ref{Hardy} with $\alp=\mu_1=0$, $\mu_2=2\gamma$ to get
\begin{equation*}
\begin{split}
&\int_T^t \|\frac 1{r^2} |\Gamma^I h|w^{\f12}\|_{L^2_x(\mathcal R_{3,\tau})}d\tau\\
\ls &\int_T^t \left(\int_{\mathcal R_{3,\tau}\cap \{r-\tau\geq 0\}} \f{|\Gamma^I h|^2 (1+|r-\tau|)^{1+2\gamma}}{(1+s)^4} dx\right)^{\f12}d\tau+\int_T^t \left(\int_{\mathcal R_{3,\tau}\cap \{r-\tau< 0\}} \f{|\Gamma^I h|^2}{(1+s)^4} dx\right)^{\f12}d\tau\\
\ls &|U_2|\int_T^t \left(\int_{\mathcal R_{3,\tau}\cap \{r-\tau\geq 0\}} \f{|\Gamma^I h|^2 }{(1+s)^4(1+|r-\tau|)^{1-2\gamma}} dx\right)^{\f12}d\tau\\
&+|U_3|\int_T^t \left(\int_{\mathcal R_{3,\tau}\cap \{r-\tau< 0\}} \f{|\Gamma^I h|^2}{(1+s)^4(1+|r-\tau|)^{2}} dx\right)^{\f12}d\tau\\
\ls & \max\{|U_2|,|U_3|\}\left(\sup_{\tau\in [T,t]} \left(\int_{\mathcal R_{3,\tau}} |\rd\Gamma^I h|^2  w(r-\tau)  dx\right)^{\f12}\right)\left(\int_T^t \f{d\tau}{(1+\tau)^2}\right)\\
&+\max\{|U_2|,|U_3|\} \left(\sup_{\tau\in[T,t]}\|v(U_2) |\Gamma^I h|\|_{L^2(\mathbb S^2(U_2,\tau))}\right)\left(\int_T^t \f{d\tau}{(1+\tau)^2}\right)\\
\ls & T^{-1}\max\{|U_2|,|U_3|\}\sup_{\tau\in [T,t]} \left(\||\rd\Gamma^I h| w^{\f12}\|_{L^2_x(\mathcal R_{3,\tau})}+\|v(U_2) |\Gamma^I h|\|_{L^2(\mathbb S^2(U_2,\tau))}\right). \qedhere
\end{split}
\end{equation*}
\end{proof}

We now proceed to obtaining the energy estimates for $h$ in this region. We first derive the estimates for $|\rd\Gamma^I h|_{\mathcal T\mathcal U}$. Recall from Proposition \ref{schematic.eqn.improved} that the right hand side of the $(\tBox_g\Gamma^I h)_{\mathcal T\mathcal U}$ does not contain the $\mathfrak B_I$ term. Therefore, we have the following estimate for $|\rd\Gamma^I h|_{\mathcal T\mathcal U}$ in region $\mathcal R_3$:
\begin{proposition}\label{EE.R3.1}
For every $U_3> 0 >U_2$, there exists $\ep'_3\in (0,\ep_2]$, $T'_3> T_2$ such that if $\ep\leq \ep'_3$, $t>T> T'_3$ and $k\leq N$, we have
\begin{equation*}
\begin{split}
&\sum_{|I|\leq k}\sup_{\tau\in[T,t]}\| \left(|\rd\Gamma^I h|_{\mathcal T\mathcal U}+|\rd\Gamma^I \beta|\right) w^{\f12} \|_{L^2_x(\mathcal R_{3,\tau})}\\
&+\sum_{|I|\leq k}\left(\int_{\mathcal B_{U_3}\cap\{T\leq \tau\leq t\}}(|\db\Gamma^I \beta|^2+\f{|\rd\Gamma^I \beta|^2}{(1+t)^{\f{\gamma}{4}+1}})w(r-\tau) \,dx\right)^{\f12}\\
\ls &C_T\,\ep(1+t)^{(2k+3)\de_0}+(1+t)^{\de_0}\sum_{|J|\leq k-1} \sup_{\tau\in [T,t]}\|\left(|\rd\Gamma^J h|+|\rd\Gamma^J \beta|\right) w^{\f12}\|_{L^2_x(\mathcal R_{3,\tau})}\\
&+T^{-\f\gamma 2}\sum_{|J|\leq k}\left(\sup_{\tau\in[T,t]}\||\rd\Gamma^J h| w^{\f12}\|_{L^2_x(\mathcal R_{3,\tau})}+\left(\int_T^t \int_{ \mathcal R_{3,\tau}} |\bar{\rd}\Gamma^J h|^2 w'(r-\tau) dx\, d\tau\right)^{\f12}\right).
\end{split}
\end{equation*}
\end{proposition}
\begin{proof}
We first bound $|\rd\Gamma^I h|_{\mathcal T\mathcal U}$. Take ${\bf E}_1^\alp\in \mathcal T$ and ${\bf E}_2^\alp\in\mathcal U$. We apply Proposition \ref{EE.3} to $h_{\mu\nu}{\bf E}_1^\mu {\bf E}_2^\nu$. By Proposition \ref{Cauchy.stability} and \eqref{II.bdry.est.1} in Proposition \ref{EE.R2}, we know that for $|I|\leq k$, the ``initial data'' terms, i.e.~the first two terms on the right hand side in Proposition \ref{EE.3} can be controlled by $C_T\,\ep(1+t)^{(2k+2)\de_0}$. Therefore, by Proposition \ref{EE.3}, for $|I|\leq k$, we have
\begin{equation}\label{EE.R3.1.1}
\begin{split}
&\sup_{\tau\in[T,t]}\| |\rd\Gamma^I h|_{\mathcal T\mathcal U} w^{\f12} \|_{L^2_x(\mathcal R_{3,\tau})}\\
\ls &C_T \ep(1+t)^{(2k+2)\de_0}+\left(\int_{T}^{t}\||\tBox_g \Gamma^I h|_{\mathcal T\mathcal U}w^{\f12}\|_{L^2_x(\mathcal R_{3,\tau})}\,d\tau\right)^{\f12}\\
&+\left(\int_{T}^{t}\|\left((\tBox_g \Gamma^I h_{\mu\nu}){\bf E}^{\mu}{\bf E}^{\nu}-\tBox_g (\Gamma^I h_{\mu\nu}{\bf E}^{\mu}{\bf E}^{\nu})\right)w^{\f12}\|_{L^2_x(\mathcal R_{3,\tau})}\,d\tau\right)^{\f12}.
\end{split}
\end{equation}
The main observation is that $|\tBox_g \Gamma^I h|_{\mathcal T\mathcal U}$ does not contain the $\mathfrak B_I$ term (see Proposition~\ref{schematic.eqn.improved}) and therefore, we can apply Proposition \ref{EE.main} to control the $|\tBox_g \Gamma^I h|_{\mathcal T\mathcal U}$ term. Using Proposition \ref{EE.main} together with the bound \eqref{II.bdry.est.2} in Proposition \ref{EE.R2} for $\sum_{|J|\leq |I|}\|v(U'(\mathcal R_3))|\Gamma^J h|\|_{L^2(\mathbb S^2(U'(\mathcal R_3)))}$, we get
\begin{equation}\label{EE.R3.1.2}
\begin{split}
&(\int_{T}^{t}\||\tBox_g \Gamma^I h|_{\mathcal T\mathcal U} w^{\f12}\|_{L^2_x(\mathcal R_{3,\tau})}\,d\tau)^{\f12}\\
\ls &C_T\,\ep(1+t)^{(2k+3)\de_0}+(1+t)^{\de_0}\sum_{|J|\leq |I|-1} \sup_{\tau\in [T,t]}\|\left(|\rd\Gamma^J h|+|\rd\Gamma^J \beta|\right) w^{\f12}\|_{L^2_x(\mathcal R_{3,\tau})}\\
&+T^{-\f\gamma 2}\sum_{|J|\leq |I|}\sup_{\tau\in[T,t]}\|\left(|\rd\Gamma^J h|+|\rd\Gamma^J \beta|\right) w^{\f12}\|_{L^2_x(\mathcal R_{3,\tau})}\\
&+T^{-\f\gamma 2}\sum_{|J|\leq |I|}\left(\int_T^t \int_{ \mathcal R_{3,\tau}} \left(|\bar{\rd}\Gamma^J h|^2+|\bar{\rd}\Gamma^J \beta|^2\right) w'(r-\tau) dx\, d\tau\right)^{\f12}.
\end{split}
\end{equation}
The final term in \eqref{EE.R3.1.1} can be controlled by combining the estimates in Propositions \ref{main.frame.proj}, \ref{proj.est.0} and \ref{proj.est}, i.e.~
\begin{equation}\label{EE.R3.1.3}
\begin{split}
&\left(\int_{T}^{t}\|\left((\tBox_g \Gamma^I h_{\mu\nu}){\bf E}^{\mu}{\bf E}^{\nu}-\tBox_g (\Gamma^I h_{\mu\nu}{\bf E}^{\mu}{\bf E}^{\nu})\right)w^{\f12}\|_{L^2_x(\mathcal R_{3,\tau})}\,d\tau\right)^{\f12}\\
\ls &T^{-\f 12}\max\{|U_2|^{\f 12},|U_3|^{\f 12+\f{\gamma}{4}}\}(\int_T^t\int_{\mathcal R_{3,\tau}} |\db\Gamma^I h|^2w'(q) dx\,d\tau)^{\f12}\\
&+T^{-1}\max\{|U_2|,|U_3|\}\sup_{\tau\in [T,t]} \left(\||\rd\Gamma^I h| w^{\f12}\|_{L^2_x(\mathcal R_{3,\tau})}+\|v(U_2) |\Gamma^I h|\|_{L^2(\mathbb S^2(U_2,\tau))}\right)\\
&+T^{-\f\gamma 2}\sup_{\tau\in [T,t]}\| |\rd\Gamma^J h| w^{\f 12}\|_{L^2_x(\mathcal R_{3,\tau})}\\
\ls &T^{-\f\gamma 2}\sum_{|J|\leq |I|}\left(\sup_{\tau\in[T,t]}\||\rd\Gamma^J h| w^{\f12}\|_{L^2_x(\mathcal R_{3,\tau})}+\left(\int_T^t \int_{ \mathcal R_{3,\tau}} |\bar{\rd}\Gamma^J h|^2 w'(r-\tau) dx\, d\tau\right)^{\f12}\right),
\end{split}
\end{equation}
where in the last line we have chosen $T'_3$ to be sufficiently large depending on $|U_2|$ and $|U_3|$ so that for $T> T_3'$, it holds that $T^{-\f 12}\max\{|U_2|^{\f 12},|U_3|^{\f 12+\f{\gamma}{4}}\}+T^{-1}\max\{|U_2|,|U_3|\}\leq T^{-\f{\gamma}{2}}$. 

Next, we control the scalar field. Since there are no bad terms in the equation for $\tBox_g \Gamma^I \beta$, we apply Propositions \ref{EE.3}, \ref{Cauchy.stability}, \ref{EE.R2} and \ref{SF.EE.main} to obtain the following bound for $|I|\leq k$:
\begin{equation}\label{EE.R3.1.4}
\begin{split}
&\sup_{\tau\in[T,t]}\| |\rd\Gamma^I \beta| w^{\f12} \|_{L^2_x(\mathcal R_{3,\tau})}+\left(\int_T^t \int_{ \mathcal R_{3,\tau}} |\bar{\rd}\Gamma^I \beta|^2 w'(r-\tau) dx\, d\tau\right)^{\f12}\\
&+\left(\int_{\mathcal B_{U_3}\cap\{T\leq \tau\leq t\}}(|\db\Gamma^I \beta|^2+\f{|\rd\Gamma^I \beta|^2}{(1+t)^{\f{\gamma}{4}+1}})w(r-\tau) \,dx\right)^{\f12}\\
\ls &C_T\,\ep(1+t)^{(2k+3)\de_0}+(1+t)^{\de_0}\sum_{|J|\leq |I|-1} \sup_{\tau\in [T,t]}\|\left(|\rd\Gamma^J h|+|\rd\Gamma^J \beta|\right) w^{\f12}\|_{L^2_x(\mathcal R_{3,\tau})}\\
&+T^{-\f\gamma 2}\sum_{|J|\leq |I|}\sup_{\tau\in[T,t]}\|\left(|\rd\Gamma^J h|+|\rd\Gamma^J \beta|\right) w^{\f12}\|_{L^2_x(\mathcal R_{3,\tau})}\\
&+T^{-\f\gamma 2}\sum_{|J|\leq |I|}\left(\int_T^t \int_{ \mathcal R_{3,\tau}} \left(|\bar{\rd}\Gamma^J h|^2+|\bar{\rd}\Gamma^J \beta|^2\right) w'(r-\tau) dx\, d\tau\right)^{\f12}.
\end{split}
\end{equation}
Combining \eqref{EE.R3.1.1}, \eqref{EE.R3.1.2}, \eqref{EE.R3.1.3} and \eqref{EE.R3.1.4} and summing over $|I|\leq k$, we get
\begin{equation*}
\begin{split}
&\sum_{|I|\leq k}\sup_{\tau\in[T,t]}\| \left(|\rd\Gamma^I h|_{\mathcal T\mathcal U}+|\rd\Gamma^I \beta|\right) w^{\f12} \|_{L^2_x(\mathcal R_{3,\tau})}+\sum_{|I|\leq k}\left(\int_T^t \int_{ \mathcal R_{3,\tau}} |\bar{\rd}\Gamma^I \beta|^2 w'(r-\tau) dx\, d\tau\right)^{\f12}\\
&+\sum_{|I|\leq k}\left(\int_{\mathcal B_{U_3}\cap\{T\leq \tau\leq t\}}(|\db\Gamma^I \beta|^2+\f{|\rd\Gamma^I \beta|^2}{(1+t)^{\f{\gamma}{4}+1}})w(r-\tau) \,dx\right)^{\f12}\\
\ls &C_T\,\ep(1+t)^{(2k+3)\de_0}+(1+t)^{\de_0}\sum_{|J|\leq k-1} \sup_{\tau\in [T,t]}\|\left(|\rd\Gamma^J h|+|\rd\Gamma^J \beta|\right) w^{\f12}\|_{L^2_x(\mathcal R_{3,\tau})}\\
&+T^{-\f\gamma 2}\sum_{|J|\leq k}\sup_{\tau\in[T,t]}\|\left(|\rd\Gamma^J h|+\underbrace{|\rd\Gamma^J \beta|}_{=:I}\right) w^{\f12}\|_{L^2_x(\mathcal R_{3,\tau})}\\
&+T^{-\f\gamma 2}\sum_{|J|\leq k}\left(\int_T^t \int_{ \mathcal R_{3,\tau}} \left(|\bar{\rd}\Gamma^J h|^2+\underbrace{|\bar{\rd}\Gamma^J \beta|^2}_{=:II}\right) w'(r-\tau) dx\, d\tau\right)^{\f12}.
\end{split}
\end{equation*}
Choosing $T_3'> T_2$ to be sufficiently large, we can then absorb the terms $I$ and $II$ to the left hand side for $T> T'_3$ and obtain the desired conclusion.
\end{proof}

In region $\mathcal R_3$, we can only use the naive estimate to control the $\mathfrak B_I$ term. Nevertheless, we make use of the fact that the estimates for $\mathfrak B_I$ depends only on $|\rd\Gamma^I h|_{\mathcal T\mathcal U}$ but not general components $|\rd\Gamma^I h|$. We then apply the estimates for $|\rd\Gamma^I h|_{\mathcal T\mathcal U}$ from Proposition \ref{EE.R3.1} which give us a smallness constant to close the estimates.
\begin{proposition}\label{EE.R3}
For every $U_3>U_2$, there exists $\ep_3\in (0,\ep_2]$, $T_3> T_2$ such that if $\ep\leq \ep_3$, $t>T> T_3$ and $k\leq N$, we have
\begin{equation}\label{EE.R3.main}
\begin{split}
\sum_{|I|\leq k}\sup_{\tau\in[T,t]}\| \left(|\rd\Gamma^I h|+|\rd\Gamma^I \beta|\right) w^{\f12} \|_{L^2_x(\mathcal R_{3,\tau})}\ls C_T\,\ep(1+t)^{(2k+4)\de_0}.
\end{split}
\end{equation}
Moreover, on the boundary $\mathcal B_{U_3}$, the following estimates are verified:
\begin{equation}\label{III.bdry.est.1}
\begin{split}
\sum_{|I|\leq k}\left(\int_{\mathcal B_{U_3}\cap\{T\leq \tau\leq t\}}(|\db\Gamma^I h|^2+|\db\Gamma^I \beta|^2+\f{|\rd\Gamma^I h|^2+|\rd\Gamma^I \beta|^2}{(1+t)^{\f{\gamma}{4}+1}})w(r-\tau) \,dx\right)^{\f12}\ls C_T\,\ep(1+t)^{(2k+4)\de_0}
\end{split}
\end{equation}
and
\begin{equation}\label{III.bdry.est.2}
\sum_{|I|\leq k}\|v(U_3)|\Gamma^I h|\|_{L^2(\mathbb S^2(U_3,t))}\ls C_T\,\ep(1+t)^{(2k+4)\de_0}.
\end{equation}
\end{proposition}
\begin{proof}
Using the bootstrap assumptions \eqref{BA3} and \eqref{BASF3}, we can bound $\mathfrak B_I$ as follows:
\begin{equation*}
\begin{split}
\int_T^t\|\mathfrak B_I w^{\f12}\|_{L^2_x(\mathcal R_{3,\tau})}d\tau\ls \int_T^t\f{\|\left(|\rd\Gamma^I h|_{\mathcal T\mathcal U}+|\rd\Gamma^I \beta|\right) w^{\f12}\|_{L^2_x(\mathcal R_{3,\tau})}}{1+\tau}d\tau.
\end{split}
\end{equation*}
Applying Proposition \ref{EE.R3.1}, we then have, for $|I|\leq k$, 
\begin{equation}\label{EE.R3.2.B}
\begin{split}
&\int_T^t\|\mathfrak B_I w^{\f12}\|_{L^2_x(\mathcal R_{3,\tau})}d\tau\\
\ls &C_T\,\ep(1+t)^{(2k+3)\de_0}+(1+t)^{\de_0}\sum_{|J|\leq k-1} \sup_{\tau\in [T,t]}\|\left(|\rd\Gamma^J h|+|\rd\Gamma^J \beta|\right) w^{\f12}\|_{L^2_x(\mathcal R_{3,\tau})}\\
&+T^{-\f\gamma 2}\sum_{|J|\leq k}\int_T^t \f{\left(\sup_{\tau\in[T,t']}\||\rd\Gamma^J h| w^{\f12}\|_{L^2_x(\mathcal R_{3,\tau})}+\left(\int_T^{t'} \int_{ \mathcal R_{3,\tau}} |\bar{\rd}\Gamma^J h|^2 w'(r-\tau) dx\, d\tau\right)^{\f12}\right)}{1+t'}\, dt'.
\end{split}
\end{equation}
By Proposition \ref{EE.3} and the bounds in Proposition \ref{Cauchy.stability} and \eqref{II.bdry.est.1} in Proposition \ref{EE.R2}, we have 
\begin{equation*}
\begin{split}
&\sum_{|I|\leq k}\left(\sup_{\tau\in[T,t]}\||\rd\Gamma^I h| w^{\f12} \|_{L^2_x(\mathcal R_{3,\tau})}+\left(\int_T^t \int_{ \mathcal R_{3,\tau}} |\bar{\rd}\Gamma^J h|^2 w'(r-\tau) dx\, d\tau\right)^{\f12}\right)\\
&+\sum_{|I|\leq k}\left(\int_{\mathcal B_{U_3}\cap\{T\leq \tau\leq t\}}\left(|\db\Gamma^I h|^2+\f{|\rd\Gamma^I h|^2}{(1+t)^{\f{\gamma}{4}+1}}\right) w(r-\tau) \,dx\right)^{\f12}\\
\ls &C_T\,\ep(1+t)^{(2k+3)\de_0}+\sum_{|I|\leq k}(\int_{T}^{t}\||\tBox_g(\Gamma^I h)|w^{\f12}\|_{L^2_x(\mathcal R_{3,\tau})}\,d\tau)^{\f12}.
\end{split}
\end{equation*}
Now $|\tBox_g(\Gamma^I h)|$ contains the term $\mathfrak B_I$ together with the terms in Proposition \ref{EE.main}. We bound the $\mathfrak B_I$ term by combining \eqref{EE.R3.2.B} and the the bounds from Proposition \ref{EE.R3.1}. The remaining terms can be estimated using Proposition \ref{EE.main}. Therefore, we have
\begin{equation}\label{EE.R3.2.1}
\begin{split}
&\sum_{|I|\leq k}\left(\sup_{\tau\in[T,t]}\||\rd\Gamma^I h| w^{\f12} \|_{L^2_x(\mathcal R_{3,\tau})}+\left(\int_T^t \int_{ \mathcal R_{3,\tau}} |\bar{\rd}\Gamma^J h|^2 w'(r-\tau) dx\, d\tau\right)^{\f12}\right)\\
&+\sum_{|I|\leq k}\left(\int_{\mathcal B_{U_3}\cap\{T\leq \tau\leq t\}}\left(|\db\Gamma^I h|^2+\f{|\rd\Gamma^I h|^2}{(1+t)^{\f{\gamma}{4}+1}}\right) w(r-\tau) \,dx\right)^{\f12}\\
\ls &C_T\,\ep(1+t)^{(2k+3)\de_0}+T^{-\f\gamma 2}\sum_{|J|\leq k}\left(\sup_{\tau\in[T,t]}\||\rd\Gamma^J h| w^{\f12}\|_{L^2_x(\mathcal R_{3,\tau})}+(\int_T^t \int_{ \mathcal R_{3,\tau}} |\bar{\rd}\Gamma^J h|^2 w'(r-\tau) dx\, d\tau)^{\f12}\right)\\
&+(1+t)^{\de_0}\sum_{|J|\leq k-1} \sup_{\tau\in [T,t]}\|\left(|\rd\Gamma^J h|+|\rd\Gamma^J \beta|\right) w^{\f12}\|_{L^2_x(\mathcal R_{3,\tau})}\\
&+T^{-\f\gamma 2}\sum_{|J|\leq k}\int_T^t \f{\left(\sup_{\tau\in[T,t']}\||\rd\Gamma^J h| w^{\f12}\|_{L^2_x(\mathcal R_{3,\tau})}+\left(\int_T^{t'} \int_{ \mathcal R_{3,\tau}} |\bar{\rd}\Gamma^J h|^2 w'(r-\tau) dx\, d\tau\right)^{\f12}\right)}{1+t'}\, dt'\\
\ls &C_T\,\ep(1+t)^{(2k+3)\de_0}+(1+t)^{\de_0}\sum_{|J|\leq k-1} \sup_{\tau\in [T,t]}\|\left(|\rd\Gamma^J h|+|\rd\Gamma^J \beta|\right) w^{\f12}\|_{L^2_x(\mathcal R_{3,\tau})}\\
&+T^{-\f\gamma 2}\sum_{|J|\leq k}\int_T^t \f{\left(\sup_{\tau\in[T,t']}\||\rd\Gamma^J h| w^{\f12}\|_{L^2_x(\mathcal R_{3,\tau})}+\left(\int_T^{t'} \int_{ \mathcal R_{3,\tau}} |\bar{\rd}\Gamma^J h|^2 w'(r-\tau) dx\, d\tau\right)^{\f12}\right)}{1+t'}\, dt',
\end{split}
\end{equation}
where the last step is achieved by choosing $T_3>T_2$ to be sufficiently large so that for $T>T_3$, we can absorb the terms to the left hand side. Applying Gr\"onwall's inequality to \eqref{EE.R3.2.1} gives\footnote{Here, and below, $C_3>0$ is some constant (which can be different from line to line) depending on $C$, $N$, $\gamma$ and $\de_0$.}
\begin{equation}\label{EE.R3.2.2}
\begin{split}
&\sum_{|I|\leq k}\left(\sup_{\tau\in[T,t]}\||\rd\Gamma^I h| w^{\f12} \|_{L^2_x(\mathcal R_{3,\tau})}+\left(\int_T^t \int_{ \mathcal R_{3,\tau}} |\bar{\rd}\Gamma^J h|^2 w'(r-\tau) dx\, d\tau\right)^{\f12}\right)\\
&+\sum_{|I|\leq k}\left(\int_{\mathcal B_{U_3}\cap\{T\leq \tau\leq t\}}\left(|\db\Gamma^I h|^2+\f{|\rd\Gamma^I h|^2}{(1+t)^{\f{\gamma}{4}+1}}\right) w(r-\tau) \,dx\right)^{\f12}\\
\ls &\left(C_T\,\ep(1+t)^{(2k+3)\de_0}+(1+t)^{\de_0}\sum_{|J|\leq k-1} \sup_{\tau\in [T,t]}\|\left(|\rd\Gamma^J h|+|\rd\Gamma^J \beta|\right) w^{\f12}\|_{L^2_x(\mathcal R_{3,\tau})}\right)\\
&\qquad\times\exp(C_3T^{-\f{\gamma}{2}} \log (2+t))\\
\ls &C_T\,\ep(1+t)^{(2k+4)\de_0}+(1+t)^{2\de_0}\sum_{|J|\leq k-1} \sup_{\tau\in [T,t]}\|\left(|\rd\Gamma^J h|+|\rd\Gamma^J \beta|\right) w^{\f12}\|_{L^2_x(\mathcal R_{3,\tau})},
\end{split}
\end{equation}
for $T_3>T_2$ chosen to be sufficiently large.
Combining \eqref{EE.R3.2.2} with Proposition \ref{EE.R3.1} gives
\begin{equation*}
\begin{split}
&\sum_{|I|\leq k}\sup_{\tau\in[T,t]}\|\left(|\rd\Gamma^I h|+|\rd\Gamma^I \beta|\right) w^{\f12} \|_{L^2_x(\mathcal R_{3,\tau})}\\
&+\sum_{|I|\leq k}\left(\int_{\mathcal B_{U_3}\cap\{T\leq \tau\leq t\}}\left(|\db\Gamma^I h|^2+\f{|\rd\Gamma^I h|^2}{(1+t)^{\f{\gamma}{4}+1}}\right) w(r-\tau) \,dx\right)^{\f12}\\
\ls &C_T\,\ep(1+t)^{(2k+4)\de_0}+(1+t)^{2\de_0}\sum_{|J|\leq k-1} \sup_{\tau\in [T,t]}\|\left(|\rd\Gamma^J h|+|\rd\Gamma^J \beta|\right) w^{\f12}\|_{L^2_x(\mathcal R_{3,\tau})},
\end{split}
\end{equation*}
A simple induction in $k$ as\footnote{In fact, the induction here is slightly simpler as there are no ``borderline'' terms for which we need to apply the Gr\"onwall's inequality.} in the proof of Proposition \ref{EE.R2} then allows us to conclude that
\begin{equation}\label{EE.R3.2.3}
\begin{split}
&\sum_{|I|\leq k}\sup_{\tau\in[T,t]}\|\left(|\rd\Gamma^I h|+|\rd\Gamma^I \beta|\right) w^{\f12} \|_{L^2_x(\mathcal R_{3,\tau})}\\
&+\sum_{|I|\leq k}\left(\int_{\mathcal B_{U_3}\cap\{T\leq \tau\leq t\}}\left(|\db\Gamma^I h|^2+\f{|\rd\Gamma^I h|^2}{(1+t)^{\f{\gamma}{4}+1}}\right) w(r-\tau) \,dx\right)^{\f12}
\ls C_T\,\ep(1+t)^{(2k+4)\de_0}.
\end{split}
\end{equation}
This proves \eqref{EE.R3.main} and the estimates for $\sum_{|I|\leq k}|\rd\Gamma^I h|$ in \eqref{III.bdry.est.1}. To obtain the estimates for $\sum_{|I|\leq k}|\rd\Gamma^I \beta|$ in \eqref{III.bdry.est.1}, we combine the estimates from Proposition \ref{EE.R3.1}, \eqref{EE.R3.2.2} and \eqref{EE.R3.2.3}. Finally, \eqref{III.bdry.est.2} follows from the estimate \eqref{EE.R3.2.3} together with the Hardy inequality in Proposition \ref{Hardy} with $\alp=\mu_1=0$, $\mu_2=2\gamma$, $R_1=t-\f{1}{(1+t)^{\f{\gamma}{4}}}-U_3$ and $R_2=t-\f{1}{(1+t)^{\f{\gamma}{4}}}-U_2$.
\end{proof}

This concludes the estimates in the region $\mathcal R_3$. Notice that the parameter $U_3$ that is used to define the region $\mathcal R_3$ is not yet chosen. This will be chosen in Section \ref{sec.IV}. In particular, it is important that Proposition \ref{EE.R3} holds for \underline{every} $U_3>U_2$.

\section{Region near timelike infinity}\label{sec.IV}

In this section, we prove the energy estimates for the region near timelike infinity, i.e.~the region $\mathcal R_4$ in Figure~\ref{fig:regions}. This will then conclude the proof of the energy estimates in all regions of the spacetime. In fact, the region $\mathcal R_4$ is treated in a manner analogous to the region $\mathcal R_2$ (see Section \ref{sec.II}), except that we now have to choose $U_3$ to be large and positive instead of negative and we pick up extra boundary terms from region $\mathcal R_3$. 

As in Section \ref{sec.II}, we  begin with the estimate for $\mathfrak B_I$. 
\begin{proposition}\label{B.EE.IV}
In the region $\mathcal R_4$, for $U_3>0$, the following bounds for $\mathfrak B_I$ hold for $|I|\leq N$ and for $t>T$:
\begin{equation*}
\begin{split}
\int_T^t \| \mathfrak B_I w^{\f12}\|_{L^2_x(\mathcal R_{4,\tau})} d\tau
\ls &\max\{\ep^{\f12},\f{1}{(1+|U_3|)^{\gamma}}\}\sum_{|J|= |I|}\int_T^t \f{\|\left(|\rd\Gamma^J h|+|\rd\Gamma^J \beta|\right) w^{\f12}\|_{L^2_x(\mathcal R_{4,\tau})}}{1+\tau} d\tau.
\end{split}
\end{equation*}
\end{proposition}
\begin{proof}
The proof is completely analogous to Proposition \ref{B.EE.II} except that we now use that we are in the region $\mathcal R_4$ and obtain smallness using the parameter $U_3$.
\end{proof}
We then prove the energy estimates in the region $\mathcal R_4$ in an analogous manner as Proposition \ref{EE.R2}:
\begin{proposition}\label{EE.R4}
There exists $U_3>0$ sufficiently large, $T_4>T_3$ and $\ep_4\in (0,\ep_3]$ such that
\begin{equation*}
\begin{split}
\sum_{|I|\leq k}\sup_{\tau\in[T,t]} \|\left(|\rd\Gamma^I h|+|\rd\Gamma^I \beta|\right) w^{\f12} \|_{L^2_x(\mathcal R_{4,\tau})}\ls C_T\ep(1+t)^{(2k+6)\de_0}
\end{split}
\end{equation*}
for $t> T> T_4$, $\ep\in (0,\ep_4]$ and $k\leq N$.
\end{proposition}
\begin{proof}
We apply the energy estimates in Proposition \ref{EE.4}. Using Proposition \ref{Cauchy.stability} and the bound \eqref{III.bdry.est.1}, the first two terms on the right hand side of Proposition \ref{EE.4} are bounded by $C_T\ep(1+t)^{(2k+4)\de_0}$. Therefore, we have
\begin{equation*}
\begin{split}
&\sum_{|I|\leq k}\sup_{\tau\in[T,t]}\|\left(|\rd\Gamma^I h|+|\rd\Gamma^I \beta|\right) w^{\f12} \|_{L^2_x(\mathcal R_{4,\tau})}\\
&+\sum_{|I|\leq k}\left(\int_T^t \int_{ \mathcal R_{4,\tau}} \left(|\bar{\rd}\Gamma^J h|^2+|\bar{\rd}\Gamma^J \beta|^2\right) w'(r-\tau) dx\, d\tau\right)^{\f12}\\
\ls &C_T\,\ep(1+t)^{(2k+4)\de_0}+\sum_{|I|\leq k}\left(\int_{T}^{t}\||\tBox_g(\Gamma^I h)|w^{\f12}\|_{L^2_x(\mathcal R_{4,\tau})}\,d\tau\right)^{\f12}.
\end{split}
\end{equation*}
To control $|\tBox_g(\Gamma^I h)|$, we use Proposition \ref{B.EE.IV} to bound the $\mathfrak B_I$ term and use Proposition \ref{EE.main} to estimate the remaining terms. More precisely, we have
\begin{equation}\label{EE.IV.1}
\begin{split}
&\sum_{|I|\leq k}\sup_{\tau\in[T,t]}\|\left(|\rd\Gamma^I h|+|\rd\Gamma^I \beta|\right) w^{\f12} \|_{L^2_x(\mathcal R_{4,\tau})}\\
&+\sum_{|I|\leq k}\left(\int_T^t \int_{ \mathcal R_{4,\tau}} \left(|\bar{\rd}\Gamma^J h|^2+|\bar{\rd}\Gamma^J \beta|^2\right) w'(r-\tau) dx\, d\tau\right)^{\f12}\\
\ls &C_T\,\ep(1+t)^{(2k+5)\de_0}+(1+t)^{\de_0}\sum_{|J|\leq k-1} \sup_{\tau\in [T,t]}\|\left(|\rd\Gamma^J h|+|\rd\Gamma^J \beta|\right) w^{\f12}\|_{L^2_x(\mathcal R_{3,\tau})}\\
&+\underbrace{T^{-\f\gamma 2}\sum_{|J|\leq k}\sup_{\tau\in[T,t]}\|\left(|\rd\Gamma^J h|+|\rd\Gamma^J \beta|\right) w^{\f12}\|_{L^2_x(\mathcal R_{3,\tau})}}_{=:I}\\
&+\underbrace{T^{-\f\gamma 2}\sum_{|J|\leq k}\left(\int_T^t \int_{ \mathcal R_{3,\tau}} \left(|\bar{\rd}\Gamma^J h|^2+|\bar{\rd}\Gamma^J \beta|^2\right) w'(r-\tau) dx\, d\tau\right)^{\f12}}_{=:II}\\
&+\max\{\ep^{\f12},\f{1}{(1+|U_3|)^{\gamma}}\}\sum_{|J|\leq k}\int_T^t \f{\|\left(|\rd\Gamma^J h|+|\rd\Gamma^J \beta|\right) w^{\f12}\|_{L^2_x(\mathcal R_{4,\tau})}}{1+\tau} d\tau\\
\ls &C_T\,\ep(1+t)^{(2k+5)\de_0}+(1+t)^{\de_0}\sum_{|J|\leq k-1} \sup_{\tau\in [T,t]}\|\left(|\rd\Gamma^J h|+|\rd\Gamma^J \beta|\right) w^{\f12}\|_{L^2_x(\mathcal R_{3,\tau})}\\
&+\max\{\ep^{\f12},\f{1}{(1+|U_3|)^{\gamma}}\}\sum_{|J|\leq k}\int_T^t \f{\|\left(|\rd\Gamma^J h|+|\rd\Gamma^J \beta|\right) w^{\f12}\|_{L^2_x(\mathcal R_{4,\tau})}}{1+\tau} d\tau,
\end{split}
\end{equation}
where in the last line we have used that we can choose $T_4'$ to be sufficiently large such that whenever $T> T_4'$, the terms $I$ and $II$ can be absorbed to the left hand side. We now proceed to an induction argument in $k$ to prove the proposition. First, for $k=0$, \eqref{EE.IV.1} gives\footnote{Here, and below, $C_4>0$ is some constant (which can be different from line to line) depending on $C$, $N$, $\gamma$ and $\de_0$.}
\begin{equation*}
\begin{split}
&\sup_{\tau\in[T,t]}\|\left(|\rd h|+|\rd \beta|\right) w^{\f12}\|_{L^2_x(\mathcal R_{4,\tau})}\\
\ls &C_T\,\ep(1+t)^{5\de_0}+\max\{\ep^{\f12},\f{1}{(1+|U_3|)^{\gamma}}\}\int_T^t \f{\|\left(|\rd h|+|\rd \beta|\right) w^{\f12}\|_{L^2_x(\mathcal R_{4,\tau})}}{1+\tau} d\tau\\
\ls &C_T\,\ep(1+t)^{5\de_0}\exp\left(C_4\max\{\ep^{\f12},\f{1}{(1+|U_3|)^{\gamma}}\}\int_T^t \f{d\tau}{1+\tau}\right) \ls \ep(1+t)^{6\de_0},
\end{split}
\end{equation*}
using Gr\"onwall's inequality, as long as $U_3$ is sufficiently large and $\ep_4$ is sufficiently small. Now assume that for some $k_0\geq 1$, we have
\begin{equation}\label{EE.IV.induction}
\sum_{|I|\leq k_0-1}\sup_{\tau\in[T,t]}\|\left(|\rd\Gamma^I h|+|\rd\Gamma^I \beta|\right) w^{\f12} \|_{L^2_x(\mathcal R_{4,\tau})}\ls \ep(1+t)^{(2(k_0-1)+6)\de_0}.
\end{equation}
Then by \eqref{EE.IV.1} and \eqref{EE.IV.induction}, we have
\begin{equation*}
\begin{split}
&\sum_{|I|\leq k_0}\sup_{\tau\in[T,t]}\|\left(|\rd\Gamma^I h|+|\rd\Gamma^I \beta|\right) w^{\f12} \|_{L^2_x(\mathcal R_{4,\tau})}\\
\ls &C_T\,\ep(1+t)^{(2k_0+5)\de_0}+\max\{\ep^{\f12},\f{1}{(1+|U_3|)^{\gamma}}\}\sum_{|J|\leq k_0}\int_T^t \f{\|\left(|\rd\Gamma^J h|+|\rd\Gamma^J \beta|\right) w^{\f12}\|_{L^2_x(\mathcal R_{4,\tau})}}{1+\tau} d\tau.
\end{split}
\end{equation*}
For $U_3$ sufficiently large and $\ep_4$ sufficiently small, Gr\"onwall's inequality implies
$$\sum_{|I|\leq k_0}\sup_{\tau\in[T,t]}\|\left(|\rd\Gamma^I h|+|\rd\Gamma^I \beta|\right) w^{\f12} \|_{L^2_x(\mathcal R_{4,\tau})}\ls C_T\,\ep(1+t)^{(2k_0+6)\de_0}.$$
This concludes the induction step. Once we have fixed $U_3$ we then choose $T_4> \max\{T_4', T_3\}$ sufficiently large. This then concludes the proof of the proposition. 
\end{proof}
This also concludes the proof of energy estimates in all regions of the spacetime. {\bf At this point, we fix $U_3$ and $T>T_4$ according to Proposition \ref{EE.R4}. Since $T$ is fixed, from now on, we allow the implicit constant in $\ls$ to depend on $T$.}

We end this section by summarizing the energy estimates that have been proven:
\begin{theorem} \label{thm:energy.bound.combined}
For $\ep_4>0$ as in Proposition \ref{EE.R4}, there following estimate holds\footnote{We are using the convention that we just introduced in the previous paragraph: We now drop the constant $C_T$ in the estimate and allow the implicit constant in $\ls$ to depend on $T$.} for all $t\geq 0$ and for all $\ep\in (0,\ep_4]$:
\begin{equation}\label{energy.bound.combined}
\sum_{|I|\leq N}\left(\sup_{\tau\in[0,t]}\int_{\Sigma_\tau} \left(|\rd\Gamma^I h|^2+|\rd\Gamma^I \beta|^2\right)(\tau,x)\, w(q) \,dx\right)^{\f 12}\ls \ep(1+t)^{(2N+6)\de_0},
\end{equation}
\end{theorem}
\begin{proof}
Combine the estimates in Propositions \ref{Cauchy.stability}, \ref{EE.R2}, \ref{EE.R3} and \ref{EE.R4}.
\end{proof}

\section{Recovering the bootstrap assumptions}\label{sec.recover.bootstrap}

Our main goal in this section is to show that the energy bound \eqref{energy.bound.combined} implies decay estimates that in particular improve the bootstrap assumptions \eqref{BA1}-\eqref{BA5} and \eqref{BASF1}-\eqref{BASF3} (see Proposition \ref{BS.close}).

As a preliminary step, we need a lemma which allows us to control in a pointwise fashion any function by its derivative via integrating along constant\footnote{Recall here that $\om:=(\theta,\varphi)$ is the standard spherical coordinates.} $(s,\om)$ curves. We will repeatedly use this lemma below.
\begin{lemma}\label{int.lemma}
Let $f(q)$ be a positive function such that\footnote{Here, we have used the notation that $A\sim B$ is $A\leq CB$ and $B\leq CA$ for some constant $C>0$.} $f(q)\sim (1+|q|)^{\beta}$ for $\beta>1$ if $q\geq 0$ and $f(q)\sim (1+|q|)^\sigma$ for $\sigma < 1$ if $q<0$. Also let $\alp>0$ and $k(s)$ be a positive function of $s$ such that $|k(s)|\leq (1+s)^{\alp}$. Then for every sufficiently regular scalar function $\xi:[0,\infty)\times\mathbb R^3\to \mathbb R$ and for $t\geq 0$, we have
\begin{equation*}
\begin{split}
&\sup_{x}k(t+r)(1+|r-t|)^{-1}f(r-t)|\xi|(t,x)\\
\ls &\sup_x (1+r)^{-1+\alp+\bt}|\xi|(0,x)+\sup_{\substack{\tau\in [0, t]\\ x\in\mathbb R^3}}k(\tau+r) f(r-\tau)|\rd_q \xi|(\tau,x).
\end{split}
\end{equation*}
\end{lemma}
\begin{proof}
Since $\beta>1$, $f(q)^{-1}$ is integrable for $q\geq 0$ and moreover 
\begin{equation}\label{int.lemma.1}
\int_q^{\infty} f(q')^{-1}\,dq'\sim f(q)^{-1}(1+|q|).
\end{equation}
The lemma then follows from integrating $\rd_q \xi$ along curves with constant $(s,\om)$ in the $-\rd_q$ direction, i.e.~for every fixed $(t,x)$, we have
\begin{equation*}
\begin{split}
&k(t+|x|)|\xi|(t,x)\\
\leq &k(t+|x|)|\xi|(0,x+t\tfrac{x}{|x|}) +k(t+|x|)\int_{|x|-t}^{\infty}\sup_{\substack{\tau\in [0, t]\\ x'\in \{x'\in \mathbb R^3:\tau+|x'|=t+|x|,\,|x'|-\tau=q'\}}}|\rd_q \xi|(\tau,x') \, dq' \\
\leq &k(t+|x|)|\xi|(0,x+t\tfrac{x}{|x|}) +\sup_{\substack{\tau\in [0, t]\\ x'\in \{x'\in \mathbb R^3:\tau+|x'|=t+|x|\}}}k(\tau+|x'|) f(|x'|-\tau)|\rd_q \xi|(\tau,x')\int_{r-t}^{\infty}f(q')^{-1} \, dq'
\end{split}
\end{equation*}
and using \eqref{int.lemma.1}.
\end{proof}

We now begin the proof of the decay estimates. First, as an immediate consequence of \eqref{energy.bound.combined} and the Klainerman-Sobolev inequality (Proposition \ref{KS.ineq}), we have:
\begin{proposition}\label{KS.cons}
The energy bound \eqref{energy.bound.combined} implies that
\begin{equation}\label{KS.cons.1}
\sum_{|I|\leq N-3}(1+s)^{1-(2N+6) \de_0}(1+|q|)^{\f12}w(q)^{\f12}(|\rd\Gamma^I h|+|\rd\Gamma^I \beta|)(t,x)\ls \ep, 
\end{equation}
\begin{equation}\label{KS.cons.2}
\sum_{|I|\leq N-3}(1+s)^{1-(2N+6) \de_0}(1+|q|)^{-\f12}w(q)^{\f12}(|\Gamma^I h|+|\Gamma^I \beta|)(t,x)\ls \ep,
\end{equation}
\begin{equation}\label{KS.cons.3}
\sum_{|I|\leq N-4}(1+s)^{2-(2N+6) \de_0}(1+|q|)^{-\f12}w(q)^{\f12}(|\bar\rd\Gamma^I h|+|\bar\rd\Gamma^I \beta|)(t,x)\ls \ep.
\end{equation}
\end{proposition}
\begin{proof}
\eqref{KS.cons.1} is a direct consequence of Proposition \ref{KS.ineq} and \eqref{energy.bound.combined}. \eqref{KS.cons.2} follows from applying Lemma \ref{int.lemma} with $k(s)=(1+s)^{1-(2N+6)\de_0}$ and $f(q)=(1+|q|)^{\f12}w(q)^{\f12}$. Notice in particular that the term on $\{t=0\}$ can be controlled thanks to Remark \ref{rmk.initial.pointwise}. Finally, \eqref{KS.cons.3} is a direct consequence of \eqref{KS.cons.2} and Proposition \ref{decay.weights}.
\end{proof}

Observe that the decay estimates in Proposition \ref{KS.cons} alone are insufficient to recover the bootstrap assumptions \eqref{BA1}, \eqref{BA2}, \eqref{BA4}, \eqref{BASF1}, \eqref{BASF2} and \eqref{BASF3}. We therefore need to combine them with Propositions \ref{decay.est} and \ref{decay.est.2} to prove stronger decay estimates.

Recall from Proposition \ref{decay.est} that $\varpi(q)$ is a weight function defined by\footnote{Recall also the definition of $w(q)$ in Definition \ref{w.def}.}
\begin{equation}\label{varpi.def}
\varpi(q):=(1+|q|)^{\f 12-\f{\gamma}{4}}w(q)^{\f12}.
\end{equation}
Define 
\begin{equation}\label{pi.def}
\pi_k(t):=\sum_{|I|\leq k}\sup_{\substack{\tau\in [0, t]\\x\in \mathbb R^3}} (1+\tau+r)\varpi(r-\tau)\left(|\rd\Gamma^I h|+|\rd\Gamma^I \beta|\right)(\tau,x).
\end{equation}
We will also use the notation that $\pi_{-k}=0$ if $-k<0$. Define also the notation $\sigma_k(t)$ in a similar way as $\pi_k(t)$, but keeps track only of the $\mathcal T\mathcal U$ components of $\Gamma^I h$ and the derivative of the scalar field:
\begin{equation}\label{sigma.def}
\sigma_k(t):=\sum_{|I|\leq k}\sup_{\substack{\tau\in [0, t]\\x\in \mathbb R^3}} (1+\tau+r)\varpi(r-\tau)\left(|\rd \Gamma^I h|_{\mathcal T\mathcal U}+|\rd\Gamma^I \beta|\right)(\tau,x).
\end{equation}
We now proceed to control \eqref{pi.def} and \eqref{sigma.def}. We first need the following proposition, which combines Propositions \ref{decay.est} and \ref{decay.est.2} with the decay bounds we have obtained in Proposition \ref{KS.cons}.
\begin{proposition}\label{pointwise.main.lemma}
The following estimates hold for $|I|\leq \lfloor\f N2\rfloor +1$:
\begin{equation*}
\begin{split}
\sup_x(1+t)\varpi(q)|\rd \Gamma^I h(t,x)|
\ls &\ep+\int_0^t (1+\tau)\|\varpi(r-\tau)|\tBox_g (\Gamma^I h)|(\tau,\cdot)\|_{L^\infty(D_\tau)}d\tau
\end{split}
\end{equation*}
and
\begin{equation*}
\begin{split}
\sup_x(1+t)\varpi(q)|\rd \Gamma^I h(t,x)|_{\mathcal T\mathcal U}
\ls &\ep+\int_0^t (1+\tau)\|\varpi(r-\tau)|\tBox_g (\Gamma^I h)|_{\mathcal T\mathcal U}(\tau,\cdot)\|_{L^\infty(D_\tau)}d\tau
\end{split}
\end{equation*}
and
\begin{equation*}
\begin{split}
\sup_x(1+t)\varpi(q)|\rd \Gamma^I \beta(t,x)|
\ls &\ep+\int_0^t (1+\tau)\|\varpi(r-\tau)|\tBox_g (\Gamma^I \beta)|(\tau,\cdot)\|_{L^\infty(D_\tau)}d\tau.
\end{split}
\end{equation*}
\end{proposition}

\begin{proof}
By Propositions \ref{decay.est} and \ref{decay.est.2}, in order to show the desired estimates it suffices to show that
\begin{equation}\label{pointwise.main.lemma.1}
\sup_{0\leq \tau\leq \infty}\sum_{|J|\leq |I|+1}\left(\|\varpi(r-\tau)\Gamma^J h(\tau,\cdot)\|_{L^\infty}+\|\varpi(r-\tau)\Gamma^J \beta(\tau,\cdot)\|_{L^\infty}\right)\ls \ep
\end{equation}
and 
\begin{equation}\label{pointwise.main.lemma.2}
\sum_{|J|\leq |I|+2}\int_0^\infty (1+\tau)^{-1}\left(\|\varpi(r-\tau)\Gamma^J h(\tau,\cdot)\|_{L^\infty(D_\tau)}+\|\varpi(r-\tau)\Gamma^J \beta(\tau,\cdot)\|_{L^\infty(D_\tau)}\right)d\tau\ls \ep.
\end{equation}
Since $N\geq 11$, we have $\lfloor\f N2\rfloor +3\leq N-3$. Therefore, by \eqref{KS.cons.2} in Proposition \ref{KS.cons}, \eqref{pointwise.main.lemma.1} and \eqref{pointwise.main.lemma.2} both hold, since $\de_0$ satisfy \eqref{de_0.def}.
\end{proof}
Using Proposition \ref{pointwise.main.lemma}, we can obtain the required pointwise bounds by estimating $\tBox_g \Gamma^I h$ and $\tBox_g \Gamma^I \beta$. We now control the contributions from each of the terms in Propositions \ref{schematic.eqn} and \ref{SF.schematic.eqn}. 
The bounds for $\mathfrak I_I$, $\mathfrak T_I$, $\mathfrak L_I$, $\mathfrak W_I$ and $\mathfrak B_I$ (and their $^{(\phi)}$-counterparts) are relatively straightforward and we will begin with them, starting with the inhomogeneous terms $\mathfrak I_I$ and $\mathfrak I_I^{(\phi)}$:
\begin{proposition}\label{decay.I.est}
The following estimate holds for $\mathfrak I_I$ and $\mathfrak I_I^{(\phi)}$ for $|I|\leq \lfloor \f N2\rfloor +1$:
\begin{equation*}
\begin{split}
\int_0^t (1+\tau)\|\varpi(r-\tau)(\mathfrak I_I+\mathfrak I_I^{(\phi)})(\tau,\cdot)\|_{L^\infty(D_\tau)}d\tau\ls 
\begin{cases}
\ep & \mbox{if }|I|=0\\
 \ep\log^3(2+t) &\mbox{if }|I|\geq 1.
\end{cases}
\end{split}
\end{equation*}
\end{proposition}
\begin{proof}
For $|I|=0$ (recall the better bounds that we have for $\mathfrak I_0$ and $\mathfrak I_0^{(\phi)}$ in Propositions \ref{schematic.eqn} and \ref{SF.schematic.eqn}), we have
\begin{equation*}
\begin{split}
&\int_0^t (1+\tau)\|\varpi(r-\tau)(\mathfrak I_0+\mathfrak I_0^{(\phi)})(\tau,\cdot)\|_{L^\infty(D_\tau)}d\tau\\
\ls & \int_0^t (1+\tau)\sup_x\frac{\ep \varpi(r-\tau)\log^2(2+\tau)}{(1+\tau)^{2+\f{\gamma}{2}}(1+|r-\tau|)^{\f 12-\f\gamma 2-\de_0}w(r-\tau)^{\f12}}\, d\tau\\
\ls & \int_0^t \frac{\ep \log^2(2+\tau)}{(1+\tau)^{1+\f{\gamma}{2}-\f\gamma 4-\de_0}}\, d\tau\ls \ep, 
\end{split}
\end{equation*}
as long as $\de_0<\f{\gamma}{4}$.

For $|I|\geq 1$, we estimate the term $\mathfrak I_I$ as follows:
\begin{equation*}
\begin{split}
\int_0^t (1+\tau)\|\varpi(r-\tau)(\mathfrak I_I+\mathfrak I_I^{(\phi)})(\tau,\cdot)\|_{L^\infty(D_\tau)}d\tau
\ls &\int_0^t \f{\ep\log^2 \tau}{1+\tau} d\tau\ls \ep\log^3(2+t). 
\end{split}
\end{equation*}
\end{proof}
We next bound the top order term $\mathfrak T_I$ and its ${ }^{(\phi)}$-counterpart. Since this is straightforward, we omit the proof.
\begin{proposition}\label{decay.T.est}
The following estimate holds for $\mathfrak T_I$ and $\mathfrak T_I^{(\phi)}$ for $|I|\leq \lfloor \f N2\rfloor +1$:
\begin{equation*}
\begin{split}
\int_0^t (1+\tau)\|\varpi(r-\tau)(\mathfrak T_I+\mathfrak T_I^{(\phi)})(\tau,\cdot)\|_{L^\infty(D_\tau)}d\tau
\ls &\int_0^t \f{\pi_{|I|}(\tau)}{(1+\tau)^{1+\f{\gamma}{2}}} \,d\tau.
\end{split}
\end{equation*}
\end{proposition}
We now turn to the lower order term $\mathfrak L_I$ and its ${ }^{(\phi)}$-counterpart.
\begin{proposition}\label{decay.L.est}
The following estimate holds for $\mathfrak L_I$ and $\mathfrak L_I^{(\phi)}$ for $|I|\leq \lfloor \f N2\rfloor +1$: 
\begin{equation*}
\begin{split}
\int_0^t (1+\tau)\|\varpi(r-\tau)(\mathfrak L_I+\mathfrak L_I^{(\phi)})(\tau,\cdot)\|_{L^\infty(D_\tau)}d\tau
\ls &\eps \log(2+t)+ \int_0^t \f{\log(2+\tau)\pi_{|I|-1}(\tau)}{(1+\tau)} \,d\tau.
\end{split}
\end{equation*}
\end{proposition}
\begin{proof}
We will only deal with the term $\sum_{|J|\leq |I|-1}\f{|\Gamma^J h|}{(1+s)(1+|q|)^{1+\gamma}}$ for the other terms are trivial. For this we note that
\begin{equation*}
\begin{split}
&\sum_{|J|\leq |I|-1}\sup_{x\in D_\tau}(1+\tau)\varpi(r-\tau)\f{|\Gamma^J h|(\tau,x)}{(1+\tau+r)(1+|r-\tau|)^{1+\gamma}}\\
\ls &\sum_{|J|\leq |I|-1}\f{1}{(1+\tau)}\sup_{x\in D_\tau}(1+\tau+r)\f{\varpi(r-\tau)}{(1+|r-\tau|)}|\Gamma^J h|(\tau,x)\ls \frac{\eps}{1+\tau}+ \f{\pi_{|I|-1}(\tau)}{1+\tau} ,
\end{split}
\end{equation*}
where in the last line we have used Lemma \ref{int.lemma} with $k(s)=(1+s)$, $f(q)=\varpi(q)$, noting that the term on $\{t=0\}$ can be controlled thanks to Remark~\ref{rmk.initial.pointwise} and gives rise to the term $\frac{\eps}{1+\tau}$. After integrating in $\tau$ over $[0, t]$, the proposition follows.
\end{proof}

To control the $\mathfrak W_I$ term, we have
\begin{proposition}\label{decay.V.est}
The following estimate holds for $\mathfrak W_I$ and $\mathfrak W_I^{(\phi)}$ for $|I|\leq \lfloor \f N2\rfloor +1$: 
\begin{equation*}
\begin{split}
\int_0^t (1+\tau)\|\varpi(r-\tau)(\mathfrak W_I+\mathfrak W_I^{(\phi)})(\tau,\cdot)\|_{L^\infty(D_\tau)}d\tau
\ls &\eps + \int_0^t  \f{\log(2+\tau)\pi_{|I|}(\tau)}{(1+\tau)^{1+\gamma}}d\tau.
\end{split}
\end{equation*}
\end{proposition}
\begin{proof}
Using Lemma \ref{int.lemma} with $k(s)=(1+s)$, $f(q)=\varpi(q)$ (and Remark \ref{rmk.initial.pointwise} to bound the term on $\{t=0\}$), we have
\begin{equation*}
\begin{split}
&\sup_{x\in D_\tau}(1+\tau)\varpi(r-\tau)(\mathfrak W_I+\mathfrak W_I^{(\phi)})(\tau,x)\\
\ls &\sum_{|J|\leq |I|}\sup_{x\in D_\tau}\f{\log(2+\tau)}{(1+\tau)^{1-2\de_0}}\f{\varpi(r-\tau)}{(1+|r-\tau|)^{\gamma+2\de_0}}|\Gamma^J h|(\tau,x)\\
\ls &\sum_{|J|\leq |I|}\f{\log(2+\tau)}{(1+\tau)^{1+\gamma}}\sup_{x\in D_\tau}(1+\tau+r)\f{\varpi(r-\tau)}{(1+|r-\tau|)}|\Gamma^J h|(\tau,x)\\
\ls &\frac{\log(2+\tau)}{(1+\tau)^{1+\gmm}}\eps + \f{\log(2+\tau)}{(1+\tau)^{1+\gamma}} \pi_{|I|}(\tau).
\end{split}
\end{equation*}
Integrating thus gives the desired estimate.
\end{proof}
The bad term $\mathfrak B_I$ can also be controlled easily. Let us emphasize again that this term is only present when $\tBox_g \Gamma^I h$ is projected to $\Lb\Lb$, as this structure will be important later.
\begin{proposition}\label{decay.B.est}
The following estimate holds for $\mathfrak B_I$ for $|I|\leq \lfloor \f N2\rfloor +1$:
\begin{equation*}
\begin{split}
\int_0^t (1+\tau)\|\varpi(q)(\mathfrak B_I)(\tau,\cdot)\|_{L^\infty(D_\tau)}d\tau
\ls &\int_0^t \f{\sigma_{|I|}(\tau)}{1+\tau} \,d\tau.
\end{split}
\end{equation*}
\end{proposition}
\begin{proof}
This follows directly from the assumptions on $|\rd h_B|_{\mathcal T\mathcal U}$ and $|\rd \phi_B|$ in Definition \ref{def.dispersivespt} and bootstrap assumptions \eqref{BA3} and \eqref{BASF3} for $|\rd h|_{\mathcal T\mathcal U}$ and $|\rd \beta|$. Notice that the estimates indeed depend only on $\sigma_{|I|}$ but not $\pi_{|I|}$.
\end{proof}
Before we proceed to the nonlinear term $\mathfrak N_I$ and the the good term $\mathfrak G_I$. We need to estimate $|\Gamma^J h|_{LL}$ using the generalized wave coordinate condition. More precisely, we have the following lemma: 
\begin{lemma}\label{lemma.Gamma.hLL}
The following estimate for $|\Gamma^J h|_{LL}$ holds for $|I|\leq \lfloor \f N2\rfloor +1$:
$$\sup_x(1+t)(1+|q|)^{-1}\varpi(q)|\Gamma^I h|_{LL}(t,x)\ls \f{\ep}{(1+t)^{\f{\gamma}{8}}}+\pi_{|I|-2}(t).$$
\end{lemma}
\begin{proof}
By Proposition \ref{wave.con.higher},
\begin{equation*}
\begin{split}
&|\rd\Gamma^I h|_{LL}\\
\ls &\f{\ep\log(2+s)}{(1+s)^2 w(q)^{\f{\gamma}{1+2\gamma}}}+\f{\log(2+s)}{(1+s)(1+|q|)^\gamma}\sum_{|J|\leq |I|}|\Gamma^J h|\\
&+\f{\log(2+s)}{1+s}\sum_{|J|\leq |I|}|\rd \Gamma^J h|+\sum_{|J_1|+|J_2|\leq |I|}|\Gamma^{J_1}h||\rd \Gamma^{J_2}h|+\sum_{|J|\leq |I|}|\db\Gamma^J h|+\sum_{|J|\leq |I|-2}|\rd \Gamma^J h|.
\end{split}
\end{equation*}
Using the estimates from Proposition \ref{KS.cons}, we have
$$\sum_{|J|\leq |I|}|\bar\rd \Gamma^J h|\ls \f{\ep(1+|q|)^{\f12}}{(1+s)^{2-(2N+6) \de_0} w(q)^{\f12}},$$ 
$$\f{\log (2+s)}{1+s}\sum_{|J|\leq |I|}(|\Gamma^J h|+|\rd \Gamma^J h|)\ls \f{\ep (1+|q|)^{\f12}\log (2+s)}{(1+s)^{2-(2N+6) \de_0} w(q)^{\f12}}.$$
and
$$\sum_{|J_1|+|J_2|\leq |I|}|\Gamma^{J_1}h||\rd \Gamma^{J_2}h|\ls \f{\ep^2}{(1+s)^{2-2(2N+6) \de_0} w(q)}.$$
By definition of $\pi_{|I|-2}$,
$$\sum_{|J|\leq |I|-2}|\rd\Gamma^J h|\ls \f{\pi_{|I|-2}(t)}{(1+s)\varpi(q)}.$$
Therefore, combining the estimates above, we obtain
$$|\rd \Gamma^I h|_{LL}(t,x)\ls \f{\ep (1+|q|)^{\f12}\log (2+s)}{(1+s)^{2-2(2N+6) \de_0} w(q)^{\f12}}+\f{\pi_{|I|-2}(t)}{(1+s)\varpi(q)},$$
which implies
$$(1+s)\varpi(q)|\rd \Gamma^I h|_{LL}(t,x)\ls \f{\ep}{(1+s)^{\f{\gamma}{8}}}+\pi_{|I|-2}(t),$$
by \eqref{varpi.def} since $\de_0\leq \f{\gamma}{16(2N+6)}$ by \eqref{de_0.def}. 
The conclusion then follows from Lemma \ref{int.lemma} after noting that $\rd_q$ commutes with the projection to $L$.
\end{proof}

We now turn to the nonlinear terms $\mathfrak N_I$ and its ${ }^{(\phi)}$-counterpart. 
\begin{proposition}\label{decay.N.est}
The following estimate holds for $\mathfrak N_I+\mathfrak N_I^{(\phi)}$ for $|I|\leq \lfloor \f N2\rfloor +1$:
\begin{equation*}
\begin{split}
&\int_0^t (1+\tau)\|\varpi(r-\tau)(\mathfrak N_I+\mathfrak N_I^{(\phi)})(\tau,\cdot)\|_{L^\infty(D_\tau)}d\tau\\
\ls &\int_0^t \f{\ep \pi_{|I|}(\tau)}{(1+\tau)^{1+\f{\gamma}{8}}}\,d\tau+\int_0^t \f{\pi_{|I|-1}(\tau)\pi_{|I|-1}(\tau)}{1+\tau} d\tau.
\end{split}
\end{equation*}
\end{proposition}
\begin{proof}
We will bound $\mathfrak N_I$ as $\mathfrak N_I^{(\phi)}$ can be controlled in a completely identical manner. We have three contributions
\begin{equation}\label{decay.N.1}
\sum_{\substack{|J_1|+|J_2|\leq |I|\\ \max\{|J_1|,|J_2|\}\leq |I|-1}}|\rd\Gamma^{J_1}h||\rd\Gamma^{J_2}h|,
\end{equation}
\begin{equation}\label{decay.N.2}
\sum_{\substack{|J_1|+|J_2|\leq |I|\\ \max\{|J_1|,|J_2|\}\leq |I|-1}}\frac{|\Gamma^{J_1}h||\rd\Gamma^{J_2}h|}{1+|q|},
\end{equation}
and
\begin{equation}\label{decay.N.4}
\sum_{|J_1|+|J_2|\leq |I|}\f{|\Gamma^{J_1} h|_{LL}|\rd\Gamma^{J_2} h|}{1+|q|}.
\end{equation}
We first consider \eqref{decay.N.1}.
\begin{equation}\label{decay.N.pf.1.2}
\begin{split}
&\sum_{\substack{|J_1|+|J_2|\leq |I|\\ \max\{|J_1|,|J_2|\}\leq |I|-1}}\int_0^t (1+\tau)\|\varpi(r-\tau)|\rd\Gamma^{J_1}h||\rd\Gamma^{J_2}h|(\tau,\cdot)\|_{L^\infty(D_\tau)}d\tau\\
\ls &\sum_{\substack{k_1+k_2\leq |I|\\\max\{k_1,k_2\}\leq |I|-1}}\int_0^t \f{\pi_{k_1}(\tau)\pi_{k_2}(\tau)}{1+\tau}d\tau.
\end{split}
\end{equation}
In view of Lemma \ref{int.lemma}, the term \eqref{decay.N.2} can be estimated in an identical manner as \eqref{decay.N.1}.
For the final term \eqref{decay.N.4}, we use Lemma \ref{lemma.Gamma.hLL} to bound $\f{(1+\tau)\varpi(r-\tau)|\Gamma^{J_1} h|_{LL}}{1+|q|}$ and obtain
\begin{equation}\label{decay.N.pf.4.2}
\begin{split}
&\sum_{|J_1|+|J_2|\leq |I|}\int_0^t (1+\tau)\|\varpi(r-\tau)\f{|\Gamma^{J_1} h|_{LL}|\rd\Gamma^{J_2} h|}{1+|r-\tau|}(\tau,\cdot)\|_{L^\infty(D_\tau)}d\tau\\
\ls & \int_0^t \f{\ep \pi_{|I|}(\tau)}{(1+\tau)^{1+\f{\gamma}{8}}}\,d\tau+\sum_{k=0}^{|I|-2}\int_0^t \f{\pi_k(\tau)\pi_{|I|-k-2}(\tau)}{1+\tau} d\tau.
\end{split}
\end{equation}
Finally, notice that the terms \eqref{decay.N.pf.1.2} and \eqref{decay.N.pf.4.2} are both acceptable. This concludes the proof of the proposition.
\end{proof}

Finally, we control the good term $\mathfrak G_I$ and its ${ }^{(\phi)}$-counterpart:
\begin{proposition}\label{decay.G.est}
The following estimate holds for $\mathfrak G_I+\mathfrak G_I^{(\phi)}$ for $|I|\leq \lfloor \f N2\rfloor +1$:
\begin{equation*}
\begin{split}
&\int_0^t (1+\tau)\|\varpi(r-\tau)(\mathfrak G_I+\mathfrak G_I^{(\phi)})(\tau,\cdot)\|_{L^\infty(D_\tau)}d\tau\ls \ep +\int_0^t \f{\pi_{|I|-2}(\tau)\, d\tau}{1+\tau}.
\end{split}
\end{equation*}
\end{proposition}
\begin{proof}
We will only need to bound $\mathfrak G_I$, as $\mathfrak G_I^{(\phi)}$ contains a strict subset of terms. By Proposition \ref{schematic.eqn}, we have
\begin{equation}\label{decay.G.est.0}
\begin{split}
&\int_0^t (1+\tau)\|\varpi(r-\tau)(\mathfrak G_I)(\tau,\cdot)\|_{L^\infty(D_\tau)}d\tau\\
\ls &\int_0^t \sum_{|J|\leq |I|}\sup_x\f{\varpi(r-\tau)|\Gamma^J h|_{LL}}{(1+|r-\tau|)^{1+\gamma}}(\tau,x)d\tau\\
&+\int_0^t(1+\tau)^{\de_0}\sum_{|J|\leq |I|}\sup_x \f{\varpi(r-\tau)(|\bar\rd\Gamma^J h|+|\bar\rd\Gamma^J \beta|)(\tau,x)}{(1+|r-\tau|)^{\gamma+\de_0}} d\tau. 
\end{split}
\end{equation}
To control the first term, we apply Lemma \ref{lemma.Gamma.hLL}, which gives
\begin{equation*}
\begin{split}
&\int_0^t \sum_{|J|\leq |I|}\sup_x\f{\varpi(r-\tau)|\Gamma^J h|_{LL}}{(1+|r-\tau|)^{1+\gamma}}(\tau,x)d\tau\\
\ls &\int_0^t \big(\f{\ep}{(1+\tau)^{1+\f{\gamma}{8}}}+\f{\ep \pi_{|I|-2}(\tau)}{1+\tau}\big)d\tau\ls \ep+\int_0^t \f{\ep \pi_{|I|-2}(\tau)}{1+\tau}\,d\tau.
\end{split}
\end{equation*}
For the second term in \eqref{decay.G.est.0}, we can use \eqref{KS.cons.3} in Proposition \ref{KS.cons} to get
\begin{equation*}
\begin{split}
&\int_0^t(1+\tau)^{\de_0}\sum_{|J|\leq |I|}\sup_x \f{\varpi(r-\tau)(|\bar\rd\Gamma^J h|+|\bar\rd\Gamma^J \beta|)(\tau,x)}{(1+|r-\tau|)^{\gamma+\de_0}} d\tau\\
\ls & \int_0^t \sup_x\f{\ep(1+|r-\tau|)^{1-\f{\gamma}{4}}}{(1+\tau+r)^{2-(2N+6)\de_0}(1-|r-\tau|)^{\gamma+\de_0}}d\tau\ls \int_0^t\ep(1+\tau)^{-1-\f{3\gamma}{4}+(2N+6)\de_0}d\tau\ls \ep,
\end{split}
\end{equation*}
since $\de_0$ satisfy \eqref{de_0.def}. Combining these estimates, we get the desired conclusion.
\end{proof}

We have now estimated each of the error terms in $\tBox_g (\Gamma^I h)$ and $\tBox_g (\Gamma^I \beta)$. We are now ready to apply Proposition \ref{pointwise.main.lemma} to obtain the desired pointwise bounds. We start with the lowest order estimates:
\begin{proposition}\label{lowest.order.decay}
The following estimates hold:
$$\sigma_0(t)\ls \ep,\quad \pi_0(t)\ls \ep\log(2+t).$$
\end{proposition}
\begin{proof}
By definition, the lower order term $\mathfrak{L}_{I}$ is missing in $\tBox_{g} h$ for $|I|=0$. Moreover, recall from Proposition \ref{schematic.eqn.improved} that the bad term $\mathfrak B_I$ is absent when the inhomogeneous term is projected to $\mathcal T\mathcal U$. Therefore, using Proposition \ref{pointwise.main.lemma}, we can combine the bounds in Propositions \ref{decay.I.est}, \ref{decay.T.est}, \ref{decay.V.est}, \ref{decay.N.est} and \ref{decay.G.est} to obtain
\begin{equation}\label{LO.est.1}
\sigma_0(t)\ls \ep+\int_0^t \f{\pi_0(\tau)}{(1+\tau)^{1+\f{\gamma}{8}}} \,d\tau.
\end{equation}
On the other hand, for a general component, we also have the contribution from the $\mathfrak B_I$ term in Proposition \ref{decay.B.est}. Therefore, we have
\begin{equation*}
\begin{split}
\pi_0(t)\ls &\ep+\int_0^t \f{\sigma_0(\tau)}{(1+\tau)} \,d\tau+\int_0^t \f{\pi_0(\tau)}{(1+\tau)^{1+\f{\gamma}{8}}} \,d\tau.
\end{split}
\end{equation*}
Since $\f{1}{(1+\tau)^{1+\f{\gamma}{2}}}$ in integrable in $\tau$, it then follows from Gr\"onwall's inequality that 
\begin{equation}\label{LO.est.2}
\pi_0(t)\ls \ep+\int_0^t \f{\sigma_0(\tau)}{(1+\tau)} \,d\tau.
\end{equation}
A simple continuity argument shows that \eqref{LO.est.1} and \eqref{LO.est.2} together imply the desired conclusion.
\end{proof}

Using the estimates we have obtained, we can show by induction the following pointwise bounds up to $\lfloor \f N2\rfloor+1$ derivatives of $h$ and $\beta$:

\begin{proposition}\label{imp.decay}
Let $\de>0$ be sufficiently small. For $1\leq k\leq \lfloor \f N2\rfloor+1$, the following holds with an implicit constant depending on $\de$ (in addition to $C$, $\gamma$ and $\de_0$, but independent of $\ep$):
$$\pi_k(t)+\sigma_k(t)\ls \ep(1+\tau)^{2^k\de}.$$
\end{proposition}
\begin{proof}
By Proposition \ref{pointwise.main.lemma} together with the estimates in Propositions \ref{decay.I.est},  \ref{decay.T.est}, \ref{decay.L.est}, \ref{decay.V.est}, \ref{decay.B.est}, \ref{decay.N.est} and \ref{decay.G.est}, we have
\begin{equation}\label{pi.est}
\begin{split}
\pi_k(t)\ls &\sum_{|I|\leq k}\int_0^t (1+\tau)\|\varpi(q)(|\tBox_g (\Gamma^I h)|+|\tBox_g (\Gamma^I \beta)|)(\tau,\cdot)\|_{L^\infty(D_\tau)}d\tau\\
\ls &\ep\log^3(2+t)+\int_0^t \f{ \pi_{k}(\tau)\, d\tau}{(1+\tau)^{1+\f{\gamma}{8}}}+\int_0^t \f{ \log(2+\tau)\pi_{k-1}(\tau)\, d\tau}{1+\tau} \\
&+\int_0^t \f{ \pi_{k-1}(\tau)\pi_{k-1}(\tau)\, d\tau}{1+\tau} +\int_0^t \f{\sigma_{k}(\tau)}{1+\tau} \,d\tau.
\end{split}
\end{equation}
On the other hand, recall that the bad term is absent in $(\tBox_g (\Gamma^I h))_{\mathcal T\mathcal U}$, Therefore, by Propositions \ref{pointwise.main.lemma}, \ref{decay.I.est}, \ref{decay.T.est}, \ref{decay.L.est}, \ref{decay.V.est}, \ref{decay.N.est} and \ref{decay.G.est}, we have
\begin{equation}\label{sigma.est}
\begin{split}
\sigma_k(t)\ls &\sum_{|I|\leq k}\int_0^t (1+\tau)\|\varpi(q)(|\tBox_g (\Gamma^I h)|_{\mathcal T\mathcal U}+|\tBox_g (\Gamma^I \beta)|)(\tau,\cdot)\|_{L^\infty(D_\tau)}d\tau\\
\ls &\ep\log^3(2+t)+\int_0^t \f{ \pi_{k}(\tau)\, d\tau}{(1+\tau)^{1+\f{\gamma}{8}}}+\int_0^t \f{ \log(2+\tau)\pi_{k-1}(\tau)\, d\tau}{1+\tau} +\int_0^t \f{ \pi_{k-1}(\tau)\pi_{k-1}(\tau)\, d\tau}{1+\tau}.
\end{split}
\end{equation}
We claim that our desired estimates follow from \eqref{pi.est} and \eqref{sigma.est}. We prove this by induction in $k$. For the $k=1$ case, using the estimates in Proposition \ref{lowest.order.decay} and Gr\"onwall's inequality, \eqref{pi.est} and \eqref{sigma.est} reduce to
\begin{equation}\label{pi.1}
\begin{split}
\pi_1(t)\ls &\ep\log^3(2+t) +\int_0^t \f{\sigma_{1}(\tau)}{1+\tau} \,d\tau.
\end{split}
\end{equation}
and
\begin{equation}\label{sigma.1}
\sigma_1(t)\ls \ep\log^3(2+t)+\int_0^t \f{ \pi_{1}(\tau)\, d\tau}{(1+\tau)^{1+\f{\gamma}{8}}}.
\end{equation}
Substituting \eqref{pi.1} into \eqref{sigma.1} and using the monotonicity of $\sigma_1$, we get
\begin{equation*}
\begin{split}
\sigma_1(t)\ls &\ep\log^3(2+t)+\int_0^t \f{ \left(\ep\log^3(2+\tau)+\int_0^\tau\f{\sigma_{1}(\tau')\,d\tau'}{1+\tau'}\right) \, d\tau}{(1+\tau)^{1+\f{\gamma}{8}}}\\
\ls &\ep\log^3(2+t)+\int_0^t \f{\log(2+\tau)\sigma_{1}(\tau) \, d\tau}{(1+\tau)^{1+\f{\gamma}{8}}}\ls \ep\log^3(2+t),
\end{split}
\end{equation*}
where in the last step we have used Gr\"onwall's inequality. Plugging this estimate back to \eqref{pi.1}, we also obtain
$$\pi_1(t)\ls \ep\log^4(2+t).$$
We have thus proved the desired (and in fact much stronger) estimates for $\pi_1$ and $\sigma_1$. Now assume we have the desired estimate for $\pi_{k-1}$ and $\sigma_{k-1}$. Then, \eqref{pi.est} and \eqref{sigma.est} reduce to
\begin{equation}\label{pi.k}
\begin{split}
\pi_k(t)\ls &\ep(1+t)^{2^k\de} +\int_0^t \f{\sigma_{k}(\tau)}{1+\tau} \,d\tau.
\end{split}
\end{equation}
and
\begin{equation}\label{sigma.k}
\sigma_k(t)\ls \ep(1+t)^{2^k\de}+\int_0^t \f{ \pi_{k}(\tau)\, d\tau}{(1+\tau)^{1+\f{\gamma}{8}}}.
\end{equation}
As long as $2^N \de\ll \f{\gamma}{8}$, we argue as before to substitute \eqref{pi.k} into \eqref{sigma.k} to obtain
$$\sigma_k(t) \ls \ep(1+t)^{2^k\de},$$
which then implies 
$$\pi_k(t) \ls \ep(1+t)^{2^k\de}$$
after plugging the estimate for $\sigma_k(t)$ into \eqref{pi.k}.
\end{proof}
Finally, we improve all the bootstrap assumptions: 
\begin{proposition}\label{BS.close}
For $\de>0$ sufficiently small such that $2^{\lfloor \f N2\rfloor +1}\de\ll \de_0$, the following pointwise bounds hold:
\begin{equation}\label{BA1.re}
\sup_{t,x}\sum_{|I|\leq \lfloor\frac{N}{2}\rfloor+1} (1+s)^{1-\delta_0}(1+|q|)^{\frac 12-\f\gamma 4}w(q)^{\frac 12}|\rd \Gamma^I h (t,x)|\ls \ep,
\end{equation}
\begin{equation}\label{BA2.re}
\sup_{t,x}\sum_{|I|\leq \lfloor\frac{N}{2}\rfloor} (1+s)^{2-\delta_0}(1+|q|)^{-\frac 12-\f\gamma 4}w(q)^{\frac 12}|\bar\rd \Gamma^I h (t,x)|\ls \ep,
\end{equation}
\begin{equation}\label{BA3.re}
\sup_{t,x}(1+s)|\rd h(t,x)|_{\mathcal T\mathcal U}\ls \ep,
\end{equation}
\begin{equation}\label{BA4.re}
\sup_{t,x}\sum_{|I|\leq \lfloor\frac{N}{2}\rfloor+1} (1+s)^{1-\delta_0}(1+|q|)^{-\frac 12-\f\gamma 4}w(q)^{\frac 12}|\Gamma^I h (t,x)| \ls \ep,
\end{equation}
\begin{equation}\label{BA5.re}
\sup_{t,x}(1+s)^{1+\frac{\gamma}{2}}(1+|q|)^{-\frac 12-\gamma}w(q)^{\frac 12}\big(|h (t,x)|_{L\mathcal T}+\sum_{|I|\leq 1}|\Gamma^I h(t,x)|_{LL}\big)\ls \ep,
\end{equation}
\begin{equation}\label{BASF1.re}
\sup_{t,x}\sum_{|I|\leq \lfloor\frac{N}{2}+1\rfloor} (1+s)^{1-\delta_0}(1+|q|)^{\frac 12-\f\gamma 4}w(q)^{\frac 12}|\rd \Gamma^I \beta (t,x)|\ls \ep,
\end{equation}
\begin{equation}\label{BASF2.re}
\sup_{t,x}\sum_{|I|\leq \lfloor\frac{N}{2}\rfloor} (1+s)^{2-\delta_0}(1+|q|)^{-\frac 12-\f\gamma 4}w(q)^{\frac 12}|\bar\rd \Gamma^I \beta (t,x)|\ls \ep
\end{equation}
\begin{equation}\label{BASF3.re}
\sup_{t,x}(1+s)|\rd\beta(t,x)|\ls \ep.
\end{equation}
In particular, there exists $\ep_0\in (0,\ep_4]$ sufficiently small such that if $\ep\in (0,\ep_0]$, we have improved the bootstrap assumptions \eqref{BA1}-\eqref{BA5}, \eqref{BASF1}-\eqref{BASF3}.
\end{proposition}
\begin{proof}
\eqref{BA1.re} follows directly from the estimate of $\pi_k$ in Propositions \ref{lowest.order.decay} and \ref{imp.decay}. \eqref{BA3.re} follows from the bound for $\sigma_0$ in Proposition \ref{lowest.order.decay}. \eqref{BA4.re} follows from combining \eqref{BA1.re} and Lemma \ref{int.lemma}. \eqref{BA2.re} then follows from \eqref{BA4.re} and Proposition \ref{decay.weights}. Finally it remains to prove \eqref{BA5.re}. This requires the use of the generalized wave coordinate condition. More precisely, by Proposition \ref{wave.con.lower} and the bounds \eqref{BA1.re}, \eqref{BA2.re} and \eqref{BA4.re}, we obtain the bound
$$|\rd h(t,x)|_{L\mathcal T}\ls \f{\ep\log(2+s)(1+|q|)^{\f 12+\f{\gamma}{4}}}{w(q)^{\f 12}(1+s)^{2-2\de_0}}\ls \f{\ep(1+|q|)^{\f{7\gamma}{8}}}{w(q)^{\f12}(1+|q|)^{\f12}(1+s)^{1+\f{5\gamma}{8}-2\de_0}}.$$
Similarly, by Proposition \ref{wave.con.lower} and \eqref{BA1.re}, \eqref{BA2.re} and \eqref{BA4.re}, we have 
$$\sum_{|I|\leq 1}|\rd \Gamma^I h(t,x)|_{LL}\ls \f{\ep(1+|q|)^{\f{7\gamma}{8}}}{w(q)^{\f12}(1+|q|)^{\f12}(1+s)^{1+\f{5\gamma}{8}-2\de_0}}.$$
Notice that $\rd_q$ commutes with the projection to $\{L,\Lb,E^1,E^2,E^3\}$. By \eqref{de_0.def}, $\f{5\gamma}{8}-2\de_0\geq \f{\gamma}{2}$. By Lemma \ref{int.lemma} with $k(s)=(1+s)^{1+\f{5\gamma}{8}-2\de_0}$ and $f(q)=w(q)^{\f12}(1+|q|)^{\f12-\f{7\gamma}{8}}$, we therefore obtain \eqref{BA5.re}. 

We now turn to the bounds for the scalar field: \eqref{BASF1.re} and \eqref{BASF3.re} follow from Propositions \ref{lowest.order.decay} and \ref{imp.decay}. \eqref{BASF2.re} can be obtained in a similar manner as \eqref{BA2.re}, i.e.~first use \eqref{BASF1.re} and Lemma \ref{int.lemma} to obtain an estimate for $\sum_{|I|\leq \lfloor \f N 2\rfloor +1}|\Gamma^I\beta|$ and then apply Proposition \ref{decay.weights}.

This concludes the proof of the proposition.
\end{proof}
These estimates easily allow us to conclude the proof of Theorem \ref{main.thm}:
\begin{proof}[Proof of Theorem \ref{main.thm}]
For $\ep\in (0,\ep_0]$, where $\ep_0>0$ is as in Proposition \ref{BS.close}, we have closed all the bootstrap assumptions. It is therefore standard to conclude that all the estimates that are proven indeed hold for $g$ and $\phi$ satisfying the equations in Proposition \ref{Einstein.eqn.g}. In particular, \eqref{energy.bound.combined} holds. Standard results on local existence of solutions then imply that the solution is global in ($t$-)time.

Finally, the estimate \eqref{main.thm.bound} follows from \eqref{energy.bound.combined} if $(2N+6)\de_0\leq \de_1$. For every $\de_1>0$, we can therefore choose $(2N+6)\de_0\leq \de_1$ so that \eqref{main.thm.bound} holds for $\ep$ appropriately small (depending in particular on $\de_1$). This concludes the proof of Theorem \ref{main.thm}.
\end{proof}

\bibliographystyle{amsplain}


\end{document}